%% file: lord_of_entropies_rewritten.tex
\pdfoutput=1
\documentclass[11pt]{article}

\usepackage[margin=1in]{geometry}
\usepackage[T1]{fontenc}
\usepackage[utf8]{inputenc}
\usepackage[english]{babel}
\usepackage{amsmath,amssymb,amsthm,mathtools}
\usepackage{bbm}
\usepackage{graphicx}
\usepackage{tikz}
\usepackage{placeins}
\usepackage{xcolor}
\usepackage{tcolorbox}
\tcbuselibrary{breakable,skins}
\usepackage{needspace}
\usepackage{enumitem}
\usepackage{physics}
\usepackage{aliascnt}
\usepackage[colorlinks=true,
  linkcolor=blue!60!black,
  citecolor=red!55!black,
  urlcolor=blue!60!black,
  bookmarks=true,
  hypertexnames=false]{hyperref}
\usepackage[capitalise,nameinlink]{cleveref}

\allowdisplaybreaks
\newtheorem{definition}{Definition}[section]
\newaliascnt{theorem}{definition}
\newtheorem{theorem}[theorem]{Theorem}
\aliascntresetthe{theorem}
\newaliascnt{lemma}{definition}
\newtheorem{lemma}[lemma]{Lemma}
\aliascntresetthe{lemma}
\newaliascnt{proposition}{definition}
\newtheorem{proposition}[proposition]{Proposition}
\aliascntresetthe{proposition}
\newaliascnt{corollary}{definition}
\newtheorem{corollary}[corollary]{Corollary}
\aliascntresetthe{corollary}
\theoremstyle{remark}
\newaliascnt{remark}{definition}
\newtheorem{remark}[remark]{Remark}
\aliascntresetthe{remark}
\newaliascnt{example}{definition}
\newtheorem{example}[example]{Example}
\aliascntresetthe{example}

\definecolor{mainobjectbackground}{RGB}{246,249,253}
\definecolor{mainobjectframe}{RGB}{94,119,145}
\definecolor{mainobjecttitle}{RGB}{231,239,247}
\newtcolorbox{mainobject}[1]{
  enhanced,
  breakable,
  colback=mainobjectbackground,
  colframe=mainobjectframe,
  colbacktitle=mainobjecttitle,
  coltitle=black,
  fonttitle=\bfseries,
  title={#1},
  boxrule=0.6pt,
  arc=1.5mm,
  outer arc=1.5mm,
  left=2.5mm,
  right=2.5mm,
  top=1.5mm,
  bottom=1.5mm,
  toptitle=1mm,
  bottomtitle=1mm,
  before skip=8pt plus 2pt minus 1pt,
  after skip=8pt plus 2pt minus 1pt
}

\DeclareMathOperator{\supp}{supp}
\DeclareMathOperator{\id}{id}
\DeclareMathOperator{\os}{os}
\DeclareMathOperator{\ts}{ts}
\DeclareMathOperator{\rev}{rev}
\newcommand{\1}{\mathbbm{1}}

\DeclareMathOperator{\R}{\mathbb{R}}

\DeclareMathOperator{\cH}{\mathcal{H}}

\DeclareMathOperator{\cK}{\mathcal{K}}

\title{Geodesic Quantum $f$-Divergences}

\author{%
\'Angela Capel\\
{\small Department of Applied Mathematics and Theoretical Physics,}\\
{\small University of Cambridge, Wilberforce Road,}\\
{\small Cambridge CB3 0WA, United Kingdom, and}\\
{\small Fachbereich Mathematik, Universit\"at T\"ubingen,}\\
{\small 72076 T\"ubingen, Germany}\\
{\small \href{mailto:ac2722@cam.ac.uk}{\texttt{ac2722@cam.ac.uk}}}\\[1.1em]
Pablo Costa Rico\\
{\small Institute for Quantum Information,}\\
{\small RWTH Aachen University,
}\\
{\small 52074 Aachen, Germany}\\
{\small \href{mailto:pablo.costa@rwth-aachen.de}{\texttt{pablo.costa@rwth-aachen.de}}}\\[1.1em]
}

\date{}

\begin{document}
\maketitle

\begin{abstract}
We introduce the geodesic quantum $f$-divergences $D_f^t$,
$0\leq t\leq1$, obtained from the affine-invariant geodesic between
the standard and maximal relative modular operators.  They reduce to
the classical $f$-divergence for commuting states.  The logarithmic
generator yields geodesic relative entropies joining the Umegaki and
Belavkin--Staszewski entropies, while the power generators yield
$(t,\alpha)$-R\'enyi divergences joining the Petz and geometric
families.

Our first main result is data processing of $D_f^t$ for every finite
operator-convex generator $f$ and every $t\in[0,1]$.  In particular, this gives DPI for the geodesic relative entropies and for the $(t,\alpha)$-R\'enyi divergences when
$0<\alpha<1$ or $1<\alpha\leq2$. Our second main result identifies equality in the DPI: for invertible states, every equality-determining operator-convex generator
has, at each nonmaximal parameter $0\leq t<1$, exactly the Petz
sufficiency class; at $t=1$, this changes to the generally larger
maximal, or BS, class, which also coincides with the equality class of the divergences associated with the quadratic generator
for every $t$. Our third main result is the corresponding collapse of invertible
geodesic quantum Markov chains: for each of the three ordered
conditional-mutual-information constructions, the $t$-quantum Markov
chains are precisely the quantum Markov chains for $0\leq t<1$,
whereas at $t=1$ they are the invertible BS quantum Markov chains, a
class that can be strictly larger.

We also determine parameter-monotonicity regimes and the universal
intersections with the $(\alpha,z)$ family.  For $C^1$ generators, we
give an exact ordinary-matrix characterization of the universal
$t$-monotonicity cone; on its finite operator-convex slice, we obtain
an equivalent canonical-measure polar characterization.  We prove strengthened data-processing and
reconstruction estimates; we compare the three conditional
orientations; we establish continuity bounds under positive
lower-eigenvalue assumptions together with complementary
discontinuity results; and we give a capacity-per-unit-cost
interpretation.  The Hessian-metric,
likelihood-statistical and restricted hypothesis-testing, and
horospherical-geometric developments are treated in a companion paper.
\end{abstract}

\newpage
\tableofcontents

\input{rewritten_front_sections}
\input{dpi_and_recovery_sections}
\input{continuity_conditional_sections}

\input{operational_interpretations_section}
\input{discussion_outlook_section}

\section*{Acknowledgments}
The authors are grateful to Andreas Bluhm for fruitful conversations related to the topic of this paper, to Michael Wolf for posing the question that gave rise to this project, and to Matteo Scandi for suggesting the connection between QMCs and symmetry with respect to power-weighted inner products.
 Á.C. acknowledges support from the Deutsche Forschungsgemeinschaft
(DFG, German Research Foundation), Project-ID 470903074, TRR 352.
This project was funded within the QuantERA II Programme, which
received funding from the European Union's Horizon 2020 research and
innovation programme under Grant Agreement No.~101017733.

\section*{AI statement}

The project began 1,5 years ago and many of the results contained in this paper, such as the DPI for the $(t,f)$-divergences, had already been obtained before the use of AI. However, the authors wanted to show a hierarchy for the $t$-QMCs, and even conjectured that for $t<1$ they would collapse to QMCs (for their connections to weighted inner products), but failed to prove these results for a long time. ChatGPT 5.6 Sol finally helped obtain the solution to this and finalise some other results in the paper. Hence, the authors used ChatGPT 5.6 Sol to assist with proofs, numerical computations, figure preparation, and manuscript drafting; the scientific ideas, research questions, and conceptual direction underlying the project originated with the authors, who take full responsibility for the final content, and have reviewed the proofs provided by the LLM carefully.

\bibliographystyle{abbrv}
\bibliography{literatura}

\end{document}

%% file: rewritten_front_sections.tex
% Paper-ready front matter for the rewritten manuscript.
% This file intentionally contains no preamble.

\section{Introduction}
\label{rw:sec:introduction}

In classical information theory and statistics, the
Kullback--Leibler (KL) divergence
\begin{equation}
 D_{\rm KL}(p\Vert q):=\sum_xp(x)\log\frac{p(x)}{q(x)}
 \label{rw:eq:intro-KL}
\end{equation}
is the canonical measure of relative information.  Since its
introduction by Kullback and Leibler, it has become ubiquitous in
statistical inference, coding, information processing, and related
areas \cite{kullback1951information,cover2006elements}.

Passing from probability distributions to quantum states creates a
basic ambiguity.  The requirement that a quantum quantity reduces to
\eqref{rw:eq:intro-KL} when the two states commute does not determine
how the noncommuting operators should be ordered.  Consequently,
infinitely many inequivalent quantum extensions of the KL divergence
can be constructed.  More generally, the same ambiguity arises when
extending arbitrary classical divergences to noncommuting states.
Standard, measured, and maximal quantum divergences, together with
sandwiched and $\alpha$--$z$ R\'enyi divergences, are prominent
quantum extensions of their corresponding classical divergences
\cite{petz1986quasi,hiai2011quantum,hiai2017different,
matsumoto2018new,audenaert2015alphaz}.  In practice, the Umegaki
relative entropy is the customary extension of the KL divergence.
Its central role is supported by its data-processing, convexity,
strong-subadditivity, and sufficiency properties, and by its many
uses throughout quantum information theory
\cite{Umegaki-RelativeEntropy-1962,lieb1973strong,
Petz-SufficiencyChannels-1978,vedral2002role,petz2007quantum}.
More concretely, it generates the standard quantum mutual information
and the relative entropy of entanglement, while its value relative to
a Gibbs state equals the nonequilibrium free-energy difference up to
the inverse-temperature factor
\cite{vedral2002role,wilming2017axiomatic}.

There is, however, another distinguished extension.  The
Belavkin--Staszewski (BS) relative entropy was introduced in
\cite{BelavkinStaszewski-BSentropy-1982} and is the logarithmic
member of the maximal quantum $f$-divergences
\cite{hiai2017different}.  It and its geometric R\'enyi relatives
have recently attracted renewed attention through strengthened data
processing and recoverability \cite{bluhm2020strengthened},
quasi-factorization
\cite{BluhmCapelPerezHernandez-WeakQFBSentropy-2021}, continuity of
conditional quantities \cite{bluhm2023continuity}, quantum channel
estimation, discrimination, and capacities
\cite{fang2021geometric,katariya2021geometric}, conditional
independence of Gibbs states and efficient learning
\cite{alhambra2026conditional}, and quantum Markov structure
\cite{Bluhm:2025kai}.  The Umegaki and BS
entropies coincide on commuting states, but they are genuinely
different in the noncommutative regime.

Geometric and interpolation-based constructions have previously been
used to organize larger families of quantum distinguishability
measures.  The $\alpha$--$z$ R\'enyi divergences place, among other
examples, the Petz and sandwiched families within a common
two-parameter scheme \cite{audenaert2015alphaz}.  At a more general
modular level, Kubo--Ando operator means and noncommutative $L_p$
interpolation have been used to generate broad classes of monotone
multi-state quantum $f$-divergences \cite{furuya2023monotonic}.
Other geometry-driven generalizations instead use geodesics defined
by deformed exponential functions on the manifold of invertible states
\cite{andrade2020generalized}.
Most directly related to the two relative-entropy endpoints above,
the construction of \cite{mosonyi2024geometric} deforms a given
quantum relative entropy $D^q$ according to
$D^{q,\#_\gamma}(\rho\Vert\sigma)
=(1-\gamma)^{-1}D^q(\rho\Vert\sigma\#_\gamma\rho)$, where
$\#_\gamma$ is the weighted Kubo--Ando geometric mean.  Starting from
the Umegaki entropy, this gives a different monotone and additive
interpolation to the BS entropy on invertible states; the same work also
constructs further barycentric R\'enyi divergences.  More recently,
varying the base point in Bures--Wasserstein geometry has produced a
family unifying several quantum fidelities and associated R\'enyi
trace functionals, including the Petz, sandwiched, reverse-sandwiched,
and geometric cases \cite{afham2025riemannian}.

Against this background, we take a different geometric route: we
follow the affine-invariant geodesic directly between the standard and
maximal relative modular operators before applying functional
calculus.  The resulting family of quantum relative entropies always
reduces to the KL divergence on commuting states and joins Umegaki to
BS.  More generally, we introduce the \emph{geodesic
$f$-divergences}, also called the \emph{$(t,f)$-divergences}, which
interpolate between the standard and maximal quantum
$f$-divergences \cite{hiai2017different}.  For the logarithmic generator $f(x)=x\log x$ they
give the $t$-relative entropies $D^t$; for the power generators they
give the $(t,\alpha)$-R\'enyi divergences $D_\alpha^t$, interpolating
between the Petz and geometric R\'enyi divergences.  The two
independent operations---choosing the generator and choosing an
endpoint of the geodesic---are displayed in
\Cref{rw:fig:divergence-family-roadmap}.

The central question is whether the intermediate points remain
genuine quantum divergences.  We answer it by proving data processing
and then study their monotonicity, recovery theory, conditional
quantities, quantitative remainders, continuity, and a
capacity-per-unit-cost interpretation.  The companion paper constructs
the canonical finite likelihood experiment associated with $\Gamma_t$
and proves its exact classical representation of the geodesic quantum
$f$-divergences for operator-convex generators
\cite[Theorem~3.1]{capel2026informationgeometry}.  It also develops the
induced Hessian metrics, their universal RLD endpoint, and their
BKM--RLD interpolation
\cite[Theorem~2.3, Corollary~2.4, and Theorem~2.5]{capel2026informationgeometry},
identifies $t$ intrinsically as a boundary index and compares this
metric path with the Hasegawa--Petz interpolation
\cite[Theorem~2.6 and Proposition~2.7]{capel2026informationgeometry},
derives the statistical and restricted hypothesis-testing consequences
of the canonical experiment
\cite[Propositions~3.2--3.3]{capel2026informationgeometry}, and
develops its Busemann and horospherical geometry
\cite[Propositions~4.1--4.2 and Section~4.3]{capel2026informationgeometry}.

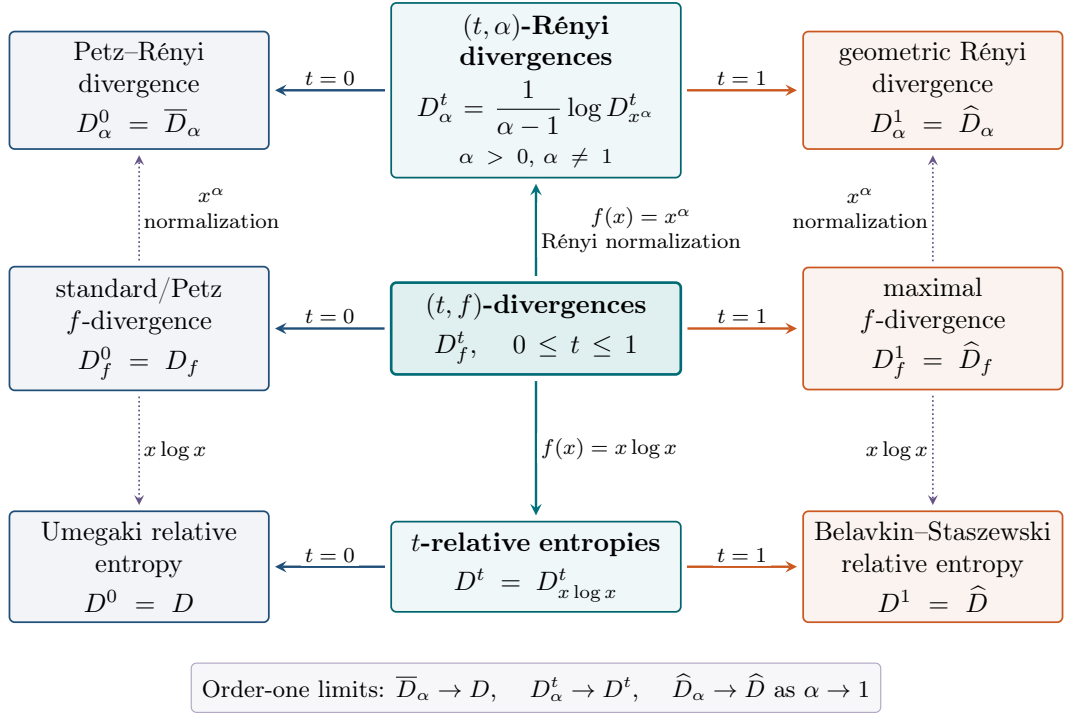
\begin{figure}[!t]
\centering
\definecolor{roadmapPetz}{RGB}{31,78,121}
\definecolor{roadmapMaximal}{RGB}{202,88,32}
\definecolor{roadmapInterior}{RGB}{0,126,137}
\definecolor{roadmapTransform}{RGB}{104,85,135}
\begin{tikzpicture}[
  x=1cm,
  y=1cm,
  >=stealth,
  font=\small,
  roadmap node/.style={
    rounded corners=2pt,
    minimum height=1.15cm,
    text width=3.15cm,
    align=center,
    inner sep=4pt,
    line width=0.65pt
  },
  central node/.style={
    roadmap node,
    text width=3.55cm,
    draw=roadmapInterior!85!black,
    fill=roadmapInterior!12,
    line width=1pt
  },
  family node/.style={
    roadmap node,
    text width=3.55cm,
    draw=roadmapInterior!80!black,
    fill=roadmapInterior!6
  },
  petz node/.style={
    roadmap node,
    draw=roadmapPetz,
    fill=roadmapPetz!6
  },
  maximal node/.style={
    roadmap node,
    draw=roadmapMaximal,
    fill=roadmapMaximal!7
  },
  petz arrow/.style={
    ->,
    draw=roadmapPetz,
    line width=0.9pt,
    shorten <=2pt,
    shorten >=2pt
  },
  maximal arrow/.style={
    ->,
    draw=roadmapMaximal,
    line width=0.9pt,
    shorten <=2pt,
    shorten >=2pt
  },
  generator arrow/.style={
    ->,
    draw=roadmapInterior!90!black,
    line width=0.9pt,
    shorten <=2pt,
    shorten >=2pt
  },
  commuting arrow/.style={
    ->,
    draw=roadmapTransform,
    densely dotted,
    line width=0.65pt,
    shorten <=2pt,
    shorten >=2pt
  },
  arrow label/.style={
    fill=white,
    inner sep=1.5pt,
    align=center,
    font=\scriptsize
  },
  limit note/.style={
    rounded corners=2pt,
    draw=roadmapTransform!55,
    fill=roadmapTransform!5,
    inner sep=4pt,
    font=\footnotesize,
    align=center
  }
]
  \node[central node] (tf) at (0,0)
    {\textbf{$(t,f)$-divergences}\\[1mm]
     $D_f^t$, \quad $0\leq t\leq1$};

  \node[petz node] (petzf) at (-5.25,0)
    {standard/Petz\\$f$-divergence\\[1mm]
     $D_f^0=D_f$};
  \node[maximal node] (maxf) at (5.25,0)
    {maximal\\$f$-divergence\\[1mm]
     $D_f^1=\widehat D_f$};

  \node[family node] (trenyi) at (0,3.15)
     {\textbf{$(t,\alpha)$-R\'enyi divergences}\\[1mm]
     $\displaystyle D_\alpha^t
       =\frac{1}{\alpha-1}\log D_{x^\alpha}^t$\\[1.2mm]
     {\scriptsize $\alpha>0$, $\alpha\neq1$}};
  \node[petz node] (petzrenyi) at (-5.25,3.15)
    {Petz--R\'enyi\\divergence\\[1mm]
     $D_\alpha^0=\overline D_\alpha$};
  \node[maximal node] (georenyi) at (5.25,3.15)
    {geometric R\'enyi\\divergence\\[1mm]
     $D_\alpha^1=\widehat D_\alpha$};

  \node[family node] (trel) at (0,-3.15)
    {\textbf{$t$-relative entropies}\\[1mm]
     $D^t=D_{x\log x}^t$};
  \node[petz node] (umegaki) at (-5.25,-3.15)
    {Umegaki relative\\entropy\\[1mm]
     $D^0=D$};
  \node[maximal node] (bs) at (5.25,-3.15)
    {Belavkin--Staszewski\\relative entropy\\[1mm]
     $D^1=\widehat D$};

  % Endpoint specializations of the three interpolating families.
  \draw[petz arrow] (tf) -- node[arrow label,above] {$t=0$} (petzf);
  \draw[maximal arrow] (tf) -- node[arrow label,above] {$t=1$} (maxf);
  \draw[petz arrow] (trenyi) -- node[arrow label,above] {$t=0$} (petzrenyi);
  \draw[maximal arrow] (trenyi) -- node[arrow label,above] {$t=1$} (georenyi);
  \draw[petz arrow] (trel) -- node[arrow label,above] {$t=0$} (umegaki);
  \draw[maximal arrow] (trel) -- node[arrow label,above] {$t=1$} (bs);

  % Generator specializations in the interior column.
  \draw[generator arrow] (tf) --
    node[arrow label,right]
      {$f(x)=x^\alpha$\\R\'enyi normalization}
    (trenyi);
  \draw[generator arrow] (tf) --
    node[arrow label,right] {$f(x)=x\log x$}
    (trel);

  % The dotted arrows show that endpoint and generator specialization commute.
  \draw[commuting arrow] (petzf) --
    node[arrow label,right] {$x^\alpha$\\normalization}
    (petzrenyi);
  \draw[commuting arrow] (petzf) --
    node[arrow label,right] {$x\log x$}
    (umegaki);
  \draw[commuting arrow] (maxf) --
    node[arrow label,left] {$x^\alpha$\\normalization}
    (georenyi);
  \draw[commuting arrow] (maxf) --
    node[arrow label,left] {$x\log x$}
    (bs);

  \node[limit note] at (0,-4.75)
    {Order-one limits:
     $\overline D_\alpha\to D$, \quad
     $D_\alpha^t\to D^t$, \quad
     $\widehat D_\alpha\to\widehat D$
     as $\alpha\to1$};
\end{tikzpicture}
\caption{Specialization diagram for the $(t,f)$-divergences.  Teal
arrows select a generator (with the R\'enyi normalization included on
the upper branch), while dark blue and orange arrows select the Petz
endpoint $t=0$ and the maximal endpoint $t=1$, respectively.  The
dotted vertical arrows show that the two operations commute.  Thus
$D^t$ interpolates between the Umegaki and BS relative entropies, and
$D_\alpha^t$ interpolates between the Petz and geometric R\'enyi
divergences.  At $\alpha=2$ the entire upper row coincides even for
noncommuting states; for commuting states every horizontal row
collapses to its common classical quantity.}
\label{rw:fig:divergence-family-roadmap}
\end{figure}

\paragraph{Standing convention.}
Except when behavior at non-invertible states is itself under
discussion, all results below are stated and proved for invertible
states.  Extensions to non-invertible states are result-dependent and
are stated locally, together with the required support convention.
In \Cref{rw:sec:interpolated-divergences,rw:sec:t-monotonicity}, every
statement concerning a pair with common support $P$ remains valid
after compression to $P\mathcal H$, with inverses, logarithms,
traces, and functional calculi taken in the compressed algebra.
Some inequalities and algebraic identities extend further by
lower-semicontinuous regularization, as recorded after the
corresponding results.  No extension is claimed for strict
equality characterizations, recovery and fixed-point results,
analyticity, real-$t$ extrapolations involving negative powers, or
estimates whose constants contain inverse spectral bounds.

\subsection{Main results}
\label{rw:subsec:informal-main-results}

We now state the main results in a form intended to make the new
quantities and the two endpoint phenomena visible at once.  Precise
support conventions and parameter ranges are given with the
corresponding results.  Throughout this overview, $\rho$ and $\sigma$
are invertible states.  We call $0\leq t<1$ the \emph{nonmaximal}
range and $0<t<1$ the \emph{strict interior}.

\medskip
\noindent\textbf{Definitions and accessible reformulations.}
We set
\begin{equation*}
 \Delta_{\rho\mid\sigma}=L_\rho R_{\sigma^{-1}},
 \qquad
 \widehat\Delta_{\rho\mid\sigma}
 =R_{\sigma^{-1/2}\rho\sigma^{-1/2}}.
\end{equation*}
The affine-invariant geodesic between these two positive
Hilbert--Schmidt operators simplifies to
\begin{equation*}
 \Gamma_t
 :=\Delta_{\rho\mid\sigma}\#_t
   \widehat\Delta_{\rho\mid\sigma}
 =L_{\rho^{1-t}}
  R_{\sigma^{-1/2}\rho^t\sigma^{-1/2}},
 \qquad 0\leq t\leq1.
\end{equation*}
Our basic definition for the $(t,f)$-divergences is
\begin{equation}
 \boxed{
 D_f^t(\rho\Vert\sigma)
 :=\big\langle\sigma^{1/2},
 f(\Gamma_t)[\sigma^{1/2}]\big\rangle_{\rm HS}.}
 \label{rw:eq:intro-Dft}
\end{equation}
Thus $D_f^0=D_f$ is the standard/Petz $f$-divergence,
$D_f^1=\widehat D_f$ is the maximal $f$-divergence, and every $t$
gives the classical $f$-divergence when $[\rho,\sigma]=0$, as recorded
in \Cref{rw:prop:gamma-formula,rw:def:interpolated-f-divergence}.

The two principal specializations, $t$-relative entropies and $(t,\alpha)$-Rényi divergences, respectively, have particularly usable formulas:
\begin{align}
 D^t(\rho\Vert\sigma)
 &=(1-t)\operatorname{Tr}(\rho\log\rho)
 +\operatorname{Tr}\!\left[
 \rho\log\!\left(\rho^{t/2}\sigma^{-1}\rho^{t/2}\right)
 \right],
 \label{rw:eq:intro-Dt-explicit}\\
 D_\alpha^t(\rho\Vert\sigma)
 &=\frac{1}{\alpha-1}\log Q_{\alpha,t}(\rho\Vert\sigma),
 \qquad
 Q_{\alpha,t}
 =\operatorname{Tr}\!\left[
  \sigma^{1/2}\rho^{\alpha(1-t)}\sigma^{1/2}
  \left(\sigma^{-1/2}\rho^t\sigma^{-1/2}\right)^\alpha
 \right].
 \label{rw:eq:intro-Renyi-explicit}
\end{align}
The endpoints of \eqref{rw:eq:intro-Dt-explicit} are Umegaki and BS,
whereas those of \eqref{rw:eq:intro-Renyi-explicit} are the Petz and
geometric R\'enyi divergences; moreover,
$D_\alpha^t\to D^t$ as $\alpha\to1$.  These statements are proved in
\Cref{rw:prop:Dt-explicit,rw:subsec:renyi-special-case}.
Equivalently, a probability measure $\nu_t$ associated with the
spectral resolution of $\Gamma_t$ satisfies
\begin{equation*}
 D_f^t=\int f(\lambda)\,\mathrm d\nu_t(\lambda),
 \qquad
 Q_{\alpha,t}=\int\lambda^\alpha\,\mathrm d\nu_t(\lambda),
 \qquad
 \int\lambda\,\mathrm d\nu_t(\lambda)=1.
\end{equation*}
This scalar-moment formulation drives several of the basic results
below; further reformulations, including one through sandwiched
R\'enyi divergences, are given in
\Cref{rw:prop:basic-spectral-properties,rw:prop:Dt-sandwiched-reformulation}.

\medskip
\noindent\textbf{Basic properties and comparison with known families.}
The geodesic quantities have the standard structural properties
available in this setting: classical reduction, unitary invariance,
exact flag rules, and continuity.  For a general generator, the
normalization is
$D_f^t(\rho\Vert\rho)=f(1)$.  If $f(1)=0$ and $f$ is scalar convex
on an interval containing $\{1\}\cup\operatorname{spec}(\Gamma_t)$,
then $D_f^t\geq0$; if, in addition, $f$ is strictly convex on that
interval, then
\begin{equation}
 D_f^t(\rho\Vert\sigma)=0
 \quad\Longleftrightarrow\quad \rho=\sigma.
 \label{rw:eq:intro-Dft-equality-separation}
\end{equation}
Thus the $t$-relative entropies are nonnegative and separate states.
The normalized R\'enyi divergences have the same two properties for
every $\alpha>0$, with $D_1^t=D^t$.  The $t$-relative entropies and
the normalized R\'enyi divergences are tensor additive.  For
operator-convex generators the
unnormalized $(t,f)$-divergences are jointly convex; in the R\'enyi
DPI range, $D_\alpha^t$ is jointly convex below order one and jointly
quasi-convex above order one.  Most importantly for the R\'enyi
parameter, for every fixed $t$,
\begin{equation}
 \alpha\longmapsto D_\alpha^t(\rho\Vert\sigma)
 \quad\text{is non-decreasing,}
 \label{rw:eq:intro-alpha-monotonicity}
\end{equation}
strictly so unless $\rho=\sigma$, as shown in
\Cref{rw:prop:elementary-structural-rules,%
rw:prop:elementary-structural-rules-2,%
rw:prop:basic-spectral-properties,cor:consolidated-joint-convexity}.

The complete Petz--geometric band for a fixed noncommuting pair is
shown in \Cref{rw:fig:renyi-interpolation-band}.  The geodesic
R\'enyi family is distinct from the $\alpha$--$z$ family: their exact
universal intersections, even allowing real $t$, are determined in
\Cref{rw:prop:alpha-z-intersections}, and their value ranges on the
common DPI region are compared in
\Cref{rw:fig:renyi-family-comparison}.  These two figures show,
respectively, how $D_\alpha^t$ fills the interval between the Petz and
geometric endpoints and how this new interpolation sits relative to
the established $\alpha$--$z$ family.

\medskip
\noindent\textbf{Monotonicity along the geodesic.}
The interpolation parameter has a sharp and nontrivial meaning.  For
the relative entropies,
\begin{equation}
 D(\rho\Vert\sigma)=D^0(\rho\Vert\sigma)
 \leq D^t(\rho\Vert\sigma)
 \leq D^1(\rho\Vert\sigma)=\widehat D(\rho\Vert\sigma),
 \label{rw:eq:intro-Dt-monotonicity}
\end{equation}
and every inequality between two distinct interpolation times is
strict exactly when $[\rho,\sigma]\neq0$.  For the R\'enyi family,
$t\mapsto D_\alpha^t$ is non-decreasing for $0<\alpha<2$, constant
at $\alpha=2$, and non-increasing for $2<\alpha\leq3$.  For
$\alpha>3$ there is no universal direction in dimension at least
three; instead we obtain an exact pair-of-times criterion.  These
claims are proved in
\Cref{rw:thm:Dt-strict-monotonicity,%
rw:thm:power-t-monotonicity,rw:fig:power-t-monotonicity}.

Operator convexity of a general generator does not by itself impose
a direction in $t$.  We therefore define the universal cone
$\mathfrak M$ as the set of generators $f$ for which
$t\mapsto D_f^t(\rho\Vert\sigma)$ is non-decreasing on $[0,1]$ for
every finite-dimensional invertible pair $(\rho,\sigma)$.  Thus
``universal'' quantifies simultaneously over all finite dimensions
and all such ordered pairs.
The universal characterization obtained here concerns $C^1$
generators.  More precisely, for $f\in C^1(0,\infty)$, set
$g(x):=f(x)/x$ and let $\mathfrak g(A):=g(A)$.  If
$Pe_i=p_i e_i$, then
\begin{equation*}
 f\in\mathfrak M
 \quad\Longleftrightarrow\quad
 \mathcal L_f(a;P,K)
 :=a\sum_i p_i^{a+1}
 \left\langle e_i,
 \mathrm D\mathfrak g(p_iK)
 \left[\frac12\{\log P-(\log p_i)I,K\}\right]e_i
 \right\rangle\geq0
\end{equation*}
for every dimension, every $a>1$, and all positive definite $P,K$
satisfying
$\operatorname{Tr}P^a=\operatorname{Tr}(P^{a-1}K^{-1})=1$.
On the finite operator-convex slice, this ordinary-matrix test is
equivalently the positive-polar condition
\begin{equation*}
 \int_{[0,1]}\kappa_{\rho,\sigma,t}(u)\,
 \mathrm d\nu_f(u)\geq0
\end{equation*}
for every finite-dimensional invertible pair and every $t\in(0,1)$,
where $\nu_f$ is the canonical representing measure and
$\kappa_{\rho,\sigma,t}$ is the response function defined in
\eqref{rw:eq:universal-cone-response-function}.  These characterizations
are established in
\Cref{rw:thm:universal-cone-matrix-characterization,%
rw:thm:universal-cone-measure-characterization}.

\medskip
\noindent\textbf{Data processing and strengthened data processing.}
The fundamental result of the paper is that for every finite
operator-convex $f$, every channel $\Phi$, and every
$0\leq t\leq1$,
\begin{equation}
 \boxed{
 D_f^t\!\left(\Phi(\rho)\Vert\Phi(\sigma)\right)
 \leq D_f^t(\rho\Vert\sigma).}
 \label{rw:eq:intro-DPI}
\end{equation}
Together with normalization, this is what justifies calling the new
quantities quantum divergences.  In particular, $D^t$ satisfies DPI
along the whole geodesic, while $D_\alpha^t$ satisfies it for
$0<\alpha<1$ and $1<\alpha\leq2$, as stated in
\Cref{thm:consolidated-dpi,cor:consolidated-basic-examples}.

We also strengthen \eqref{rw:eq:intro-DPI}.  The DPI deficit has a
positive resolvent remainder, which can be compressed into a single
intrinsic Hilbert--Schmidt recovery defect.  This yields a quartic
remainder for $D^t$, order-dependent remainders for $D_\alpha^t$,
and explicit tripartite lower bounds.  These results extend the
remainder strategy for the Umegaki and BS endpoints
\cite{carlen2018recovery,bluhm2020strengthened}; the corresponding
statements are
\Cref{thm:consolidated-resolvent-dpi,%
thm:consolidated-Dt-quartic,thm:consolidated-Renyi-dpi,%
cor:consolidated-tripartite-dpi}.

\medskip
\noindent\textbf{Equality, recovery, and the endpoint dichotomy.}
Our main rigidity theorem identifies the equality cases of
\eqref{rw:eq:intro-DPI}.  We let
\begin{equation*}
 \mathcal P_{\sigma,\Phi}(X)
 :=\sigma^{1/2}\Phi^*\!\left(
   \Phi(\sigma)^{-1/2}X\Phi(\sigma)^{-1/2}
  \right)\sigma^{1/2}
\end{equation*}
be the Petz recovery map.  If the input and output states are
invertible, $0\leq t<1$, and $f$ is an equality-determining
operator-convex generator, then
\begin{equation}
 \boxed{
 D_f^t(\rho\Vert\sigma)
 =D_f^t\!\left(\Phi(\rho)\Vert\Phi(\sigma)\right)
 \quad\Longleftrightarrow\quad
 \rho=\mathcal P_{\sigma,\Phi}(\Phi(\rho)).}
 \label{rw:eq:intro-Petz-rigidity}
\end{equation}
Thus the recovery condition does \emph{not} interpolate throughout
the nonmaximal range: every such equality set is the ordinary Petz
sufficiency class.  At $t=1$ it changes to the generally larger BS
equality class.  This applies in particular to the relative entropies
and to R\'enyi orders $\alpha\in(0,1)\cup(1,2)$, with the order-one
limit included.  The quadratic order $\alpha=2$ is the distinguished
exception: its $t$-independent divergence has the BS equality class
for every $t$.  Affine or spectrally sparse, non-equality-determining
generators likewise require separate treatment.

The same theorem gives several equivalent formulations of equality:
in terms of likelihood ratios, half powers, the full modular
functional calculus, Petz recovery, and preservation of Umegaki
relative entropy.  The complete hierarchy and its R\'enyi
fixed-point diagram are given in
\Cref{thm:consolidated-equality-recovery,%
prop:consolidated-general-hierarchy,%
thm:consolidated-interior-rigidity,%
fig:consolidated-Renyi-fixed-point-phase}.  In this precise sense,
the equality-determining nonquadratic families exhibit only two
recovery classes: the Petz class before the maximal endpoint and the
BS class at that endpoint.  The exceptional quadratic family belongs
to the latter class throughout the interpolation.

\medskip
\noindent\textbf{Conditional mutual informations and quantum Markov
chains.}
Specializing DPI to partial traces produces three natural ordered
$t$-conditional mutual informations---one-sided, two-sided, and
reverse---defined in \eqref{eq:consolidated-three-CMIs}.  An invertible
state is called a $t$-quantum Markov chain when the corresponding
deficit vanishes.  The equality theorem gives the second central
dichotomy of the paper:
\begin{equation}
 \boxed{
 \mathfrak Q_t^{\os}=\mathfrak Q_t^{\ts}=\mathfrak Q_t^{\rev}
 =\{\text{invertible quantum Markov chains}\},
 \qquad 0\leq t<1.}
 \label{rw:eq:intro-tQMC-collapse}
\end{equation}
At $t=1$, all three classes are the invertible BS quantum Markov
chains.  Hence every invertible nonmaximal $t$-QMC is a QMC and
conversely, while the BS endpoint can be strictly larger.  The
Werner-state construction of \cite{Bluhm:2025kai} witnesses this
strictness.  The analogous statement holds for equality-determining
$(t,f)$-CMIs and for the R\'enyi orders in the nonquadratic DPI range,
as shown in
\Cref{cor:consolidated-tQMC-hierarchy,%
cor:consolidated-f-Renyi-tQMC-hierarchy}.

The abrupt collapse to the Petz class has a modular explanation.  It
parallels the rigidity of adjointness with respect to the family of
power-weighted inner products: all non-KMS weights impose the same
modular covariance, whereas the KMS weight may have a larger symmetry
class, as explained in
\Cref{subsec:consolidated-weighted-symmetry}.  Numerically, however,
the three $t$-CMIs need not be ordered.  We prove their universal
incomparability, establish an additive projective comparison for the
two forward orientations, obtain a genuine multiplicative comparison
at R\'enyi order two, and show that the analogous BS multiplicative
factor cannot exist in general, as proved in
\Cref{prop:consolidated-CMI-incomparability,%
prop:consolidated-forward-projective-comparison,%
prop:consolidated-quadratic-CMI-comparison,%
prop:consolidated-no-BS-projective-factor}.

\medskip
\noindent\textbf{Further results.}
We prove continuity bounds for conditional $t$-entropies, the three
$t$-CMIs, and the general $(t,f)$-divergences when the relevant
smallest eigenvalues are bounded below by fixed positive constants,
compare the Umegaki and BS endpoint specializations with the
corresponding bounds in the literature,
and show that, for the nonoptimized quantities at any fixed $t>0$,
no dimension-only continuity modulus can hold uniformly as the
smallest eigenvalues approach zero; these results appear in
\Cref{sec:continuity-interpolated-conditionals}.  Finally, the canonical
finite experiment $(p_t,q_t)$ constructed in the companion paper
satisfies
$D^t(\rho\Vert\sigma)=D_{\rm KL}(p_t\Vert q_t)$
\cite[Theorem~3.1 and Proposition~3.3(i)]{capel2026informationgeometry}.
Combining this identity with the classical zero-cost-symbol theorem
gives the capacity-per-unit-cost interpretation proved here for every
$D^t$, including its BS specialization, as stated in \Cref{thm:op-capacity}.
The companion paper separately develops the Hessian-metric,
likelihood-statistical and restricted hypothesis-testing, and
horospherical-geometric consequences of the same construction
\cite[Theorems~2.3 and~2.5--2.6, Propositions~3.2--3.3,
and Section~4]{capel2026informationgeometry}.

\FloatBarrier

\section{Preliminaries}
\label{rw:sec:preliminaries}

\subsection{Basic notions and notations}
\label{rw:subsec:states-channels}

We work with finite-dimensional Hilbert spaces $\cH$ and $\cK$.  We write
$\mathcal B(\mathcal H)$ for the algebra of bounded linear operators on
$\cH$, $\mathcal B(\mathcal H)_+$ for its positive cone, and
\begin{equation}
   \mathcal S(\mathcal H)
 :=\{\rho\in\mathcal B(\mathcal H)_+:
          \operatorname{Tr}\rho=1\}  
\end{equation}
for the set of density matrices, or states, and
$\mathcal S_{+}(\mathcal H)$
for the set of invertible states. The support projection of a
positive operator $X$ is denoted by $\operatorname{supp}(X)$.
Throughout, $\operatorname{Tr}$ denotes the ordinary unnormalized
matrix trace, and $\operatorname{Tr}_A$ denotes the usual partial
trace over $A$; in particular, no factor $d_A^{-1}$ is included in
$\operatorname{Tr}_A$.

A quantum channel
$\Phi:\mathcal B(\mathcal H)\to\mathcal B(\mathcal K)$ is a
completely positive and trace-preserving linear map.  Its Hilbert--Schmidt
adjoint $\Phi^*$ is unital and completely positive.  Equivalently,
there are an ancillary space $\mathcal E$ and an isometry
$W:\mathcal H\to\mathcal K\otimes\mathcal E$ such that
\begin{equation}
 \Phi(X)=\operatorname{Tr}_{\mathcal E}(WXW^*),
 \qquad
 \Phi^*(Y)=W^*(Y\otimes I_{\mathcal E})W ,
 \label{rw:eq:stinespring}
\end{equation}
where $\operatorname{Tr}_{\mathcal E}$ denotes the partial trace over ${\mathcal E}$ and $I_{\mathcal E}$ is the identity matrix in ${\mathcal E}$.
This is the finite-dimensional form of Stinespring's dilation theorem
\cite[Theorem~1]{stinespring1955positive}.

Our basic formulas are first stated for invertible pairs.  If
$\rho,\sigma>0$, then there exist constants $0<c\leq C<\infty$
such that
\begin{equation}
 c\sigma\leq\rho\leq C\sigma.
 \label{rw:eq:invertible-comparison}
\end{equation}
Positivity of $\Phi$ implies the same comparison for
$\Phi(\rho)$ and $\Phi(\sigma)$, and
\begin{equation}
 \operatorname{supp}(\Phi(\rho))=\operatorname{supp}(\Phi(\sigma)).
 \label{rw:eq:common-output-support}
\end{equation}
All output inverses are taken on this support.  Hence,
we replace $\mathcal K$ by
$\operatorname{supp}(\Phi(\sigma))\mathcal K$, where both outputs are invertible.  

The Petz map associated with a reference state $\sigma>0$ and a
channel $\Phi$ \cite{Petz-SufficiencyChannels-1978} is
\begin{equation}
 \mathcal P_{\sigma,\Phi}(X)
 :=\sigma^{1/2}\Phi^*\!\left(
[\Phi(\sigma)]^{-1/2}X[\Phi(\sigma)]^{-1/2}\right)\sigma^{1/2},
 \label{rw:eq:petz-map}
\end{equation}
where the inverse is taken on $\operatorname{supp}(\Phi(\sigma))$.  It is completely
positive and trace-preserving on the supported output algebra, and
\begin{equation}
 \mathcal P_{\sigma,\Phi}(\Phi(\sigma))=\sigma.
 \label{rw:eq:petz-recovers-reference}
\end{equation}
Indeed, invertibility of $\sigma$ and
$\operatorname{Tr}[\Phi(\sigma)(I-\operatorname{supp}(\Phi(\sigma)))]=0$ imply
$\Phi^*(\operatorname{supp}(\Phi(\sigma)))=I$.  We say that $\Phi$ is reversible, or
sufficient, for $\{\rho,\sigma\}$ when
$\mathcal P_{\sigma,\Phi}(\Phi(\rho))=\rho$; in finite dimensions
this is equivalent to the existence of some channel recovering both
states \cite{petz2003monotonicity,jenvcova2006sufficiency}.

\subsection{The Hilbert--Schmidt representation}
\label{rw:subsec:hs-representation}

We regard $\mathcal B(\mathcal H)$ as the Hilbert space $\mathcal{H}_{HS}$ with the Hilbert-Schmidt inner product
\begin{equation}
 \langle X,Y\rangle_{\rm HS}:=\operatorname{Tr}(X^*Y),
 \qquad \|X\|_2:=\langle X,X\rangle_{\rm HS}^{1/2}.
 \label{rw:eq:hs-inner-product}
\end{equation}
For $a,b\in\mathcal B(\mathcal H)$, we define
\begin{equation}
 L_a(X):=aX,
 \qquad
 R_b(X):=Xb.
 \label{rw:eq:left-right}
\end{equation}
Then, $L_aR_b=R_bL_a$.  If $a>0$, both $L_a$ and $R_a$ are
positive definite operators, and
functional calculus gives
\begin{equation}
 h(L_a)=L_{h(a)},\qquad h(R_a)=R_{h(a)}.
 \label{rw:eq:left-right-functional-calculus}
\end{equation}
If $a,b>0$, commutation of left and right multiplication further
implies
\begin{equation}
 \log(L_aR_b)=L_{\log a}+R_{\log b}.
 \label{rw:eq:log-left-right}
\end{equation}
For an invertible pair, the standard relative modular operator is
\begin{equation}
 \Delta_{\rho\mid\sigma}:=L_\rho R_{\sigma^{-1}}.
 \label{rw:eq:relative-modular}
\end{equation}
We shall also use the positive definite commutant relative modular
operator
\begin{equation}
 \widehat\Delta_{\rho\mid\sigma}
 :=R_{\sigma^{-1/2}\rho\sigma^{-1/2}}.
 \label{rw:eq:commutant-derivative}
\end{equation}
Both $\Delta_{\rho\mid\sigma}$ and
$\widehat\Delta_{\rho\mid\sigma}$ are positive definite operators on
$\mathcal H_{\rm HS}$; they act on Hilbert--Schmidt vectors, not on
the state space itself.

\subsection{Standard and maximal quantum \texorpdfstring{$f$}{f}-divergences}
\label{rw:subsec:standard-maximal-f-divergences}

We fix $\rho,\sigma\in\mathcal S_{+}(\mathcal H)$ and a function
$f:(0,\infty)\to\mathbb R$ defined on the spectra relevant below.
The \textit{standard}, or \textit{Petz}, \textit{quantum $f$-divergence} \cite{hiai2017different} is
\begin{equation}
 D_f(\rho\Vert\sigma)
 :=\big\langle\sigma^{1/2},
 f(\Delta_{\rho\mid\sigma})[\sigma^{1/2}]
 \big\rangle_{\rm HS}.
 \label{rw:eq:standard-f-divergence}
\end{equation}
This is Petz's quasi-entropy with auxiliary operator $I$
\cite{petz1986quasi,hiai2012quasi}.  Its maximal counterpart, the
\textit{maximal quantum $f$-divergence}
\cite[Definition~3.21]{hiai2017different}, is
\begin{equation}
 \begin{aligned}
 \widehat D_f(\rho\Vert\sigma)
 :=\big\langle\sigma^{1/2},
 f(\widehat\Delta_{\rho\mid\sigma})[\sigma^{1/2}]
 \big\rangle_{\rm HS}=\operatorname{Tr}\!\left[
 \sigma f(\sigma^{-1/2}\rho\sigma^{-1/2})\right].
 \end{aligned}
 \label{rw:eq:maximal-f-divergence}
\end{equation}
The second equality follows from
$f(R_a)=R_{f(a)}$.  The adjective ``maximal'' refers to its position
among monotone quantum extensions of the corresponding classical
$f$-divergence in the operator-convex regime; this is Matsumoto's
reverse-test maximality theorem \cite[Lemma~3.1]{matsumoto2018new}.

For strictly positive probability distributions $p=(p_i)_i$ and $q=(q_i)_i$,
so that $p_i,q_i>0$, our convention for the \textit{classical $f$-divergence} is
\begin{equation}
 D_f(p\Vert q):=\sum_i q_i f(p_i/q_i).
 \label{rw:eq:classical-f-divergence}
\end{equation}

The logarithmic generator $f(x)=x\log x$ gives the two relative
entropies at the ends of our interpolation: The \textit{Umegaki relative entropy} \cite{Umegaki-RelativeEntropy-1962}
\begin{equation}
 \begin{aligned}
 D(\rho\Vert\sigma)
 :=D_{x\log x}(\rho\Vert\sigma)
 &=\operatorname{Tr}\!\left[\rho(\log\rho-\log\sigma)\right],
 \end{aligned}
 \label{rw:eq:Umegaki-relative-entropies}
\end{equation}
and the \textit{Belavkin-Staszewski relative entropy} \cite{BelavkinStaszewski-BSentropy-1982}
\begin{equation}
 \begin{aligned}
 \widehat D(\rho\Vert\sigma)
 :=\widehat D_{x\log x}(\rho\Vert\sigma)
 &=\operatorname{Tr}\!\left[
 \rho\log(\rho^{1/2}\sigma^{-1}\rho^{1/2})\right].
 \end{aligned}
 \label{rw:eq:Belavkin-Stasszewski-relative-entropies}
\end{equation}
For the first identity, we are using that, by functional calculus,
$\Delta_{\rho\mid\sigma}^{1/2}[\sigma^{1/2}]=\rho^{1/2}$;
for the second identity, putting
$X=\sigma^{-1/2}\rho^{1/2}$, the push-through identity
$(XX^*)\log(XX^*)=X\log(X^*X)X^*$ converts
\eqref{rw:eq:maximal-f-divergence} into the displayed formula.  

For the power generator
$f_\alpha(x)=x^\alpha$, $\alpha>0$, one instead obtains
\begin{equation}
 \begin{aligned}
 D_{f_\alpha}(\rho\Vert\sigma)
 =\operatorname{Tr}(\rho^\alpha\sigma^{1-\alpha}),
 \end{aligned}
 \label{rw:eq:standard-power-functionals}
\end{equation}
and
\begin{equation}
 \begin{aligned}
 \widehat D_{f_\alpha}(\rho\Vert\sigma)
 =\operatorname{Tr}\!\left[
 \sigma(\sigma^{-1/2}\rho\sigma^{-1/2})^\alpha\right].
 \end{aligned}
 \label{rw:eq:maximal-power-functionals}
\end{equation}
respectively. Thus, for $\alpha>0$ with $\alpha\neq1$, the \textit{Petz} and \textit{geometric
R\'enyi divergences} \cite[Definition~2]{fang2021geometric} are obtained as normalized logarithmic transforms
of the standard and maximal power $f$-functionals generated by
$f_\alpha(x)=x^\alpha$:
\begin{equation}
 \begin{aligned}
 \overline{D}_\alpha(\rho\Vert\sigma)
 &:=\frac{1}{\alpha-1}
 \log D_{f_\alpha}(\rho\Vert\sigma)
 =\frac{1}{\alpha-1}
 \log\operatorname{Tr}(\rho^\alpha\sigma^{1-\alpha}),\\
\widehat D_\alpha(\rho\Vert\sigma)
 &:=\frac{1}{\alpha-1}
 \log \widehat D_{f_\alpha}(\rho\Vert\sigma)
 =\frac{1}{\alpha-1}
 \log\operatorname{Tr}\!\left[
 \sigma(\sigma^{-1/2}\rho\sigma^{-1/2})^\alpha\right].
 \end{aligned}
 \label{rw:eq:petz-geometric-renyi-as-f-divergences}
\end{equation}
Equivalently, multiplying by $\alpha-1$ and exponentiating recovers
exactly the standard and maximal power $f$-divergences.  The logarithm and the factor
$(\alpha-1)^{-1}$ are the R\'enyi normalization; hence the normalized
quantities in \eqref{rw:eq:petz-geometric-renyi-as-f-divergences} are
logarithmic transforms of quantum $f$-divergences.  For $\alpha>2$,
$x^\alpha$ lies outside the operator-convex maximal-divergence range,
and the displayed geometric quantity has no universal data-processing
inequality \cite[Remark~3.23 and Corollary~3.31]{hiai2017different}.
The quantities from \eqref{rw:eq:petz-geometric-renyi-as-f-divergences}
will be the endpoints of the geodesic Rényi divergences introduced
in \Cref{rw:subsec:renyi-special-case}.

Their order-one limits recover the two relative entropies above:
\begin{equation}
 \lim_{\alpha\to1} \overline{D}_\alpha(\rho\Vert\sigma)
 =D(\rho\Vert\sigma),
 \qquad
 \lim_{\alpha\to1}\widehat D_\alpha(\rho\Vert\sigma)
 =\widehat D(\rho\Vert\sigma).
 \label{rw:eq:petz-geometric-renyi-order-one-limits}
\end{equation}
The geometric limit is recorded, for example, in
\cite[Remark~1]{fang2021geometric}.
For the raw power functionals, $\alpha=1$ is affine and $\alpha=2$
is quadratic, and the standard and maximal functionals coincide at
both orders for every invertible pair.  At $\alpha=1$ this equality is
only the trivial identity $D_{f_1}=\widehat D_{f_1}=1$; their
first-order logarithmic limits in
\eqref{rw:eq:petz-geometric-renyi-order-one-limits} generally differ.
At $\alpha=2$, both the power functionals and their R\'enyi
normalizations coincide, even for noncommuting pairs:
\begin{equation}
 \overline D_2(\rho\Vert\sigma)
 =\widehat D_2(\rho\Vert\sigma)
 =\log\operatorname{Tr}(\rho^2\sigma^{-1}).
 \label{rw:eq:petz-geometric-renyi-quadratic-coincidence}
\end{equation}
\Cref{rw:subsec:chi-square} strengthens this observation by
showing that the entire interpolating path is constant at order two.
Away from these exceptional orders, the
power functionals generally differ when $\rho$ and $\sigma$ do not
commute, as do the Petz and geometric R\'enyi divergences.
If $[\rho,\sigma]=0$, however, both reduce to the same classical
quantity,
\begin{equation}
 D_f(\rho\Vert\sigma)
 =\widehat D_f(\rho\Vert\sigma)
 =\sum_i\sigma_i f(\rho_i/\sigma_i).
 \label{rw:eq:standard-maximal-classical-collapse}
\end{equation}
Thus the standard and maximal constructions differ only through their
noncommutative operator ordering.  \Cref{rw:sec:interpolated-divergences} will
construct a path joining them, while the operator-convex assumptions
needed for data processing are recorded below.

\subsection{Quantum Markov chains and their
Belavkin--Staszewski extension}
\label{rw:subsec:qmc-bsqmc}

We recall here the two endpoint notions of quantum conditional
independence that motivate the geodesic Markov classes studied
later.  We take $\rho_{ABC}$ to be a tripartite state and write
\begin{equation}
 S(\omega):=-\operatorname{Tr}(\omega\log\omega)
\end{equation}
for the von Neumann entropy, and suppress identity embeddings of
marginal operators.  The \emph{conditional mutual information} (CMI)
of $A$ and $C$ conditioned on $B$ is
\begin{equation}
 I(A;C\mid B)_\rho
 :=S(\rho_{AB})+S(\rho_{BC})-S(\rho_B)-S(\rho_{ABC}).
 \label{rw:eq:ordinary-cmi}
\end{equation}
For any invertible normalized state $\tau_C$, it is equivalently the
relative-entropy DPI deficit
\begin{equation}
 \begin{aligned}
 I(A;C\mid B)_\rho
 &=D(\rho_{ABC}\Vert\rho_{AB}\otimes\tau_C)
   -D(\rho_{BC}\Vert\rho_B\otimes\tau_C)
 \geq0.
 \end{aligned}
 \label{rw:eq:cmi-as-DPI-deficit}
\end{equation}
The inequality is the strong subadditivity of quantum entropy
\cite{lieb1973strong}.  A state for which
$I(A;C\mid B)_\rho=0$ is called a \emph{quantum Markov chain}
(QMC) in the order $A-B-C$.

For an invertible state, Petz's equality theorem and Ruskai's logarithmic
characterization of equality in strong subadditivity give the following
equivalences
\cite{Petz-SufficiencyChannels-1978,petz2003monotonicity}%
\cite[Theorem~1]{ruskai2002inequalities}:
\begin{equation}
 \begin{aligned}
 \rho_{ABC}\text{ is a QMC}
 &\Longleftrightarrow I(A;C\mid B)_\rho=0\\
 &\Longleftrightarrow
 \rho_{ABC}
 =\rho_{AB}^{1/2}\rho_B^{-1/2}\rho_{BC}
   \rho_B^{-1/2}\rho_{AB}^{1/2}\\
 &\Longleftrightarrow
 \log\rho_{ABC}+\log\rho_B
 =\log\rho_{AB}+\log\rho_{BC}.
 \end{aligned}
 \label{rw:eq:qmc-equivalent-conditions}
\end{equation}
The middle line is precisely recovery of $\rho_{BC}$ by the Petz
channel associated with $\operatorname{Tr}_A$ and the reference
$\rho_{AB}\otimes\tau_C$; the factors involving $\tau_C$ cancel.
The structure theorem of Hayden, Jozsa, Petz, and Winter
\cite{hayden-2004} equivalently shows that the middle system has a
decomposition
\begin{equation}
 \mathcal H_B
 =\bigoplus_j\mathcal H_{B_j^L}\otimes\mathcal H_{B_j^R},
 \qquad
 \rho_{ABC}
 =\bigoplus_j p_j\,
   \rho_{AB_j^L}\otimes\rho_{B_j^RC}.
 \label{rw:eq:qmc-block-structure}
\end{equation}
\begin{remark}[Non-invertible QMCs]
For an arbitrary tripartite state, vanishing conditional mutual
information, recovery by the Petz map with support inverses, and the
block decomposition \eqref{rw:eq:qmc-block-structure} remain
equivalent
\cite{Petz-SufficiencyChannels-1978,petz2003monotonicity,hayden-2004}.
The logarithmic identity in
\eqref{rw:eq:qmc-equivalent-conditions} is the convenient invertible
formulation and is not used as a definition for the non-invertible case.
\end{remark}
Thus a QMC is the exact quantum analogue of conditional independence:
all correlations between $A$ and $C$ can be mediated through $B$,
the full state can be reconstructed from $\rho_{BC}$ by a quantum
channel, and the rigid block form in
\eqref{rw:eq:qmc-block-structure} implies that $\rho_{AC}$ is
separable.  Small CMI is correspondingly a quantitative notion of
approximate recoverability \cite{fawzi2015approximate}.

Replacing $D$ in \eqref{rw:eq:cmi-as-DPI-deficit} by the
Belavkin--Staszewski relative entropy produces several inequivalent
numbers because the latter is not symmetric under changes of
operator ordering.  For invertible $\rho_{ABC}$, we define
\begin{align}
 \widehat I^{\os}(A;C\mid B)_\rho
 &:=
 \widehat D(\rho_{ABC}\Vert\rho_{AB}\otimes\tau_C)
 -\widehat D(\rho_{BC}\Vert\rho_B\otimes\tau_C),\notag\\
 \widehat I^{\ts}(A;C\mid B)_\rho
 &:=
 \widehat D(\rho_{ABC}\Vert\rho_{AB}\otimes\rho_C)
 -\widehat D(\rho_{BC}\Vert\rho_B\otimes\rho_C),\notag\\
 \widehat I^{\rev}(A;C\mid B)_\rho
 &:=
 \widehat D(\rho_{AB}\otimes\tau_C\Vert\rho_{ABC})
 -\widehat D(\rho_B\otimes\tau_C\Vert\rho_{BC}).
 \label{rw:eq:three-BS-CMIs}
\end{align}
The superscripts stand for \emph{one-sided}, \emph{two-sided}, and
\emph{reverse}.  Each is a nonnegative DPI deficit, and although
their numerical values generally differ, their zero sets on invertible states
coincide
\cite{bluhm2023continuity,alhambra2026conditional,Bluhm:2025kai}.
An invertible state in this common zero set is called a
\emph{Belavkin--Staszewski quantum Markov chain} (BS-QMC).  The common
equality condition has the intrinsic forms
\begin{equation}
 \begin{aligned}
 \rho_{ABC}\text{ is a BS-QMC}
 &\Longleftrightarrow
 \widehat I^\bullet(A;C\mid B)_\rho=0
 \quad\text{for some, hence every, }
 \bullet\in\{\os,\ts,\rev\}\\
 &\Longleftrightarrow
 \rho_{ABC}=\rho_{AB}\rho_B^{-1}\rho_{BC}\\
 &\Longleftrightarrow
 \rho_{ABC}
 =\left(\rho_{AB}\rho_B^{-1}\rho_{BC}^{\,2}
        \rho_B^{-1}\rho_{AB}\right)^{1/2}.
 \end{aligned}
 \label{rw:eq:bsqmc-equivalent-conditions}
\end{equation}
The product identity is the asymmetric BS recovery condition; the
last line is its positive symmetric form.  The associated linear
operation is not generally a quantum channel, but an equivalent
completely positive recovery condition is available and will be
recalled in \Cref{sec:consolidated-recovery}.

The significance of the BS notion is that it is exactly the
partial-trace equality class naturally selected by the maximal
relative entropy.  It contains the QMC class, but the
inclusion is strict:
\begin{equation}
 \{\text{invertible QMCs}\}
 \subset
 \{\text{invertible BS-QMCs}\},
 \label{rw:eq:qmc-in-bsqmc}
\end{equation}
and the inclusion is strict whenever the systems are sufficiently
large to support the known witnesses.  In particular, a BS-QMC can
have an entangled $AC$ marginal
\cite[Proposition~4.3]{Bluhm:2025kai}, which is impossible for a
QMC.  Nevertheless its structure is controlled by a QMC: if
$d_B:=\dim\mathcal H_B$ and
\begin{equation}
 \eta_{ABC}
 :=\frac1{d_B}\rho_B^{-1/2}\rho_{ABC}\rho_B^{-1/2},
 \qquad
 \eta_B=\frac{I_B}{d_B},
 \label{rw:eq:canonical-BS-QMC-shadow}
\end{equation}
then
\begin{equation}
 \rho_{ABC}\text{ is a BS-QMC}
 \quad\Longleftrightarrow\quad
 \eta_{ABC}\text{ is a QMC}.
 \label{rw:eq:BS-QMC-QMC-correspondence}
\end{equation}
The strict inclusion, the correspondence
\eqref{rw:eq:BS-QMC-QMC-correspondence}, and the recovery and
structural theory of BS-QMCs were established in
\cite{Bluhm:2025kai}; the role of BS conditional quantities in
correlation decay and quantum Gibbs states is also discussed in
\cite{bluhm2023continuity,alhambra2026conditional}.  The three endpoint quantities in
\eqref{rw:eq:three-BS-CMIs} are interpolated in
\Cref{subsec:consolidated-tripartite-tcmi}, where the resulting
$t$-QMC hierarchy is determined.

\subsection{Affine-invariant geometry of positive operators}
\label{rw:subsec:positive-geometry}

Let $\mathcal V$ be a finite-dimensional complex Hilbert space.  We
write $\mathbb H(\mathcal V)$ for the real vector space of
self-adjoint operators on $\mathcal V$ and
\begin{equation*}
 \mathbb P(\mathcal V)
 :=\{A\in\mathbb H(\mathcal V):A>0\}
\end{equation*}
for its open cone of positive-definite operators.  When
$\mathcal V=\mathbb C^n$, we also write $\mathbb H_n$ and
$\mathbb P_n$, respectively.  Since $\mathbb P(\mathcal V)$ is an
open submanifold of $\mathbb H(\mathcal V)$, its tangent space is
canonically identified with
$T_A\mathbb P(\mathcal V)\cong\mathbb H(\mathcal V)$ at every
$A\in\mathbb P(\mathcal V)$ \cite[Chap.~6]{bhatia2009positive}.  We
endow it with the affine-invariant Riemannian metric
$g_A:T_A\mathbb P(\mathcal V)\times T_A\mathbb P(\mathcal V)
\to\mathbb R$ given by
\begin{equation}
 g_A(H,K):=\operatorname{Tr}_{\mathcal V}(A^{-1}HA^{-1}K),
 \qquad A\in\mathbb P(\mathcal V),
 \label{rw:eq:affine-invariant-metric}
\end{equation}
for Hermitian tangent vectors $H,K$.  The resulting distance is
\begin{equation}
 d_2(A,B)
 =\big\|\log(A^{-1/2}BA^{-1/2})\big\|_2.
 \label{rw:eq:affine-invariant-distance}
\end{equation}
In particular, $d_2(A,B)$ is not generally
$\|\log A-\log B\|_2$; that simplification is only valid when $A$ and
$B$ commute.  The unique constant-speed minimizing geodesic from
$A$ to $B$ is the weighted geometric mean
\begin{equation}
 \gamma_{A,B}(t):=A\#_tB
 =A^{1/2}(A^{-1/2}BA^{-1/2})^tA^{1/2},
 \qquad 0\leq t\leq1.
 \label{rw:eq:weighted-geometric-mean}
\end{equation}
Thus $\gamma_{A,B}(0)=A$, $\gamma_{A,B}(1)=B$, and
$\gamma_{A,B}(t)\in\mathbb P(\mathcal V)$ for every $t\in[0,1]$.
In the application below, $\mathcal V$ is the Hilbert--Schmidt space
and $A,B$ are positive-definite superoperators on that space.  Both
the metric and the distance are invariant under invertible
congruences \cite{bhatia2009positive,nielsen2013matrix}.

\subsection{Operator-convex generators}
\label{rw:subsec:operator-convex}

For the data-processing results, the natural class consists of
operator-convex functions.  A finite operator-convex function
$f:[0,\infty)\to\mathbb R$ admits the representation
\cite[Theorem~8.1]{hiai2011quantum}
\begin{equation}
 f(x)=f(0)+ax+bx^2
 +\int_{(0,\infty)}\left(
    \frac{x}{1+s}-\frac{x}{x+s}\right)\,\mathrm d\mu_f(s),
 \qquad b\geq0,
 \label{rw:eq:operator-convex-representation}
\end{equation}
where $a\in\mathbb R$, $\mu_f$ is a positive measure, and
$\int(1+s)^{-2}\,\mathrm d\mu_f(s)<\infty$.  Conversely, such a
representation defines a finite operator-convex function.  Singular
operator-convex functions on $(0,\infty)$ are handled by truncation
or regularization on the relevant support.  For normalized states,
adding $a(x-1)$ does not change the divergence, but it also leaves
$f(1)$ unchanged.  To impose the normalization $f(1)=0$, one instead
replaces $f$ by
\begin{equation*}
 f_0(x):=f(x)-f(1),
 \qquad
 D_{f_0}^t(\rho\Vert\sigma)
 =D_f^t(\rho\Vert\sigma)-f(1).
\end{equation*}
This constant shift leaves every DPI deficit unchanged.  These two
forms of affine freedom are made explicit in
\eqref{rw:eq:affine-generator-freedom}; see also
\cite{hiai2011quantum,hiai2017different}.

\section{Geodesic divergences and their special cases}
\label{rw:sec:interpolated-divergences}

\subsection{Definition and endpoints}
\label{rw:subsec:definition-interpolation}

The Umegaki relative entropy of two states $\rho$ and $\sigma$ can be defined in terms of their relative modular operator as 
\begin{equation}
    D(\rho \| \sigma):= \big\langle\rho^{1/2},
 \log(\Delta_{\rho\mid\sigma})\rho^{1/2}
 \big\rangle_{\rm HS}.
\end{equation}
If instead we use the positive-definite commutant derivative
$\widehat\Delta_{\rho\mid\sigma}$ from
\eqref{rw:eq:commutant-derivative}, then the
Belavkin--Staszewski relative entropy has the modular representation
\begin{equation}
 \begin{aligned}
 \widehat D(\rho\Vert\sigma)
 &=
 \big\langle\sigma^{1/2},
 \bigl(\widehat\Delta_{\rho\mid\sigma}
 \log\bigl(\widehat\Delta_{\rho\mid\sigma}\bigr)\bigr)
 [\sigma^{1/2}]\big\rangle_{\rm HS}\\
 &=\operatorname{Tr}\!\left[
 \sigma K\log K\right],
 \qquad K:=\sigma^{-1/2}\rho\sigma^{-1/2}.
 \end{aligned}
\end{equation}
Therefore, in order to construct our geodesic relative entropies, we first define the following geodesic modular operators, whose basic properties we collect in the following result.

\begin{proposition}[The geodesic modular operator]
\label{rw:prop:gamma-formula}
For invertible states $\rho,\sigma$ and $0\leq t\leq1$, we define
\begin{equation}
 \Gamma_t(\rho,\sigma)
 :=\Delta_{\rho\mid\sigma}\#_t
   \widehat\Delta_{\rho\mid\sigma}.
 \label{rw:eq:gamma-geodesic-definition}
\end{equation}
Then
\begin{equation}
 \Gamma_t(\rho,\sigma)
 =L_{\rho^{1-t}}R_{B_t(\rho,\sigma)},
 \qquad
 B_t(\rho,\sigma):=\sigma^{-1/2}\rho^t\sigma^{-1/2}>0.
 \label{rw:eq:gamma-explicit}
\end{equation}
When the pair $(\rho,\sigma)$ is fixed, we abbreviate
$B_t(\rho,\sigma)$ to $B_t$.  Thus $B_t$ acts on the original Hilbert
space, while
$L_{\rho^{1-t}}R_{B_t}$ acts on the Hilbert--Schmidt space.  More
explicitly, for every $X\in\mathcal B(\mathcal H)$,
\begin{equation}
 \Gamma_t(\rho,\sigma)[X]
 =\rho^{1-t}X\sigma^{-1/2}\rho^t\sigma^{-1/2}.
 \label{rw:eq:gamma-action}
\end{equation}
In particular, the endpoints are
\begin{equation}
 \Gamma_0(\rho,\sigma)=\Delta_{\rho\mid\sigma},
 \qquad
 \Gamma_1(\rho,\sigma)=\widehat\Delta_{\rho\mid\sigma},
 \label{rw:eq:gamma-endpoints}
\end{equation}
and $\sigma^{1/2}$ satisfies
\begin{equation}
 \Gamma_t(\rho,\sigma)[\sigma^{1/2}]=\rho\sigma^{-1/2},
 \qquad
 \big\langle\sigma^{1/2},
 \Gamma_t(\rho,\sigma)[\sigma^{1/2}]\big\rangle_{\rm HS}=1.
 \label{rw:eq:gamma-normalization}
\end{equation}
\end{proposition}

\begin{proof}
We put
\begin{equation}
 K:=\sigma^{-1/2}\rho\sigma^{-1/2},
 \qquad
 \widehat\Delta_{\rho\mid\sigma}=R_K.
\end{equation}
The left and right multiplication operators commute, so their
functional calculi give
\begin{equation}
 \Delta_{\rho\mid\sigma}^{1/2}
 =L_{\rho^{1/2}}R_{\sigma^{-1/2}},
 \qquad
 \Delta_{\rho\mid\sigma}^{-1/2}
 =L_{\rho^{-1/2}}R_{\sigma^{1/2}}.
 \label{rw:eq:delta-half-powers}
\end{equation}
Using $ R_aR_b=R_{ba}$ together with $L_aR_b=R_bL_a$ for any operators $a,b$, we obtain
\begin{equation}
\Delta_{\rho\mid\sigma}^{-1/2}
  \widehat\Delta_{\rho\mid\sigma}
  \Delta_{\rho\mid\sigma}^{-1/2}=
 (L_{\rho^{-1/2}}R_{\sigma^{1/2}})
 R_K
(L_{\rho^{-1/2}}R_{\sigma^{1/2}})=L_{\rho^{-1}}
 R_{\sigma^{1/2}K\sigma^{1/2}}
 =L_{\rho^{-1}}R_\rho, 
\label{rw:eq:gamma-conjugated-endpoint}
\end{equation}
and by functional calculus
\begin{equation}
 (L_{\rho^{-1}}R_\rho)^t
 =L_{\rho^{-t}}R_{\rho^t}.
 \label{rw:eq:gamma-conjugated-power}
\end{equation}
Substituting \eqref{rw:eq:gamma-conjugated-power} and
\eqref{rw:eq:delta-half-powers} into the definition of the weighted
geometric mean proves \eqref{rw:eq:gamma-explicit}, and applying the last
superoperator to an arbitrary $X$ proves \eqref{rw:eq:gamma-action}.
Setting $t=0$ and $t=1$ gives \eqref{rw:eq:gamma-endpoints}. \eqref{rw:eq:gamma-normalization} is a particular case of 
\eqref{rw:eq:gamma-action}.
\end{proof}

The previous proposition motivates the definition of the main objects of this paper. 

\refstepcounter{definition}
\label{rw:def:interpolated-f-divergence}
\begin{mainobject}{Definition~\thedefinition: Geodesic
$f$-divergence}
For invertible states $\rho,\sigma$, $t\in[0,1]$, and a real
function $f$ defined on the spectrum of $\Gamma_t$, we define the \emph{geodesic $f$-divergence} of $\rho$ and $\sigma$ with parameter $t$ by
\begin{equation}
 D_f^t(\rho\Vert\sigma)
 :=\big\langle\sigma^{1/2},
 f(\Gamma_t(\rho,\sigma))[\sigma^{1/2}]
 \big\rangle_{\rm HS}.
 \label{rw:eq:interpolated-f-divergence}
\end{equation}
When $\rho$ and $\sigma$ are fixed, this is called in short their \textit{$(t,f)$-divergence}.
\end{mainobject}

At $t=0$, since $\Gamma_0=\Delta_{\rho\mid\sigma}$ we recover the standard $f$-divergences,
\begin{equation}
 D_f^0(\rho\Vert\sigma)=D_f(\rho\Vert\sigma),
 \label{rw:eq:standard-endpoint}
\end{equation}
whereas at $t=1$, $\Gamma_1=\widehat\Delta_{\rho\mid\sigma}$ gives
\begin{equation}
 D_f^1(\rho\Vert\sigma)
 =\operatorname{Tr}\!\left[
 \sigma f(\sigma^{-1/2}\rho\sigma^{-1/2})\right]
 =:\widehat D_f(\rho\Vert\sigma),
 \label{rw:eq:maximal-endpoint}
\end{equation}
the maximal quantum $f$-divergence.  If $[\rho,\sigma]=0$, then
$\sigma^{1/2}$ is supported on a joint eigenbasis sector on which
$\Gamma_t$ has eigenvalues $\rho_i/\sigma_i$, independently of
$t$.  Consequently,
\begin{equation}
 D_f^t(\rho\Vert\sigma)
 =\sum_i\sigma_i f(\rho_i/\sigma_i),
 \qquad [\rho,\sigma]=0,
 \label{rw:eq:classical-collapse}
\end{equation}
for every interpolation parameter. In order to gain a better understanding of this quantity, in the next subsections we particularize this definition to some cases in which the endpoints are very well-known objects, such as relative entropies or Rényi divergences.

\subsection{Geodesic relative entropies: Umegaki to Belavkin--Staszewski}
\label{rw:subsec:relative-special-case}
 We consider now the particular case of $f(x)=x \log x$, in which the endpoints reduce to the Umegaki, respectively Belavkin-Staszewski, relative entropy, as described in \Cref{rw:subsec:standard-maximal-f-divergences}.

\refstepcounter{definition}
\begin{mainobject}{Definition~\thedefinition: Geodesic
relative entropy}
For invertible states $\rho,\sigma$ and $t\in[0,1]$, we define the \emph{geodesic relative entropy} of $\rho$ and $\sigma$ with parameter $t$ by
\begin{equation}
 D^t(\rho\Vert\sigma):=D_{x\log x}^t(\rho\Vert\sigma).
 \label{rw:eq:Dt-and-Kt}
\end{equation}
When $\rho$ and $\sigma$ are fixed, we call this in short their \textit{$t$-relative entropy}.
\end{mainobject}

The following proposition contains an explicit, simple formula for the t-relative entropies for any $t \in [0,1]$ that allows to easily verify that its endpoints are the Umegaki and Belavkin-Staszewski relative entropy.

\begin{proposition}[Explicit logarithmic interpolation]
\label{rw:prop:Dt-explicit}
For every invertible pair $\rho$ and $\sigma$ and every $t\in[0,1]$,
\begin{equation}
 D^t(\rho\Vert\sigma)
 =(1-t)\operatorname{Tr}(\rho\log\rho)
 +\operatorname{Tr}\!\left[
 \rho\log\!\left(\rho^{t/2}\sigma^{-1}\rho^{t/2}\right)
 \right].
 \label{rw:eq:Dt-explicit}
\end{equation}
In particular,
\begin{align}
 D^0(\rho\Vert\sigma)
 &=\operatorname{Tr}\rho(\log\rho-\log\sigma)
 =:D(\rho\Vert\sigma),
 \label{rw:eq:umegaki-endpoint}\\
 D^1(\rho\Vert\sigma)
 &=\operatorname{Tr}\rho\log(\rho^{1/2}\sigma^{-1}\rho^{1/2})
 =:\widehat D(\rho\Vert\sigma).
 \label{rw:eq:bs-endpoint}
\end{align}
\end{proposition}

\begin{proof}
By \eqref{rw:eq:log-left-right} and
\eqref{rw:eq:gamma-explicit},
\begin{equation}
 \log\Gamma_t
 =L_{(1-t)\log\rho}+R_{\log B_t}.
\end{equation}
The left-multiplication term in
$\langle\sigma^{1/2},\Gamma_t\log\Gamma_t
\sigma^{1/2}\rangle_{\rm HS}$ equals
$(1-t)\operatorname{Tr}(\rho\log\rho)$.  For the other term, we put
$X=\sigma^{-1/2}\rho^{t/2}$, so that
$B_t=XX^*$ and $X^*X=\rho^{t/2}\sigma^{-1}\rho^{t/2}$.  The identity
$X^*\log(XX^*)=\log(X^*X)X^*$ gives
\begin{equation}
 \sigma^{1/2}B_t\log(B_t)\sigma^{1/2}
 =\rho^{t/2}
 \log\!\left(\rho^{t/2}\sigma^{-1}\rho^{t/2}\right)
 \rho^{t/2}.
\end{equation}
Cyclicity now yields the second term in
\eqref{rw:eq:Dt-explicit}.  The endpoint identities follow directly;
the last expression in \eqref{rw:eq:bs-endpoint} is equivalent to
$\operatorname{Tr}[\sigma C\log C]$, for
$C=\sigma^{-1/2}\rho\sigma^{-1/2}$, by the same push-through
identity \cite{fujii1989relative}.
\end{proof}

\begin{remark}[Extension to unequal supports]
Suppose that
$\operatorname{supp}\rho\leq\operatorname{supp}\sigma$.
Working on $\operatorname{supp}\sigma$ and applying
\eqref{rw:eq:Dt-explicit} to
$\rho_\varepsilon=(1-\varepsilon)\rho+\varepsilon\sigma$ gives, as
$\varepsilon\downarrow0$, the same formula with $0\log0=0$ and
$\sigma^{-1}$ interpreted as a support inverse.  For $t>0$, the
logarithm in the second term is taken on
$\operatorname{supp}\rho$; at $t=0$, the limit is the usual supported
Umegaki expression.  If
$\operatorname{supp}\rho\nleq\operatorname{supp}\sigma$, the
lower-semicontinuous extension is $D^t(\rho\Vert\sigma)=+\infty$.
\end{remark}

\subsection{Geodesic Rényi divergences: Petz to geometric}
\label{rw:subsec:renyi-special-case}

The next significant particular case to consider is that of $f(x)=x^\alpha$, for $\alpha >0$. 

\refstepcounter{definition}
\begin{mainobject}{Definition~\thedefinition: Geodesic R\'enyi divergence}
For invertible states $\rho,\sigma$, $t\in[0,1]$, and $\alpha>0$,
we define the unnormalized power functional
\begin{equation}
 Q_{\alpha,t}(\rho\Vert\sigma)
 :=D_{x^\alpha}^t(\rho\Vert\sigma)
 =\big\langle\sigma^{1/2},
 \Gamma_t^\alpha\sigma^{1/2}\big\rangle_{\rm HS}=\operatorname{Tr}\!\left[
 \sigma^{1/2}\rho^{\alpha(1-t)}\sigma^{1/2}B_t^\alpha
 \right],
 \label{rw:eq:power-functional}
\end{equation}
since the two factors in $\Gamma_t=L_{\rho^{1-t}}R_{B_t}$
commute. For $\alpha\neq1$, the \textit{normalized geodesic Rényi divergence} of $\rho$ and $\sigma$ with parameter $t$ is
\begin{equation}
 D_\alpha^t(\rho\Vert\sigma)
 :=\frac{1}{\alpha-1}\log Q_{\alpha,t}(\rho\Vert\sigma).
 \label{rw:eq:renyi-interpolation}
\end{equation}
When $\rho$ and $\sigma$ are fixed, this is called in short their \textit{$(t,\alpha)$-Rényi divergence}.
\end{mainobject}

At the endpoints,
\begin{align}
 D_\alpha^0(\rho\Vert\sigma)
 &=\frac{1}{\alpha-1}
   \log\operatorname{Tr}(\rho^\alpha\sigma^{1-\alpha}),
 \label{rw:eq:petz-renyi-endpoint}\\
 D_\alpha^1(\rho\Vert\sigma)
 &=\frac{1}{\alpha-1}
   \log\operatorname{Tr}\!\left[
    \sigma(\sigma^{-1/2}\rho\sigma^{-1/2})^\alpha\right],
 \label{rw:eq:geometric-renyi-endpoint}
\end{align}
namely the Petz and geometric (maximal) Rényi divergences.  Since
$Q_{1,t}=1$, differentiation in $\alpha$ at 1 shows that
\begin{equation}
 \lim_{\alpha\to1}D_\alpha^t(\rho\Vert\sigma)
 =D^t(\rho\Vert\sigma),
 \label{rw:eq:renyi-alpha-one-limit}
\end{equation}
recovering the geodesic relative entropy with parameter $t$. The definition \eqref{rw:eq:renyi-interpolation} makes sense for all
positive $\alpha\neq1$; a data-processing inequality will only be asserted in
the operator-convex range $0<\alpha<1$ and $1<\alpha\leq2$.

\begin{example}[The complete Petz--geometric R\'enyi band]
\label{rw:ex:renyi-interpolation-band}
We consider the invertible, noncommuting qubit states
\begin{equation}
 \rho=
 \begin{pmatrix}
  \frac12&\frac9{20}\\[1mm]
  \frac9{20}&\frac12
 \end{pmatrix},
 \qquad
 \sigma=
 \begin{pmatrix}
  \frac9{10}&0\\[1mm]
  0&\frac1{10}
 \end{pmatrix}.
 \label{rw:eq:renyi-band-states}
\end{equation}
Indeed, the eigenvalues of $\rho$ are $19/20$ and $1/20$, and
\begin{equation}
 [\rho,\sigma]
 =\begin{pmatrix}0&-9/25\\ 9/25&0\end{pmatrix}\neq0.
\end{equation}
For this pair,
\Cref{rw:thm:Dt-strict-monotonicity,rw:thm:power-t-monotonicity}
give
\begin{equation}
 D_\alpha^0(\rho\Vert\sigma)
 \leq D_\alpha^t(\rho\Vert\sigma)
 \leq D_\alpha^1(\rho\Vert\sigma),
 \qquad 0<\alpha<2,\quad 0\leq t\leq1,
 \label{rw:eq:renyi-band-order}
\end{equation}
where $D_1^t:=D^t$.  Since the path is continuous in $t$, its
pointwise range is the entire interval
\begin{equation}
 \bigl\{D_\alpha^t(\rho\Vert\sigma):0\leq t\leq1\bigr\}
 =\bigl[D_\alpha^0(\rho\Vert\sigma),
        D_\alpha^1(\rho\Vert\sigma)\bigr].
 \label{rw:eq:renyi-band-range}
\end{equation}
Thus the shaded region in \Cref{rw:fig:renyi-interpolation-band}
is the exact continuum of geodesic divergences.  At the two ends of the plotted order
interval,
\begin{equation}
 D_0^t(\rho\Vert\sigma):=\lim_{\alpha\downarrow0}D_\alpha^t
 =0,
 \qquad
 D_2^t(\rho\Vert\sigma)
 =\log\operatorname{Tr}(\rho^2\sigma^{-1})
 =\log\frac{181}{36},
 \label{rw:eq:renyi-band-meeting-points}
\end{equation}
so all the curves meet at both $\alpha=0$ and $\alpha=2$.
\end{example}

\begin{figure}[htbp]
 \centering
 \includegraphics[width=0.96\linewidth]{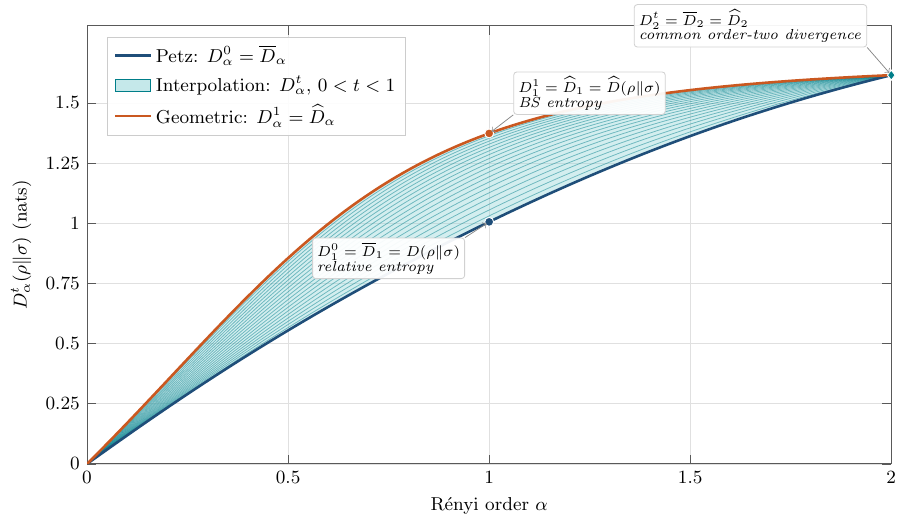}
 \caption{R\'enyi interpolation for the states in
 \eqref{rw:eq:renyi-band-states}.
 The dark blue and orange boundary curves are, respectively, the Petz-Rényi
 divergences $D_\alpha^0=\overline D_\alpha$ and the geometric Rényi
 divergence $D_\alpha^1=\widehat D_\alpha$.
 The contrasting teal curves show $D_\alpha^t$ for
 $t=k/20$, $k=1,\ldots,19$, and the light teal band represents the
 full continuum $0\leq t\leq1$.  The values at $\alpha=0$ and
 $\alpha=1$ are the continuous extensions $D_0^t=0$ and
 $D_1^t=D^t$.  The marked order-one points are the relative and BS
 entropies, and the order-two marker is their common R\'enyi value .}
 \label{rw:fig:renyi-interpolation-band}
\end{figure}

\subsection{The quadratic fixed point of the interpolation}
\label{rw:subsec:chi-square}

As mentioned above, the order-two power functional does not see the geodesic parameter.
Indeed,
\begin{equation}
 Q_{2,t}(\rho\Vert\sigma)
 =\operatorname{Tr}(\rho^2\sigma^{-1}),
 \qquad 0\leq t\leq1.
 \label{rw:eq:Q2-constant}
\end{equation}
Consequently, for $f_\chi(x)=(x-1)^2$,
\begin{align}
 D_{f_\chi}^t(\rho\Vert\sigma)
 &=\operatorname{Tr}(\rho^2\sigma^{-1})-1
 =\operatorname{Tr}\!\left[
 (\rho-\sigma)\sigma^{-1}(\rho-\sigma)\right]
 =:\chi^2(\rho\Vert\sigma),
 \label{rw:eq:chi-square-constant}
\end{align}
independently of $t$, and hence
\begin{equation}
D_2^t(\rho\Vert\sigma)=\log(1+\chi^2(\rho\Vert\sigma)).
 \label{rw:eq:renyi-two-chi-square}
\end{equation}
This exceptional quadratic behavior is important later: a pure
quadratic generator cannot distinguish the nonmaximal Petz equality
class from the maximal BS equality class.

\section{Basic properties and relations}
\label{rw:sec:basic-properties}

Beyond their definitions and distinguished examples, the geodesic
divergences admit several elementary comparisons and structural rules.
We first locate the R\'enyi path relative to the standard
two-parameter family, and then collect the covariance, tensor,
positivity, and order properties used later.

\subsection{Comparison with the
\texorpdfstring{$(\alpha,z)$}{(alpha,z)} family}
\label{rw:subsec:alpha-z-comparison}

To locate the present interpolation among the standard two-parameter
R\'enyi divergences, we write
\begin{equation}
 D_{\alpha,z}(\rho\Vert\sigma)
 :=\frac{1}{\alpha-1}\log Q_{\alpha,z}(\rho\Vert\sigma),
 \qquad z>0,
 \label{rw:eq:alpha-z-divergence}
\end{equation}
where
\begin{equation}
 Q_{\alpha,z}(\rho\Vert\sigma)
 :=\operatorname{Tr}\!\left(
  \rho^{\alpha/(2z)}
  \sigma^{(1-\alpha)/z}
  \rho^{\alpha/(2z)}
 \right)^z .
 \label{rw:eq:alpha-z-functional}
\end{equation}
This is the $\alpha$--$z$ family of
\cite[Eqs.~(1)--(3)]{audenaert2015alphaz}.  Thus $z=1$ is the Petz
divergence and $z=\alpha$ is the sandwiched divergence
\cite[Eqs.~(5)--(6)]{audenaert2015alphaz}.  We distinguish
the two-parameter family $D_{\alpha,z}$ from the geodesic family
$D_\alpha^t$ by the position and number of its parameters.

For invertible inputs, the explicit operator in
\eqref{rw:eq:gamma-explicit} is positive for every $t\in\mathbb R$.
Thus \eqref{rw:eq:power-functional} and \eqref{rw:eq:renyi-interpolation} provide
an algebraic continuation of $Q_{\alpha,t}$ and $D_\alpha^t$ to all
real $t$.  The push-through identity gives the useful formula
\begin{equation}
 Q_{\alpha,t}
 =\operatorname{Tr}\!\left[
  \rho^{\alpha-(\alpha-1)t}
  \bigl(\rho^{t/2}\sigma^{-1}\rho^{t/2}\bigr)^{\alpha-1}
 \right],
 \qquad t\in\mathbb R.
 \label{rw:eq:power-functional-push-through}
\end{equation}
Outside $[0,1]$ this is a geodesic extrapolation only; no
nonnegativity or data-processing assertion for the resulting
divergence is implicit.

\begin{proposition}[Exact universal intersections for real $t$]
\label{rw:prop:alpha-z-intersections}
We assume $\dim\mathcal H\geq2$ and fix $\alpha>0$, $\alpha\neq1$,
$z>0$, and $t\in\mathbb R$.  Then
\begin{equation}
 D_\alpha^t(\rho\Vert\sigma)
 =D_{\alpha,z}(\rho\Vert\sigma)
 \quad\text{for every invertible pair }(\rho,\sigma)
 \label{rw:eq:alpha-z-universal-equality}
\end{equation}
if and only if at least one of the following conditions holds:
\begin{equation}
 \begin{aligned}
 &t=0,\quad z=1;\qquad
 &&\alpha=2,\quad z=1;\\
 &t=\frac{\alpha}{\alpha-1},\quad
   z=\lvert\alpha-1\rvert;\qquad
 &&(\alpha,t,z)=(3,3,1).
 \end{aligned}
 \label{rw:eq:alpha-z-intersection-locus}
\end{equation}
The alternatives form a union and may overlap.  Equivalently, the
universal identities are the Petz line $D_\alpha^0=D_{\alpha,1}$,
the full real order-two fibre
\begin{equation}
 D_2^t=D_{2,1}
 =\log\operatorname{Tr}(\rho^2\sigma^{-1}),
 \qquad t\in\mathbb R,
 \label{rw:eq:alpha-z-order-two-fibre}
\end{equation}
the extrapolated branch
\begin{equation}
 D_\alpha^{\,\alpha/(\alpha-1)}
 =D_{\alpha,\lvert\alpha-1\rvert},
 \label{rw:eq:alpha-z-extrapolated-intersection}
\end{equation}
and the isolated cubic coincidence
\begin{equation}
 D_3^3=D_{3,1}.
 \label{rw:eq:alpha-z-cubic-intersection}
\end{equation}
\end{proposition}

\begin{proof}
The Petz identity follows from
\eqref{rw:eq:petz-renyi-endpoint}.  We put $\beta:=\alpha-1$.  From
\eqref{rw:eq:power-functional-push-through}, cyclicity gives
\begin{equation}
 Q_{2,t}=\operatorname{Tr}(\rho^2\sigma^{-1}),
 \qquad t\in\mathbb R,
 \label{rw:eq:real-t-quadratic-identity}
\end{equation}
which proves the quadratic fibre.  At
$t_\star:=\alpha/\beta$, the exterior power of $\rho$ in
\eqref{rw:eq:power-functional-push-through} disappears.  If
$\beta>0$, the resulting trace is the $(\alpha,z)$ functional at
$z=\beta$; if $\beta<0$, we use
\begin{equation}
 \bigl(\rho^{t_\star/2}\sigma^{-1}
       \rho^{t_\star/2}\bigr)^{-1}
 =\rho^{-t_\star/2}\sigma\rho^{-t_\star/2}
 \label{rw:eq:negative-branch-inversion}
\end{equation}
to obtain the $(\alpha,z)$ functional at $z=-\beta$.  This proves
\eqref{rw:eq:alpha-z-extrapolated-intersection}.  Finally, for the
real-$t$ continuation,
\begin{align}
 Q_{3,t}
 &=\operatorname{Tr}\!\left(
   \sigma^{-1}\rho^{3-t}\sigma^{-1}\rho^t\right)\notag\\
 &=\left\|\rho^{(3-t)/2}\sigma^{-1}\rho^{t/2}\right\|_2^2
 =\left\|\rho^{t/2}\sigma^{-1}\rho^{(3-t)/2}\right\|_2^2
 =Q_{3,3-t},
 \label{rw:eq:cubic-reflection}
\end{align}
Thus $Q_{3,3}=Q_{3,0}=\operatorname{Tr}(\rho^3\sigma^{-2})$, which
is precisely the $(\alpha,z)$ functional at $(\alpha,z)=(3,1)$.

We prove exhaustiveness.  A two-dimensional calculation is enough.
Indeed, if an invertible qubit pair $(\rho_2,\sigma_2)$ gives a mismatch,
then for any invertible state $\tau$ on the complementary subspace and
$0<p<1$, we set
\begin{equation}
 \rho=p\rho_2\oplus(1-p)\tau,
 \qquad
 \sigma=p\sigma_2\oplus(1-p)\tau.
 \label{rw:eq:real-t-qubit-embedding}
\end{equation}
Each power functional is $p$ times its qubit value plus $1-p$, so
the mismatch persists in every dimension.  We take
\begin{equation}
 \rho_x:=\frac{1}{1+x}
 \begin{pmatrix}1&0\\0&x\end{pmatrix},
 \qquad
 S:=\sigma^{-1}
 =\begin{pmatrix}a&c\\ \overline c&b\end{pmatrix}>0,
 \label{rw:eq:real-t-rank-one-witness}
\end{equation}
and $P=\lvert0\rangle\!\langle0\rvert$.  We set
$m_r:=\langle0\rvert S^r\lvert0\rangle$.  As $x\downarrow0$,
\begin{equation}
 Q_{\alpha,z}(\rho_x\Vert\sigma)
 \longrightarrow m_{\beta/z}^{\,z}.
 \label{rw:eq:real-t-alpha-z-pure-limit}
\end{equation}
To record the relevant spectral calculus, for $t>0$ and up to a
scalar tending to one,
\begin{equation}
 K_x:=\rho_x^{t/2}S\rho_x^{t/2}
 =\begin{pmatrix}
   a&cx^{t/2}\\ \overline c x^{t/2}&bx^t
  \end{pmatrix}.
 \label{rw:eq:real-t-spectral-block}
\end{equation}
Its eigenvalues and the lower diagonal entry of the spectral
projection $E_+$ of the larger eigenvalue satisfy
\begin{align}
 \lambda_+(K_x)=a+O(x^t), \qquad
 \lambda_-(K_x)=\frac{ab-\lvert c\rvert^2}{a}x^t+O(x^{2t}), \qquad
 \langle1\rvert E_+\lvert1\rangle
 =\frac{\lvert c\rvert^2}{a^2}x^t+O(x^{2t}).
 \label{rw:eq:real-t-spectral-asymptotics}
\end{align}
In \eqref{rw:eq:power-functional-push-through}, the latter term is
weighted by $x^{\alpha-\beta t}$, producing the competing exponent
$\kappa:=\alpha+t(2-\alpha)$.  At $\kappa=0$ its extra contribution
is $\lvert c\rvert^2a^{\beta-2}$, so together with $a^\beta$ it
equals $m_1^{\beta-2}m_2$.  Applying the same calculation to
$K_x^{-1}=\rho_x^{-t/2}\sigma\rho_x^{-t/2}$ yields the negative-$t$
cases.  Consequently, for $\alpha\neq2$,
\begin{equation}
 Q_{\alpha,t}(\rho_x\Vert\sigma)\longrightarrow
 \begin{cases}
  m_\beta,&t=0,\\[1mm]
  m_{-1}^{-\beta},&t<0,\ 0<\alpha<2,\\[1mm]
  m_1^\beta,&t>0,\ \kappa>0,\\[1mm]
  m_1^{\beta-2}m_2,&t>0,\ \kappa=0.
 \end{cases}
 \label{rw:eq:real-t-pure-limits}
\end{equation}
For $c\neq0$ the same quantity diverges when $t<0$, $\alpha>2$,
or when $t>0$, $\kappa<0$.

We let $M_r:=m_r^{1/r}$ be the scalar power mean of the spectral values
of $S$ in the vector state $P$.  Since
$m_{\beta/z}^z=M_{\beta/z}^{\beta}$, strict monotonicity of power
means in \eqref{rw:eq:real-t-alpha-z-pure-limit} and \eqref{rw:eq:real-t-pure-limits}
forces
\begin{equation}
 \begin{array}{c|c}
  t=0&z=1\\
  t<0&\beta<0,\ z=-\beta\\
  t>0,\ \kappa>0&\beta>0,\ z=\beta.
 \end{array}
 \label{rw:eq:real-t-power-mean-candidates}
\end{equation}
At $\alpha=2$, \eqref{rw:eq:real-t-quadratic-identity} and the first
line of the same power-mean argument force $z=1$.

It remains to analyze the positive-$t$ threshold $\kappa=0$, which
can occur only for $\alpha>2$ and
$t=\alpha/(\alpha-2)$.  We put $p:=\beta/z$ and take
$S_u=\left(\begin{smallmatrix}1&u\\u&1\end{smallmatrix}\right)$,
$\lvert u\rvert<1$.  Normalizing $S_u^{-1}$ to a state only
multiplies both sides below by the same positive scalar.  Universal
equality would imply
\begin{equation}
 \left(
  \frac{(1+u)^p+(1-u)^p}{2}
 \right)^{\beta/p}=1+u^2.
 \label{rw:eq:real-t-threshold-identity}
\end{equation}
The coefficient of $u^2$ gives $\beta(p-1)=2$; after substitution,
the coefficient of $u^4$ gives
$(\beta-2)(2\beta+1)=0$.  Since $\beta>1$, this forces
$\beta=p=2$.  Hence the threshold contributes exactly
$(\alpha,t,z)=(3,3,1)$.

We finally test the two regular candidates with
$z=\lvert\beta\rvert$.  We put
\begin{equation}
 \delta:=t_\star-t,\qquad
 A:=\rho^\delta,\qquad
 B:=\rho^{t/2}\sigma^{-1}\rho^{t/2}.
 \label{rw:eq:real-t-ALT-variables}
\end{equation}
Then
\begin{equation}
 Q_{\alpha,t}=\operatorname{Tr}(A^\beta B^\beta),
 \qquad
 \left.Q_{\alpha,z}\right|_{z=\lvert\beta\rvert}
 =\operatorname{Tr}(A^{1/2}BA^{1/2})^\beta.
 \label{rw:eq:alpha-z-ALT-comparison}
\end{equation}
If $\delta=0$, this is precisely the extrapolated branch.  If
$\delta\neq0$, arbitrary positive definite $A$ and $B$ can be
realized up to positive scalar factors by taking
$\rho\propto A^{1/\delta}$ and
$\sigma\propto\rho^{t/2}B^{-1}\rho^{t/2}$.
For $\beta>0$, we choose $A=\operatorname{diag}(a,1)$, $a>0$,
$a\neq1$, and
$B_\varepsilon=(1-\varepsilon)P_++\varepsilon P_-$, where
$P_\pm$ project onto $2^{-1/2}(1,\pm1)^{\mathsf T}$.  The two traces
in \eqref{rw:eq:alpha-z-ALT-comparison} tend to
\begin{equation}
 \frac{a^\beta+1}{2}
 \quad\text{and}\quad
 \left(\frac{a+1}{2}\right)^\beta,
 \label{rw:eq:alpha-z-strict-witness}
\end{equation}
which are unequal unless $\beta=1$.  If $\beta=-q<0$, multiplying
the two traces by $\varepsilon^q$ instead gives the limits
\begin{equation}
 \frac{a^{-q}+1}{2}
 \quad\text{and}\quad
 \left(\frac{a+1}{2a}\right)^q,
 \label{rw:eq:alpha-z-negative-strict-witness}
\end{equation}
which are strictly unequal for $0<q<1$ by concavity of $x^q$.
Thus $\delta=0$, except when $\beta=1$; that exception is exactly the
quadratic fibre already found.  This proves that the list is complete.
\end{proof}

\begin{remark}[Consequences and boundary cases]
\label{rw:rem:alpha-z-boundaries}
The universal classification is unchanged if ``invertible pair'' is
replaced by ``pair with common support'', by compression to that
support.  For unequal supports, the real-$t$ extrapolation can contain
negative powers, so no blanket boundary extension is asserted; the
items below record only the particular limiting statements that are
needed.
The relevant cases are as follows.
\begin{itemize}[leftmargin=1.5em,itemsep=2pt,parsep=0pt,topsep=4pt]
 \item \emph{The range $0<\alpha<1$.}
 Here $t_\star=\alpha/(\alpha-1)<0$ and $z=1-\alpha$.
 Consequently, \eqref{rw:eq:alpha-z-extrapolated-intersection}
 identifies the negative-$t$ continuation with the
 reverse-sandwiched branch.  Equivalently, it can be read using the
 evenness in $z$, which is valid here because the inputs are invertible
 \cite[p.~4, Remark~1]{audenaert2015alphaz}.  In particular,
 $(\alpha,z)=(1/2,1/2)$ corresponds to $t=-1$.

 \item \emph{The range $\alpha>1$.}
 Here $t_\star=\alpha/(\alpha-1)>1$ and $z=\alpha-1$.
 Thus the positive continuation meets a second branch of the
 $(\alpha,z)$ family outside the interpolation segment.

 \item \emph{Quadratic and cubic orders.}
 At $\alpha=2$, \eqref{rw:eq:alpha-z-order-two-fibre} holds for every
 real $t$.  At $\alpha=3$, the reflection in
 \eqref{rw:eq:cubic-reflection} maps the Petz point $t=0$ to the
 additional point $t=3$, yielding
 \eqref{rw:eq:alpha-z-cubic-intersection}.  The latter coincidence is
 isolated.

 \item \emph{Order one.}
 Every fixed $z>0$---more generally, every $C^1$ path extending to
 $z(1)=z_0>0$---has the Umegaki limit
 $D_{\alpha,z(\alpha)}\to D$ \cite[Theorem~4]{lin2015investigating}, whereas
 $D_\alpha^t\to D^t$ by \eqref{rw:eq:renyi-alpha-one-limit}; its
 differentiation argument applies verbatim to every fixed real
 $t$.  For
 $\rho_x\to P_\psi:=\lvert\psi\rangle\!\langle\psi\rvert$, the
 real-$t$ continuation satisfies
 \begin{equation}
  D^t(\rho_x\Vert\sigma)\longrightarrow
  \begin{cases}
   \log\langle\psi\rvert\sigma^{-1}\lvert\psi\rangle,&t>0,\\
   -\langle\psi\rvert(\log\sigma)\lvert\psi\rangle,&t=0,\\
   -\log\langle\psi\rvert\sigma\lvert\psi\rangle,&t<0.
  \end{cases}
  \label{rw:eq:real-t-order-one-pure-limits}
 \end{equation}
 Strict Jensen shows that $D^t=D$ universally only for $t=0$.  The
 $t=0$ member has the same order-one limit for every fixed $z$, but
 coincides as a family only when $z=1$.  Singular paths with
 $z(\alpha)\to0$ can be path dependent
 \cite[Theorems~2--3]{audenaert2015alphaz} and are not additional finite-parameter
 intersections.

 \item \emph{Order zero.}
 If $\alpha=0$ is adjoined by continuous extension from invertible states,
 every real $t$ and every $z>0$ gives zero.

 \item \emph{Commuting and statewise coincidences.}
 If $[\rho,\sigma]=0$, then every real $t$ and every $z$ gives the
 same classical R\'enyi divergence.  Such statewise coincidences, or
 accidental equality for one noncommuting pair, do not enlarge the
 universal locus in \Cref{rw:prop:alpha-z-intersections}.  In
 particular, the geometric endpoint $D_\alpha^1$ is not an
 $(\alpha,z)$ divergence in general, apart from the order-two point.
\end{itemize}
\end{remark}

For a numerical comparison on the common data-processing range, we retain
the qubit states in \eqref{rw:eq:renyi-band-states} and, for
$0<\alpha\leq2$, set
\begin{equation}
 \mathcal Z_{\mathrm{DPI}}(\alpha)
 :=\begin{cases}
 [\max\{\alpha,1-\alpha\},\infty),&0<\alpha<1,\\[1mm]
 (0,\infty),&\alpha=1,\\[1mm]
 [\alpha/2,\alpha],&1<\alpha\leq2.
 \end{cases}
 \label{rw:eq:alpha-z-dpi-region}
\end{equation}
This is the complete parameter region in this order interval for which
$D_{\alpha,z}$ satisfies data processing
\cite[Theorem~1.1]{zhang2020alphazdpi}.  The Araki--Lieb--Thirring
inequality \cite{araki1990lieb} yields that
$z\mapsto Q_{\alpha,z}$ is non-increasing
\cite[Proposition~6]{lin2015investigating}.
Consequently, the closure of the corresponding value range for the
fixed pair is
\begin{equation}
 \overline{\bigl\{D_{\alpha,z}(\rho\Vert\sigma):
       z\in\mathcal Z_{\mathrm{DPI}}(\alpha)\bigr\}}
 =\begin{cases}
 [D_{\alpha,z_{\min}}(\rho\Vert\sigma),
  D_{\alpha,\infty}(\rho\Vert\sigma)],&0<\alpha<1,\\[1mm]
 [D_{\alpha,\alpha}(\rho\Vert\sigma),
  D_{\alpha,\alpha/2}(\rho\Vert\sigma)],&1<\alpha\leq2,
 \end{cases}
 \label{rw:eq:alpha-z-dpi-value-band}
\end{equation}
where $z_{\min}=\max\{\alpha,1-\alpha\}$ and
\begin{equation}
 D_{\alpha,\infty}(\rho\Vert\sigma)
 :=\frac{1}{\alpha-1}\log\operatorname{Tr}
 \exp\!\bigl(\alpha\log\rho+(1-\alpha)\log\sigma\bigr).
 \label{rw:eq:alpha-z-log-euclidean-limit}
\end{equation}
as in \cite[Eq.~(28)]{audenaert2015alphaz}.
For $0<\alpha<1$, the log-Euclidean endpoint is the limit as
$z\to\infty$ and need not be attained at finite $z$; this is why the
closure appears in \eqref{rw:eq:alpha-z-dpi-value-band}.
At $\alpha=1$ all $(\alpha,z)$ quantities in the plot are understood
by continuous extension and equal $D(\rho\Vert\sigma)$.  To display
representative members of every vertical section without truncating
the unbounded branch, the thin violet curves in
\Cref{rw:fig:renyi-family-comparison} use $u\in\{1/4,1/2,3/4\}$ and
\begin{equation}
 z_u(\alpha):=
 \begin{cases}
  z_{\min}/(1-u),&0<\alpha<1,\\[1mm]
  \alpha(1+u)/2,&1<\alpha\leq2.
 \end{cases}
 \label{rw:eq:alpha-z-compactified-slices}
\end{equation}

\begin{figure}[!htbp]
 \centering
 \includegraphics[width=0.90\linewidth]{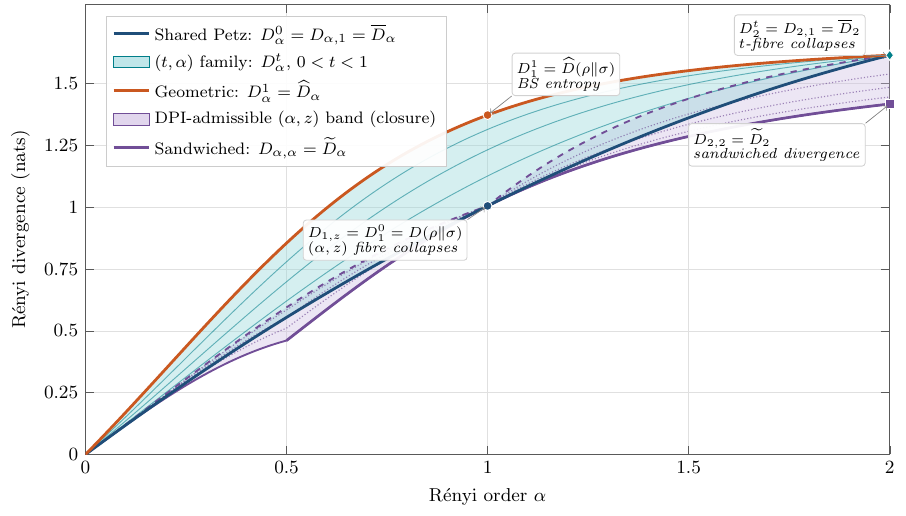}
 \caption{Comparison of the $(t,\alpha)$ and $(\alpha,z)$ R\'enyi
 families for the same noncommuting states as in
 \Cref{rw:fig:renyi-interpolation-band}.  Teal denotes the complete
 $t$-interpolation band and violet the compactified DPI-admissible
 $(\alpha,z)$ value band in
 \eqref{rw:eq:alpha-z-dpi-value-band}; thin curves sample their
 interiors.  The navy, orange, and solid-violet boundaries are,
 respectively, the shared Petz, geometric, and sandwiched
 divergences.  The labels mark the order-one Umegaki and BS limits and
 the order-two Petz and sandwiched points.  Any overlap of the shaded
 bands is specific to this fixed pair.}
 \label{rw:fig:renyi-family-comparison}
\end{figure}

\subsection{Algebraic and structural properties}
\label{rw:subsec:algebraic-properties}

The following elementary rules do not use data processing or
monotonicity in the interpolation parameter.  They also fix two
distinctions that will be useful below: a general $f$-divergence need
not be tensor additive, and a normalized R\'enyi divergence need not
be affine under classical flags.

\begin{proposition}[Covariance and tensor products]
\label{rw:prop:elementary-structural-rules}
We assume that all states below are invertible, fix $t\in[0,1]$, and
assume that every generator is defined on the relevant spectra.
\begin{enumerate}[label=\textnormal{(\roman*)}]
\item The assignment $f\mapsto D_f^t$ is linear and
\begin{equation}
 D_\ell^t(\rho\Vert\sigma)=a+b,
 \qquad \ell(x)=a+bx.
 \label{rw:eq:affine-generator-normalization}
\end{equation}
Consequently, if $g(x)=f(x)+a(x-1)+b$, then
\begin{equation}
 D_g^t(\rho\Vert\sigma)=D_f^t(\rho\Vert\sigma)+b.
 \label{rw:eq:affine-generator-freedom}
\end{equation}

\item For every unitary $U$,
\begin{equation}
 D_f^t(U\rho U^*\Vert U\sigma U^*)
 =D_f^t(\rho\Vert\sigma).
 \label{rw:eq:unitary-invariance-Dft}
\end{equation}
In particular, $D^t$, $Q_{\alpha,t}$, and $D_\alpha^t$ are
simultaneously unitarily invariant.

\item For two invertible pairs, under the canonical identification of
their Hilbert--Schmidt spaces,
\begin{equation}
 \Gamma_t(\rho_1\otimes\rho_2,\sigma_1\otimes\sigma_2)
 =\Gamma_t(\rho_1,\sigma_1)\otimes
  \Gamma_t(\rho_2,\sigma_2).
 \label{rw:eq:Gamma-tensor-product}
\end{equation}
Consequently,
\begin{align}
 Q_{\alpha,t}(\rho_1\otimes\rho_2
       \Vert\sigma_1\otimes\sigma_2)
 &=Q_{\alpha,t}(\rho_1\Vert\sigma_1)
   Q_{\alpha,t}(\rho_2\Vert\sigma_2),
 \label{rw:eq:Qalpha-tensor-product}\\
 D^t(\rho_1\otimes\rho_2
       \Vert\sigma_1\otimes\sigma_2)
 &=D^t(\rho_1\Vert\sigma_1)
   +D^t(\rho_2\Vert\sigma_2),
 \label{rw:eq:Dt-tensor-product}\\
 D_\alpha^t(\rho_1\otimes\rho_2
       \Vert\sigma_1\otimes\sigma_2)
 &=D_\alpha^t(\rho_1\Vert\sigma_1)
   +D_\alpha^t(\rho_2\Vert\sigma_2).
 \label{rw:eq:Renyi-tensor-product}
\end{align}
The last identity uses the continuous value $D_1^t=D^t$ when
$\alpha=1$.
\end{enumerate}
\end{proposition}

\begin{proof}
(i) Linearity follows directly from functional calculus.  Moreover,
$\lVert\sigma^{1/2}\rVert_2^2=1$ and
\eqref{rw:eq:gamma-normalization} give
$D_\ell^t=a+b$ for $\ell(x)=a+bx$.

\vspace{0.2cm}

\noindent (ii) We let $\mathcal C_U(X):=UXU^*$.  The explicit formula
\eqref{rw:eq:gamma-explicit} gives
\begin{equation}
 \Gamma_t(U\rho U^*,U\sigma U^*)
 =\mathcal C_U\Gamma_t(\rho,\sigma)\mathcal C_U^*,
\end{equation}
while $(U\sigma U^*)^{1/2}=\mathcal C_U(\sigma^{1/2})$.
Unitary invariance of the Hilbert--Schmidt inner product proves the
claim.

\vspace{0.2cm}

\noindent (iii)  We have
\begin{equation}
 B_t(\rho_1\otimes\rho_2,\sigma_1\otimes\sigma_2)
 =B_t(\rho_1,\sigma_1)\otimes B_t(\rho_2,\sigma_2),
\end{equation}
so \eqref{rw:eq:Gamma-tensor-product} follows from
\eqref{rw:eq:gamma-explicit}.  Powers of a tensor product factor,
which proves \eqref{rw:eq:Qalpha-tensor-product} and hence the
R\'enyi identity.  For the order-one case, we use
\begin{equation}
 \log(A\otimes B)=\log A\otimes I+I\otimes\log B
\end{equation}
together with
$\langle\sigma_k^{1/2},
\Gamma_t(\rho_k,\sigma_k)\sigma_k^{1/2}\rangle_{\rm HS}=1$.
This proves \eqref{rw:eq:Dt-tensor-product}.
\end{proof}

\begin{remark}[Relative entropy identities for non-invertible states]
With the supported extended-value convention above, unitary
invariance and \eqref{rw:eq:Dt-tensor-product} remain valid for
arbitrary states, with $+\infty+c=+\infty$.  This follows by
regularization and factorization of tensor-product supports.  No
corresponding blanket extension of the assertions for $D_f^t$ is made
without specifying the behavior of $f$ at zero and infinity.
\end{remark}

\begin{remark}{Tensor additivity does not extend to universal
superadditivity.}
For a fixed $t$, the corresponding product-reference inequality
would be
\begin{equation}
 D^t(\rho_{AB}\Vert\sigma_A\otimes\sigma_B)
 \geq D^t(\rho_A\Vert\sigma_A)
       +D^t(\rho_B\Vert\sigma_B)
 \label{rw:eq:putative-Dt-superadditivity}
\end{equation}
for every bipartite state $\rho_{AB}$, with marginals $\rho_A$ and
$\rho_B$, and all invertible $\sigma_A,\sigma_B$.  This cannot hold
universally for any
$t\in(0,1]$.  Indeed, with the standard continuous extension in the
first argument, $D^t$ is continuous, satisfies data processing by
\Cref{cor:consolidated-basic-examples}, and is tensor additive by
\eqref{rw:eq:Dt-tensor-product}.  These three properties together
with \eqref{rw:eq:putative-Dt-superadditivity} characterize a positive
multiple of the Umegaki relative entropy
\cite[Theorem~1]{wilming2017axiomatic}.  The classical collapse
\eqref{rw:eq:classical-collapse} fixes that multiple to one.  This
would give $D^t=D$, contradicting the strict $t$-monotonicity (which will be proved in
\Cref{rw:thm:Dt-strict-monotonicity}) on every noncommuting
invertible pair.  Hence, for each $t>0$, some invertible
$\rho_{AB},\sigma_A,\sigma_B$ strictly violate
\eqref{rw:eq:putative-Dt-superadditivity}; this
includes the BS endpoint
\cite[Section~3, Eq.~(13) and Fig.~1]
{BluhmCapelPerezHernandez-WeakQFBSentropy-2021}.  At $t=0$, by contrast, the difference
between the two sides is
$D(\rho_{AB}\Vert\rho_A\otimes\rho_B)=I(A:B)_\rho\geq0$.
\end{remark}

\begin{proposition}[Flags]
\label{rw:prop:elementary-structural-rules-2}

 We take $p=(p_i)_i$ and $q=(q_i)_i$ to be strictly positive probability
vectors and $\rho_i,\sigma_i$ to be invertible states on finite-dimensional
spaces $\mathcal H_i$.  We put
\begin{equation}
 \rho^\oplus:=\bigoplus_i p_i\rho_i,
 \qquad
 \sigma^\oplus:=\bigoplus_i q_i\sigma_i.
\end{equation}
Then the exact flagged-sum rule is
\begin{equation}
 D_f^t(\rho^\oplus\Vert\sigma^\oplus)
 =\sum_iq_i
 D_{f((p_i/q_i)\,\cdot)}^t(\rho_i\Vert\sigma_i).
 \label{rw:eq:general-flag-rule-Dft}
\end{equation}
In particular,
\begin{align}
 D^t(\rho^\oplus\Vert\sigma^\oplus)
 &=D_{\rm KL}(p\Vert q)
   +\sum_i p_iD^t(\rho_i\Vert\sigma_i),
 \label{rw:eq:general-flag-rule-Dt}
\end{align}
and
\begin{align}
 Q_{\alpha,t}(\rho^\oplus\Vert\sigma^\oplus)
 &=\sum_i p_i^\alpha q_i^{1-\alpha}
   Q_{\alpha,t}(\rho_i\Vert\sigma_i).
 \label{rw:eq:general-flag-rule-Qalpha}
\end{align}
Here $D_{\rm KL}(p\Vert q):=\sum_i p_i\log(p_i/q_i)$ is the Kullback-Leibler divergence.
Thus, for $\alpha\ne1$,
\begin{equation}
 D_\alpha^t(\rho^\oplus\Vert\sigma^\oplus)
 =\frac1{\alpha-1}\log\sum_i
 p_i^\alpha q_i^{1-\alpha}
 \exp\!\left((\alpha-1)
 D_\alpha^t(\rho_i\Vert\sigma_i)\right).
 \label{rw:eq:general-flag-rule-Renyi}
\end{equation}
For a common flag, $p=q$,
\eqref{rw:eq:general-flag-rule-Dft} reduces to
$D_f^t(\bigoplus_i p_i\rho_i\Vert\bigoplus_i p_i\sigma_i)
=\sum_i p_iD_f^t(\rho_i\Vert\sigma_i)$.

\end{proposition}

\begin{proof}
The diagonal Hilbert--Schmidt sector
$\mathcal B(\mathcal H_i)$ is invariant under the modular operator
of $(\rho^\oplus,\sigma^\oplus)$, and the restriction to that sector
is
\begin{equation}
 \Gamma_t(\rho^\oplus,\sigma^\oplus)
 \big|_{\mathcal B(\mathcal H_i)}
 =\frac{p_i}{q_i}\,\Gamma_t(\rho_i,\sigma_i).
 \label{rw:eq:Gamma-flag-restriction}
\end{equation}
The distinguished vector has $i$-th component
$q_i^{1/2}\sigma_i^{1/2}$.  Functional calculus on the mutually
orthogonal diagonal sectors proves
\eqref{rw:eq:general-flag-rule-Dft}.  Taking
$f(x)=x\log x$ and $f(x)=x^\alpha$ gives
\eqref{rw:eq:general-flag-rule-Dt} and
\eqref{rw:eq:general-flag-rule-Qalpha}; logarithmic normalization
then gives \eqref{rw:eq:general-flag-rule-Renyi}.

\end{proof}

\begin{remark}[Zero weights and support conventions]
The identity \eqref{rw:eq:general-flag-rule-Dt} extends to
nonnegative probability vectors and arbitrary component states by
omitting terms with $p_i=0$ and using
$0\log(0/q)=0$, $p\log(p/0)=+\infty$, and the supported extended value
of $D^t$.  Extending \eqref{rw:eq:general-flag-rule-Dft} itself
requires the boundary values and recession slope of $f$, and is not
asserted here.
\end{remark}

\subsection{Positivity and monotonicity in the R\'enyi order}
\label{rw:subsec:spectral-properties}

The preceding identities are algebraic.  We now use the spectral
measure of the geodesic modular operator to obtain normalization,
positivity, and monotonicity in the R\'enyi parameter. The monotonicity with respect to $t$ will be analyzed in the following section. 

\begin{proposition}[Normalization, positivity, and R\'enyi order]
\label{rw:prop:basic-spectral-properties}
For invertible states $\rho,
\sigma$ and fixed $t\in[0,1]$:
\begin{enumerate}[label=\textnormal{(\roman*)}]
\item One has
\begin{equation}
 D_f^t(\rho\Vert\rho)=f(1).
 \label{rw:eq:Dft-diagonal-normalization}
\end{equation}
If $f$ is scalar convex on an interval containing
$\{1\}\cup\operatorname{spec}(\Gamma_t)$, then
\begin{equation}
 D_f^t(\rho\Vert\sigma)\geq f(1).
 \label{rw:eq:Dft-scalar-Jensen}
\end{equation}
If $f$ is strictly convex, equality holds in
\eqref{rw:eq:Dft-scalar-Jensen} if and only if $\rho=\sigma$.

\item With the continuous definitions
\begin{equation}
 D_0^t(\rho\Vert\sigma):=0,
 \qquad
 D_1^t(\rho\Vert\sigma):=D^t(\rho\Vert\sigma),
 \label{rw:eq:Renyi-zero-one-extensions}
\end{equation}
the map
\begin{equation}
 \alpha\longmapsto D_\alpha^t(\rho\Vert\sigma)
 \label{rw:eq:Renyi-order-map}
\end{equation}
is non-decreasing on $[0,\infty)$.  It is strictly increasing unless
$\rho=\sigma$.  In particular, for every $\alpha>0$,
\begin{equation}
 D^t(\rho\Vert\sigma)\geq0,
 \qquad
 D_\alpha^t(\rho\Vert\sigma)\geq0,
 \label{rw:eq:basic-nonnegativity}
\end{equation}
and equality in either inequality holds exactly when $\rho=\sigma$.

\item For a fixed invertible pair, $t\mapsto D^t$ is continuous on
$[0,1]$, and the extension in
\eqref{rw:eq:Renyi-zero-one-extensions} makes
$(t,\alpha)\mapsto D_\alpha^t$ continuous on
$[0,1]\times[0,\infty)$.  The first map is real analytic in $t$;
the second is real analytic locally at every point with $\alpha>0$.
\end{enumerate}
\end{proposition}

\begin{proof}
(i) We let $E_t$ be the spectral measure of the positive definite
superoperator $\Gamma_t$ and define the scalar probability measure
\begin{equation}
 \nu_t(S):=
 \big\langle\sigma^{1/2},E_t(S)\sigma^{1/2}
 \big\rangle_{\rm HS}.
 \label{rw:eq:cyclic-spectral-measure}
\end{equation}
The normalization in \eqref{rw:eq:gamma-normalization} gives
\begin{equation}
 \int 1\,\mathrm d\nu_t=1,
 \qquad
 \int\lambda\,\mathrm d\nu_t(\lambda)=1,
 \label{rw:eq:cyclic-spectral-moments}
\end{equation}
and spectral calculus gives
\begin{equation}
 D_f^t(\rho\Vert\sigma)
 =\int f(\lambda)\,\mathrm d\nu_t(\lambda),
 \qquad
 Q_{\alpha,t}(\rho\Vert\sigma)
 =\int\lambda^\alpha\,\mathrm d\nu_t(\lambda).
 \label{rw:eq:cyclic-spectral-representation}
\end{equation}
Jensen's inequality now proves
\eqref{rw:eq:Dft-scalar-Jensen}.  If $f$ is strictly convex, equality
forces $\nu_t$ to be concentrated at its mean $1$.  Equivalently,
\begin{equation}
 (\Gamma_t-I)\sigma^{1/2}=0.
\end{equation}
Using the first identity in \eqref{rw:eq:gamma-normalization}, this is
$\rho\sigma^{-1/2}=\sigma^{1/2}$, hence $\rho=\sigma$.  Conversely,
if $\rho=\sigma$, the distinguished vector is a unit eigenvector of
$\Gamma_t$, which allows us to conclude.

\vspace{0.2cm}

\noindent (ii) We fix $t\in[0,1]$.  The operators
$L_{\rho^{1-t}}$ and $R_{B_t}$ are positive definite operators on
Hilbert--Schmidt space.  They commute, so their product
$\Gamma_t=L_{\rho^{1-t}}R_{B_t}$ is positive definite as well.  By
the spectral theorem there is a projection-valued measure $E_t$ on
$(0,\infty)$ such that
\begin{equation}
 \Gamma_t=\int_{(0,\infty)}\lambda\,\mathrm dE_t(\lambda).
 \label{rw:eq:Gamma-spectral-resolution}
\end{equation}
Functional calculus therefore gives, for every $a\geq0$,
\begin{equation}
 \Gamma_t^a
 =\int_{(0,\infty)}\lambda^a\,\mathrm dE_t(\lambda).
 \label{rw:eq:Gamma-power-spectral-resolution}
\end{equation}
We use the scalar measure $\nu_t$ defined in
\eqref{rw:eq:cyclic-spectral-measure}.  Since
$\lVert\sigma^{1/2}\rVert_2^2=1$, this is a probability measure, and
\eqref{rw:eq:gamma-normalization} gives
\begin{equation}
 \int_{(0,\infty)}\lambda\,\mathrm d\nu_t(\lambda)=1.
 \label{rw:eq:cyclic-spectral-first-moment}
\end{equation}
Writing $Q_t(a):=Q_{a,t}(\rho\Vert\sigma)$,
\eqref{rw:eq:power-functional} and
\eqref{rw:eq:Gamma-power-spectral-resolution} yield
\begin{equation}
Q_t(a)=\int_{(0,\infty)}\lambda^a\,
 \mathrm d\nu_t(\lambda).
 \label{rw:eq:Renyi-spectral-moment}
\end{equation}
The companion paper gives this spectral measure its statistical
meaning: $\nu_t$ is the null law $q_t$ of the canonical finite
experiment, while weighting it by $\lambda$ gives the alternative law
$p_t$, as shown in \cite[Theorem~3.1]{capel2026informationgeometry}.  In the
present section we use only the scalar moment identity
\eqref{rw:eq:Renyi-spectral-moment}.

At $t=0$ and $0\leq a\leq1$, $Q_t(a)$ is the Petz overlap appearing
in the quantum Chernoff bound, as described, for example, in
\cite{audenaert2007discriminating}.  The same log-convexity argument
applies to every $t$.  Indeed, we let $a,b\geq0$ and $0<s<1$.  H\"older's
inequality with conjugate exponents $(1-s)^{-1}$ and $s^{-1}$ gives
\begin{align}
 Q_t((1-s)a+sb)
 &=\int_{(0,\infty)}
   (\lambda^a)^{1-s}(\lambda^b)^s\,
   \mathrm d\nu_t(\lambda)\notag\\
 &\leq
 \left(\int_{(0,\infty)}\lambda^a\,
   \mathrm d\nu_t(\lambda)\right)^{1-s}
 \left(\int_{(0,\infty)}\lambda^b\,
   \mathrm d\nu_t(\lambda)\right)^s\notag\\
 &=Q_t(a)^{1-s}Q_t(b)^s.
 \label{rw:eq:Renyi-Holder-log-convexity}
\end{align}
The endpoint cases $s=0,1$ are identities.  Since $Q_t(a)>0$, taking
logarithms in \eqref{rw:eq:Renyi-Holder-log-convexity} shows that
\begin{equation}
 \psi_t(a):=\log Q_t(a)
 \label{rw:eq:Renyi-log-moment}
\end{equation}
is convex:
\begin{equation}
 \psi_t((1-s)a+sb)
 \leq(1-s)\psi_t(a)+s\psi_t(b).
 \label{rw:eq:Renyi-log-moment-convexity}
\end{equation}
Moreover,
\begin{equation}
 Q_t(0)=1,
 \qquad
 Q_t(1)=
 \big\langle\sigma^{1/2},
 \Gamma_t\sigma^{1/2}\big\rangle_{\rm HS}
 =\operatorname{Tr}\rho=1.
 \label{rw:eq:Renyi-moment-zero-one}
\end{equation}
Thus $\psi_t(0)=\psi_t(1)=0$, and, for $a\ne1$,
\begin{equation}
 D_a^t(\rho\Vert\sigma)
 =\frac{\psi_t(a)-\psi_t(1)}{a-1}.
 \label{rw:eq:Renyi-as-secant-slope}
\end{equation}
This is the secant slope of the convex function $\psi_t$ between
$a$ and $1$.  Such slopes are non-decreasing in the moving endpoint,
which proves monotonicity in the R\'enyi order.  At $a=1$, the
continuous extension is
\begin{equation}
 \begin{aligned}
 D_1^t(\rho\Vert\sigma) =\psi_t'(1) =\frac{Q_t'(1)}{Q_t(1)} =\big\langle\sigma^{1/2},
 \Gamma_t\log(\Gamma_t)\sigma^{1/2}
 \big\rangle_{\rm HS}
 =D^t(\rho\Vert\sigma).
 \end{aligned}
 \label{rw:eq:Renyi-order-one-spectral-limit}
\end{equation}
The equality condition in H\"older's inequality gives the strictness
criterion recorded in \Cref{rw:rem:Renyi-order-sharpness} below.
Finally, monotonicity and $D_0^t=0$ give nonnegativity and the
corresponding zero-equality statements.

\vspace{0.2cm}

\noindent (iii) All operators are positive definite and finite dimensional.
The formulas \eqref{rw:eq:gamma-explicit} and
\eqref{rw:eq:power-functional}, together with analytic
functional calculus, give analytic dependence of
$\psi_t(\alpha)=\log Q_{\alpha,t}$ on $(t,\alpha)$.  Since
$\psi_t(1)=0$ identically in $t$, analytic division by $\alpha-1$
gives the removable analytic extension at $\alpha=1$, whose value is
$D^t$.  The identity $Q_{0,t}=1$ gives continuity at $\alpha=0$.
\end{proof}

\begin{remark}[Non-invertible states]
The proposition holds without change for common-support pairs after
compression, but not verbatim for unequal supports.  For example, if
$\rho=P_\psi$ is pure and $\sigma>0$, then for every $t>0$ the
supported extension satisfies
\begin{equation}
 D_\alpha^t(\rho\Vert\sigma)
 =\log\langle\psi|\sigma^{-1}|\psi\rangle,
 \qquad \alpha>0.
\end{equation}
Thus strict R\'enyi-order monotonicity fails for non-invertible states, and the
limit as $\alpha\downarrow0$ need not agree with the invertible-pair
convention $D_0^t=0$.  We therefore retain invertibility in the
strictness and analyticity clauses.
\end{remark}

\begin{remark}[Sharpness]
\label{rw:rem:Renyi-order-sharpness}
For two distinct orders $a$ and $b$, equality in the nontrivial
H\"older step \eqref{rw:eq:Renyi-Holder-log-convexity} holds if and
only if $\lambda^{a-b}$ is constant $\nu_t$-almost surely.  Thus
$\psi_t$ is strictly convex unless $\nu_t$ is concentrated at one
point.  \eqref{rw:eq:cyclic-spectral-first-moment} forces
that point to be $1$, so
\begin{equation}
 E_t(\{1\})\sigma^{1/2}=\sigma^{1/2},
 \qquad
 \Gamma_t\sigma^{1/2}=\sigma^{1/2}.
\end{equation}
By \eqref{rw:eq:gamma-normalization}, this is equivalent to
$\rho=\sigma$.  Conversely, if $\rho=\sigma$, then
$\nu_t=\delta_1$, $Q_t(a)=1$, and $D_a^t(\rho\Vert\sigma)=0$ for
every $a\geq0$.  Hence the R\'enyi-order monotonicity is strict
exactly when $\rho\ne\sigma$.
\end{remark}

\subsection{Representation through sandwiched R\'enyi divergences}
\label{rw:subsec:sandwiched-representation}

The preceding result varies the R\'enyi order while keeping $t$
fixed.  There is a different, indirect connection with the
sandwiched R\'enyi family obtained by using a R\'enyi parameter to
reparametrize $t$ itself.  The following identity is a decomposition
of the order-one quantity $D^t$; it does not identify the family
$D_\alpha^t$ with an $(\alpha,z)$ branch, and therefore does not
conflict with \Cref{rw:prop:alpha-z-intersections}.

\begin{proposition}[Umegaki--sandwiched representation]
\label{rw:prop:Dt-sandwiched-reformulation}
We take $\rho,\sigma$ to be invertible states and $t\in[0,1)$.  We set
\begin{equation}
 a:=\frac1{1-t},
 \qquad
 A_a:=\rho^{\frac{1-a}{2a}}
       \sigma
       \rho^{\frac{1-a}{2a}},
 \qquad
 Z_a:=\operatorname{Tr}A_a^a,
 \qquad
 \tau_a:=\frac{A_a^a}{Z_a}.
 \label{rw:eq:Dt-sandwiched-escort-state}
\end{equation}
For $a>1$, we write
\begin{equation}
 \widetilde D_a(\sigma\Vert\rho)
 :=\frac1{a-1}\log Z_a
 =D_{a,a}(\sigma\Vert\rho)
 \label{rw:eq:reverse-sandwiched-definition}
\end{equation}
for the sandwiched R\'enyi divergence
\cite[Section~II]{mullerLennert2013quantum,wildeWinterYang2014strong},
with its continuous
order-one value when $a=1$.  The equality with the $(a,a)$
expression follows from the common nonzero spectrum of $XX^*$ and
$X^*X$.  Then
\begin{equation}
 \boxed{
 D^t(\rho\Vert\sigma)
 =\frac1aD(\rho\Vert\tau_a)
  -\frac{a-1}{a}\,
   \widetilde D_a(\sigma\Vert\rho).}
 \label{rw:eq:Dt-sandwiched-reformulation}
\end{equation}
Consequently,
\begin{equation}
 D(\rho\Vert\tau_a)
 \geq(a-1)\widetilde D_a(\sigma\Vert\rho),
 \label{rw:eq:Dt-sandwiched-comparison}
\end{equation}
with equality if and only if $\rho=\sigma$.
\end{proposition}

\begin{proof}
Since $t=1-a^{-1}$,
\begin{equation}
 K_t^{-1}
 =\rho^{-t/2}\sigma\rho^{-t/2}
 =A_a.
 \label{rw:eq:Kt-inverse-sandwiched}
\end{equation}
The explicit formula \eqref{rw:eq:Dt-explicit} therefore gives
\begin{align}
 D^t(\rho\Vert\sigma)
 &=\frac1a\operatorname{Tr}(\rho\log\rho)
   -\operatorname{Tr}(\rho\log A_a)\notag\\
 &=\frac1a\operatorname{Tr}\!\left[
   \rho\bigl(\log\rho-\log A_a^a\bigr)\right]\notag\\
 &=\frac1aD(\rho\Vert\tau_a)-\frac1a\log Z_a.
 \label{rw:eq:Dt-sandwiched-proof-calculation}
\end{align}
\eqref{rw:eq:reverse-sandwiched-definition} proves
\eqref{rw:eq:Dt-sandwiched-reformulation}.  Finally,
\Cref{rw:prop:basic-spectral-properties} says that
$D^t\geq0$, with equality exactly for $\rho=\sigma$; multiplying
\eqref{rw:eq:Dt-sandwiched-reformulation} by $a$ gives
\eqref{rw:eq:Dt-sandwiched-comparison} and its equality condition.
\end{proof}

\begin{remark}[Support of the sandwiched representation]
The proposition remains valid for states with common support after
compression, with $A_a$, $Z_a$, and $\tau_a$ computed in the
compressed algebra.  For unequal supports,
\eqref{rw:eq:Dt-sandwiched-reformulation} is not asserted as an
extended-real identity, since its two right-hand terms can separately
be infinite.
\end{remark}

\begin{remark}[The BS limit of the representation]
\label{rw:rem:Dt-sandwiched-BS-limit}
The two terms in \eqref{rw:eq:Dt-sandwiched-reformulation} have
finite limits as $a\to\infty$, corresponding to $t\uparrow1$.  We put
\begin{equation}
 A_\infty:=\rho^{-1/2}\sigma\rho^{-1/2},
 \qquad
 r_a:=Z_a^{1/a}.
\end{equation}
Then $A_a\to A_\infty$ in operator norm, and
\begin{equation}
 \lVert A_a\rVert_\infty
 \leq r_a
 \leq(\dim\mathcal H)^{1/a}\lVert A_a\rVert_\infty.
 \label{rw:eq:sandwiched-radius-squeeze}
\end{equation}
Hence $r_a\to\lVert A_\infty\rVert_\infty$.  Moreover,
\begin{align}
 \frac1aD(\rho\Vert\tau_a)
 &=\frac1a\operatorname{Tr}(\rho\log\rho)
   -\operatorname{Tr}(\rho\log A_a)+\log r_a
 \notag\\
 &\longrightarrow
 -\operatorname{Tr}(\rho\log A_\infty)
 +\log\lVert A_\infty\rVert_\infty,
 \label{rw:eq:sandwiched-first-term-limit}\\
 -\frac{a-1}{a}\widetilde D_a(\sigma\Vert\rho)
 &=-\frac1a\log Z_a=-\log r_a
 \longrightarrow-\log\lVert A_\infty\rVert_\infty.
 \label{rw:eq:sandwiched-second-term-limit}
\end{align}
The norm terms cancel.  Since
$A_\infty^{-1}=\rho^{1/2}\sigma^{-1}\rho^{1/2}$, this yields
\begin{equation}
 \lim_{t\uparrow1}D^t(\rho\Vert\sigma)
 =-\operatorname{Tr}(\rho\log A_\infty)
 =\operatorname{Tr}\!\left[
   \rho\log(\rho^{1/2}\sigma^{-1}\rho^{1/2})\right]
 =\widehat D(\rho\Vert\sigma).
 \label{rw:eq:sandwiched-representation-BS-limit}
\end{equation}
Thus the decomposition extends continuously to the BS endpoint even
though its auxiliary order $a=(1-t)^{-1}$ diverges.
\end{remark}

\begin{remark}[What these properties do not imply]
\label{rw:rem:basic-properties-scope}
The positivity and R\'enyi-order monotonicity above hold for every
$\alpha>0$ because they concern a fixed scalar spectral law.  They do
not imply data processing for $\alpha>2$.  Channel monotonicity and
joint convexity are treated in \Cref{sec:consolidated-dpi},
whereas the dependence on $t$ is the subject of
\Cref{rw:sec:t-monotonicity}.
\end{remark}

\section{Dependence on the interpolation parameter}
\label{rw:sec:t-monotonicity}

We now determine how the principal divergences vary with the
interpolation parameter.  The $t$-relative entropies have a strict
direction of monotonicity.  For the $(t,\alpha)$-R\'enyi divergences,
the direction depends on the R\'enyi order and changes at the
quadratic point.  We conclude by showing that operator convexity alone
does not determine the behavior of a general generator.

\subsection{Strict monotonicity of the
\texorpdfstring{$t$}{t}-relative entropies}
\label{rw:subsec:Dt-monotonicity}

\begin{theorem}[Strict $t$-monotonicity]
\label{rw:thm:Dt-strict-monotonicity}
For invertible states $\rho,\sigma$, the function
$t\mapsto D^t(\rho\Vert\sigma)$ is non-decreasing on $[0,1]$.
For any $0\leq t_1<t_2\leq1$,
\begin{equation}
 D^{t_1}(\rho\Vert\sigma)=D^{t_2}(\rho\Vert\sigma)
 \quad\Longleftrightarrow\quad [\rho,\sigma]=0.
 \label{rw:eq:Dt-strict-equality}
\end{equation}
Thus the $t$-relative entropies are independent of $t$ for commuting
pairs and strictly increasing in $t$ for noncommuting pairs.
\end{theorem}

\begin{proof}
We recall from \Cref{rw:prop:Dt-explicit} that
\[
 K_t:=\rho^{t/2}\sigma^{-1}\rho^{t/2},
 \qquad
 D^t(\rho\Vert\sigma)
 =(1-t)\operatorname{Tr}(\rho\log\rho)
  +\operatorname{Tr}(\rho\log K_t).
\]
To compare the two interpolation times, we set
$\delta:=t_2-t_1>0$ and write the corresponding increment as
\[
 \begin{split}
 D^{t_2}(\rho\Vert\sigma)-D^{t_1}(\rho\Vert\sigma)
 ={}&-\delta\operatorname{Tr}(\rho\log\rho)+\operatorname{Tr}\!\left[
 \rho(\log K_{t_2}-\log K_{t_1})\right].
 \end{split}
\]
Thus it remains to control the change of the logarithmic term when
the parameter advances from $t_1$ to $t_2$.  The definition of $K_t$
gives
\[
 \begin{split}
 K_{t_2} =\rho^{\delta/2}K_{t_1}\rho^{\delta/2}.
 \end{split}
\]
Hence the time increment is encoded by a positive congruence, which
is precisely the form controlled by the logarithmic Hiai--Petz
inequality, a consequence of Araki log-majorization.  For positive
definite $X,Y$ and every $p>0$, this inequality states that
\[
 \operatorname{Tr}\!\left[X(\log X^p+\log Y^p)\right]
 \leq
 \operatorname{Tr}\!\left[
 X\log\!\left(X^{p/2}Y^pX^{p/2}\right)\right],
\]
with equality if and only if $[X,Y]=0$, as proved in
\cite[Theorem~2.1]{HIAI1993153} and
\cite[Theorem~2.3]{carlen2018some}.  Taking
$X=\rho$, $Y=K_{t_1}^{1/\delta}$, and $p=\delta$ gives
\begin{equation}
 \operatorname{Tr}\rho\log K_{t_2}
 \geq \delta\operatorname{Tr}(\rho\log\rho)
      +\operatorname{Tr}\rho\log K_{t_1}.
 \label{rw:eq:complementary-GT-step}
\end{equation}
Substitution in the preceding increment proves that
$D^{t_2}(\rho\Vert\sigma)-D^{t_1}(\rho\Vert\sigma)\geq0$.
Equality holds precisely when $[\rho,K_{t_1}]=0$.  Finally,
\[
 [\rho,K_{t_1}]
 =\rho^{t_1/2}[\rho,\sigma^{-1}]\rho^{t_1/2},
\]
so this is equivalent to $[\rho,\sigma]=0$.  The commuting case also
follows directly from \eqref{rw:eq:classical-collapse}.
\end{proof}

\begin{remark}[Non-invertible pairs]
With the supported extended-value definition, the non-decreasing
inequality extends to arbitrary states by positive-definite
regularization.  The equality characterization extends only to
common-support pairs.  Indeed, if $\rho=P_\psi$ is pure and
$\sigma>0$, then
\begin{equation}
 D^t(\rho\Vert\sigma)
 =\log\langle\psi|\sigma^{-1}|\psi\rangle,
 \qquad 0<t\leq1,
\end{equation}
although $[\rho,\sigma]$ need not vanish.
\end{remark}

\subsection{Sharp regimes for the
\texorpdfstring{$(t,\alpha)$}{(t,alpha)}-R\'enyi divergences}
\label{rw:subsec:renyi-t-monotonicity}

The monotonicity direction of $D_\alpha^t$ changes at order two.  Its
proof is most transparent at the level of the unnormalized trace
functional $Q_{\alpha,t}$, before applying the logarithmic
normalization.

\begin{theorem}[Monotonicity of the $(t,\alpha)$-R\'enyi divergences]
\label{rw:thm:power-t-monotonicity}
We take $\rho,\sigma>0$.  With the continuous order-one convention
$D_1^t:=D^t$, the map
$t\mapsto D_\alpha^t(\rho\Vert\sigma)$ is
\begin{itemize}
 \item Non-decreasing for $0<\alpha<2$;
 \item Constant for $\alpha=2$;
 \item Non-increasing for $2<\alpha\leq3$.
\end{itemize}
At $\alpha=1$, the increase is strict precisely when
$[\rho,\sigma]\neq0$.
For $\alpha>3$, $0\leq t_1<t_2\leq1$, the comparison
\begin{equation}
 D_\alpha^{t_1}(\rho\Vert\sigma)
 \geq D_\alpha^{t_2}(\rho\Vert\sigma)
 \label{rw:eq:power-comparison-alpha-large}
\end{equation}
holds for every invertible pair in every finite dimension if and only if
\begin{equation}
 \frac{1+(\alpha-1)(1-t_2)}{t_2-t_1}
 \leq\frac{\alpha-1}{\alpha-3}.
 \label{rw:eq:power-sharp-hyperbola}
\end{equation}
If \eqref{rw:eq:power-sharp-hyperbola} fails, an invertible qutrit pair
can violate \eqref{rw:eq:power-comparison-alpha-large}.
Consequently, for every $\alpha>3$ there are invertible qutrit states
for which $t\mapsto D_\alpha^t$ is neither non-increasing nor
non-decreasing.  These regimes are summarized schematically in
\Cref{rw:fig:power-t-monotonicity}.
\end{theorem}

\begin{proof}
For $\alpha\neq1$, logarithmic normalization preserves comparisons of
$Q_{\alpha,t}$ when $\alpha>1$ and reverses them when $0<\alpha<1$.
We put $\beta=\alpha-1$.  The push-through identity used in the proof
of \Cref{rw:prop:Dt-explicit} yields
\begin{equation}
 Q_{\alpha,t}
 =\operatorname{Tr}\!\left[
 \rho^{1+\beta(1-t)}K_t^\beta\right].
 \label{rw:eq:power-reduction-one}
\end{equation}
We fix $t_1<t_2$, let $\delta=t_2-t_1$, and set
\begin{equation}
 A:=\rho^\delta,\qquad B:=K_{t_1},\qquad
 r:=\frac{1+\beta(1-t_2)}{\delta}.
 \label{rw:eq:power-reduction-parameters}
\end{equation}
Then $K_{t_2}=A^{1/2}BA^{1/2}$ and
\begin{align}
 Q_{\alpha,t_1}
 &=\operatorname{Tr}(A^{r+\beta}B^\beta),
 \notag\\
 Q_{\alpha,t_2}
 &=\operatorname{Tr}\!\left[
 A^r(A^{1/2}BA^{1/2})^\beta\right].
 \label{rw:eq:power-two-traces}
\end{align}
The needed comparisons are the following weighted extensions of the
Araki--Lieb--Thirring trace inequality.  If $A,B>0$, $0<s<1$, and
$f$ is nonnegative and non-decreasing on the spectrum of $A$, then
\begin{equation*}
 \operatorname{Tr}\!\left[f(A)A^sB^s\right]
 \leq
 \operatorname{Tr}\!\left[
 f(A)(A^{1/2}BA^{1/2})^s\right].
\end{equation*}
If instead $x\mapsto x^s g(x)$ is nonnegative and non-increasing,
then
\begin{equation*}
 \operatorname{Tr}\!\left[
 g(A)(A^{1/2}BA^{1/2})^s\right]
 \leq
 \operatorname{Tr}\!\left[g(A)A^sB^s\right].
\end{equation*}
These are \cite[Theorem~4 and Proposition~5]{liu2026araki},
respectively.

For \underline{$0<\beta<1$}, we take $s=\beta$ and $f(x)=x^r$ in the first
inequality.  Since $r\geq0$, it gives directly from
\eqref{rw:eq:power-two-traces}
\begin{equation*}
 Q_{\alpha,t_1}\leq Q_{\alpha,t_2}.
\end{equation*}
For \underline{$-1<\beta<0$}, we put $s=-\beta$ and
$g(x)=x^{-r}$.  Here
\begin{equation*}
 \delta(r-s)=1-s+st_1>0,
\end{equation*}
so $x^sg(x)=x^{s-r}$ is non-increasing.  Applying the second
inequality to $A^{-1},B^{-1}$ yields
\begin{align*}
 Q_{\alpha,t_2}
 &=\operatorname{Tr}\!\left[
   A^r(A^{1/2}BA^{1/2})^{-s}\right]\leq \operatorname{Tr}(A^{r-s}B^{-s})
 =Q_{\alpha,t_1}.
\end{align*}
Thus $Q_{\alpha,t}$ is non-increasing for $0<\alpha<1$ and
non-decreasing for $1<\alpha<2$.  When $\beta=0$, one has
$Q_{1,t}=1$; when $\beta=1$, the two traces in
\eqref{rw:eq:power-two-traces} agree by cyclicity.

For \underline{$1<\beta\leq2$}, the reverse power-weight comparison
\begin{equation*}
 \operatorname{Tr}(A^{r+\beta}B^\beta)
 \geq
 \operatorname{Tr}\!\left[
 A^r(A^{1/2}BA^{1/2})^\beta\right],
 \qquad r\geq0,
\end{equation*}
holds for all positive definite $A,B$ by
\cite[Theorem~2.1]{shi2026liu}.  Hence $Q_{\alpha,t}$ is
non-increasing for $2<\alpha\leq3$.  After logarithmic normalization,
these are exactly the three stated regimes for $D_\alpha^t$; the
order-one assertion is \Cref{rw:thm:Dt-strict-monotonicity}.

For \underline{$\beta>2$}, the same reverse power-weight comparison holds in every
finite dimension exactly when $r\leq\beta/(\beta-2)$, by
\cite[Theorem~3.3]{vuong2026hyperbolic}.  Since $\alpha>1$, this is
also the condition for $D_\alpha^{t_1}\geq D_\alpha^{t_2}$, and it is
exactly \eqref{rw:eq:power-sharp-hyperbola}; beyond this region a
positive definite qutrit pair gives a strict violation.

The positive-matrix counterexamples can be realized by invertible
states.  Given positive definite qutrit
matrices $A_0,B_0$ violating the trace comparison, we define
\begin{equation}
 \rho:=\frac{A_0^{1/\delta}}{\operatorname{Tr}A_0^{1/\delta}},
 \quad
 \widetilde\sigma:=\rho^{t_1/2}B_0^{-1}\rho^{t_1/2},
 \quad
 \sigma:=\frac{\widetilde\sigma}
 {\operatorname{Tr}\widetilde\sigma}.
 \label{rw:eq:qutrit-state-realization}
\end{equation}
Then $\rho^\delta$ and $K_{t_1}$ are positive scalar multiples
of $A_0$ and $B_0$, and homogeneity preserves the strict
violation.

We finally spell out why no universal direction remains when
$\alpha>3$.  We choose
\begin{equation*}
 \frac{2}{\alpha-1}<t_*<1,
\end{equation*}
which is possible precisely because $\alpha>3$.  For
$(t_1,t_2)=(t_*,1)$, the condition
\eqref{rw:eq:power-sharp-hyperbola} fails, because its left-hand side is
\begin{equation*}
 \frac{1}{1-t_*}>\frac{\alpha-1}{\alpha-3}.
\end{equation*}
Hence the strict qutrit counterexample above satisfies
\begin{equation*}
 D_\alpha^{t_*}(\rho\Vert\sigma)
 <D_\alpha^1(\rho\Vert\sigma),
\end{equation*}
so its $t$-dependence is not non-increasing.  On the other hand, the
endpoint choice $(t_1,t_2)=(0,1)$ always satisfies the cutoff, since
\begin{equation*}
 \frac{1+(\alpha-1)(1-1)}{1-0}=1
 \leq\frac{\alpha-1}{\alpha-3}.
\end{equation*}
The endpoint comparison is universal, so it applies in particular to
the very same counterexample pair constructed for $(t_*,1)$:
$D_\alpha^0\geq D_\alpha^1$.  Combining the two facts gives
\begin{equation*}
 D_\alpha^0(\rho\Vert\sigma)
 \geq D_\alpha^1(\rho\Vert\sigma)
 >D_\alpha^{t_*}(\rho\Vert\sigma).
\end{equation*}
Thus this single function $t\mapsto D_\alpha^t(\rho\Vert\sigma)$
has a strict downward net change from $0$ to $t_*$ but a strict upward
change from $t_*$ to $1$.
It is neither non-decreasing nor non-increasing, which rules out a
universal monotonicity direction.
\end{proof}

\begin{remark}[Boundary scope]
All comparisons in the theorem remain valid for common-support pairs
by compression.  For unequal supports, corresponding weak
comparisons require a separately specified regularized definition of
$D_\alpha^t$, while the strict order-one assertion does not survive;
no such broader statement is made here.
\end{remark}

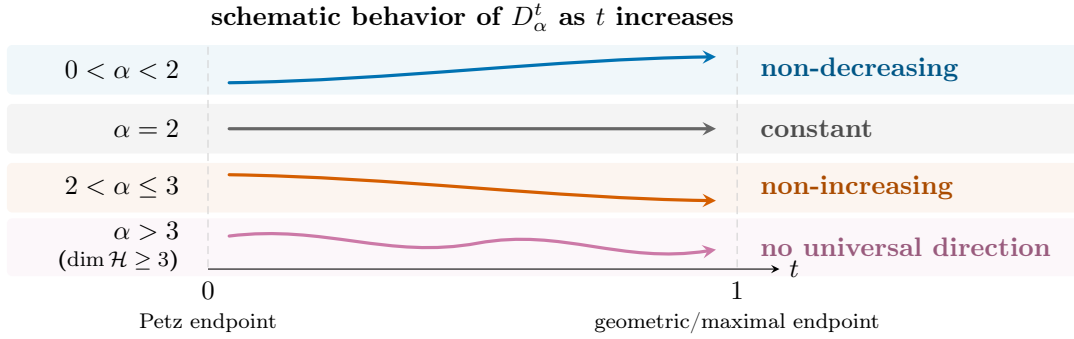
\begin{figure}[!ht]
\centering
\definecolor{phaseDecrease}{HTML}{D55E00}
\definecolor{phaseConstant}{HTML}{666666}
\definecolor{phaseIncrease}{HTML}{0072B2}
\definecolor{phaseMixed}{HTML}{CC79A7}
\begin{tikzpicture}[
  x=7.0cm,
  y=0.78cm,
  >=stealth,
  font=\small,
  phase curve/.style={very thick,->,line cap=round,line join=round},
  phase range/.style={anchor=east,font=\small\bfseries},
  phase result/.style={anchor=west,font=\small\bfseries}
]
  % Keep one common padded span for every tint.  The left edge lies beyond
  % the widest two-line alpha constraint, and the right edge lies beyond
  % the longest qualitative label.
  \def\phaseShadeLeft{-0.38}
  \def\phaseShadeRight{1.65}
  \fill[phaseIncrease!6,rounded corners=2pt]
    (\phaseShadeLeft,2.58) rectangle (\phaseShadeRight,3.42);
  \fill[phaseConstant!7,rounded corners=2pt]
    (\phaseShadeLeft,1.58) rectangle (\phaseShadeRight,2.42);
  \fill[phaseDecrease!6,rounded corners=2pt]
    (\phaseShadeLeft,0.58) rectangle (\phaseShadeRight,1.42);
  \fill[phaseMixed!6,rounded corners=2pt]
    (\phaseShadeLeft,-0.50) rectangle (\phaseShadeRight,0.48);

  \draw[black!18,densely dashed] (0,-0.35)--(0,3.45);
  \draw[black!18,densely dashed] (1,-0.35)--(1,3.45);
  \node[above,font=\small\bfseries] at (0.5,3.48)
    {schematic behavior of $D_\alpha^t$ as $t$ increases};

  \node[phase range]
    at (-0.035,3) {$0<\alpha<2$};
  \draw[phase curve,draw=phaseIncrease]
    (0.04,2.78)..controls (0.35,2.82) and (0.65,3.18)..(0.96,3.22);
  \node[phase result,text=phaseIncrease!75!black]
    at (1.025,3) {non-decreasing};

  \node[phase range]
    at (-0.035,2) {$\alpha=2$};
  \draw[phase curve,draw=phaseConstant] (0.04,2)--(0.96,2);
  \node[phase result,text=phaseConstant!75!black]
    at (1.025,2) {constant};

  \node[phase range]
    at (-0.035,1) {$2<\alpha\leq3$};
  \draw[phase curve,draw=phaseDecrease]
    (0.04,1.22)..controls (0.35,1.18) and (0.65,0.82)..(0.96,0.78);
  \node[phase result,text=phaseDecrease!80!black]
    at (1.025,1) {non-increasing};

  \node[phase range,align=right] at (-0.035,0)
    {$\alpha>3$\\[-2pt]\scriptsize($\dim\mathcal H\geq3$)};
  \draw[phase curve,draw=phaseMixed]
    (0.04,0.18)..controls (0.22,0.38) and (0.34,-0.22)..
    (0.51,0.06)..controls (0.67,0.32) and (0.78,-0.30)..(0.96,-0.05);
  \node[phase result,align=left,text=phaseMixed!75!black] at (1.025,0)
    {no universal direction};

  \draw[->] (0,-0.38)--(1.08,-0.38) node[right] {$t$};
  \node[below,align=center] at (0,-0.42)
    {$0$\\[-1pt]\scriptsize Petz endpoint};
  \node[below,align=center] at (1,-0.42)
    {$1$\\[-1pt]\scriptsize geometric/maximal endpoint};
\end{tikzpicture}
\caption{Monotonicity phase diagram for the $(t,\alpha)$-R\'enyi
divergences in \Cref{rw:thm:power-t-monotonicity}.  At order one we
use the continuous extension $D_1^t=D^t$.  Colors encode the
qualitative class: vermillion is non-increasing, gray is constant,
blue is non-decreasing, and purple indicates no universal direction.  Only the weak direction
of variation is represented; the curves are not numerical plots and
may be flat, for example for commuting pairs.  When $\alpha>3$ in
dimension at least three, the map $t\mapsto D_\alpha^t$ can change
direction, although the endpoint inequality
$D_\alpha^0\geq D_\alpha^1$ remains universal.  The decreasing
comparison between two prescribed parameters is universal exactly under
\eqref{rw:eq:power-sharp-hyperbola}; in dimension two it remains valid
throughout the range $\alpha>2$.}
\label{rw:fig:power-t-monotonicity}
\end{figure}

\begin{example}[An explicit nonmonotone $(t,4)$-R\'enyi divergence]
\label{rw:ex:explicit-alpha-four}
We take $t_1=9/10$, $t_2=1$, and
\begin{equation}
 A=\operatorname{diag}(1,1/40,1/80),\qquad
 B=\begin{pmatrix}
 1&1/\sqrt2&1/\sqrt2\\
 1/\sqrt2&80001&-79999\\
 1/\sqrt2&-79999&80001
 \end{pmatrix}.
\end{equation}
Both matrices are positive definite: in the basis
$(e_1,(e_2+e_3)/\sqrt2,(e_2-e_3)/\sqrt2)$, the second one is
$\left(\begin{smallmatrix}1&1\\1&2\end{smallmatrix}\right)
 \oplus(160000)$.  We define
\begin{equation}
 \rho:=\frac{A^{10}}{\operatorname{Tr} A^{10}},\qquad
 R:=\rho^{9/20},\qquad
 \sigma:=\frac{RB^{-1}R}{\operatorname{Tr}(RB^{-1}R)}.
\end{equation}
Direct substitution in \eqref{rw:eq:power-reduction-one} gives the
following exact identity; no numerical approximation is used:
\begin{equation}
 \frac{Q_{4,9/10}(\rho\Vert\sigma)}
      {Q_{4,1}(\rho\Vert\sigma)}
 =\frac{54975614947328000000106509}
 {80134659699076265308694002}<1.
 \label{rw:eq:explicit-alpha-four-ratio}
\end{equation}
Since the R\'enyi normalization preserves inequalities at order four,
this says that $D_4^{9/10}<D_4^1$.  On the other hand, the universally
valid endpoint instance of \eqref{rw:eq:power-sharp-hyperbola} gives
$D_4^0\geq D_4^1$.  For this single invertible qutrit pair, the map
$t\mapsto D_4^t$ is therefore neither non-increasing nor
non-decreasing.
\end{example}

\subsection{General operator-convex generators}
\label{rw:subsec:general-f-t-dependence}

Differently from the particular cases of the $t$-relative entropies
and the $(t,\alpha)$-R\'enyi divergences, a general $(t,f)$-divergence
has no fixed direction of monotonicity.  We first give an exact
criterion for one ordered pair.  We then derive a complete
ordinary-matrix characterization of universal monotonicity for
$C^1$ generators, identify
the operator-convex slice as an explicit polar cone of canonical
measures, and extract both scalar obstructions and concrete
subcones.

\subsubsection{A pairwise derivative criterion}

We first recall the finite-dimensional Daleckii--Krein formula.  If
$A=A^*$ has spectral decomposition $A=\sum_k\lambda_kP_k$ into
distinct eigenvalues and $f$ is $C^1$ on an open interval containing
$\operatorname{spec}(A)$, we let $\mathfrak f:A\mapsto f(A)$ denote the
functional-calculus map on the real vector space of self-adjoint
operators.  This map is Fr\'echet differentiable and, for every
self-adjoint perturbation $H$,
\begin{align}
 f(A+H)
 &=f(A)+\mathrm D\mathfrak f(A)[H]+o(\lVert H\rVert),\nonumber\\
 \mathrm D\mathfrak f(A)[H]
 &=\sum_{k,l}f^{[1]}(\lambda_k,\lambda_l)P_kHP_l,
 \label{rw:eq:daleckii-krein}
\end{align}
as $\lVert H\rVert\to0$, where
$f^{[1]}(x,y)=(f(x)-f(y))/(x-y)$ for $x\neq y$ and
$f^{[1]}(x,x)=f'(x)$.  This formula is given in
\cite[(2.38), p.~60]{bhatia2009positive} and originates in
\cite{daletskii1965integration}.

\begin{proposition}[Pairwise derivative criterion]
\label{rw:prop:general-f-pairwise-criterion}
We take $\rho,\sigma$ to be invertible states and
$f:(0,\infty)\to\mathbb R$ to be $C^1$, and we put
$F_f(t):=D_f^t(\rho\Vert\sigma)$.  The curve
$t\mapsto\Gamma_t$ is differentiable, with
\begin{equation}
 \dot\Gamma_t
 =-L_{\rho^{1-t}\log\rho}R_{B_t}
  +L_{\rho^{1-t}}
   R_{\sigma^{-1/2}\rho^t(\log\rho)\sigma^{-1/2}}.
 \label{rw:eq:gamma-derivative}
\end{equation}
Applying \eqref{rw:eq:daleckii-krein} on the Hilbert--Schmidt space
with the positive-definite self-adjoint operator $A=\Gamma_t$ and
direction $H=\dot\Gamma_t$, and using the chain rule, gives
\begin{equation}
 F_f'(t)
 =\big\langle\sigma^{1/2},
 \mathrm D\mathfrak f(\Gamma_t)[\dot\Gamma_t]
 \sigma^{1/2}\big\rangle_{\rm HS}.
 \label{rw:eq:general-f-derivative}
\end{equation}
More explicitly, if
$\Gamma_t=\sum_k\lambda_kP_k$ is its decomposition into distinct
eigenvalues, then
\begin{equation}
 F_f'(t)
 =\sum_{k,l}f^{[1]}(\lambda_k,\lambda_l)
 \big\langle P_k\sigma^{1/2},
 \dot\Gamma_t(P_l\sigma^{1/2})\big\rangle_{\rm HS},
 \label{rw:eq:general-f-divided-difference}
\end{equation}
Consequently, $F_f$ is non-decreasing for this
fixed pair if and only if the right-hand side of
\eqref{rw:eq:general-f-divided-difference} is nonnegative for every
$t\in(0,1)$.
\end{proposition}

\begin{proof}
Differentiating the two commuting multiplication factors in
$\Gamma_t=L_{\rho^{1-t}}R_{B_t}$ gives
\eqref{rw:eq:gamma-derivative}.  The finite-dimensional
Daleckii--Krein formula \eqref{rw:eq:daleckii-krein} then yields
\eqref{rw:eq:general-f-derivative} and its spectral form
\eqref{rw:eq:general-f-divided-difference}.  The final equivalence is
the elementary derivative criterion for a differentiable function on
an interval.
\end{proof}

\begin{remark}[Boundary scope]
The criterion remains valid for common-support pairs after
compression.  It is not extended to unequal supports, because
$\Gamma_t$ can acquire zero eigenvalues and the displayed
Fr\'echet derivative requires boundary regularity of $f$ that is not
assumed.
\end{remark}

In general, the Fréchet derivative in
\eqref{rw:eq:general-f-derivative} cannot be replaced by
$f'(\Gamma_t)\dot\Gamma_t$, since $\Gamma_t$ and $\dot\Gamma_t$
need not commute.

\subsubsection{The universal cone: characterization and obstructions}

A natural candidate for an example in which there is operator monotonicity is when the function $f$ defining the $(t,f)$-divergence is operator convex. However, the following example shows that operator convexity does not fix the
sign of the derivative in
\eqref{rw:eq:general-f-divided-difference}.

\begin{example}[Operator convexity does not imply $t$-monotonicity]
\label{rw:ex:weighted-chi-counterexample}
The function
\begin{equation}
 h(x):=x-2+x^{-1}=\frac{(x-1)^2}{x}
 \label{rw:eq:weighted-chi-generator}
\end{equation}
is operator convex on $(0,\infty)$.  Direct functional calculus
gives
\begin{equation}
 D_h^t(\rho\Vert\sigma)
 =\operatorname{Tr}(\sigma\rho^{t-1}\sigma\rho^{-t})-1.
 \label{rw:eq:weighted-chi-path}
\end{equation}
In an eigenbasis $\rho=\sum_i\lambda_i|i\rangle\langle i|$, the
term before $-1$ is
\begin{equation}
 \sum_{i,j}|\sigma_{ij}|^2
 \lambda_i^{t-1}\lambda_j^{-t}.
 \label{rw:eq:weighted-chi-spectral}
\end{equation}
Pairing the $(i,j)$ and $(j,i)$ terms shows that this expression
is symmetric under $t\leftrightarrow1-t$, non-increasing on
$[0,1/2]$, and non-decreasing on $[1/2,1]$.  It is nonconstant
whenever $[\rho,\sigma]\neq0$ \cite[Proposition~8]{temme2010}.  Hence even for an
operator-convex generator, $t\mapsto D_h^t(\rho\Vert\sigma)$ need not
be globally monotone.
\end{example}

\begin{definition}[Universal $t$-monotonicity cone]
\label{rw:def:universal-t-monotonicity-cone}
We define
\begin{equation}
 \boxed{\;
 \mathfrak M
 :=\left\{f:(0,\infty)\to\mathbb R\ \middle|\
 \begin{array}{l}
 t\longmapsto D_f^t(\rho\Vert\sigma)\text{ is non-decreasing on }[0,1]\\
 \text{for every finite-dimensional invertible pair }(\rho,\sigma)
 \end{array}
 \right\}
 \;}.
 \label{rw:eq:universal-t-monotonicity-cone}
\end{equation}
Thus ``universal'' refers both to all finite dimensions and to all
ordered pairs of invertible states.
\end{definition}

We next give a characterization of \(\mathfrak M\) which no longer
quantifies over a path of Hilbert--Schmidt superoperators.  It involves
only the divided differences of one scalar function and ordinary
positive matrices.  For a \(C^1\) function \(f\), we put
\begin{equation}
 g(x):=\frac{f(x)}{x},\qquad x>0,
 \label{rw:eq:universal-cone-g-transform}
\end{equation}
and we let \(\mathfrak g:A\mapsto g(A)\) denote its functional-calculus
map.

\begin{theorem}[Complete matrix characterization of the universal cone
for $C^1$ generators]
\label{rw:thm:universal-cone-matrix-characterization}
We take \(f\in C^1(0,\infty)\) and define \(g\) by
\eqref{rw:eq:universal-cone-g-transform}.  For \(a>1\), we take
\(P,K\in\mathcal B(\mathbb C^n)\) to be positive definite and to satisfy
\begin{equation}
 \operatorname{Tr}P^a=1,
 \qquad
 \operatorname{Tr}(P^{a-1}K^{-1})=1.
 \label{rw:eq:universal-cone-matrix-normalization}
\end{equation}
We choose an orthonormal eigenbasis \(Pe_i=p_i e_i\) and set
\begin{equation}
 C_i(P,K)
 :=\frac12\Big\{\log P-(\log p_i)I,K\Big\}.
 \label{rw:eq:universal-cone-Ci}
\end{equation}
Then \(f\in\mathfrak M\) if and only if
\begin{equation}
 \boxed{
 \mathcal L_f(a;P,K)
 :=a\sum_{i=1}^n p_i^{a+1}
 \big\langle e_i,
 \mathrm D\mathfrak g(p_iK)[C_i(P,K)]e_i
 \big\rangle
 \geq0
 }
 \label{rw:eq:universal-cone-matrix-characterization}
\end{equation}
for every \(n\), every \(a>1\), and every normalized pair
\((P,K)\) as in
\eqref{rw:eq:universal-cone-matrix-normalization}.

Equivalently, if
\(K=\sum_j\kappa_j|u_j\rangle\langle u_j|\), the matrix element in
\eqref{rw:eq:universal-cone-matrix-characterization} is the explicit
L\"owner expression
\begin{equation}
 \begin{split}
 &\big\langle e_i,
 \mathrm D\mathfrak g(p_iK)[C_i(P,K)]e_i\big\rangle
=\sum_{j,k}
 g^{[1]}(p_i\kappa_j,p_i\kappa_k)
 \langle e_i,u_j\rangle
 \langle u_j,C_i(P,K)u_k\rangle
 \langle u_k,e_i\rangle .
 \end{split}
 \label{rw:eq:universal-cone-Loewner-form}
\end{equation}
Thus \eqref{rw:eq:universal-cone-matrix-characterization} is a
condition on the generator alone, modulo tests by finite positive
matrices; it contains neither \(\sigma\) nor the geodesic modular
operator \(\Gamma_t\).
\end{theorem}

\begin{proof}
We fix an invertible pair and put
\begin{equation}
 K_t:=\rho^{t/2}\sigma^{-1}\rho^{t/2},
 \qquad
 \Xi_t:=L_{\rho^{1-t}}R_{K_t}.
 \label{rw:eq:universal-cone-Xi}
\end{equation}
The reduction to \(\Xi_t\) is the key point.  Indeed, we let
\(X_t:=\sigma^{-1/2}\rho^{t/2}=U_tK_t^{1/2}\) be the polar
decomposition.  Since
\(B_t=X_tX_t^*=U_tK_tU_t^*\), the Hilbert--Schmidt unitary
\(\mathcal U_t(Z):=ZU_t\) satisfies
\begin{equation}
 \mathcal U_t\Gamma_t\mathcal U_t^*=\Xi_t,
 \qquad
 \mathcal U_t(\sigma^{1/2})
 =\rho^{t/2}K_t^{-1/2}.
 \label{rw:eq:universal-cone-polar-conjugation}
\end{equation}
Moreover,
\begin{equation}
 \Xi_t^{1/2}[\rho^{t/2}K_t^{-1/2}]=\rho^{1/2}.
 \label{rw:eq:universal-cone-cyclic-vector}
\end{equation}
Since \(f(x)=xg(x)\), functional calculus and
\eqref{rw:eq:universal-cone-polar-conjugation}--
\eqref{rw:eq:universal-cone-cyclic-vector} give the useful identity
\begin{equation}
 D_f^t(\rho\Vert\sigma)
 =\big\langle\rho^{1/2},
 g(\Xi_t)[\rho^{1/2}]\big\rangle_{\rm HS}.
 \label{rw:eq:universal-cone-reduced-divergence}
\end{equation}

We now encode an arbitrary interior interpolation time by
\begin{equation}
 a:=\frac1{1-t}>1,
 \qquad
 P:=\rho^{1-t},
 \qquad
 K:=K_t.
 \label{rw:eq:universal-cone-time-reparametrization}
\end{equation}
Then \(\rho=P^a\), and differentiation of
\(K_t=\rho^{t/2}\sigma^{-1}\rho^{t/2}\) yields
\begin{equation}
 \dot K_t=\frac12\{\log\rho,K_t\}.
 \label{rw:eq:universal-cone-K-derivative}
\end{equation}
Consequently,
\begin{equation}
 \dot\Xi_t
 =a\left(-L_{P\log P}R_K
 +\frac12L_PR_{\{\log P,K\}}\right).
 \label{rw:eq:universal-cone-Xi-derivative}
\end{equation}
Both \(\Xi_t=L_PR_K\) and \(\dot\Xi_t\) leave invariant every
left spectral sector \(E_i\mathcal B(\mathbb C^n)\), where
\(E_i=|e_i\rangle\langle e_i|\).  On this sector they act,
respectively, as
\begin{equation}
 p_iR_K,
 \qquad
 ap_iR_{C_i(P,K)}.
 \label{rw:eq:universal-cone-sector-actions}
\end{equation}
The vector \(\rho^{1/2}=P^{a/2}\) has component
\(p_i^{a/2}E_i\) in the same sector.  Applying the
Daleckii--Krein formula in \eqref{rw:eq:daleckii-krein} to
\eqref{rw:eq:universal-cone-reduced-divergence}, and identifying a
right-multiplication sector with \(\mathbb C^n\), therefore gives
\begin{equation}
 \frac{\mathrm d}{\mathrm dt}D_f^t(\rho\Vert\sigma)
 =\mathcal L_f(a;P,K).
 \label{rw:eq:universal-cone-derivative-reduction}
\end{equation}
The spectral form of the same derivative is exactly
\eqref{rw:eq:universal-cone-Loewner-form}.

It remains to check that the tests in the statement are neither more
nor less general than state pairs.  The trace-one conditions for
\(\rho\) and \(\sigma\) become
\eqref{rw:eq:universal-cone-matrix-normalization}, because
\begin{equation}
 \sigma=P^{(a-1)/2}K^{-1}P^{(a-1)/2}.
 \label{rw:eq:universal-cone-state-reconstruction}
\end{equation}
Conversely, every \(a,P,K\) obeying
\eqref{rw:eq:universal-cone-matrix-normalization} defines invertible
states by \(\rho=P^a\) and
\eqref{rw:eq:universal-cone-state-reconstruction}; at
\(t=1-a^{-1}\), their matrix \(K_t\) is precisely \(K\).
Thus \eqref{rw:eq:universal-cone-matrix-characterization} is
equivalent to nonnegativity of the derivative for every pair and every
\(t\in(0,1)\).  Continuity at the endpoints completes the
equivalence with membership in \(\mathfrak M\).
\end{proof}

For operator-convex generators, the matrix criterion has an equivalent
dual form in the unique canonical representing measure.  This form
makes cancellations between different resolvent kernels explicit.
Every finite operator-convex function on \((0,\infty)\) can be written
uniquely as
\begin{equation}
 f(x)=f(1)+f'(1)(x-1)
 +\int_{[0,1]}q_u(x)\,\mathrm d\nu_f(u),
 \qquad
 q_u(x):=\frac{(x-1)^2}{1+u(x-1)},
 \label{rw:eq:universal-cone-compact-representation}
\end{equation}
where \(\nu_f\) is a finite positive measure, as shown in
\cite[Theorem~8.1]{hiai2011quantum}.

\begin{theorem}[Canonical-measure characterization of the
operator-convex slice]
\label{rw:thm:universal-cone-measure-characterization}
We take \(f\) to be a finite operator-convex function with canonical measure
\(\nu_f\) in
\eqref{rw:eq:universal-cone-compact-representation}.  For an invertible
pair, we put
\begin{equation}
 v:=\sigma^{1/2},
 \qquad
 \delta_{\rho,\sigma}:=(\Gamma_t-I)v
 =\rho\sigma^{-1/2}-\sigma^{1/2},
 \qquad
 C_{u,t}:=(1-u)I+u\Gamma_t.
 \label{rw:eq:universal-cone-response-data}
\end{equation}
The vector \(\delta_{\rho,\sigma}\) is independent of \(t\).  We define
the scalar response function
\begin{equation}
 \kappa_{\rho,\sigma,t}(u)
 :=-u\big\langle C_{u,t}^{-1}\delta_{\rho,\sigma},
 \dot\Gamma_t C_{u,t}^{-1}\delta_{\rho,\sigma}
 \big\rangle_{\rm HS},
 \qquad 0\leq u\leq1.
 \label{rw:eq:universal-cone-response-function}
\end{equation}
Then the operator-convex part of the universal cone is exactly
\begin{equation}
 \boxed{
 f\in\mathfrak M
 \quad\Longleftrightarrow\quad
 \int_{[0,1]}\kappa_{\rho,\sigma,t}(u)\,
 \mathrm d\nu_f(u)\geq0
 }
 \label{rw:eq:universal-cone-measure-polar}
\end{equation}
for every finite-dimensional invertible pair \((\rho,\sigma)\) and
every \(t\in(0,1)\).  Equivalently, without differentiation, for
every \(0\leq t_1<t_2\leq1\),
\begin{equation}
 \int_{[0,1]}
 \big\langle\delta_{\rho,\sigma},
 (C_{u,t_2}^{-1}-C_{u,t_1}^{-1})
 \delta_{\rho,\sigma}\big\rangle_{\rm HS}
 \,\mathrm d\nu_f(u)\geq0.
 \label{rw:eq:universal-cone-measure-finite-difference}
\end{equation}
In other words, modulo affine functions, the operator-convex slice of
\(\mathfrak M\) is the positive polar of the explicit family of
response functions in
\eqref{rw:eq:universal-cone-response-function}.
\end{theorem}

\begin{proof}
Since \(C_{u,t}\) is a function of \(\Gamma_t\), it commutes with
\(\Gamma_t-I\), and functional calculus gives
\begin{equation}
 \begin{split}
 D_{q_u}^t(\rho\Vert\sigma)
 &=\big\langle v,
 (\Gamma_t-I)^2C_{u,t}^{-1}v\big\rangle_{\rm HS}\\
 &=\big\langle\delta_{\rho,\sigma},
 C_{u,t}^{-1}\delta_{\rho,\sigma}\big\rangle_{\rm HS}.
 \end{split}
 \label{rw:eq:universal-cone-kernel-resolvent}
\end{equation}
The identity for \(\delta_{\rho,\sigma}\) follows from
\eqref{rw:eq:gamma-normalization}.  In particular, its derivative
vanishes.  Differentiating the inverse in
\eqref{rw:eq:universal-cone-kernel-resolvent} gives
\begin{equation}
 \frac{\mathrm d}{\mathrm dt}D_{q_u}^t(\rho\Vert\sigma)
 =\kappa_{\rho,\sigma,t}(u).
 \label{rw:eq:universal-cone-kernel-derivative}
\end{equation}
The affine part of
\eqref{rw:eq:universal-cone-compact-representation} is independent of
\(t\).  Integration of
\eqref{rw:eq:universal-cone-kernel-resolvent} and
\eqref{rw:eq:universal-cone-kernel-derivative} against the finite
measure \(\nu_f\) proves both
\eqref{rw:eq:universal-cone-measure-polar} and
\eqref{rw:eq:universal-cone-measure-finite-difference}.  The exchange
of differentiation and integration is harmless: for a fixed invertible
pair, \(C_{u,t}\) and its inverse are uniformly bounded on
\([0,1]\times[\varepsilon,1-\varepsilon]\), and \(\nu_f\) is finite.
\end{proof}

The characterization in
\Cref{rw:thm:universal-cone-matrix-characterization} applies to all
\(C^1\) generators, including the nonconvex powers occurring beyond
the quadratic order.  Its specialization in
\Cref{rw:thm:universal-cone-measure-characterization} is stronger than
a support condition on \(\nu_f\): the response functions can change
sign, and universal monotonicity can result from cancellations across
the whole representing measure.  The atom at \(u=0\) is invisible,
since \(q_0(x)=(x-1)^2\) and
\(\kappa_{\rho,\sigma,t}(0)=0\).

Linearity in $f$ makes $\mathfrak M$ a convex cone.  There is a useful
elementary subcone.  For
\begin{equation}
 f(x)=a+bx+cx^2+d\,x\log x,
 \label{rw:eq:elementary-monotone-cone}
\end{equation}
linearity of the construction and \eqref{rw:eq:Q2-constant} give
\begin{equation}
 D_f^t(\rho\Vert\sigma)
 =a+b+c\operatorname{Tr}(\rho^2\sigma^{-1})
  +dD^t(\rho\Vert\sigma).
 \label{rw:eq:elementary-monotone-cone-value}
\end{equation}
Therefore, it follows that $t\mapsto D_f^t(\rho\Vert\sigma)$ is universally
non-decreasing when $d\geq0$, and it is strict on noncommuting pairs
when $d>0$.
Within the operator-convex subclass one takes $c,d\geq0$.  This is an
elementary sufficient cone inside the exact cones characterized by
\eqref{rw:eq:universal-cone-matrix-characterization} and
\eqref{rw:eq:universal-cone-measure-polar}; the weighted-$\chi^2$
example shows why the answer cannot be reduced to operator convexity
alone.

The non-polynomial extreme kernels in
\eqref{rw:eq:operator-convex-representation} show why the preceding
problem is genuinely global in the representing measure.

\begin{proposition}[No individual resolvent kernel is universally
monotone]
\label{rw:prop:kernel-not-monotone}
For $s>0$, we set
\begin{equation}
 k_s(x):=\frac{x}{1+s}-\frac{x}{x+s}.
 \label{rw:eq:resolvent-kernel}
\end{equation}
Every $k_s$ is operator convex, but $k_s\notin\mathfrak M$ for
every $s>0$.
\end{proposition}

\begin{proof}
It is enough to disprove membership in $\mathfrak M$ at one convenient
base point $s_0>0$.  Indeed, we shall show below that all the kernels
$k_s$ are obtained from any one of them by dilating the argument and
adding a linear function, two operations that do not affect membership
in $\mathfrak M$.

We choose $s_0=1/10$, for which an exact rational counterexample is
particularly simple, and take
\begin{equation}
 \rho=\begin{pmatrix}9/10&0\\0&1/10\end{pmatrix},\qquad
 \sigma=\begin{pmatrix}1/2&2/5\\2/5&1/2\end{pmatrix}.
\end{equation}
Direct exact calculation on Hilbert--Schmidt space gives
\begin{align}
 D_{k_{s_0}}^0(\rho\Vert\sigma)=\frac{3200}{19019},\qquad
 D_{k_{s_0}}^{1/2}(\rho\Vert\sigma)
 =\frac{1014400}{6175697},\qquad
 D_{k_{s_0}}^{1/2}-D_{k_{s_0}}^0
 =-\frac{42668800}{10677780113}<0.
 \label{rw:eq:kernel-exact-counterexample}
\end{align}
Hence $k_{s_0}\notin\mathfrak M$.

We now justify the transport from this single base point.  First, for
every $\lambda>0$,
\begin{equation*}
 f\in\mathfrak M
 \quad\Longleftrightarrow\quad
 f(\lambda\,\cdot)\in\mathfrak M.
\end{equation*}
To prove the forward implication, we choose $p,q\in(0,1)$ with
$p/q=\lambda$ and let $\tau,\omega$ be an arbitrary invertible pair.
The flagged-sum rule in
\eqref{rw:eq:general-flag-rule-Dft} gives
\begin{align}
 &D_f^t\bigl(p\tau\oplus(1-p)\,\big\Vert\,
                 q\omega\oplus(1-q)\bigr)
 =qD_{f(\lambda\,\cdot)}^t(\tau\Vert\omega)
 +(1-q)f\left(\frac{1-p}{1-q}\right).
 \label{rw:eq:universal-cone-dilation}
\end{align}
The last term is independent of $t$ and $q>0$.  Thus universal
non-decrease of the left-hand side implies universal non-decrease of
$D_{f(\lambda\,\cdot)}^t$.  Applying this implication to
$f(\lambda\,\cdot)$ with dilation factor $1/\lambda$ proves the
converse.  Moreover, adding an affine function does not affect
membership in $\mathfrak M$, because
\eqref{rw:eq:affine-generator-normalization} shows that its divergence
is independent of $t$.

We now fix an arbitrary target value $s>0$ and set
$\lambda=s/s_0$.  The identity
\begin{equation}
 k_s(\lambda x)=k_{s/\lambda}(x)
 +\left(\frac{\lambda}{1+s}
       -\frac{1}{1+s/\lambda}\right)x.
\end{equation}
then becomes $k_s(\lambda x)=k_{s_0}(x)+c_{s,\lambda}x$, where
$c_{s,\lambda}:=\lambda/(1+s)-1/(1+s/\lambda)$.  Consequently,
\begin{equation*}
 k_s\in\mathfrak M
 \quad\Longleftrightarrow\quad
 k_s(\lambda\,\cdot)\in\mathfrak M
 \quad\Longleftrightarrow\quad
 k_{s_0}+c_{s,\lambda}\,\mathrm{id}\in\mathfrak M
 \quad\Longleftrightarrow\quad
 k_{s_0}\in\mathfrak M.
\end{equation*}
The last statement is false by
\eqref{rw:eq:kernel-exact-counterexample}.  Since $s>0$ was arbitrary,
$k_s\notin\mathfrak M$ for every $s>0$.
\end{proof}

The ordinary-matrix characterization also yields a useful scalar test
which is invisible in the first-order formula
\eqref{rw:eq:general-f-divided-difference}.

\begin{proposition}[A necessary differential test for the universal cone]
\label{rw:prop:universal-cone-differential-test}
If \(f\in\mathfrak M\cap C^4(0,\infty)\), then, for every \(a>1\),
\begin{equation}
 x\longmapsto x^{a+2}f'''(x)
 \quad\text{is non-increasing on }(0,\infty).
 \label{rw:eq:universal-cone-all-a-test}
\end{equation}
Equivalently,
\begin{equation}
 f'''(x)\leq0,
 \qquad
 \frac{\mathrm d}{\mathrm dx}\bigl[x^3f'''(x)\bigr]\leq0,
 \qquad x>0.
 \label{rw:eq:universal-cone-third-derivative-test}
\end{equation}
\end{proposition}

\begin{proof}
We fix \(a>1\) and take \(P\) diagonal in the basis
\((e_l)_l\), with two distinct eigenvalues \(p_i<p_j\).  We perturb a
scalar \(K=\kappa I\) in a single off-diagonal direction:
\begin{equation*}
 K_\varepsilon
 =c_\varepsilon\kappa(I+\varepsilon X),
 \qquad
 X=X^*,\qquad X_{kl}=0
 \quad\text{unless }\{k,l\}=\{i,j\}.
\end{equation*}
Here \(\kappa=\operatorname{Tr}P^{a-1}\), and
\(c_\varepsilon=1+O(\varepsilon^2)\) is chosen so that the second
condition in \eqref{rw:eq:universal-cone-matrix-normalization} is
preserved.  Expanding the L\"owner expression in
\eqref{rw:eq:universal-cone-matrix-characterization} at
\(\varepsilon=0\) gives
\begin{equation}
 \begin{split}
 \mathcal L_f(a;P,K_\varepsilon)
 ={}&\frac{a\kappa^2\varepsilon^2}{6}|X_{ij}|^2
 (\log p_j-\log p_i)\\
 &\quad\times\left[
 p_i^{a+2}f'''(\kappa p_i)
 -p_j^{a+2}f'''(\kappa p_j)
 \right]+O(\varepsilon^3).
 \end{split}
 \label{rw:eq:universal-cone-local-expansion}
\end{equation}
The normalization factor does not change the displayed coefficient:
\(\mathcal L_f(a;P,\lambda I)=0\) for every \(\lambda>0\).
Nonnegativity for both signs of sufficiently small \(\varepsilon\)
therefore orders the two bracketed values.  Arbitrary eigenvalue
ratios are obtained by varying \(P\), and arbitrary common scales by
the dilation invariance in
\eqref{rw:eq:universal-cone-dilation}.  This proves
\eqref{rw:eq:universal-cone-all-a-test}.

We put \(h(x):=x^3f'''(x)\).  Since
\begin{equation*}
 \frac{\mathrm d}{\mathrm dx}
 \bigl[x^{a+2}f'''(x)\bigr]
 =x^{a-2}\bigl[(a-1)h(x)+xh'(x)\bigr],
\end{equation*}
letting \(a\downarrow1\) gives \(h'\leq0\), while letting \(a\) grow
forces \(h\leq0\).  The converse implication is immediate, proving
the equivalence with
\eqref{rw:eq:universal-cone-third-derivative-test}.
\end{proof}

\begin{remark}[Why scalar and higher-tone tests are not sufficient]
\label{rw:rem:universal-cone-false-simplifications}
The test in
\eqref{rw:eq:universal-cone-third-derivative-test} already detects
both elementary obstructions above.  For the resolvent kernel,
\begin{equation*}
 k_s'''(x)=-\frac{6s}{(x+s)^4},
 \qquad
 \frac{\mathrm d}{\mathrm dx}[x^3k_s'''(x)]
 =\frac{6sx^2(x-3s)}{(x+s)^5},
\end{equation*}
which is positive for \(x>3s\).  For
\(h(x)=x-2+x^{-1}\), one has
\((x^3h'''(x))'=6x^{-2}>0\).  The condition is not sufficient:
\(-x^\alpha\) satisfies
\eqref{rw:eq:universal-cone-third-derivative-test} for every
\(\alpha>2\), whereas
\Cref{rw:thm:power-t-monotonicity} rules out universal monotonicity
when \(\alpha>3\).

Nor is \(\mathfrak M\) the negative operator-
\(3\)-tone cone suggested by the alternating power signs.  The
canonical representation of operator-\(3\)-tone functions
\cite[Theorem~5.1]{franz2014higher} contains
\(x^3/(x+s)\) as a positive kernel and arbitrary quadratic
polynomials as null directions.  But
\begin{equation}
 -k_s(x)=\frac{x^3}{s^2(x+s)}-\frac{x^2}{s^2}
 +\frac{x}{s(1+s)},
 \label{rw:eq:resolvent-negative-three-tone}
\end{equation}
so \(-k_s\) is operator \(3\)-tone while
\(k_s\notin\mathfrak M\) by
\Cref{rw:prop:kernel-not-monotone}.  Hence the full matrix condition
\eqref{rw:eq:universal-cone-matrix-characterization}, or its exact
measure-polar specialization, is genuinely stronger than these
scalar function classes.
\end{remark}

\begin{remark}[A necessary condition on the representing measure]
The preceding proposition does not prevent cancellations between
different kernels: for instance,
$x\log x=\int_0^\infty k_s(x)\,\mathrm ds$ belongs to
$\mathfrak M$.  There is nevertheless a simple necessary condition
on the measure \(\nu_f\) in
\eqref{rw:eq:universal-cone-compact-representation}.
If $f\in\mathfrak M$, then
\begin{equation}
 \nu_f(\{1\})=0.
 \label{rw:eq:no-weighted-chi-atom}
\end{equation}
Indeed, for $r>0$ we set $f_r(x):=r f(rx)$.  Dilation invariance
\eqref{rw:eq:universal-cone-dilation} and closure under positive
scaling give $f_r\in\mathfrak M$.  If $\nu_f(\{1\})>0$, the compact
representation above shows, as $r\downarrow0$, that the contribution
of the atom at $u=1$ converges to $\nu_f(\{1\})x^{-1}$, whereas the
affine terms and the part supported on $[0,1)$ vanish.  This
convergence is locally uniform on $(0,\infty)$.  Indeed, on
$K=[a,b]\Subset(0,\infty)$ and for $rb\leq1$, the rescaled kernel is
bounded by $2r$ on $u\leq1/2$ and by $2/a$ on $u>1/2$; for each
$u<1$ it converges uniformly on $K$ to zero, while at $u=1$ it
converges uniformly to $x^{-1}$.  Dominated convergence therefore
gives the claimed local-uniform limit.
Consequently the defining inequalities for $\mathfrak M$ pass to the
limit on every finite-dimensional spectrum, so that a positive
multiple of $x\mapsto x^{-1}$ would belong to $\mathfrak M$.  This is
impossible because it differs by an affine function from the
weighted-$\chi^2$ generator
\eqref{rw:eq:weighted-chi-generator}.  Condition
\eqref{rw:eq:no-weighted-chi-atom} is necessary but not sufficient.
Indeed, a point mass at any \(u\in(0,1)\) satisfies it, while its
kernel \(q_u\) differs by a positive factor and an affine function
from \(k_{(1-u)/u}\), which is excluded by
\Cref{rw:prop:kernel-not-monotone}.  The exact replacement is the
polar condition
\eqref{rw:eq:universal-cone-measure-polar}.
\end{remark}

\begin{corollary}[A larger operator-convex universal cone]
\label{rw:cor:larger-monotone-cone}
For $\alpha\in(0,1)\cup(1,2)$, we define
$g_\alpha(x):=\operatorname{sgn}(\alpha-1)x^\alpha$.  We let $\mu$
be a finite positive Borel measure on
$(0,1)\cup(1,2)$.  Every function
\begin{equation}
 \begin{split}
 f(x)={}&a+bx+cx^2+d\,x\log x+e(-\log x)+\int_{(0,1)\cup(1,2)}g_\alpha(x)\,\mathrm d\mu(\alpha),
 \qquad c,d,e\geq0,
 \end{split}
 \label{rw:eq:larger-monotone-cone}
\end{equation}
is operator convex on $(0,\infty)$ and belongs to $\mathfrak M$.
\end{corollary}

\begin{proof}
For $0<\alpha<1$, $-x^\alpha$ is operator convex and
$t\mapsto D_{-x^\alpha}^t=-Q_{\alpha,t}$ is non-decreasing; for
$1<\alpha<2$, the same holds for
$t\mapsto D_{x^\alpha}^t=Q_{\alpha,t}$, by
\Cref{rw:thm:power-t-monotonicity}.
The polynomial and logarithmic terms were treated above.  Moreover,
$(1-x^\alpha)/\alpha\in\mathfrak M$ for $0<\alpha<1$ and
converges locally uniformly to $-\log x$ as $\alpha\downarrow0$.
Finite-dimensional functional calculus, positive integration, and
positive linear combinations preserve the conclusion.
\end{proof}

\begin{remark}[The power rays in canonical coordinates]
The explicit subcone in
\Cref{rw:cor:larger-monotone-cone} has a simple description inside
the polar cone
\eqref{rw:eq:universal-cone-measure-polar}.  We normalize the power
generators by
\begin{equation}
 \ell_\alpha(x)
 :=\frac{x^\alpha-\alpha x+\alpha-1}{\alpha(\alpha-1)},
 \qquad 0<\alpha<2,
 \label{rw:eq:normalized-power-generators}
\end{equation}
with continuous extensions
\begin{equation*}
 \ell_0(x)=x-1-\log x,
 \qquad
 \ell_1(x)=x\log x-x+1,
 \qquad
 \ell_2(x)=\frac12(x-1)^2.
\end{equation*}
For \(0\leq\alpha<2\), the canonical measure of
\(\ell_\alpha\) is the beta measure
\begin{equation}
 \mathrm d\nu_\alpha(u)
 =\frac{u^{1-\alpha}(1-u)^\alpha}
 {\Gamma(\alpha+1)\Gamma(2-\alpha)}\,\mathrm du,
 \qquad 0<u<1,
 \label{rw:eq:power-beta-representing-measure}
\end{equation}
while \(\nu_2=\frac12\delta_0\).  Indeed,
\(\ell_\alpha''(x)=x^{\alpha-2}\),
\(q_u''(x)=2[1+u(x-1)]^{-3}\), and Euler's beta integral gives
\eqref{rw:eq:power-beta-representing-measure}.  Thus every positive
mixture of these beta measures satisfies the polar inequalities.  In
the noncompact representation
\eqref{rw:eq:operator-convex-representation}, the same rays are the
Mellin densities proportional to \(s^{\alpha-1}\,\mathrm ds\).
\end{remark}

\begin{remark}[Dimension and the equality hierarchy]
Algebraically, $-x^\alpha\in\mathfrak M$ also holds for
$2<\alpha\leq3$, although these generators are not convex.  For
$\alpha>3$, neither sign of $x^\alpha$ belongs to $\mathfrak M$
in dimensions at least three.  In dimension two, however,
$D_\alpha^t$ remains non-increasing in $t$ for every $\alpha>2$
\cite[Theorem~4.4]{vuong2026hyperbolic}; dimension three is the first
dimension admitting the failure.

Finally, monotonicity of a divergence on a fixed pair does not order
the corresponding DPI deficits, which are differences between two
evaluations of the same $t$-dependent divergence.  The equality
hierarchy in
\Cref{sec:consolidated-recovery,sec:consolidated-tqmc} comes instead
from Stinespring
rigidity and Petz sufficiency.
\end{remark}

In summary, monotonicity for a fixed pair is governed by the divided
difference criterion in
\Cref{rw:prop:general-f-pairwise-criterion}, whereas universal
monotonicity is a property of the generator.  The $t$-relative
entropies $D^t$ and the $(t,\alpha)$-R\'enyi divergences
$D_\alpha^t$ provide the principal universal examples, but the
resolvent and weighted-$\chi^2$ counterexamples show that operator
convexity is not sufficient.  The full \(C^1\) cone is characterized
by the ordinary-matrix L\"owner test in
\Cref{rw:thm:universal-cone-matrix-characterization}, while its
operator-convex slice is equivalently the resolvent polar cone in
\Cref{rw:thm:universal-cone-measure-characterization}.

%% file: dpi_and_recovery_sections.tex
% Paper-ready replacement for the DPI, recovery, Markov, and
% quantitative-reconstruction material.

\section{Data processing and quantitative remainders}
\label{sec:consolidated-dpi}

A quantum distinguishability functional is required to satisfy the
data-processing inequality under quantum channels in order to qualify
as a quantum divergence; this requirement is discussed, for example, in
\cite[Section~1]{hiai2017different}.  Thus, to justify the terminology
introduced in \Cref{rw:def:interpolated-f-divergence}, we must first
prove data processing for \(D_f^t\).  The procedure followed in this
section is inspired by Petz's proof of the monotonicity of quantum
relative entropy \cite{petz2003monotonicity} and by the quantitative
refinement of that method due to Carlen and Vershynina
\cite{carlen2018recovery}.  We do so from a single transformer
contraction and then record the principal special cases.
Next, we retain the positive resolvent terms discarded by that proof,
compress them into one intrinsic defect, and derive the relative entropies,
R\'enyi divergences, and tripartite quantitative remainders used later.  Exact
equality and recovery are separated from these estimates and treated
in \Cref{sec:consolidated-recovery}.

Throughout, all Hilbert spaces are finite dimensional.  We take
$\rho,\sigma>0$ to be states and
$\Phi:\mathcal B(\mathcal H)\to\mathcal B(\mathcal K)$ to be a
channel, and we assume
\begin{equation}
 \Phi(\rho)>0,\qquad \Phi(\sigma)>0.
\end{equation}
The output assumption entails no loss for an invertible input pair.
Indeed, $c\sigma\leq\rho\leq C\sigma$ for some $c,C>0$, so
$c\Phi(\sigma)\leq\Phi(\rho)\leq C\Phi(\sigma)$ and the two outputs
have the same support.  Compressing the output algebra to that support
makes both output states invertible; we keep the displayed assumption
only to avoid repeating this compression in every formula.
For $t\in[0,1]$, we define the geodesic modular operators
\begin{align}
 \Gamma_t
 &:=L_{\rho^{1-t}}
 R_{\sigma^{-1/2}\rho^t\sigma^{-1/2}},&
 \Gamma_t^\Phi
 &:=L_{[\Phi(\rho)]^{1-t}}
 R_{[\Phi(\sigma)]^{-1/2}[\Phi(\rho)]^t[\Phi(\sigma)]^{-1/2}}.
 \label{eq:consolidated-modular-operators}
\end{align}
Thus
\begin{equation}
 D_f^t(\rho\Vert\sigma)
 =\left\langle\sigma^{1/2},
 f(\Gamma_t)\sigma^{1/2}\right\rangle_{\rm HS}.
 \label{eq:consolidated-Dft}
\end{equation}

\subsection{Data processing in the general case}

The main result of this section is that every $(t,f)$-divergence
generated by a finite operator-convex $f$ satisfies a data-processing
inequality. To prove that, the input and output modular operators are compared by
\begin{equation}
 \mathsf V_{\sigma,\Phi}(X)
 :=\Phi^*(X[\Phi(\sigma)]^{-1/2})\sigma^{1/2}.
 \label{eq:consolidated-V}
\end{equation}
This map is a Hilbert--Schmidt contraction, not an isometry in
general; it is the standard comparison map for quasi-entropies
\cite{petz1986quasi,hiai2011quantum}.

\begin{lemma}[Transformer contraction]
\label{lem:consolidated-transformer}
For every $t\in[0,1]$,
\begin{equation}
 \mathsf V_{\sigma,\Phi}([\Phi(\sigma)]^{1/2})=\sigma^{1/2},
 \qquad
 \mathsf V_{\sigma,\Phi}^{*}\Gamma_t
 \mathsf V_{\sigma,\Phi}\leq\Gamma_t^\Phi.
 \label{eq:consolidated-transformer}
\end{equation}
\end{lemma}

\begin{proof}
The first identity follows from $\Phi^*(I)=I$, and the
Kadison--Schwarz inequality for $\Phi^*$
\cite{choi1974schwarz} gives
\begin{equation}
 \lVert\mathsf V_{\sigma,\Phi}\rVert_{2\to2}\leq1.
\end{equation}
For the second relation, we recall that the power quasi-entropy
\begin{equation*}
 \mathcal S_t^K(A\Vert B)
 :=\operatorname{Tr}\!\left[K^*A^{1-t}KB^t\right],
 \qquad 0\leq t\leq1,
\end{equation*}
satisfies Petz's transformer inequality
\begin{equation*}
 \mathcal S_t^{\Phi^*(K)}(A\Vert B)
 \leq \mathcal S_t^K\bigl(\Phi(A)\Vert\Phi(B)\bigr).
\end{equation*}
Proofs of this inequality appear in \cite[Theorem~4]{petz1986quasi} and
\cite[Section~3.6]{petz2007quantum}.
We fix $X\in\mathcal B(\mathcal K)$ and set
$K:=X[\Phi(\sigma)]^{-1/2}$.  Then
\begin{align*}
 \left\langle X,\mathsf V_{\sigma,\Phi}^*\Gamma_t
          \mathsf V_{\sigma,\Phi}X\right\rangle_{\rm HS}
 &=\operatorname{Tr}\!\left[
   \bigl(\Phi^*(K)\sigma^{1/2}\bigr)^*
   \rho^{1-t}\Phi^*(K)\sigma^{1/2}
   \sigma^{-1/2}\rho^t\sigma^{-1/2}\right]\\
 &=\operatorname{Tr}\!\left[\Phi^*(K)^*\rho^{1-t}
                    \Phi^*(K)\rho^t\right]\\
 &=\mathcal S_t^{\Phi^*(K)}(\rho\Vert\rho)\\
 &\leq\mathcal S_t^K\bigl(\Phi(\rho)\Vert\Phi(\rho)\bigr)\\
 &=\operatorname{Tr}\!\left[[\Phi(\sigma)]^{-1/2}X^*
      [\Phi(\rho)]^{1-t}X[\Phi(\sigma)]^{-1/2}
      [\Phi(\rho)]^t\right]\\
 &=\left\langle X,\Gamma_t^\Phi X\right\rangle_{\rm HS},
\end{align*}
where the second and last equalities use cyclicity of the trace.
Since $X$ is arbitrary, the second relation in
\eqref{eq:consolidated-transformer} follows.
\end{proof}

We recall from \eqref{rw:eq:operator-convex-representation} that every
finite operator-convex function $f$ on $[0,\infty)$ admits the
representation
\begin{equation}
 f(x)=f(0)+ax+bx^2+
 \int_0^\infty\left(
 \frac{x}{1+s}-\frac{x}{x+s}\right)\,\mathrm d\mu_f(s),
 \qquad b\geq0.
 \label{eq:consolidated-f-representation}
\end{equation}
Singular endpoints can instead be handled on the appropriate support
by regularization.

\begin{theorem}[Data processing for the geodesic divergences]
\label{thm:consolidated-dpi}
For every $t\in[0,1]$,
\begin{equation}
 D_f^t(\Phi(\rho)\Vert\Phi(\sigma))
 \leq D_f^t(\rho\Vert\sigma).
 \label{eq:consolidated-dpi}
\end{equation}
\end{theorem}

\begin{proof}
The constant contribution is unchanged because both $\sigma$ and
$\Phi(\sigma)$ have trace one.  The linear contribution is also
unchanged, since
\begin{equation}
 \left\langle\sigma^{1/2},\Gamma_t\sigma^{1/2}\right\rangle
 =\left\langle[\Phi(\sigma)]^{1/2},
 \Gamma_t^\Phi[\Phi(\sigma)]^{1/2}\right\rangle=1.
\end{equation}

We next treat the integral kernels in
\eqref{eq:consolidated-f-representation}.  For $s>0$, we put
$\varphi_s(x):=-x/(x+s)$ and
$k_s(x):=x/(1+s)+\varphi_s(x)$.  The function $\varphi_s$ is operator convex,
operator monotone decreasing, and satisfies $\varphi_s(0)=0$.
The contractive Jensen inequality
\cite[Corollary~2.3]{hansen2003jensen}, applied to
$\mathsf V_{\sigma,\Phi}$, and the transformer inequality in
\eqref{eq:consolidated-transformer} therefore give, in two separate
steps,
\begin{equation}
 \varphi_s(\Gamma_t^\Phi)
 \leq\varphi_s\!\left(
       \mathsf V_{\sigma,\Phi}^*\Gamma_t
       \mathsf V_{\sigma,\Phi}\right)
 \leq\mathsf V_{\sigma,\Phi}^*\varphi_s(\Gamma_t)
       \mathsf V_{\sigma,\Phi}.
\end{equation}
Here the first inequality uses that $\varphi_s$ reverses the operator
order, whereas the second one is operator Jensen.  Using the first
identity in \eqref{eq:consolidated-transformer}, we consequently obtain
\begin{align*}
 \left\langle[\Phi(\sigma)]^{1/2},
   \varphi_s(\Gamma_t^\Phi)[\Phi(\sigma)]^{1/2}\right\rangle_{\rm HS}
&\leq
 \left\langle[\Phi(\sigma)]^{1/2},
   \mathsf V_{\sigma,\Phi}^*\varphi_s(\Gamma_t)
   \mathsf V_{\sigma,\Phi}[\Phi(\sigma)]^{1/2}
 \right\rangle_{\rm HS}\\
 &=
 \left\langle\sigma^{1/2},
   \varphi_s(\Gamma_t)\sigma^{1/2}\right\rangle_{\rm HS}.
\end{align*}
The remaining part $x/(1+s)$ of $k_s$ is linear and hence has
the same expectation at the input and output.  Thus every integral
kernel in \eqref{eq:consolidated-f-representation} satisfies the
desired DPI.

It remains to justify the quadratic term.  Since $\Gamma_t$ is
positive and self-adjoint on Hilbert--Schmidt space, its apparent
$t$-dependence cancels explicitly:
\begin{align}
 D_{x^2}^t(\rho\Vert\sigma)
 &=\left\lVert\Gamma_t\sigma^{1/2}\right\rVert_2^2\notag\\
 &=\left\lVert
   \rho^{1-t}\sigma^{1/2}\sigma^{-1/2}
   \rho^t\sigma^{-1/2}\right\rVert_2^2\notag\\
 &=\left\lVert\rho\sigma^{-1/2}\right\rVert_2^2
 =\operatorname{Tr}(\rho^2\sigma^{-1}).
 \label{eq:consolidated-quadratic-independence}
\end{align}
To prove data processing directly, we abbreviate
$\rho_\Phi:=\Phi(\rho)$ and $\sigma_\Phi:=\Phi(\sigma)$.  For every
$X\in\mathcal B(\mathcal K)$, cyclicity of the trace and the
definition of the Hilbert--Schmidt adjoint give
\begin{equation*}
\begin{split}
 \left\langle X,
 \mathsf V_{\sigma,\Phi}^*(\rho\sigma^{-1/2})\right\rangle_{\rm HS}
 &=\operatorname{Tr}\!\left[
   \sigma^{1/2}\Phi^*(\sigma_\Phi^{-1/2}X^*)
   \rho\sigma^{-1/2}\right]\\
 &=\operatorname{Tr}\!\left[
   \Phi^*(\sigma_\Phi^{-1/2}X^*)\rho\right]\\
 &=\operatorname{Tr}\!\left[\sigma_\Phi^{-1/2}X^*\rho_\Phi\right]
 =\langle X,\rho_\Phi\sigma_\Phi^{-1/2}\rangle_{\rm HS}.
\end{split}
\end{equation*}
Hence
$\mathsf V_{\sigma,\Phi}^*(\rho\sigma^{-1/2})
=\rho_\Phi\sigma_\Phi^{-1/2}$.
Since $\mathsf V_{\sigma,\Phi}$ is a contraction, its adjoint satisfies
$\lVert\mathsf V_{\sigma,\Phi}^*\rVert_{2\to2}\leq1$, and therefore
\begin{align*}
 D_{x^2}^t(\rho_\Phi\Vert\sigma_\Phi)
 =\left\lVert\rho_\Phi\sigma_\Phi^{-1/2}\right\rVert_2^2
 =\left\lVert
   \mathsf V_{\sigma,\Phi}^*(\rho\sigma^{-1/2})
  \right\rVert_2^2
 \leq\left\lVert\rho\sigma^{-1/2}\right\rVert_2^2
 =D_{x^2}^t(\rho\Vert\sigma).
\end{align*}
Finally, the constant and linear terms cancel in the DPI deficit, so
the representation \eqref{eq:consolidated-f-representation} yields
\begin{align*}
 D_f^t(\rho\Vert\sigma)
   -D_f^t(\rho_\Phi\Vert\sigma_\Phi)
 &=b\!\left[
   D_{x^2}^t(\rho\Vert\sigma)
   -D_{x^2}^t(\rho_\Phi\Vert\sigma_\Phi)\right]
+\int_0^\infty\!\left[
   D_{k_s}^t(\rho\Vert\sigma)
   -D_{k_s}^t(\rho_\Phi\Vert\sigma_\Phi)
   \right]\,\mathrm d\mu_f(s)\geq0,
\end{align*}
because $b\geq0$ and $\mu_f$ is a positive measure.  This is
\eqref{eq:consolidated-dpi}.
\end{proof}

This yields a data-processing inequality in the following particular cases.

\begin{corollary}[Relative, R\'enyi, and quadratic cases]
\label{cor:consolidated-basic-examples}
We put
\begin{equation}
 D^t:=D_{x\log x}^t,\qquad
 Q_{\alpha,t}:=
 \left\langle\sigma^{1/2},\Gamma_t^\alpha
 \sigma^{1/2}\right\rangle,\qquad
 D_\alpha^t:=\frac{\log Q_{\alpha,t}}{\alpha-1}.
\end{equation}
Then $D^t$ satisfies the DPI for every $t$.  The R\'enyi
interpolation satisfies the DPI for
\begin{equation}
 0<\alpha<1,\qquad 1<\alpha\leq2;
\end{equation}
Below one, we use $-x^\alpha$, and above one, we use $x^\alpha$.
At $\alpha=2$ the divergence is independent of $t$.  The
geodesic $\chi^2$-divergence generated by $(x-1)^2$ is also
independent of $t$.  In the displayed DPI ranges, nonnegativity also
follows by applying the DPI to the trace channel.  The stronger fact
that $D_\alpha^t\geq0$ for every $\alpha>0$, independently of whether
DPI holds, is \Cref{rw:prop:basic-spectral-properties}.
No general DPI is asserted for $\alpha>2$, as happens for geometric Rényi divergences.
\end{corollary}

\begin{corollary}[Joint convexity and the R\'enyi distinction]
\label{cor:consolidated-joint-convexity}
For every operator-convex $f$ in the range of
\Cref{thm:consolidated-dpi}, the map
$(\rho,\sigma)\mapsto D_f^t(\rho\Vert\sigma)$ is jointly convex.
Consequently, $Q_{\alpha,t}$ is jointly concave for
$0<\alpha<1$ and jointly convex for $1<\alpha\leq2$.  The
normalized R\'enyi divergence $D_\alpha^t$ is jointly convex for
$0<\alpha<1$, but for $1<\alpha\leq2$ the general consequence is
only joint quasi-convexity:
\begin{equation}
 D_\alpha^t\!\left(\sum_i p_i\rho_i\middle\Vert
                         \sum_i p_i\sigma_i\right)
 \leq\max_i D_\alpha^t(\rho_i\Vert\sigma_i).
 \label{eq:consolidated-Renyi-quasiconvexity}
\end{equation}
\end{corollary}

\begin{proof}
We use the standard block-diagonal argument from Petz's proof of
joint convexity of quasi-entropies
\cite[Theorem~6]{petz1986quasi}.
We take $p_i>0$ with $\sum_i p_i=1$; zero weights may simply be omitted.
We introduce an auxiliary classical register $F$ with orthonormal basis
$\{\lvert i\rangle\}_i$ and define the block-diagonal states
\begin{align*}
 \widehat\rho_{FA}
 &:=\sum_i p_i\lvert i\rangle\!\langle i\rvert_F\otimes\rho_i,
 &
 \widehat\sigma_{FA}
 &:=\sum_i p_i\lvert i\rangle\!\langle i\rvert_F\otimes\sigma_i.
\end{align*}
Their partial traces over $F$ are precisely the mixtures
\begin{equation*}
 \operatorname{Tr}_F\widehat\rho_{FA}=\sum_i p_i\rho_i,
 \qquad
 \operatorname{Tr}_F\widehat\sigma_{FA}=\sum_i p_i\sigma_i,
\end{equation*}
and orthogonality of the blocks 
gives
\begin{equation}
 D_f^t(\widehat\rho_{FA}\Vert\widehat\sigma_{FA})
 =\sum_i p_iD_f^t(\rho_i\Vert\sigma_i).
\end{equation}
The map $\operatorname{Tr}_F$ is a quantum channel.  Applying
\Cref{thm:consolidated-dpi} to this particular channel yields
\begin{align*}
 D_f^t\!\left(\sum_i p_i\rho_i\middle\Vert
                    \sum_i p_i\sigma_i\right)
 =D_f^t\!\left(
     \operatorname{Tr}_F\widehat\rho_{FA}\middle\Vert
     \operatorname{Tr}_F\widehat\sigma_{FA}\right)
 \leq D_f^t(\widehat\rho_{FA}\Vert\widehat\sigma_{FA})
 =\sum_i p_iD_f^t(\rho_i\Vert\sigma_i),
\end{align*}
which is exactly joint convexity.

It remains to spell out the R\'enyi consequences.  We write
$\bar\rho:=\sum_i p_i\rho_i$ and
$\bar\sigma:=\sum_i p_i\sigma_i$.  For $0<\alpha<1$, the
operator-convex generator is $-x^\alpha$, and the preceding joint
convexity inequality becomes
\begin{equation*}
 Q_{\alpha,t}(\bar\rho\Vert\bar\sigma)
 \geq\sum_i p_iQ_{\alpha,t}(\rho_i\Vert\sigma_i).
\end{equation*}
Thus $Q_{\alpha,t}$ is jointly concave.  The scalar function
\begin{equation*}
 h_\alpha(q):=\frac{\log q}{\alpha-1}
\end{equation*}
is convex and decreasing in this range.  Consequently,
\begin{align*}
 D_\alpha^t(\bar\rho\Vert\bar\sigma)
 =h_\alpha\!\left(Q_{\alpha,t}(\bar\rho\Vert\bar\sigma)\right)
 \leq h_\alpha\!\left(
       \sum_i p_iQ_{\alpha,t}(\rho_i\Vert\sigma_i)\right)
 \leq\sum_i p_i
       h_\alpha\!\left(Q_{\alpha,t}(\rho_i\Vert\sigma_i)\right)
 =\sum_i p_iD_\alpha^t(\rho_i\Vert\sigma_i),
\end{align*}
where the last inequality is the scalar Jensen inequality.

For $1<\alpha\leq2$, applying joint convexity to $x^\alpha$ instead
gives
\begin{equation*}
 Q_{\alpha,t}(\bar\rho\Vert\bar\sigma)
 \leq\sum_i p_iQ_{\alpha,t}(\rho_i\Vert\sigma_i).
\end{equation*}
Since $q\mapsto(\alpha-1)^{-1}\log q$ is increasing,
\begin{align*}
 D_\alpha^t(\bar\rho\Vert\bar\sigma)
 \leq\frac{1}{\alpha-1}
       \log\!\left(\sum_i p_i
       Q_{\alpha,t}(\rho_i\Vert\sigma_i)\right)
 \leq\frac{1}{\alpha-1}
       \log\!\left(\max_i
       Q_{\alpha,t}(\rho_i\Vert\sigma_i)\right)
 =\max_iD_\alpha^t(\rho_i\Vert\sigma_i),
\end{align*}
which is \eqref{eq:consolidated-Renyi-quasiconvexity}.  The logarithm
is concave in this range, so joint convexity of $Q_{\alpha,t}$ alone
does not imply joint convexity of $D_\alpha^t$.
\end{proof}

\begin{remark}[Non-invertible inputs]
The qualitative DPI and its convexity consequences extend to the
support-regularized divergences.  If
$\operatorname{supp}\rho\leq\operatorname{supp}\sigma=:P$, we work
on $P\mathcal H$, restrict the output to
$Q:=\operatorname{supp}\Phi(\sigma)$, and interpret inverses on these
supports.  Applying the invertible result to
$\rho_\varepsilon=(1-\varepsilon)\rho+\varepsilon\sigma$ and letting
$\varepsilon\downarrow0$ proves the claim.  The usual simultaneous
full-rank regularization gives the corresponding extended-valued
statement for arbitrary pairs whenever that boundary extension is
defined.  This limiting argument proves inequalities only and is not
used below to infer equality conditions.
\end{remark}

\subsection{Strengthened data-processing inequality}

As we saw above, the qualitative DPI is obtained by discarding positive terms.  However,
to quantify the deficit in the DPI, we retain those terms and follow the resolvent
strategy of
\cite{carlen2018recovery,bluhm2020strengthened}.

We write
\begin{equation}
 \delta_{f,t}^{\Phi}(\rho,\sigma)
 :=D_f^t(\rho\Vert\sigma)
 -D_f^t(\Phi(\rho)\Vert\Phi(\sigma)).
 \label{eq:consolidated-dpi-defect}
\end{equation}
For the logarithmic generator associated to the $t$-relative entropies, we abbreviate
$\delta_t^\Phi:=\delta_{x\log x,t}^\Phi(\rho,\sigma)$.

The next lemma is an abstract square-completion form of
\cite[Lemma~2.1]{carlen2020recovery-map} and is closely related to
\cite[Lemma~5.1]{bluhm2020strengthened}.

\begin{lemma}[Resolvent square completion]
\label{lem:consolidated-resolvent}
We take $A,B>0$ and a contraction $W$ with $W^*AW\leq B$,
and we suppose that $x=Wy$ and
$\lVert x\rVert=\lVert y\rVert$.  For $s>0$, we set
\begin{equation}
 u_s:=(A+s)^{-1}x-W(B+s)^{-1}y.
\end{equation}
Then
\begin{equation}
 \langle x,(A+s)^{-1}x\rangle
 -\langle y,(B+s)^{-1}y\rangle
 \geq s\lVert u_s\rVert^2.
 \label{eq:consolidated-resolvent-square}
\end{equation}
\end{lemma}

\begin{proof}
We put $P=I-W^*W\geq0$, $Q=B-W^*AW\geq0$,
$a=(A+s)^{-1}x$, and $b=(B+s)^{-1}y$.  Since
$Py=0$, completing the square gives
\begin{align*}
 \langle x,(A+s)^{-1}x\rangle
 -\langle y,(B+s)^{-1}y\rangle=
 \langle a-Wb,(A+s)(a-Wb)\rangle
 +\langle b,Qb\rangle+s\langle b,Pb\rangle.
\end{align*}
Since $A+sI\geq sI$ and $P,Q\geq0$, the right-hand side is at least
$s\lVert a-Wb\rVert^2=s\lVert u_s\rVert^2$, which is precisely
\eqref{eq:consolidated-resolvent-square}.
\end{proof}

\begin{remark}[Semidefinite version]
\Cref{lem:consolidated-resolvent} remains valid for $A,B\geq0$:
the proof uses only the shifted inverses $(A+sI)^{-1}$ and
$(B+sI)^{-1}$ with $s>0$.
\end{remark}

This lemma is used to prove the following general lower bound for the deficit. 

\begin{theorem}[General resolvent strengthening]
\label{thm:consolidated-resolvent-dpi}
For
\begin{equation}
 w_{s,t}:=(\Gamma_t+s)^{-1}\sigma^{1/2}
 -\mathsf V_{\sigma,\Phi}
  (\Gamma_t^\Phi+s)^{-1}[\Phi(\sigma)]^{1/2},
 \label{eq:consolidated-wst}
\end{equation}
one has
\begin{align}
 \delta_{f,t}^{\Phi}(\rho,\sigma)
 \geq{}&b\left[
 D_{x^2}^t(\rho\Vert\sigma)
 -D_{x^2}^t(\Phi(\rho)\Vert\Phi(\sigma))\right]
 +\int_0^\infty s^2\lVert w_{s,t}\rVert_2^2
 \,\mathrm d\mu_f(s).
 \label{eq:consolidated-general-remainder}
\end{align}
\end{theorem}

\begin{proof}
For the integral kernel, we use
$-x/(x+s)=-1+s(x+s)^{-1}$.  Its DPI defect is $s$ times the
left-hand side of \eqref{eq:consolidated-resolvent-square} with
\begin{equation}
 (A,B,W,x,y)
 =(\Gamma_t,\Gamma_t^\Phi,\mathsf V_{\sigma,\Phi},
   \sigma^{1/2},[\Phi(\sigma)]^{1/2}).
\end{equation}
\Cref{lem:consolidated-resolvent} supplies the second factor of
$s$, and integration proves the result.
\end{proof}

\begin{remark}[Support extension of the resolvent remainder]
The preceding remainder also holds when
$\operatorname{supp}\rho\leq\operatorname{supp}\sigma$, after
compression to the input and output reference supports and with
support inverses.  The resulting modular operators may be
semidefinite, but every resolvent in the proof is shifted by $s>0$.
The endpoint formulas are understood through the same invertible
regularization as in the qualitative DPI.
\end{remark}

\subsection{From resolvent errors to a single intrinsic defect}

The preceding theorem controls a continuum of resolvent errors.  For
the explicit estimates below, we compress that information into one
Hilbert--Schmidt quantity.  We first pass to the \emph{transpose
function} of $f$, defined by
$\widetilde f(x):=xf(x^{-1})$, and introduce the comparison
contraction
\begin{equation}
 \mathsf U_t:=\Gamma_t^{1/2}
 \mathsf V_{\sigma,\Phi}(\Gamma_t^\Phi)^{-1/2}.
 \label{eq:consolidated-transpose-objects}
\end{equation}
This is the standard transpose construction for quasi-entropies
\cite[Sections~2 and~7]{carlen2018recovery}; a related treatment appears in
\cite[Theorem~6.1]{bluhm2020strengthened}.  More precisely,
substituting \eqref{eq:consolidated-transpose-objects} and using
\Cref{lem:consolidated-transformer}, together with
$\mathsf V_{\sigma,\Phi}^*\mathsf V_{\sigma,\Phi}\leq I$, gives
\begin{align}
 \mathsf U_t^*\mathsf U_t\leq I,\quad \qquad
 \mathsf U_t(\Gamma_t^\Phi)^{1/2}[\Phi(\sigma)]^{1/2}
 =\Gamma_t^{1/2}\sigma^{1/2}, \quad \qquad
 \mathsf U_t^*\Gamma_t^{-1}\mathsf U_t
 \leq(\Gamma_t^\Phi)^{-1},
 \label{eq:consolidated-transpose-relations}
\end{align}
and
\begin{equation}
 D_f^t(\rho\Vert\sigma)
 =\left\langle\Gamma_t^{1/2}\sigma^{1/2},
 \widetilde f(\Gamma_t^{-1})
 \Gamma_t^{1/2}\sigma^{1/2}\right\rangle.
\end{equation}
We set
\begin{equation}
 q_t^\Phi(\rho,\sigma)
 :=\left\lVert\sigma^{1/2}
 -\mathsf U_t[\Phi(\sigma)]^{1/2}\right\rVert_2.
 \label{eq:consolidated-qt}
\end{equation}
We define
\begin{equation}
 \mathcal E_t^\Phi(\rho,\sigma)
 :=\mathcal L_t(\rho\mid\sigma)
 -\Phi^*\!\left(\mathcal L_t(\Phi(\rho)\mid\Phi(\sigma))\right).
 \label{eq:consolidated-recovery-residual}
\end{equation}
Here, anticipating the recovery section, this last term is defined in terms of the geodesic likelihood
ratio, which is given by
\begin{equation}
 \mathcal L_t(R\mid S)
 :=R^{(t-1)/2}S^{1/2}
 (S^{-1/2}R^tS^{-1/2})^{-1/2}S^{-1/2}.
\end{equation}
The definitions give the exact identity
\begin{equation}
 q_t^\Phi(\rho,\sigma)
 =\left\lVert\Gamma_t^{1/2}
 \mathcal E_t^\Phi(\rho,\sigma)\sigma^{1/2}\right\rVert_2.
 \label{eq:consolidated-qt-likelihood}
\end{equation}
Consequently, denoting by $\lambda_{\min}(\Gamma_t)$ the smallest
eigenvalue of $\Gamma_t$,
\begin{equation}
 \lambda_{\min}(\Gamma_t)^{1/2}\lVert
 \mathcal E_t^\Phi(\rho,\sigma)\sigma^{1/2}\rVert_2
 \leq q_t^\Phi
 \leq\lVert\Gamma_t\rVert_\infty^{1/2}
       \lVert\mathcal E_t^\Phi(\rho,\sigma)
       \sigma^{1/2}\rVert_2.
 \label{eq:consolidated-qt-weighted-residual}
\end{equation}
Thus $q_t^\Phi$ is exactly a weighted likelihood-ratio residual,
not a trace distance from a recovered state.
Its exact zero set and recovery interpretation are deliberately
deferred to \Cref{sec:consolidated-recovery}; the quantitative
arguments below use only the intrinsic residual itself.

We assume that $\widetilde f$ is operator monotone decreasing and,
modulo affine terms, has the Stieltjes representation
\begin{equation}
 \widetilde f(x)=c_0+c_1x+
 \int_0^\infty\left(\frac1{x+s}-\frac1{1+s}\right)
 \,\mathrm d\nu_{\widetilde f}(s).
 \label{eq:consolidated-transpose-measure}
\end{equation}
This is the representation in
\cite[Section~2, Eq.~(2.1)]{carlen2018recovery}.
We suppose that there are $a\geq0$ and $C>0$ such that, for every
$R\geq1$,
\begin{equation}
 \mathrm ds\leq C R^{2a}\,\mathrm d\nu_{\widetilde f}(s)
 \quad\text{on }[R^{-1},R].
 \label{eq:consolidated-measure-regularity}
\end{equation}

\begin{theorem}[Regular-transpose strengthening]
\label{thm:consolidated-regular-transpose}
We put
\begin{equation}
 M_t:=\max\{\lVert\Gamma_t^{-1}\rVert_\infty,
                  \lVert(\Gamma_t^\Phi)^{-1}\rVert_\infty\}.
\end{equation}
For every $R\geq1$,
\begin{equation}
 \delta_{f,t}^{\Phi}(\rho,\sigma)
 \geq\frac{\pi^2}{C R^{2a+1}}
 \left(q_t^\Phi(\rho,\sigma)
 -\frac{4(1+M_t)}{\pi\sqrt R}\right)_+^2.
 \label{eq:consolidated-scale-bound}
\end{equation}
In particular,
\begin{equation}
 \delta_{f,t}^{\Phi}(\rho,\sigma)
 \geq
 \frac{\pi^{4a+4}}
 {4C\,8^{4a+2}(1+M_t)^{4a+2}}
 \bigl(q_t^\Phi(\rho,\sigma)\bigr)^{4(a+1)}.
 \label{eq:consolidated-power-bound}
\end{equation}
\end{theorem}

\begin{proof}
We apply \Cref{lem:consolidated-resolvent} with
\begin{equation}
 (A,B,W,x,y)=\left(
 \Gamma_t^{-1},(\Gamma_t^\Phi)^{-1},\mathsf U_t,
 \Gamma_t^{1/2}\sigma^{1/2},
 (\Gamma_t^\Phi)^{1/2}[\Phi(\sigma)]^{1/2}\right).
\end{equation}
This gives
\begin{equation}
 \delta_{f,t}^{\Phi}\geq
 \int_0^\infty s\lVert e_s\rVert_2^2
 \,\mathrm d\nu_{\widetilde f}(s),
\quad
 e_s=(\Gamma_t^{-1}+s)^{-1}\Gamma_t^{1/2}\sigma^{1/2}
 -\mathsf U_t\bigl((\Gamma_t^\Phi)^{-1}+s\bigr)^{-1}
 (\Gamma_t^\Phi)^{1/2}[\Phi(\sigma)]^{1/2}.
\end{equation}
The required normalization and first-moment identities read
\begin{equation}
 \left\lVert\Gamma_t^{1/2}\sigma^{1/2}\right\rVert_2
 =\left\lVert(\Gamma_t^\Phi)^{1/2}
 [\Phi(\sigma)]^{1/2}\right\rVert_2=1
\end{equation}
and
\begin{equation}
 \left\langle\Gamma_t^{1/2}\sigma^{1/2},
 \Gamma_t^{-1}\Gamma_t^{1/2}\sigma^{1/2}\right\rangle
 =\left\langle(\Gamma_t^\Phi)^{1/2}[\Phi(\sigma)]^{1/2},
 (\Gamma_t^\Phi)^{-1}(\Gamma_t^\Phi)^{1/2}
 [\Phi(\sigma)]^{1/2}\right\rangle=1.
\end{equation}
The square-root resolvent representation
\cite[Example~2.3 and Eq.~(7.3)]{carlen2018recovery} gives, with the
sign immaterial after taking norms,
\begin{equation}
 \sigma^{1/2}-\mathsf U_t[\Phi(\sigma)]^{1/2}
 =-\frac1\pi\int_0^\infty s^{1/2}e_s\,\mathrm ds.
\end{equation}
The square-root integral yields
\begin{equation}
 q_t^\Phi\leq\frac1\pi
 \int_0^\infty s^{1/2}\lVert e_s\rVert_2\,\mathrm ds.
\end{equation}
Following the interval splitting and tail estimates in
\cite[proof of Theorem~4.2 and Section~7]{carlen2018recovery} and
\cite[proof of Theorem~6.1]{bluhm2020strengthened}, Cauchy--Schwarz and
\eqref{eq:consolidated-measure-regularity} control the middle
interval, while elementary resolvent estimates control the tails by
$4(1+M_t)/(\pi\sqrt R)$.  This proves
\eqref{eq:consolidated-scale-bound}.  To obtain a bound depending
only on $q_t^\Phi$, we choose the scale so that the tail contribution
uses exactly one half of this defect.  If $q_t^\Phi=0$, the desired
estimate is immediate; otherwise, we set
\begin{equation}
 R=\left(\frac{8(1+M_t)}{\pi q_t^\Phi}\right)^2
\end{equation}
so that
\begin{equation}
 \frac{4(1+M_t)}{\pi\sqrt R}=\frac{q_t^\Phi}{2}.
\end{equation}
Moreover, $q_t^\Phi\leq2$ by the normalization above and the
contractivity of $\mathsf U_t$, whereas $M_t\geq1$ by the first-moment
identities; hence this choice satisfies $R\geq1$.  Substitution into
\eqref{eq:consolidated-scale-bound} leaves
$(q_t^\Phi/2)^2$ in the positive part and yields
\eqref{eq:consolidated-power-bound}.
\end{proof}

\begin{theorem}[Quartic strengthening for $D^t$]
\label{thm:consolidated-Dt-quartic}
For $D^t=D_{x\log x}^t$,
\begin{equation}
 \boxed{
 D^t(\rho\Vert\sigma)-D^t(\Phi(\rho)\Vert\Phi(\sigma))
 \geq\left(\frac\pi4\right)^4M_t^{-2}
 \bigl(q_t^\Phi(\rho,\sigma)\bigr)^4.}
 \label{eq:consolidated-Dt-quartic}
\end{equation}
where
\begin{align*}
 q_t^\Phi(\rho,\sigma)
 &:=\left\lVert\Gamma_t^{1/2}\left[
 \mathcal L_t(\rho\mid\sigma)
 -\Phi^*\!\left(\mathcal L_t(\Phi(\rho)\mid\Phi(\sigma))\right)
 \right]\sigma^{1/2}\right\rVert_2,\\
 M_t
 &:=\max\left\{\lVert\Gamma_t^{-1}\rVert_\infty,
 \lVert(\Gamma_t^\Phi)^{-1}\rVert_\infty\right\}.
\end{align*}
\end{theorem}

\begin{proof}
Here $\widetilde f(x)=-\log x$ and its Stieltjes measure is
Lebesgue measure \cite[Example~2.2]{carlen2018recovery}.  Splitting the square-root integral at $T>0$
gives
\begin{equation}
 \pi q_t^\Phi\leq
 \sqrt{T\delta_t^\Phi}+\frac{4M_t}{\sqrt T}.
\end{equation}
Optimization at $T=4M_t/\sqrt{\delta_t^\Phi}$ proves the claim.
\end{proof}

\paragraph{Consistency with the endpoint remainders.}
For comparison with the earlier results, we take $\Phi=\mathcal E$ to be a
trace-preserving conditional expectation onto $\mathcal N$, and we write
$\rho_{\mathcal N}=\mathcal E(\rho)$ and
$\sigma_{\mathcal N}=\mathcal E(\sigma)$.  At $t=0$, taking adjoints
in the definition of the defect gives
\begin{equation*}
 q_0^{\mathcal E}
 =\left\lVert\sigma_{\mathcal N}^{1/2}
 \rho_{\mathcal N}^{-1/2}\rho^{1/2}-\sigma^{1/2}\right\rVert_2,
 \qquad
 M_0=\lVert\Delta_{\sigma\mid\rho}\rVert_\infty,
 \quad
 \Delta_{\sigma\mid\rho}:=L_\sigma R_{\rho^{-1}}.
\end{equation*}
Thus $q_0^{\mathcal E}$ is the relative-entropy defect: it is
precisely the half-power remainder of
\cite[Corollary~5.1, in particular (5.3)]{carlen2018recovery}; the
present logarithmic optimization gives $M_0^{-2}$ in place of their
$(1+M_0)^{-2}$.  At $t=1$, we put
$C=\rho^{-1/2}\sigma\rho^{-1/2}$ and
$C_{\mathcal N}=\rho_{\mathcal N}^{-1/2}
\sigma_{\mathcal N}\rho_{\mathcal N}^{-1/2}$.  The polar
decompositions of $\rho^{-1/2}\sigma^{1/2}$ and its output analogue
give
\begin{equation*}
 q_1^{\mathcal E}
 =\left\lVert
 \rho^{1/2}\rho_{\mathcal N}^{-1/2}C_{\mathcal N}^{1/2}
 \rho_{\mathcal N}^{1/2}-C^{1/2}\rho^{1/2}
 \right\rVert_2,
 \qquad M_1=\lVert C\rVert_\infty.
\end{equation*}
Thus $q_1^{\mathcal E}$ is the BS defect, and
\eqref{eq:consolidated-Dt-quartic} becomes exactly the BS remainder of
\cite[Theorem~5.3]{bluhm2020strengthened}, after matching their first
state with our $\rho$.  For a general channel, the same identification
after a Stinespring dilation gives the form in
\cite[Theorem~7.1]{bluhm2020strengthened}.

\subsubsection{A Petz-distance remainder away from the BS endpoint.}
 In \Cref{sec:consolidated-recovery}, and specifically in 
\Cref{cor:consolidated-interior-recovery-dictionary}, it will be proven that,
 for every $t<1$, the vanishing of $q_t^\Phi$ is equivalent to recovery
by the ordinary Petz map.  We record its
quantitative consequence here because it is obtained by combining
that equality characterization directly with
\Cref{thm:consolidated-Dt-quartic}. With that characterization in hand, the coincidence of the
nonmaximal equality classes gives room for the possibility of lower bounding the defect in \eqref{eq:consolidated-Dt-quartic}
by a distance from the state to its Petz recovery. That is the content of this section.

We fix a channel
$\Phi:\mathcal B(\mathcal H)\to\mathcal B(\mathcal K)$ and, for
$\varepsilon,\eta>0$, we set
\begin{equation}
 \mathfrak K_{\varepsilon,\eta}^{\Phi}
 :=\left\{(\rho,\sigma)\in\mathcal S(\mathcal H)^2:
 \begin{array}{c}
  \rho,\sigma\geq\varepsilon I_{\mathcal H},\\
  \Phi(\rho),\Phi(\sigma)\geq\eta I_{\mathcal K}
 \end{array}\right\}.
 \label{eq:consolidated-invertible-Petz-slab}
\end{equation}

\begin{corollary}[H\"older Petz-recovery remainder]
\label{cor:consolidated-nonmaximal-Petz-holder}
For $0<\tau<1$, if
$\mathfrak K_{\varepsilon,\eta}^{\Phi}$ is nonempty, then there are a
constant $c_{\tau,\varepsilon,\eta}^{\Phi}>0$ and an integer
$N_{\tau,\varepsilon,\eta}^{\Phi}\geq1$ such that, simultaneously
for every $0\leq t\leq1-\tau$ and every
$(\rho,\sigma)\in\mathfrak K_{\varepsilon,\eta}^{\Phi}$,
\begin{equation}
 q_t^\Phi(\rho,\sigma)
 \geq c_{\tau,\varepsilon,\eta}^{\Phi}
 \left\lVert\rho-
 \mathcal P_{\sigma,\Phi}(\Phi(\rho))\right\rVert_1^{
 N_{\tau,\varepsilon,\eta}^{\Phi}}.
 \label{eq:consolidated-qt-Petz-holder}
\end{equation}
Consequently,
\begin{equation}
 \boxed{
 \begin{aligned}
 &D^t(\rho\Vert\sigma)
 -D^t(\Phi(\rho)\Vert\Phi(\sigma))\geq
 \left(\frac\pi4\right)^4
 \min\{\varepsilon,\eta\}^{2}
 \bigl(c_{\tau,\varepsilon,\eta}^{\Phi}\bigr)^4
 \left\lVert\rho-
 \mathcal P_{\sigma,\Phi}(\Phi(\rho))\right\rVert_1^{
 4N_{\tau,\varepsilon,\eta}^{\Phi}} .
 \end{aligned}}
 \label{eq:consolidated-Dt-Petz-holder}
\end{equation}
In particular, every fixed $t<1$ admits such a bound by choosing
$\tau<1-t$.
\end{corollary}

\begin{proof}
We identify the Hermitian trace-one affine spaces with finite-dimensional
real affine spaces.  On a neighborhood of
$[0,1-\tau]\times\mathfrak K_{\varepsilon,\eta}^{\Phi}$, the two
nonnegative functions
\begin{equation*}
 (t,\rho,\sigma)\longmapsto
 \bigl(q_t^\Phi(\rho,\sigma)\bigr)^2
 \quad\text{and}\quad
 (t,\rho,\sigma)\longmapsto
 \left\lVert\rho-
 \mathcal P_{\sigma,\Phi}(\Phi(\rho))\right\rVert_2^2
\end{equation*}
are real analytic: inverses, real powers, and square roots depend real
analytically on a positive-definite matrix.  We recall the compact
two-function form of the \L{}ojasiewicz inequality
\cite[Theorem~6.4]{bierstone1988semianalytic}: if $\mathcal X$ is a
compact semianalytic set and $u,v:\mathcal X\to\mathbb R$ are
continuous subanalytic functions, then
\begin{equation}
 u^{-1}(0)\subseteq v^{-1}(0)
 \quad\Longrightarrow\quad
 |u(x)|\geq c\,|v(x)|^r
 \quad\text{for every }x\in\mathcal X
 \label{eq:consolidated-two-function-Lojasiewicz}
\end{equation}
for some $c>0$ and $r>0$.  By
\eqref{eq:consolidated-qt-likelihood} and
\Cref{cor:consolidated-interior-recovery-dictionary}, the zero sets of
the two functions displayed above coincide for every $t<1$.
Applying \eqref{eq:consolidated-two-function-Lojasiewicz} on
$[0,1-\tau]\times\mathfrak K_{\varepsilon,\eta}^{\Phi}$ therefore
gives $a>0$ and $r>0$ such that
\begin{equation}
 \bigl(q_t^\Phi(\rho,\sigma)\bigr)^2
 \geq a
 \left\lVert\rho-
 \mathcal P_{\sigma,\Phi}(\Phi(\rho))\right\rVert_2^{2r}.
\end{equation}
Since the Petz recovery error is bounded on the compact set, we may
enlarge $r$ to an integer $N\geq1$, decreasing $a$ if necessary.  If
$d_{\mathcal H}:=\dim\mathcal H$, then
\begin{align}
 q_t^\Phi(\rho,\sigma)
 &\geq \sqrt a\left\lVert\rho-
 \mathcal P_{\sigma,\Phi}(\Phi(\rho))\right\rVert_2^N
 \geq \sqrt a\,d_{\mathcal H}^{-N/2}
 \left\lVert\rho-
 \mathcal P_{\sigma,\Phi}(\Phi(\rho))\right\rVert_1^N.
 \label{eq:consolidated-Lojasiewicz-trace-comparison}
\end{align}
This proves \eqref{eq:consolidated-qt-Petz-holder} after absorbing the
dimensional factor into $c_{\tau,\varepsilon,\eta}^{\Phi}$.

It remains to make the prefactor in \Cref{thm:consolidated-Dt-quartic}
uniform.  Since all four operators in
\eqref{eq:consolidated-invertible-Petz-slab} are states,
\begin{align*}
 \lVert\Gamma_t^{-1}\rVert_\infty
 &\leq
 \lVert\rho^{t-1}\rVert_\infty
 \lVert\sigma^{1/2}\rho^{-t}\sigma^{1/2}\rVert_\infty
 \leq\varepsilon^{-(1-t)}\varepsilon^{-t}
 =\varepsilon^{-1},\\
 \lVert(\Gamma_t^\Phi)^{-1}\rVert_\infty
 &\leq\eta^{-1}.
\end{align*}
Hence $M_t^{-2}\geq\min\{\varepsilon,\eta\}^2$.
Substituting \eqref{eq:consolidated-qt-Petz-holder} into
\eqref{eq:consolidated-Dt-quartic} proves
\eqref{eq:consolidated-Dt-Petz-holder}.
\end{proof}

The restriction away from $t=1$ is essential.  The constants above
may degenerate as $\tau\downarrow0$, and at $t=1$ no such estimate can
hold on a set containing an invertible BS-equality triple which is not
Petz sufficient: its intrinsic defect vanishes while its Petz
recovery error is strictly positive.  Thus
\Cref{cor:consolidated-nonmaximal-Petz-holder} is a nonconstructive
H\"older stability result under uniform positive lower bounds on the
relevant smallest eigenvalues, rather than a
constant-free recovery inequality.

\subsection{R\'enyi remainders and order-dependent powers}

The logarithmic generator gives the quartic estimate in
\Cref{thm:consolidated-Dt-quartic} for the $t$-relative entropies case.  The same intrinsic defect also
controls the open R\'enyi DPI ranges, with a power determined by the
order.

We take $0<\alpha<1$ or $1<\alpha<2$, and put
\begin{equation}
 r_\alpha:=|\alpha-1|,
 \qquad C_\alpha:=\frac{\pi}{\sin(\pi r_\alpha)}.
\end{equation}
The corresponding power-function Stieltjes densities are given in
\cite[Example~2.3]{carlen2018recovery}; the resulting power
quasi-entropy remainder appears in
\cite[Corollary~5.2]{carlen2018recovery}.
For $0<\alpha<1$, we apply
\Cref{thm:consolidated-regular-transpose} to
$-x^\alpha$; and for $1<\alpha<2$, we apply it to $x^\alpha$.
We define
\begin{align}
 \mathcal B_{\alpha,t}^{\Phi}
 &:=\underset{R\geq1}{\sup}\frac{\pi^2}
 {C_\alpha R^{r_\alpha+1}}
 \left(q_t^\Phi-
 \frac{4(1+M_t)}{\pi\sqrt R}\right)_+^2\geq \kappa_\alpha(M_t)
 \bigl(q_t^\Phi\bigr)^{4+2r_\alpha},
 \label{eq:consolidated-Renyi-B}\\
 \kappa_\alpha(M)
 &:=\frac{\pi^{2r_\alpha+4}}
 {4C_\alpha\,8^{2r_\alpha+2}(1+M)^{2r_\alpha+2}}.
\end{align}
Then, we can prove the following strengthened DPI for the $(t,\alpha)$-Rényi divergences.

\begin{theorem}[Strengthened R\'enyi DPI]
\label{thm:consolidated-Renyi-dpi}
In the conditions above, we have
\begin{equation}
 \boxed{
 \begin{aligned}
 D_\alpha^t(\rho\Vert\sigma)
 -D_\alpha^t(\Phi(\rho)\Vert\Phi(\sigma))
 &\geq\frac1{1-\alpha}
 \log\left(1+
 \frac{\mathcal B_{\alpha,t}^{\Phi}}
 {Q_{\alpha,t}(\rho\Vert\sigma)}\right),
 &&0<\alpha<1,\\
 D_\alpha^t(\rho\Vert\sigma)
 -D_\alpha^t(\Phi(\rho)\Vert\Phi(\sigma))
 &\geq\frac1{\alpha-1}
 \log\left(1+
 \frac{\mathcal B_{\alpha,t}^{\Phi}}
 {Q_{\alpha,t}(\Phi(\rho)\Vert\Phi(\sigma))}\right),
 &&1<\alpha<2.
 \end{aligned}}
 \label{eq:consolidated-Renyi-log-bounds}
\end{equation}
The corresponding powers of $q_t^\Phi$ are $6-2\alpha$ below
one and $2\alpha+2$ above one.
\end{theorem}

\begin{proof}
For $0<\alpha<1$, applying
\Cref{thm:consolidated-regular-transpose} to $-x^\alpha$ gives
$Q_{\alpha,t}(\Phi(\rho)\Vert\Phi(\sigma))
-Q_{\alpha,t}(\rho\Vert\sigma)
\geq\mathcal B_{\alpha,t}^{\Phi}$.  We divide by the input power
functional, take the logarithm, and use $\alpha-1<0$ to obtain the
first line of \eqref{eq:consolidated-Renyi-log-bounds}.  For
$1<\alpha<2$, the same theorem applied to $x^\alpha$ gives
$Q_{\alpha,t}(\rho\Vert\sigma)
-Q_{\alpha,t}(\Phi(\rho)\Vert\Phi(\sigma))
\geq\mathcal B_{\alpha,t}^{\Phi}$.  Division by the output power
functional gives the second line.  The two exponents follow from
the power of $q_t^\Phi$ in \eqref{eq:consolidated-Renyi-B}.
\end{proof}

\begin{remark}[Fixed common supports]
The regular-transpose, quartic, and R\'enyi remainders above remain
valid when the two input states have a common proper support and the
two output states have a common proper support, by applying the
results in the corresponding matrix corners.  No extension to
unequal supports is asserted: the constants contain inverse modular
norms and can diverge when a positive eigenvalue tends to zero.
\end{remark}

The quadratic endpoint $\alpha=2$ is excluded from this regular
remainder, and the $\alpha=1$ result is
\Cref{thm:consolidated-Dt-quartic}.

\section{Equality and recovery conditions}
\label{sec:consolidated-recovery}

The preceding section establishes inequalities and quantitative lower
bounds for their deficits.  We now identify exactly when those deficits
vanish.  The argument first promotes one half-power equality to the
full modular functional calculus, then identifies the resulting
likelihood-ratio condition, and finally proves that every nonmaximal
equality-determining member has the ordinary Petz equality class.
The fixed-point structure is therefore the standard Petz structure;
the genuinely $t$-dependent quantitative reconstruction
consequences are collected in
\Cref{sec:consolidated-quantitative-reconstruction}.

\Needspace{12\baselineskip}
\subsection{Likelihood-ratio equality and functional propagation}

\begin{mainobject}{Main definition: The geodesic likelihood ratio}
For $R,S>0$ acting on the same finite-dimensional Hilbert space and
$t\in[0,1]$, the \emph{$t$-geodesic likelihood ratio} of the
ordered pair $(R,S)$ is the operator
\begin{equation}
 \mathcal L_t(R\mid S)
 :=R^{(t-1)/2}S^{1/2}
 \left(S^{-1/2}R^tS^{-1/2}\right)^{-1/2}S^{-1/2}.
 \label{eq:consolidated-likelihood}
\end{equation}
\end{mainobject}
The terminology is tied to the canonical experiment of the companion
paper: the spectral measurement of $\Gamma_t(R,S)$ has likelihood
ratio equal to the spectrum of $\Gamma_t(R,S)$
\cite[Theorem~3.1]{capel2026informationgeometry}.  The operator
$\mathcal L_t(R\mid S)$ is not that likelihood observable itself; it
is its inverse-square-root representative on the cyclic vector:
\begin{equation}
 \Gamma_t(R,S)^{-1/2}S^{1/2}
 =\mathcal L_t(R\mid S)S^{1/2}.
 \label{eq:consolidated-likelihood-vector}
\end{equation}
If $\Phi$ is a channel and $\Phi(R),\Phi(S)>0$, we call the identity
\begin{equation}
 \boxed{
 \mathcal L_t(R\mid S)
 =\Phi^*\!\left(
 \mathcal L_t(\Phi(R)\mid\Phi(S))\right)}
 \label{eq:consolidated-channel-likelihood}
\end{equation}
the \emph{$t$-likelihood-ratio recovery condition}.  It is the
interpolation-adapted analogue of a recovery condition: rather than
reconstructing the state directly, it reconstructs its likelihood
ratio by pulling the output operator back through $\Phi^*$.
\Cref{thm:consolidated-equality-recovery} shows that this is exactly
the condition selected by equality in data processing, while
\Cref{cor:consolidated-interior-recovery-dictionary} relates it to
ordinary Petz recovery for $t<1$ and to the BS condition at $t=1$.

\paragraph{Equality-determining generators.}
For a fixed finite-dimensional DPI problem, we put
\begin{equation}
 N_t(\rho,\sigma;\Phi)
 :=\left|\operatorname{spec}(\Gamma_t)
 \cup\operatorname{spec}(\Gamma_t^\Phi)\right|.
 \label{eq:consolidated-equality-determining-count}
\end{equation}
We call the representing measure in
\eqref{eq:consolidated-f-representation} \emph{equality
determining} for this problem if
\begin{equation}
 |\supp\mu_f|\geq N_t(\rho,\sigma;\Phi).
 \label{eq:consolidated-equality-determining-support}
\end{equation}
This is the finite-dimensional Cauchy-interpolation condition used in
the reversibility theory of quantum $f$-divergences
\cite[Theorem~3.18(iv), especially Eq.~(3.21)]{hiai2017different}.
It is enough to recover the complete
functional calculus from the resolvent equalities on
$\supp\mu_f$.  The simpler, state-independent assumption that
$\supp\mu_f$ has an accumulation point in $(0,\infty)$ is more
than sufficient.  In particular, it holds for $x\log x$,
$-x^\alpha$, $0<\alpha<1$, and $x^\alpha$,
$1<\alpha<2$.  A nonempty support alone is not sufficient, while
affine and purely quadratic generators have no such measure at all.

The next lemma is the technical step that turns the half-power
identity selected by the resolvent defect into equality for the whole
functional calculus.

\begin{lemma}[Equality propagation]
\label{lem:consolidated-equality-propagation}
We take $A,B>0$ and a contraction $W$ with
\begin{equation}
 W^*AW\leq B,\qquad W\eta=\xi,\qquad
 \lVert\eta\rVert=\lVert\xi\rVert,
\end{equation}
and assume
\begin{equation}
 A^{-1/2}\xi=WB^{-1/2}\eta.
 \label{eq:consolidated-half-propagation}
\end{equation}
Then, for every continuous $g$ on the relevant spectra,
\begin{equation}
 g(A)\xi=Wg(B)\eta.
 \label{eq:consolidated-functional-propagation}
\end{equation}
Moreover, $W$ is an isometry intertwining $B$ and $A$ on the
cyclic subspaces generated by $\eta$ and $\xi$.
\end{lemma}

\begin{proof}
We set
\begin{equation}
 C:=A^{1/2}WB^{-1/2}.
\end{equation}
The inequality $W^*AW\leq B$ gives
\begin{equation}
 C^*C
 =B^{-1/2}W^*AWB^{-1/2}\leq I,
\end{equation}
so $C$ is a contraction.  Multiplying
\eqref{eq:consolidated-half-propagation} by $A^{1/2}$ gives
$C\eta=\xi$, while $W\eta=\xi$ holds by assumption.  Since
$\|\eta\|=\|\xi\|$, equality holds in the contraction estimate for
both $C$ and $W$.  For a contraction $T$, the equality
$\|Tv\|=\|v\|$ implies
$(I-T^*T)v=0$, and hence $T^*Tv=v$.  We therefore obtain
\begin{equation}
 C\eta=W\eta=\xi,
 \qquad
 C^*\xi=W^*\xi=\eta.
 \label{eq:consolidated-propagation-base}
\end{equation}

We now spell out the propagation.  Put
\begin{equation}
 \eta_k:=B^{k/2}\eta,
 \qquad
 \xi_k:=A^{k/2}\xi,
 \qquad k\geq0.
\end{equation}
We prove inductively that
\begin{equation}
 C\eta_k=W\eta_k=\xi_k,
 \qquad
 C^*\xi_k=W^*\xi_k=\eta_k
 \label{eq:consolidated-propagation-induction}
\end{equation}
for every $k\geq0$.  The case $k=0$ is
\eqref{eq:consolidated-propagation-base}.  Suppose that
\eqref{eq:consolidated-propagation-induction} holds for some $k$.
The definition of $C$ and its adjoint give the two identities
\begin{equation}
 CB^{1/2}=A^{1/2}W,
 \qquad
 B^{1/2}C^*=W^*A^{1/2}.
 \label{eq:consolidated-C-shift-identities}
\end{equation}
Consequently,
\begin{align}
 C\eta_{k+1}
 &=CB^{1/2}\eta_k
   =A^{1/2}W\eta_k
   =A^{1/2}\xi_k
   =\xi_{k+1},
 \label{eq:consolidated-C-forward-step}\\
 W^*\xi_{k+1}
 &=W^*A^{1/2}\xi_k
   =B^{1/2}C^*\xi_k
   =B^{1/2}\eta_k
   =\eta_{k+1}.
 \label{eq:consolidated-W-adjoint-step}
\end{align}
Since $C$ and $W^*$ are contractions, these identities imply
\begin{equation}
 \|\xi_{k+1}\|
 =\|C\eta_{k+1}\|
 \leq\|\eta_{k+1}\|
 =\|W^*\xi_{k+1}\|
 \leq\|\xi_{k+1}\|.
\end{equation}
Thus all inequalities are equalities.  Applying the equality case of
the contraction estimate first to $C$ at $\eta_{k+1}$ and then to
$W^*$ at $\xi_{k+1}$ yields
\begin{equation}
 C^*\xi_{k+1}=\eta_{k+1},
 \qquad
 W\eta_{k+1}=\xi_{k+1}.
\end{equation}
This closes the induction.  In particular,
\begin{equation}
 A^{k/2}\xi=WB^{k/2}\eta,
 \qquad
 W^*A^{k/2}\xi=B^{k/2}\eta,
 \qquad k=0,1,\ldots .
 \label{eq:consolidated-half-power-propagation}
\end{equation}

It follows first for every polynomial $p$ that
\begin{equation}
 p(A^{1/2})\xi=Wp(B^{1/2})\eta,
 \qquad
 W^*p(A^{1/2})\xi=p(B^{1/2})\eta.
\end{equation}
Because $x\mapsto\sqrt{x}$ is a homeomorphism on the positive
spectral interval, polynomial approximation of
$y\mapsto g(y^2)$ and continuity of the functional calculus extend
these identities to every continuous $g$ on the relevant spectra.
The first identity is
\eqref{eq:consolidated-functional-propagation}; together with the
second, it shows that $W$ is an isometry from the cyclic subspace
generated by $\eta$ onto the one generated by $\xi$.  Finally, for
every continuous $h$,
\begin{equation}
 AW h(B)\eta
 =A h(A)\xi
 =W B h(B)\eta,
\end{equation}
so $W$ intertwines $B$ and $A$ on these cyclic subspaces.
\end{proof}

We can now state the exact equality criterion in the coordinate system
adapted to the interpolation.

\begin{theorem}[Equality and likelihood-ratio recovery]
\label{thm:consolidated-equality-recovery}
We take $\rho,\sigma>0$ and assume $\Phi(\rho)>0$ and
$\Phi(\sigma)>0$.
We fix $t\in[0,1]$ and an operator-convex generator $f$ whose
representing measure $\mu_f$ in
\eqref{eq:consolidated-f-representation} satisfies the
equality-determining condition
\eqref{eq:consolidated-equality-determining-support}.  The
Hilbert--Schmidt comparison contraction
$\mathsf V_{\sigma,\Phi}:\mathcal B(\mathcal K)\to
\mathcal B(\mathcal H)$ is given by
\begin{equation}
 \mathsf V_{\sigma,\Phi}(X)
 :=\Phi^*\!\left(X[\Phi(\sigma)]^{-1/2}\right)\sigma^{1/2}.
 \label{eq:equality-theorem-comparison-map}
\end{equation}
Then the following conditions are equivalent:
\begin{enumerate}
\item The following expression holds:
\begin{equation}
 \Gamma_t^{-1/2}\sigma^{1/2}
 =\mathsf V_{\sigma,\Phi}
   (\Gamma_t^\Phi)^{-1/2}[\Phi(\sigma)]^{1/2}.
 \label{eq:consolidated-half-equality}
\end{equation}
\item The geodesic likelihood ratio is recovered, namely
\begin{equation*}
 \mathcal L_t(\rho\mid\sigma)
 =\Phi^*\!\left(
 \mathcal L_t(\Phi(\rho)\mid\Phi(\sigma))\right).
\end{equation*}
\item For every continuous $g$ on the two spectra,
\begin{equation}
 g(\Gamma_t)\sigma^{1/2}
 =\mathsf V_{\sigma,\Phi}
   g(\Gamma_t^\Phi)[\Phi(\sigma)]^{1/2}.
 \label{eq:consolidated-all-g}
\end{equation}
\item $\mathsf V_{\sigma,\Phi}$ is an isometric intertwiner on the
cyclic subspaces generated by $[\Phi(\sigma)]^{1/2}$ and
$\sigma^{1/2}$.
\item The DPI for the $D_f^t$ is saturated, that is,
\begin{equation}
 D_f^t(\rho\Vert\sigma)
 =D_f^t(\Phi(\rho)\Vert\Phi(\sigma)).
 \label{eq:consolidated-Dft-equality}
\end{equation}
\item For every operator-convex generator $h$ for which the two
divergences are defined,
\begin{equation}
 D_h^t(\rho\Vert\sigma)
 =D_h^t(\Phi(\rho)\Vert\Phi(\sigma)).
 \label{eq:consolidated-all-Dht-equality}
\end{equation}
\end{enumerate}
In particular, choosing $f(x)=x\log x$ makes the fifth condition
precisely saturation for $D^t$.  For any fixed
\begin{equation}
 0<\alpha<1,\qquad 1<\alpha<2,
\end{equation}
one may instead choose $f(x)=-x^\alpha$ in the first range and
$f(x)=x^\alpha$ in the second.  These generators are equality
determining, and injectivity of the logarithm shows that their
equality condition is equivalent to saturation for $D_\alpha^t$.
Thus saturation for any one such $D_\alpha^t$ may replace the fifth
condition.  The purely quadratic case $\alpha=2$ cannot replace it in
general.
\end{theorem}

\begin{proof}
Equivalence of the first two statements follows from
\eqref{eq:consolidated-likelihood-vector} and the definition of
$\mathsf V_{\sigma,\Phi}$.  The first statement implies the third
and fourth by
\Cref{lem:consolidated-equality-propagation}; the third implies
the first by choosing $g(x)=x^{-1/2}$.  Conversely, the fourth
implies the third because $\mathsf V_{\sigma,\Phi}$ maps
$[\Phi(\sigma)]^{1/2}$ to $\sigma^{1/2}$ and an isometric
intertwiner intertwines the continuous functional calculi on the
corresponding cyclic subspaces.  Thus the first four conditions are
equivalent.

They imply the sixth condition.  Indeed, for
$\xi=\sigma^{1/2}$ and $\eta=[\Phi(\sigma)]^{1/2}$, the third and
fourth conditions give, for every admissible $h$,
\begin{align}
 D_h^t(\rho\Vert\sigma)
 &=\langle \xi,h(\Gamma_t)\xi\rangle_{\rm HS}\notag\\
 &=\left\langle \mathsf V_{\sigma,\Phi}\eta,
 \mathsf V_{\sigma,\Phi}h(\Gamma_t^\Phi)\eta
 \right\rangle_{\rm HS}\notag\\
 &=\langle \eta,h(\Gamma_t^\Phi)\eta\rangle_{\rm HS}
 =D_h^t(\Phi(\rho)\Vert\Phi(\sigma)).
 \label{eq:consolidated-functional-calculus-DPI-equality}
\end{align}
The sixth condition plainly implies the fifth.

Finally, we suppose that the fifth condition holds.
\Cref{thm:consolidated-resolvent-dpi} shows that saturation forces
$w_{s,t}=0$ for
$\mu_f$-almost every $s$, hence at every point of
$\supp\mu_f$ by continuity.  Expanding the two resolvents in the
spectral projections of $\Gamma_t$ and $\Gamma_t^\Phi$ produces
a Cauchy system with at most
$N_t(\rho,\sigma;\Phi)$ distinct poles.  The
equality-determining support condition makes this system
invertible, so the spectral components agree separately.  Therefore
\eqref{eq:consolidated-all-g} holds for every $g$, and in
particular for $g(x)=x^{-1/2}$.
\end{proof}

The previous result can be simplified in the case that the channel is a partial trace.

\begin{corollary}[Partial traces and Stinespring form]
\label{cor:consolidated-partial-trace}
For invertible states $R_{AB},S_{AB}$, equality in the DPI under
$\operatorname{Tr}_A$ is equivalent to
\begin{equation}
 \mathcal L_t(R_{AB}\mid S_{AB})
 =I_A\otimes\mathcal L_t(R_B\mid S_B).
 \label{eq:consolidated-partial-trace-likelihood}
\end{equation}
For a Stinespring isometry
$U:\mathcal H\to\mathcal K\otimes\mathcal E$ of $\Phi$, with
$P=UU^*$, the general-channel condition is equivalently
\begin{equation}
 U\mathcal L_t(\rho\mid\sigma)U^*
 =P\left[
 \mathcal L_t(\Phi(\rho)\mid\Phi(\sigma))\otimes I_{\mathcal E}
 \right]P.
 \label{eq:consolidated-Stinespring-likelihood}
\end{equation}
The latter identity is an equality in the compressed algebra
$P\mathcal B(\mathcal K\otimes\mathcal E)P$; no ambient
invertibility of $U\rho U^*$ is required.
\end{corollary}

\begin{remark}[Endpoints and the quadratic exception]
At $t=0$, \eqref{eq:consolidated-channel-likelihood} is equivalent
to the Petz sufficiency condition
\cite[Theorem~3.18]{hiai2017different}.  At $t=1$, it is an equivalent
square-root form of the asymmetric Belavkin--Staszewski condition
\cite[Theorem~4.2 and Eq.~(14)]{bluhm2020strengthened}.
At $\alpha=2$, by contrast, $D_2^t$ is quadratic and independent
of $t$; its equality set is the BS/$\chi^2$ equality set and is
generally larger than the likelihood-recovery set for $t<1$
\cite[Example~4.2 and Theorem~3.34]{hiai2017different}.
\end{remark}

The full fixed-point hierarchy in the $(t,\alpha)$-plane is
summarized in \Cref{fig:consolidated-Renyi-fixed-point-phase}.  The
identification of the nonmaximal region with the Petz class is proved
in the following subsection, specifically in
\Cref{prop:consolidated-general-hierarchy}.

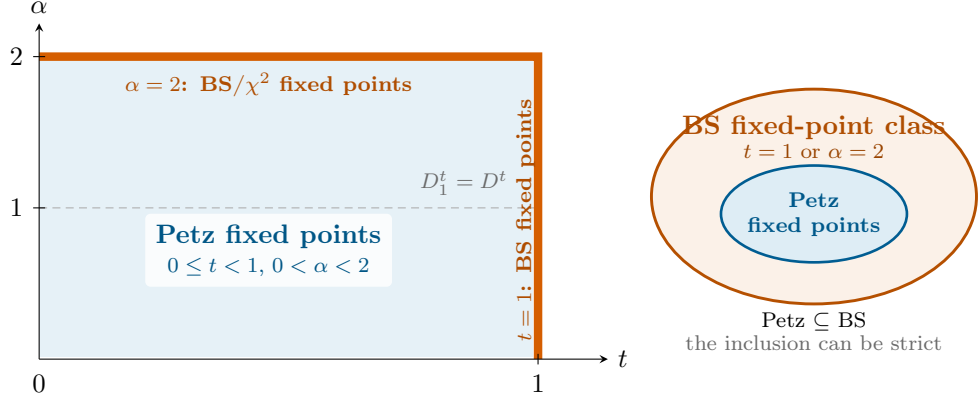
\begin{figure}[!ht]
\centering
\definecolor{fixedPetz}{HTML}{0072B2}
\definecolor{fixedBS}{HTML}{D55E00}
\begin{tikzpicture}[>=stealth,font=\small]
  % Parameter plane.
  \begin{scope}[x=6.6cm,y=2cm]
    \fill[fixedPetz!10] (0,0.015) rectangle (0.99,1.985);
    \draw[black!32,densely dashed] (0,1)--(1,1);
    \draw[fixedBS,line width=3.2pt,line cap=butt]
      (0,2)--(1,2)--(1,0);

    \draw[->] (0,0)--(1.14,0) node[right] {$t$};
    \draw[->] (0,0)--(0,2.22) node[above] {$\alpha$};
    \draw (1,-0.035)--(1,0.035);
    \draw (-0.012,1)--(0.012,1);
    \draw (-0.012,2)--(0.012,2);
    \node[below] at (0,-0.04) {$0$};
    \node[below] at (1,-0.04) {$1$};
    \node[left] at (-0.012,1) {$1$};
    \node[left] at (-0.012,2) {$2$};

    \node[align=center,text=fixedPetz!75!black,
      fill=white,fill opacity=0.82,text opacity=1,
      rounded corners=2pt,inner sep=4pt] at (0.46,0.72)
      {\bfseries Petz fixed points\\[-1pt]
       \scriptsize $0\leq t<1$, $0<\alpha<2$};
    \node[font=\scriptsize,text=black!58,anchor=south east]
      at (0.96,1.02) {$D_1^t=D^t$};
    \node[font=\scriptsize\bfseries,text=fixedBS!82!black,
      anchor=north] at (0.46,1.96)
      {$\alpha=2$: BS/$\chi^2$ fixed points};
    \node[font=\scriptsize\bfseries,text=fixedBS!82!black,
      rotate=90,anchor=south] at (1.015,0.92)
      {$t=1$: BS fixed points};
  \end{scope}

  % Inclusion of the two equality classes.
  \begin{scope}[shift={(10.25cm,2.15cm)}]
    \filldraw[fill=fixedBS!9,draw=fixedBS!85!black,line width=1.1pt]
      (0,0) ellipse (2.15cm and 1.42cm);
    \node[text=fixedBS!82!black,font=\small\bfseries]
      at (0,0.91) {BS fixed-point class};
    \node[text=fixedBS!72!black,font=\scriptsize]
      at (0,0.60) {$t=1$ or $\alpha=2$};
    \filldraw[fill=fixedPetz!13,draw=fixedPetz!82!black,line width=1pt]
      (0,-0.23) ellipse (1.23cm and 0.64cm);
    \node[text=fixedPetz!78!black,font=\scriptsize\bfseries,
      align=center] at (0,-0.23) {Petz\\fixed points};
    \node[font=\scriptsize,align=center] at (0,-1.78)
      {Petz $\subseteq$ BS\\[-1pt]
       \textcolor{black!58}{the inclusion can be strict}};
  \end{scope}
\end{tikzpicture}
\caption{Fixed-point phase diagram for the $(t,\alpha)$-R\'enyi
divergences in their DPI range.  For fixed $\sigma$ and $\Phi$, a
fixed point here means an invertible state $\rho$ saturating the
$D_\alpha^t$-DPI.  All points in the blue region have the same Petz
fixed-point class, including the continuous order $\alpha=1$.  The
orange boundary, consisting of $t=1$ or $\alpha=2$, has the common
BS/$\chi^2$ class.  The latter contains the Petz class and can be
strictly larger, although the two may coincide for particular
channels and reference states.}
\label{fig:consolidated-Renyi-fixed-point-phase}
\end{figure}
\FloatBarrier

\subsection{State recovery and the equality hierarchy}

The likelihood-ratio condition, rather than a state identity alone, is
the complete invariant for nonmaximal DPI equality.  Nevertheless,
choosing $g(x)=x$ in \eqref{eq:consolidated-all-g} yields a useful
state identity.

\subsubsection{The asymmetric state identity}

\begin{proposition}[The asymmetric state-recovery consequence]
\label{prop:consolidated-asymmetric-recovery}
If the conditions of
\Cref{thm:consolidated-equality-recovery} hold, then
\begin{equation}
 \rho=\mathcal R_{\sigma,\Phi}^{\rm asym}(\Phi(\rho)),
 \qquad
 \mathcal R_{\sigma,\Phi}^{\rm asym}(X)
 :=\Phi^*(X[\Phi(\sigma)]^{-1})\sigma.
 \label{eq:consolidated-asymmetric-map}
\end{equation}
The map is trace preserving and recovers $\Phi(\sigma)\mapsto\sigma$,
but is generally neither positive nor Hermiticity preserving.  For
$0<t<1$, \eqref{eq:consolidated-asymmetric-map} is necessary but
not sufficient for $D^t$-saturation; it characterizes the larger
BS equality class.  The full criterion is
\eqref{eq:consolidated-channel-likelihood}.
\end{proposition}

\begin{proof}
The identities
\begin{equation}
 \Gamma_t\sigma^{1/2}=\rho\sigma^{-1/2},\qquad
 \Gamma_t^\Phi[\Phi(\sigma)]^{1/2}=\Phi(\rho)[\Phi(\sigma)]^{-1/2}
\end{equation}
turn \eqref{eq:consolidated-all-g} with $g(x)=x$ into
\begin{equation}
 \rho\sigma^{-1/2}
 =\Phi^*(\Phi(\rho)[\Phi(\sigma)]^{-1})\sigma^{1/2}.
\end{equation}
Right multiplication by $\sigma^{1/2}$ proves the state identity.
Moreover,
\begin{equation}
 \operatorname{Tr}\mathcal R_{\sigma,\Phi}^{\rm asym}(X)
 =\operatorname{Tr}[X[\Phi(\sigma)]^{-1}\Phi(\sigma)]=\operatorname{Tr} X,
\end{equation}
and $\mathcal R_{\sigma,\Phi}^{\rm asym}(\Phi(\sigma))=\sigma$.
\end{proof}

\begin{remark}[Asymmetric and symmetric BS recovery]
The map in \eqref{eq:consolidated-asymmetric-map} is the
adjoint-order version of the asymmetric BS recovery map introduced
in 
\cite[Theorem~4.2 and (14)]{bluhm2020strengthened}.
Indeed, their convention is
\begin{equation}
 \mathcal B_{\sigma,\Phi}(X)
 :=\sigma\Phi^*\!\left([\Phi(\sigma)]^{-1}X\right),
 \qquad
 \mathcal R_{\sigma,\Phi}^{\rm asym}(X)
 =\mathcal B_{\sigma,\Phi}(X)^*,\quad X=X^*.
 \label{eq:consolidated-two-asymmetric-orders}
\end{equation}
Consequently, their fixed-point condition and
\eqref{eq:consolidated-asymmetric-map} are equivalent.

A positive, but nonlinear, symmetric version is obtained by setting
for $X\geq0$
\begin{equation}
 \mathcal R_{\sigma,\Phi}^{\rm sym}(X)
 :=\left[\sigma\Phi^*\!\left(
 [\Phi(\sigma)]^{-1}X^2[\Phi(\sigma)]^{-1}
 \right)\sigma\right]^{1/2}.
 \label{eq:consolidated-nonlinear-symmetric-BS-map}
\end{equation}
The symmetric BS recovery theorem
\cite[Theorem~5.1]{Bluhm:2025kai} gives the equivalences
\begin{multline}
 \widehat D(\rho\Vert\sigma)
 =\widehat D(\Phi(\rho)\Vert\Phi(\sigma))
 \quad\Longleftrightarrow\quad
 \rho=\mathcal R_{\sigma,\Phi}^{\rm asym}(\Phi(\rho))
 \Longleftrightarrow\quad
 \rho=\mathcal R_{\sigma,\Phi}^{\rm sym}(\Phi(\rho)).
 \label{eq:consolidated-symmetric-BS-equivalence}
\end{multline}
For the tripartite specialization
$\Phi=\operatorname{Tr}_A$, $\sigma=\rho_{AB}\otimes\tau_C$, this last identity
becomes, with identity embeddings suppressed,
\begin{equation}
 \rho_{ABC}
 =\left(
 \rho_{AB}\rho_B^{-1}\rho_{BC}^{2}
 \rho_B^{-1}\rho_{AB}
 \right)^{1/2}.
 \label{eq:consolidated-tripartite-nonlinear-symmetric-BS}
\end{equation}

There is also a different symmetric representative which is linear
and completely positive, although generally not trace preserving.
For an invertible tripartite state $\varrho_{ABC}$, we define
\begin{align}
 \Phi_{B\to AB}^{\rm BS}(Z)
 :={}&\varrho_B^{1/2}
 (\varrho_B^{-1/2}\varrho_{AB}\varrho_B^{-1/2})^{1/2}
 \varrho_B^{-1/2}Z\varrho_B^{-1/2}
 (\varrho_B^{-1/2}\varrho_{AB}\varrho_B^{-1/2})^{1/2}
 \varrho_B^{1/2}.
 \label{eq:consolidated-BS-CP-map}
\end{align}
This map was introduced in
\cite[(54)]{alhambra2026conditional}.  Its exact common-fixed-point
property with the asymmetric BS condition was subsequently proved in
\cite[(13) and Corollary~3.9]{Bluhm:2025kai}:
\begin{equation}
 \begin{aligned}
  \varrho_{ABC}
  =(\mathcal R_{\varrho_{AB},\operatorname{Tr}_A}^{\rm asym}
    \otimes\id_C)(\varrho_{BC})
  &\Longleftrightarrow
  \varrho_{ABC}=\varrho_{AB}\varrho_B^{-1}\varrho_{BC}\\
  &\Longleftrightarrow
  \varrho_{ABC}
  =(\Phi_{B\to AB}^{\rm BS}\otimes\id_C)(\varrho_{BC}).
 \end{aligned}
 \label{eq:consolidated-tripartite-symmetric-BS-equivalence}
\end{equation}
In the first term, the opposite product order produced by
\eqref{eq:consolidated-asymmetric-map} is equivalent after taking
adjoints, since $\varrho_{ABC}=\varrho_{ABC}^*$.  Thus the nonlinear
symmetric identity is necessary whenever the hypotheses of
\Cref{prop:consolidated-asymmetric-recovery} hold, and in the
tripartite partial-trace setting so is the completely positive
identity.  For $0\leq t<1$ these remain conditions for the larger BS
equality class and are not, in general, sufficient for
$D^t$-saturation; the full criterion remains
\eqref{eq:consolidated-channel-likelihood}.
\end{remark}

\subsubsection{Equality classes and the endpoint hierarchy}

For an operator-convex generator $f$, we define the invertible equality
class
\begin{equation}
 \mathfrak F_{f,t}(\sigma,\Phi)
 :=\{\rho>0:
 D_f^t(\rho\Vert\sigma)
 =D_f^t(\Phi(\rho)\Vert\Phi(\sigma))\}.
 \label{eq:consolidated-Fft}
\end{equation}
We abbreviate
$\mathfrak F_t:=\mathfrak F_{x\log x,t}$, and write
\begin{equation}
 \mathfrak F_{\rm P}(\sigma,\Phi)
 :=\{\rho>0:\rho=\mathcal P_{\sigma,\Phi}(\Phi(\rho))\}
\end{equation}
for the Petz class, where the Petz map is defined in
\eqref{rw:eq:petz-map}.  The BS class $\mathfrak F_{\rm BS}$ is the
common equality class at the maximal endpoint $t=1$.  We write
$\mathfrak F_{\alpha,t}$ for the equality class of
$D_\alpha^t$.

\begin{proposition}[Complete equality-set hierarchy]
\label{prop:consolidated-general-hierarchy}
We take $\Phi$ to be a channel and all states to be invertible.
For every $0\leq t<1$ and every equality-determining
operator-convex $f$,
\begin{equation}
 \mathfrak F_{f,t}(\sigma,\Phi)
 =\mathfrak F_{\rm P}(\sigma,\Phi).
\end{equation}
In particular,
\begin{equation}
 \boxed{
 \mathfrak F_0(\sigma,\Phi)
 =\mathfrak F_t(\sigma,\Phi)
 =\mathfrak F_{\rm P}(\sigma,\Phi)
 \subseteq\mathfrak F_1(\sigma,\Phi)
 =\mathfrak F_{\rm BS}(\sigma,\Phi),
 \qquad 0\leq t<1.
 }
 \label{eq:consolidated-general-hierarchy}
\end{equation}
For the R\'enyi interpolation this becomes
\begin{equation}
 \boxed{\begin{gathered}
  \mathfrak F_{\alpha,t}=\mathfrak F_{\rm P},\\[-2pt]
  0\leq t<1,\qquad \alpha\in(0,1)\cup(1,2)
 \end{gathered}}
\end{equation}
and the same formula holds at $\alpha=1$ after the continuous
definition $D_1^t:=D^t$, whereas
\begin{equation}
 \boxed{
  \mathfrak F_{2,t}=\mathfrak F_{\rm BS},
  \qquad 0\leq t\leq1
 }
\end{equation}
At $t=1$, equality for any non-affine operator-convex generator
has the common maximal/BS equality class.  Thus all measure-rich
nonmaximal classes coincide with Petz sufficiency.  Among the Umegaki
and power/R\'enyi examples, only the endpoint and the purely
quadratic exception can be larger; sparse non-equality-determining
generators may also have larger equality sets, as explained below.
The inclusion in
\eqref{eq:consolidated-general-hierarchy} can be strict, but need not
be strict for each fixed channel and reference state.
\end{proposition}

\begin{proof}
If $\rho\in\mathfrak F_{\rm P}$, Petz sufficiency supplies a channel
$\mathcal P$ recovering both $\rho$ and $\sigma$
\cite[Theorem~3.18(i),(iii),(viii)]{hiai2017different}.  Applying
the $D_f^t$-DPI first to $\Phi$ and then to $\mathcal P$ gives
\begin{equation}
 \begin{aligned}
  D_f^t(\rho\Vert\sigma)
  \geq D_f^t(\Phi(\rho)\Vert\Phi(\sigma))\geq D_f^t(\mathcal P(\Phi(\rho))\Vert
               \mathcal P(\Phi(\sigma)))
   =D_f^t(\rho\Vert\sigma)
 \end{aligned}
\end{equation}
The reverse inclusion under the equality-determining support
condition is
\Cref{thm:consolidated-interior-rigidity} below.  At $t=1$,
the maximal-divergence preservation theorem
\cite[Theorem~3.34]{hiai2017different} identifies the common
non-affine equality class with $\mathfrak F_{\rm BS}$.  The
power generators have full representing support for
$\alpha\in(0,1)\cup(1,2)$.  At $\alpha=2$, the generator is
quadratic and $D_2^t$ is independent of $t$, with precisely the
BS/$\chi^2$ equality condition
\cite[Example~4.2 and Theorem~3.34]{hiai2017different}.
\end{proof}

\begin{remark}[Why ``any $f$'' needs a hypothesis]
The equality-determining assumption is substantive and essentially
sharp for this argument.  An affine generator carries no information,
and a purely quadratic generator gives the $t$-independent
$\chi^2$ equality class.  More generally, a representing measure
with too few support points need not determine the full functional
calculus of $\Gamma_t$ and $\Gamma_t^\Phi$.  Thus the universal
nonmaximal Petz statement applies to every measure-rich
operator-convex generator---in particular to Umegaki and to all
Petz--R\'enyi orders in the DPI range except the quadratic endpoint---
but not literally to every non-affine formula without qualification.
\end{remark}

\subsubsection{Nonmaximal rigidity}

The proof of the reverse inclusion is one of the main results of the paper, and it is a Stinespring refinement of the
partial-trace argument.  A direct application of partial-trace
rigidity to $J\rho J^*$ and $J\sigma J^*$ would be incomplete:
these states are invertible only on their common support $P=JJ^*$, and the compressed
inclusion $x\mapsto P(x\otimes I_{\mathcal E})P$ is not generally
multiplicative.  Equality supplies exactly the missing one-sided
multiplicativity.

\begin{theorem}[Universal nonmaximal Petz rigidity]
\label{thm:consolidated-interior-rigidity}
For a channel
$\Phi:\mathcal B(\mathcal H)\to\mathcal B(\mathcal K)$,
states $\rho,\sigma>0$ satisfying
$\Phi(\rho)>0$ and $\Phi(\sigma)>0$, a parameter $0\leq t<1$,
and an equality-determining operator-convex function $f$ in the sense of
\eqref{eq:consolidated-equality-determining-support}, one has
\begin{equation}
 D_f^t(\rho\Vert\sigma)=D_f^t(\Phi(\rho)\Vert\Phi(\sigma))
 \quad\Longleftrightarrow\quad
 \rho=\mathcal P_{\sigma,\Phi}(\Phi(\rho))
 \quad\Longleftrightarrow\quad
 D(\rho\Vert\sigma)=D(\Phi(\rho)\Vert\Phi(\sigma)).
 \label{eq:consolidated-interior-rigidity}
\end{equation}
In the equality case one has the full Petz cocycle identity
\begin{equation}
 \rho^z\sigma^{-z}
 =\Phi^*\!\left([\Phi(\rho)]^{\,z}[\Phi(\sigma)]^{-z}\right),
 \qquad z\in\mathbb C.
\end{equation}
\end{theorem}

\begin{proof}
Only the implication from $D_f^t$-equality to Petz sufficiency is
nontrivial.  More precisely, we shall prove the Petz cocycle identity
\begin{equation}
 \rho^z\sigma^{-z}
 =\Phi^*\!\left(
 [\Phi(\rho)]^z[\Phi(\sigma)]^{-z}\right),
 \qquad z\in\mathbb C,
 \label{eq:consolidated-proof-target-Petz-cocycle}
\end{equation}
whose restriction to purely imaginary $z$ characterizes sufficiency.
We divide the argument into four steps.

\smallskip
\noindent\underline{\textbf{Step 1.}}\enspace
As stated above, the target is Petz's modular-cocycle criterion for sufficiency
\cite[Theorem~2]{petz2003monotonicity}; the equivalent
reversibility conditions appear in
\cite[Theorem~3.18(i),(vi)]{hiai2017different}.  Equality propagation,
however, initially gives identities for the powers of the geodesic
operators $\Gamma_t$ and $\Gamma_t^\Phi$, rather than for the Petz
cocycle itself.  The discrete modular orbits
introduced next provide the bridge: their ordered products represent
these powers on the corresponding cyclic vectors, consecutive
quotients isolate individual orbit points, and Step~3 promotes the
resulting lattice identities to the full complex orbit.  Thus the
proof combines the standard multiplicative-domain route to Petz
sufficiency with a finite-dimensional discrete spectral continuation.
The latter is analogous in spirit to the Cauchy-matrix interpolation
used to recover the full functional calculus in
\cite[Lemma~5.2 and Theorem~5.1]{hiai2011quantum}, but the sampled
orbit, quotient extraction, and Vandermonde implementation below are
adapted to the present geodesic operator.

Since $t<1$, one has $1-t>0$.  We define the input and output modular orbits
\begin{align*}
 X_0&:=\rho\sigma^{-1},&
 x_0&:=\Phi(\rho)[\Phi(\sigma)]^{-1},\\
 X_j&:=\rho^{j(1-t)}X_0\rho^{-j(1-t)},&
 x_j&:=[\Phi(\rho)]^{j(1-t)}x_0
       [\Phi(\rho)]^{-j(1-t)}\qquad(j\geq0).
\end{align*}
We set $M_0=m_0=I$ and, for $n\geq1$, define
\begin{equation}
 M_n:=X_{n-1}\cdots X_0,\qquad
 m_n:=x_{n-1}\cdots x_0.
\end{equation}
Indeed, direct induction from the explicit action of the geodesic
modular operator gives
\begin{equation}
 \Gamma_t^n\sigma^{1/2}
 =\rho^{n(1-t)}(\rho^t\sigma^{-1})^n\sigma^{1/2}
 =M_n\sigma^{1/2},\qquad
 (\Gamma_t^\Phi)^n[\Phi(\sigma)]^{1/2}
 =m_n[\Phi(\sigma)]^{1/2},
 \qquad n\geq0.
 \label{eq:consolidated-orbit-moment-representation}
\end{equation}
Thus $M_n$ and $m_n$ are the ordinary-operator representatives of the
geodesic modular moments.  Equality of those moments will now
propagate along the discrete orbits.  Starting from condition~\textup{(5)}
of \Cref{thm:consolidated-equality-recovery}, namely the assumed
$D_f^t$-DPI equality, the implication
\textup{(5)}$\Rightarrow$\textup{(3)}, applied to $g(x)=x^n$ in
\eqref{eq:consolidated-all-g}, gives the vector identity
\begin{equation}
 M_n\sigma^{1/2}
 =\mathsf V_{\sigma,\Phi}
   \bigl(m_n[\Phi(\sigma)]^{1/2}\bigr)
 =\Phi^*(m_n)\sigma^{1/2}.
 \label{eq:consolidated-orbit-propagated-vector}
\end{equation}
Since $\sigma^{1/2}$ is invertible, the vector identity gives the
operator equality below.  The implication
\textup{(5)}$\Rightarrow$\textup{(4)} in
\Cref{thm:consolidated-equality-recovery} says that the cyclic
intertwiner is isometric on these vectors and gives the accompanying
norm equality:
\begin{equation}
 M_n=\Phi^*(m_n),\qquad
 \lVert M_n\sigma^{1/2}\rVert_2
 =\lVert m_n[\Phi(\sigma)]^{1/2}\rVert_2.
 \label{eq:norm_identity_Petz}
\end{equation}
Together, \eqref{eq:consolidated-orbit-moment-representation} and
\eqref{eq:consolidated-orbit-propagated-vector} turn the propagated
Hilbert--Schmidt moments into ordinary operator identities and norm
equalities.  The latter expose the equality case of the
Kadison--Schwarz inequality used next.

Because $\Phi$ is trace preserving, $\Phi^*$ is unital; since it is
also completely positive, the Kadison--Schwarz inequality applies
\cite{choi1974schwarz}:
for every $y\in\mathcal B(\mathcal K)$,
\[
 \Phi^*(y^*y)\geq \Phi^*(y)^*\Phi^*(y).
\]
Applying this inequality to $y=m_n$, and using
$M_n=\Phi^*(m_n)$ together with \eqref{eq:norm_identity_Petz}, yields
\begin{align*}
 0
 &=\lVert m_n[\Phi(\sigma)]^{1/2}\rVert_2^2
   -\lVert M_n\sigma^{1/2}\rVert_2^2=\operatorname{Tr}\sigma\left[
 \Phi^*(m_n^*m_n)-\Phi^*(m_n)^*\Phi^*(m_n)\right].
\end{align*}
The operator in brackets is positive, and $\sigma>0$; hence it
vanishes.  We have therefore proved
\begin{equation}
 \Phi^*(m_n^*m_n)
 =\Phi^*(m_n)^*\Phi^*(m_n),
 \qquad n\geq0.
 \label{eq:consolidated-Kadison-equality-mn}
\end{equation}
This is the \emph{Kadison equality} for $m_n$: it is precisely the
equality case of the preceding Kadison--Schwarz inequality when
$y=m_n$.

\smallskip
\noindent\underline{\textbf{Step 2.}}\enspace
Now we interpret this equality through a one-sided multiplicative
domain of the adjoint channel $\Phi^*$, using a Stinespring
representation.  The two-sided multiplicative-domain formulation appears in
\cite[Section~2.5, especially Eqs.~(2.12)--(2.13)]{hiai2017different}.
As recalled in \Cref{rw:subsec:states-channels}, a Stinespring
isometry for $\Phi$ consists of a finite-dimensional ancillary Hilbert
space $\mathcal E$ and an isometry
$J:\mathcal H\to\mathcal K\otimes\mathcal E$, so that
$J^*J=I_{\mathcal H}$ and
\begin{equation}
 \Phi(A)=\operatorname{Tr}_{\mathcal E}(JAJ^*),\qquad
 \Phi^*(y)=J^*(y\otimes I_{\mathcal E})J
 \label{eq:consolidated-proof-Stinespring-representation}
\end{equation}
for every $A\in\mathcal B(\mathcal H)$ and
$y\in\mathcal B(\mathcal K)$; the same representation is recorded in
\eqref{rw:eq:stinespring}.  Hence $P:=JJ^*$ is the orthogonal
projection of $\mathcal K\otimes\mathcal E$ onto the embedded copy
$J\mathcal H$ of the input space.

For the right
multiplicative domain
\begin{equation}
\mathfrak R_{\Phi^*}
:=\{x:\Phi^*(x^*x)=\Phi^*(x)^*\Phi^*(x)\},
\end{equation}
the Stinespring representation gives the more explicit identity
\begin{align}
 &\Phi^*(x^*x)-\Phi^*(x)^*\Phi^*(x)
 =J^*(x^*\otimes I_{\mathcal E})(I-P)
       (x\otimes I_{\mathcal E})J\geq0.
 \label{eq:consolidated-right-domain-Stinespring-defect}
\end{align}
The right-hand side is the square of
$(I-P)(x\otimes I_{\mathcal E})J$.  Consequently,
\begin{equation}
 x\in\mathfrak R_{\Phi^*}
 \quad\Longleftrightarrow\quad
 (I-P)(x\otimes I_{\mathcal E})J=0
 \quad\Longleftrightarrow\quad
 (x\otimes I_{\mathcal E})J=J\Phi^*(x).
\end{equation}
Whenever $x\in\mathfrak R_{\Phi^*}$, it follows more explicitly that
\begin{align}
 \Phi^*(yx)
 =J^*(y\otimes I_{\mathcal E})(x\otimes I_{\mathcal E})J
 =J^*(y\otimes I_{\mathcal E})J\Phi^*(x)
 =\Phi^*(y)\Phi^*(x)
 \label{eq:consolidated-right-domain-multiplication}
\end{align}
for every $y\in\mathcal B(\mathcal K)$.
We record explicitly the algebraic consequences of the Stinespring
criterion.  The identity belongs to $\mathfrak R_{\Phi^*}$, and
linearity of
$(x\otimes I_{\mathcal E})J=J\Phi^*(x)$ shows that
$\mathfrak R_{\Phi^*}$ is a linear space.  If
$x,z\in\mathfrak R_{\Phi^*}$, then
\begin{align}
 (xz\otimes I_{\mathcal E})J
 &=(x\otimes I_{\mathcal E})(z\otimes I_{\mathcal E})J\notag\\
 &=(x\otimes I_{\mathcal E})J\Phi^*(z)
  =J\Phi^*(x)\Phi^*(z)
  =J\Phi^*(xz),
 \label{eq:consolidated-right-domain-product}
\end{align}
where the last equality is
\eqref{eq:consolidated-right-domain-multiplication} with $y=x$ and
right factor $z$.  Hence $xz\in\mathfrak R_{\Phi^*}$, so
$\mathfrak R_{\Phi^*}$ is a unital linear algebra.

It is also inverse closed under the following natural hypothesis.  If
$x\in\mathfrak R_{\Phi^*}$ and both $x$ and $\Phi^*(x)$ are
invertible, then the Stinespring relation implies
\begin{equation}
 (x^{-1}\otimes I_{\mathcal E})J
 =J\Phi^*(x)^{-1}.
 \label{eq:consolidated-right-domain-inverse-relation}
\end{equation}
Indeed, this follows by multiplying
$(x\otimes I_{\mathcal E})J=J\Phi^*(x)$ on the left by
$x^{-1}\otimes I_{\mathcal E}$ and on the right by
$\Phi^*(x)^{-1}$.  Applying $J^*$ to
\eqref{eq:consolidated-right-domain-inverse-relation} gives
\begin{equation}
 \Phi^*(x^{-1})=\Phi^*(x)^{-1},
\end{equation}
and the Stinespring criterion then shows that
$x^{-1}\in\mathfrak R_{\Phi^*}$.

Next, \eqref{eq:consolidated-Kadison-equality-mn} and the definition of
$\mathfrak R_{\Phi^*}$ therefore show that every $m_n$ lies in this
right multiplicative domain.  Since $m_n$ and
$M_n=\Phi^*(m_n)$ are invertible, $m_n^{-1}$ lies there as well.
Thus
\begin{equation}
 x_n=m_{n+1}m_n^{-1}\in\mathfrak R_{\Phi^*},\qquad
 X_n=M_{n+1}M_n^{-1}=\Phi^*(x_n).
\end{equation}
This quotient step is the reason for introducing the ordered products
$M_n$ and $m_n$: it separates each individual modular-orbit sample
from the moment identities.

\smallskip
\noindent\underline{\textbf{Step 3.}}\enspace
In this step, we extend the discrete identities to arbitrary complex
parameters.  We write
\begin{equation}
 X(z):=\rho^zX_0\rho^{-z},\qquad
 x(z):=[\Phi(\rho)]^zx_0[\Phi(\rho)]^{-z}.
\end{equation}
We now make explicit what is used from the preceding steps.
By definition,
\[
 X\bigl(n(1-t)\bigr)=X_n,\qquad
 x\bigl(n(1-t)\bigr)=x_n,\qquad n\geq0.
\]
Step~1 produced the moment identities
$M_n=\Phi^*(m_n)$, while Step~2 took consecutive quotients of those
identities and used the right multiplicative domain to isolate the
individual orbit points.  Consequently, the conclusion of Step~2 can
be rewritten as
\begin{equation}
 X\bigl(n(1-t)\bigr)
 =\Phi^*\!\left(x\bigl(n(1-t)\bigr)\right),
 \qquad
 x\bigl(n(1-t)\bigr)\in\mathfrak R_{\Phi^*},
 \qquad n\geq0.
 \label{eq:consolidated-discrete-orbit-input}
\end{equation}
These are the two discrete inputs used below.  The first identity
interpolates to the operator equality for arbitrary $z$.  The second,
together with the linearity of $\mathfrak R_{\Phi^*}$, yields
multiplicative-domain membership for arbitrary $z$.

If $\rho=\sum_a r_aP_a$, then, for example,
\begin{equation}
 X(z)=\sum_{a,b}\left(\frac{r_a}{r_b}\right)^zP_aX_0P_b,
 \label{eq:consolidated-orbit-spectral-expansion}
\end{equation}
and $x(z)$ has the analogous expansion in the spectral projections of
$\Phi(\rho)$.  We spell out the resulting finite-dimensional
interpolation.  Put $h:=1-t>0$ and
$F(z):=X(z)-\Phi^*(x(z))$.  After grouping equal spectral ratios from
the two expansions, there are distinct numbers
$\lambda_1,\ldots,\lambda_N>0$ and operator coefficients
$C_1,\ldots,C_N$ such that
\begin{equation}
 F(z)=\sum_{j=1}^N\lambda_j^z C_j.
 \label{eq:consolidated-Step3-exponential-sum}
\end{equation}
The first identity in
\eqref{eq:consolidated-discrete-orbit-input} gives, for
$n=0,\ldots,N-1$,
\begin{equation}
 0=F(nh)=\sum_{j=1}^N(\lambda_j^h)^n C_j.
 \label{eq:consolidated-Step3-Vandermonde-system}
\end{equation}
The scalar matrix
$V=((\lambda_j^h)^n)_{0\leq n\leq N-1,\,1\leq j\leq N}$ is a
Vandermonde matrix.  Its nodes $\lambda_j^h$ are distinct because
$h>0$, so $V$ is invertible.  Multiplying
\eqref{eq:consolidated-Step3-Vandermonde-system} by $V^{-1}$ recovers
each operator coefficient and gives $C_j=0$ for every $j$.  Thus
$F(z)=0$ for every $z\in\mathbb C$.

The multiplicative-domain conclusion follows from the same
coefficient recovery.  Namely, for the distinct spectral ratios
$\mu_1,\ldots,\mu_M>0$ of $\Phi(\rho)$, the corresponding grouped
expansion has the form
\[
 x(z)=\sum_{j=1}^M\mu_j^z B_j.
\]
The inverse of the Vandermonde matrix
$W=((\mu_j^h)^n)_{0\leq n\leq M-1,\,1\leq j\leq M}$ expresses every
coefficient as
\[
 B_j=\sum_{n=0}^{M-1}(W^{-1})_{jn}x(nh).
\]
Every sample $x(nh)$ belongs to $\mathfrak R_{\Phi^*}$ by
\eqref{eq:consolidated-discrete-orbit-input}; hence the linearity of
$\mathfrak R_{\Phi^*}$ gives $B_j\in\mathfrak R_{\Phi^*}$ for every
$j$.  The expansion of $x(z)$ then gives
$x(z)\in\mathfrak R_{\Phi^*}$ for arbitrary $z$.  We have proved
\begin{equation}
 X(z)=\Phi^*(x(z)),\qquad
 x(z)\in\mathfrak R_{\Phi^*},\qquad z\in\mathbb C.
 \label{eq:consolidated-full-orbit-domain}
\end{equation}
\smallskip
\noindent\underline{\textbf{Step 4.}}\enspace
We use the multiplicative-domain property to recover the Petz cocycle.
More precisely, Step~3 gives, for every integer $k\geq0$,
$X(k)=\Phi^*(x(k))$ and
$x(k)\in\mathfrak R_{\Phi^*}$, as recorded in
\eqref{eq:consolidated-full-orbit-domain}.  Successive use of the
right-multiplication rule
\eqref{eq:consolidated-right-domain-multiplication}, starting with
the rightmost factor, therefore gives
\begin{equation}
 \Phi^*\!\left(x(m-1)\cdots x(0)\right)
 =\Phi^*(x(m-1))\cdots\Phi^*(x(0))
 =X(m-1)\cdots X(0).
 \label{eq:consolidated-Step4-product-pullback}
\end{equation}
The definitions of the two orbits also give
\[
 X(k)=\rho^{k+1}\sigma^{-1}\rho^{-k},\qquad
 x(k)=[\Phi(\rho)]^{k+1}[\Phi(\sigma)]^{-1}
       [\Phi(\rho)]^{-k}.
\]
Thus the adjacent powers of $\rho$, respectively of $\Phi(\rho)$,
cancel in each ordered product.  Combining these two telescoping
identities with \eqref{eq:consolidated-Step4-product-pullback}, for
every integer $m\geq1$, yields
\begin{align*}
 \rho^m\sigma^{-m}
 &=X(m-1)\cdots X(0)\\
 &=\Phi^*\!\left(x(m-1)\cdots x(0)\right)\\
 &=\Phi^*\!\left(
 [\Phi(\rho)]^m[\Phi(\sigma)]^{-m}\right).
\end{align*}
We spell out the second finite spectral/Vandermonde continuation.  Let
\begin{equation*}
 \rho=\sum_a r_aP_a,qquad
 \sigma=\sum_b s_bQ_b,qquad
 \Phi(\rho)=\sum_c\widetilde r_c\widetilde P_c,qquad
 \Phi(\sigma)=\sum_d\widetilde s_d\widetilde Q_d
\end{equation*}
be their spectral decompositions.  Invertibility makes all the
eigenvalues displayed here strictly positive, and functional calculus
gives
\begin{align}
 \rho^z\sigma^{-z}
 &=\sum_{a,b}\left(\frac{r_a}{s_b}\right)^zP_aQ_b,
 \label{eq:consolidated-Petz-cocycle-spectral-expansion}\\
 [\Phi(\rho)]^z[\Phi(\sigma)]^{-z}
 &=\sum_{c,d}
 \left(\frac{\widetilde r_c}{\widetilde s_d}\right)^z
 \widetilde P_c\widetilde Q_d.
 \label{eq:consolidated-output-Petz-cocycle-spectral-expansion}
\end{align}
Define
\begin{equation*}
 G(z):=\rho^z\sigma^{-z}
 -\Phi^*\!\left(
 [\Phi(\rho)]^z[\Phi(\sigma)]^{-z}\right).
\end{equation*}
After grouping equal ratios in
\eqref{eq:consolidated-Petz-cocycle-spectral-expansion} and
\eqref{eq:consolidated-output-Petz-cocycle-spectral-expansion}, there
are distinct numbers $\theta_1,\ldots,\theta_N>0$ and operator
coefficients $D_1,\ldots,D_N$ such that
\begin{equation}
 G(z)=\sum_{j=1}^N\theta_j^zD_j,
 \qquad \theta_j^z:=e^{z\log\theta_j},
 \label{eq:consolidated-Petz-cocycle-exponential-sum}
\end{equation}
The identities just proved at the positive integers say that, for
$m=1,\ldots,N$,
\begin{equation}
 0=G(m)=\sum_{j=1}^N\theta_j^mD_j.
 \label{eq:consolidated-Petz-cocycle-Vandermonde-system}
\end{equation}
The coefficient matrix $(\theta_j^m)_{1\leq m,j\leq N}$ is the
product of the ordinary Vandermonde matrix
$((\theta_j)^{m-1})_{m,j}$ and the invertible diagonal matrix
$\operatorname{diag}(\theta_1,\ldots,\theta_N)$.  It is therefore
invertible, so
\eqref{eq:consolidated-Petz-cocycle-Vandermonde-system} forces
$D_j=0$ for every $j$.  Hence $G(z)=0$ for all $z\in\mathbb C$, that
is,
\begin{equation}
 \rho^z\sigma^{-z}
 =\Phi^*\!\left(
 [\Phi(\rho)]^z[\Phi(\sigma)]^{-z}\right),
 \qquad z\in\mathbb C.
\end{equation}
This is the cocycle identity displayed in the theorem.  At purely
imaginary $z$, this is Petz's modular-cocycle criterion
\cite[Theorem~3.18(i),(iii),(vi),(viii)]{hiai2017different} and hence
gives Umegaki equality.
Conversely, Umegaki equality supplies a CPTP Petz map recovering both
states; applying the $D_f^t$-DPI to $\Phi$ and then to that
recovery map gives $D_f^t$-equality for every admissible generator,
whether or not its measure is equality determining.

The restriction $t<1$ is essential: at $t=1$, one has $1-t=0$, so the
discrete orbit no longer separates spectral ratios, and the equality
class can enlarge to the BS class.
\end{proof}

\begin{remark}[Common-support extension]
The theorem, the preceding equality-set hierarchy, and the abstract
likelihood-ratio equivalences remain valid when $\rho$ and $\sigma$
have the same possibly proper support $P$.  One applies the invertible
result to the compressed channel
\begin{equation}
 \Phi_{P,Q}:P\mathcal B(\mathcal H)P
 \longrightarrow Q\mathcal B(\mathcal K)Q,
 \qquad Q:=\operatorname{supp}\Phi(\rho)
          =\operatorname{supp}\Phi(\sigma),
\end{equation}
whose adjoint is
$\Phi_{P,Q}^*(Y)=P\Phi^*(Y)P$.  All powers and inverses are then
computed in the supported corners.  This does not cover unequal
supports, for which the negative powers and full complex cocycle
require a separate formulation.
For a partial trace, the supported adjoint is
$Y\mapsto P(I_A\otimes Y)P$; hence the uncompressed tensor identity
in \eqref{eq:consolidated-partial-trace-likelihood} must be replaced
by its corresponding compressed form.
\end{remark}

\begin{remark}[The implicit half-power identity is Petz recovery]
Combining
\Cref{thm:consolidated-equality-recovery,thm:consolidated-interior-rigidity}
shows that, for every $0\leq t<1$, the apparently implicit condition
\eqref{eq:consolidated-half-equality} is equivalent to ordinary Petz
recovery.  Written entirely in terms of $\rho$, $\sigma$, $\Phi$, and
$t$, the former condition is
\begin{equation}
\begin{aligned}
 &\rho^{(t-1)/2}\sigma^{1/2}
 \left(\sigma^{-1/2}\rho^t\sigma^{-1/2}\right)^{-1/2}
 \\
 &\quad=
 \Phi^*\!\left(
 [\Phi(\rho)]^{(t-1)/2}[\Phi(\sigma)]^{1/2}
 \left(
 [\Phi(\sigma)]^{-1/2}[\Phi(\rho)]^t
 [\Phi(\sigma)]^{-1/2}
 \right)^{-1/2}
 [\Phi(\sigma)]^{-1/2}
 \right)\sigma^{1/2}.
\end{aligned}
 \label{eq:expanded-half-power-Petz-equivalence}
\end{equation}
It holds if and only if
\begin{equation}
{
 \rho
 =\sigma^{1/2}\Phi^*\!\left(
 [\Phi(\sigma)]^{-1/2}\Phi(\rho)
 [\Phi(\sigma)]^{-1/2}
 \right)\sigma^{1/2}.}
 \label{eq:expanded-Petz-recovery-after-rigidity}
\end{equation}
  This equivalence is striking from
\eqref{eq:expanded-half-power-Petz-equivalence} alone: $\rho$ occurs
there through several noncommuting powers on both sides, and there is
no evident algebraic manipulation that isolates $\rho$.  Nevertheless,
\Cref{thm:consolidated-equality-recovery} identifies that identity with
DPI saturation, and
\Cref{thm:consolidated-interior-rigidity} turns saturation into the
explicit recovered-state formula
\eqref{eq:expanded-Petz-recovery-after-rigidity}.  The restriction
$t<1$ is indispensable; at $t=1$ the corresponding equality class may
be strictly larger than the Petz class.

Another algebraic fixed-point condition was obtained in
\cite[Theorem~1 and Eq.~(30)]{leditzky2017data}.  To state it, we take
\begin{equation}
 \alpha\in(1/2,1)\cup(1,\infty),
 \qquad
 \gamma_\alpha:=\frac{1-\alpha}{2\alpha}.
\end{equation}
We recall from \Cref{rw:subsec:alpha-z-comparison} that the sandwiched
R\'enyi divergence is the branch $z=\alpha$ of the
$(\alpha,z)$-family.  In the present notation,
\begin{align}
 \widetilde Q_\alpha(\rho\Vert\sigma)
 &:=\operatorname{Tr}\!\left(
   \sigma^{\gamma_\alpha}\rho\sigma^{\gamma_\alpha}
   \right)^\alpha,
 \label{eq:sandwiched-Q-recall}\\
 \widetilde D_\alpha(\rho\Vert\sigma)
 &:=D_{\alpha,\alpha}(\rho\Vert\sigma)
 =\frac{1}{\alpha-1}
   \log\widetilde Q_\alpha(\rho\Vert\sigma).
 \label{eq:sandwiched-divergence-recall}
\end{align}
The invertible fixed points of the DPI for the sandwiched R\'enyi
divergence---that is, the pairs satisfying
$\widetilde D_\alpha(\rho\Vert\sigma)
=\widetilde D_\alpha(\Phi(\rho)\Vert\Phi(\sigma))$---are precisely
those obeying the nonlinear identity
\begin{equation}
\begin{aligned}
 &\sigma^{\gamma_\alpha}
  \left(\sigma^{\gamma_\alpha}\rho
              \sigma^{\gamma_\alpha}\right)^{\alpha-1}
  \sigma^{\gamma_\alpha}=\Phi^*\!\left(
 [\Phi(\sigma)]^{\gamma_\alpha}
 \left(
 [\Phi(\sigma)]^{\gamma_\alpha}\Phi(\rho)
 [\Phi(\sigma)]^{\gamma_\alpha}
 \right)^{\alpha-1}
 [\Phi(\sigma)]^{\gamma_\alpha}
 \right).
 \label{eq:LRD-sandwiched-equality}
\end{aligned}
\end{equation}
The results in
\cite{jencova2017preservation,jencova2021renyi} show that, throughout
the displayed range of $\alpha$, these sandwiched-R\'enyi fixed
points coincide with the Petz fixed points, equivalently with
\eqref{eq:expanded-Petz-recovery-after-rigidity}.  Thus
\eqref{eq:LRD-sandwiched-equality} is another algebraic expression for
the same recovery class as the half-power condition above.  The
endpoint $\alpha=1/2$ is excluded because preservation of fidelity
need not imply Petz sufficiency.
\end{remark}

\subsubsection{The nonmaximal recovery dictionary}

The following corollary contrasts with
\Cref{thm:consolidated-equality-recovery}.  That theorem gives a
$t$-dependent criterion for preservation at a fixed interpolation
parameter, whereas the result below shows that, for every $t<1$, the
resulting equality class collapses to the $t$-independent
Petz-reversible class.  This contrast parallels the endpoint results
of Hiai and Mosonyi: the fixed-parameter functional-calculus statement
in \Cref{thm:consolidated-equality-recovery}, particularly its $t=1$
specialization, is analogous to their maximal-$f$-divergence
preservation theorem, while the corollary below extends their
standard-$f$-divergence/Petz-reversibility dictionary from $t=0$ to
every $t<1$ \cite[Theorems~3.34 and~3.18]{hiai2017different}.

\begin{corollary}[The nonmaximal recovery dictionary]
\label{cor:consolidated-interior-recovery-dictionary}
Under the hypotheses of
\Cref{thm:consolidated-interior-rigidity}, the following are
equivalent for every $0\leq t<1$:
\begin{enumerate}
\item One equality-determining $D_f^t$ saturates its DPI;
\item Every geodesic $D_g^t$ for which the DPI is defined
saturates, including $D^t$ and every
$D_\alpha^t$, $\alpha\in(0,1)\cup(1,2)$;
\item The geodesic likelihood identity
\eqref{eq:consolidated-channel-likelihood} holds;
\item The Petz map recovers the state,
\begin{equation}
 \rho=\mathcal P_{\sigma,\Phi}(\Phi(\rho)),
 \qquad
 \sigma=\mathcal P_{\sigma,\Phi}(\Phi(\sigma));
\end{equation}
\item The cocycle identity
\begin{equation}
 \rho^z\sigma^{-z}
 =\Phi^*([\Phi(\rho)]^{\,z}[\Phi(\sigma)]^{-z}),
 \qquad z\in\mathbb C,
\end{equation}
holds;
\item With the reference-conditioned channel
\begin{equation}
 \Phi_\sigma(X):=[\Phi(\sigma)]^{-1/2}
 \Phi(\sigma^{1/2}X\sigma^{1/2})[\Phi(\sigma)]^{-1/2},
 \label{eq:consolidated-conditioned-channel}
\end{equation}
the whitened likelihood belongs to the Petz fixed-point algebra,
\begin{equation}
 \sigma^{-1/2}\rho\sigma^{-1/2}
 \in\operatorname{Fix}(\Phi^*\circ\Phi_\sigma).
\end{equation}
\end{enumerate}
Thus the recovery channel is the same Petz map for every nonmaximal
parameter; the $t$-dependent likelihood equation is an equivalent
coordinate description, not a different recovery mechanism.  At
$t=1$, the likelihood equation instead characterizes BS/maximal
equality and need not imply any positive recovery of $\rho$.
\end{corollary}

\begin{proof}
The equivalence of the first three conditions is
\Cref{thm:consolidated-equality-recovery};
\Cref{thm:consolidated-interior-rigidity} identifies them with
Umegaki equality.  Petz's reversibility theorem gives the equivalence
with the state-recovery, cocycle, and fixed-point conditions
\cite[Theorem~3.18(i),(iii),(vi),(viii),(ix)]{hiai2017different}.
The endpoint statement is
the maximal-divergence preservation theorem
\cite[Theorem~3.34]{hiai2017different}.
\end{proof}

\begin{remark}[Common-support version]
By the common-support extension above, the recovery dictionary also
holds in a fixed supported corner, with the compressed channel,
support inverses, and the Petz map of that corner.  No unequal-support
version is inferred from this observation.
\end{remark}

\begin{remark}[What the strengthened inequalities recover]
\label{rem:consolidated-strengthened-Petz-scope}
Combining the recovery dictionary in
\Cref{cor:consolidated-interior-recovery-dictionary} with
\eqref{eq:consolidated-power-bound},
\eqref{eq:consolidated-Dt-quartic}, and
\eqref{eq:consolidated-Renyi-log-bounds} yields strengthened DPIs with
the correct equality set.  These estimates do not, by themselves,
give a closed-form bound in
$\lVert\rho-\mathcal P_{\sigma,\Phi}(\Phi(\rho))\rVert_1$: passing
from the cyclic Hilbert--Schmidt defect $q_t^\Phi$ to that state
error requires a quantitative comparison between the two residuals.
At $t=0$ the standard Petz recovery estimates may be used directly.
For $0<t<1$, \Cref{cor:consolidated-nonmaximal-Petz-holder} supplies
such a comparison whenever the relevant smallest eigenvalues have
fixed positive lower bounds, uniformly
when $t$ stays a positive distance from $1$.  Its H\"older exponent and
constant are nonconstructive and depend on those lower bounds and the channel;
thus it is not a universal constant-free substitution.  Finding a
computable dimension- and spectrum-explicit comparison remains open.
At $t=1$, any Petz-state substitution is false in general because the
zero set has enlarged to the BS class.  Nor can the H\"older constants
be uniform as $t\uparrow1$: for a fixed invertible BS-equality triple
which is not Petz sufficient, continuity gives $q_t^\Phi\to0$, whereas
$\lVert\rho-\mathcal P_{\sigma,\Phi}(\Phi(\rho))\rVert_1$ stays
strictly positive.
\end{remark}

The channel equality problem is therefore complete.  We next
specialize the same recovery dictionary to partial traces, where its
zero sets become geodesic quantum Markov classes.

\section{Geodesic quantum Markov chains}
\label{sec:consolidated-tqmc}

The QMC and BS-QMC endpoint notions, together with their
significance, are recalled in \Cref{rw:subsec:qmc-bsqmc}.
This section has five aims.  We first introduce the three
$t$-conditional mutual informations and specialize the quantitative
DPI to the related tripartite pairs built from a state and its
marginals.  We then identify the invertible $t$-quantum Markov classes
and locate their unique endpoint jump.  Next, we recall from Bluhm
et al. examples witnessing strictness, including their Werner-state
construction.  We then explain the modular origin of the endpoint
exception through power-weighted symmetry.  Finally, we determine
which canonical, inverse, and fractional shadow correspondences
survive away from the BS endpoint.

\subsection{Tripartite specialization and the
\texorpdfstring{$t$-conditional mutual informations}{t-conditional mutual informations}}
\label{subsec:consolidated-tripartite-tcmi}

The definitions below are inspired by the one-sided and two-sided
BS-conditional mutual informations introduced in
\cite[Section~3.2, especially Eq.~(4)]{bluhm2023continuity}.  Here the BS-entropy is
replaced by the geodesic $t$-entropy, and the ordered-DPI
viewpoint additionally leads us to retain the reverse orientation.

Specializing the channel in \Cref{sec:consolidated-dpi} to a partial
trace, and to the case in which both states are in a tripartite space and the second state is defined from the first one as a (product of) marginal(s), produces three ordered DPI deficits.  Their distinction matters
away from the Umegaki endpoint, so we first fix the three conventions
and then record the quantitative consequence.

The superscripts $\os$, $\ts$, and $\rev$ stand for
\emph{one-sided}, \emph{two-sided}, and \emph{reverse}, respectively.

\begin{mainobject}{Main definition: The three $t$-conditional mutual
informations}
For an invertible state $\rho_{ABC}>0$, a fixed normalized state
$\tau_C>0$, and $t\in[0,1]$, we define
\begin{align}
 I_t^{\os}(A;C\mid B)
 &:=
 D^t(\rho_{ABC}\Vert\rho_{AB}\otimes\tau_C)
 -D^t(\rho_{BC}\Vert\rho_B\otimes\tau_C),\notag\\
 I_t^{\ts}(A;C\mid B)
 &:=
 D^t(\rho_{ABC}\Vert\rho_{AB}\otimes\rho_C)
 -D^t(\rho_{BC}\Vert\rho_B\otimes\rho_C),\notag\\
 I_t^{\rev}(A;C\mid B)
 &:=
 D^t(\rho_{AB}\otimes\tau_C\Vert\rho_{ABC})
 -D^t(\rho_B\otimes\tau_C\Vert\rho_{BC}).
 \label{eq:consolidated-three-CMIs}
\end{align}
\end{mainobject}
All three are nonnegative DPI deficits.  At $t=0$, the first two
equal the ordinary conditional mutual information, whereas the
reverse deficit is generally not numerically equal to it.

\begin{remark}[Non-invertible tripartite states]
For an arbitrary $\rho_{ABC}$, the one-sided and two-sided quantities
have the usual support extension and remain nonnegative.  Indeed, the
support of a joint state is contained in the tensor product of the
supports of the relevant marginals, so both ordered pairs satisfy the
support condition required by the extended DPI.  The reverse
expression is more delicate: the support of
$\rho_{AB}\otimes\tau_C$ need not be contained in that of
$\rho_{ABC}$, and its two terms can form the indeterminate difference
$+\infty-+\infty$.  We therefore retain invertibility whenever all
three orientations are considered simultaneously.
\end{remark}

\subsubsection{Comparison of the three orientations}

The common sign and, as shown later in
\Cref{cor:consolidated-tQMC-hierarchy}, the common zero sets do not
extend to a numerical ordering.

\begin{proposition}[No universal numerical hierarchy]
\label{prop:consolidated-CMI-incomparability}
For every invertible $\rho_{ABC}$ and every invertible normalized
$\tau_C$,
\begin{equation}
 I_0^{\os}(A;C\mid B)=I_0^{\ts}(A;C\mid B)
 =I(A;C\mid B)_\rho.
 \label{eq:consolidated-forward-CMI-tzero}
\end{equation}
If $\tau_C=\rho_C$, the first equality continues to hold for every
$t$ simply by definition.  These are the only universal comparisons
without state-dependent constants:
\begin{enumerate}[label=\textnormal{(\roman*)}]
 \item for every fixed $t\in(0,1]$, there are invertible examples with
 $I_t^{\os}<I_t^{\ts}$ and invertible examples with
 $I_t^{\os}>I_t^{\ts}$;
 \item for every fixed $t\in[0,1]$, the reverse quantity can lie
 strictly below or strictly above both forward quantities.
\end{enumerate}
In particular, no two of the three quantities admit a universal
ordering at the BS endpoint.
\end{proposition}

\begin{proof}
At $t=0$, we expand the logarithms of the product references in
\eqref{eq:consolidated-three-CMIs}.  The terms containing
$\log\tau_C$ cancel and give
\begin{equation}
 \begin{aligned}
 I_0^{\os}(A;C\mid B)
 &=\operatorname{Tr}(\rho_{ABC}\log\rho_{ABC})
   -\operatorname{Tr}(\rho_{AB}\log\rho_{AB})\\
 &\quad-\operatorname{Tr}(\rho_{BC}\log\rho_{BC})
   +\operatorname{Tr}(\rho_B\log\rho_B)
 =I(A;C\mid B)_\rho.
 \end{aligned}
\end{equation}
The same expansion with $\tau_C=\rho_C$ proves
\eqref{eq:consolidated-forward-CMI-tzero} for the two-sided
orientation.

We next give invertible examples which work for every fixed $t>0$.
It is enough to take $B$ one-dimensional and a cq state
\begin{equation}
 \rho_{AC}=\frac12\ketbra{0}{0}_A\otimes\eta_0
             +\frac12\ketbra{1}{1}_A\otimes\eta_1,
 \qquad
 \overline\eta:=\frac{\eta_0+\eta_1}{2}.
 \label{eq:consolidated-cq-ordering-state}
\end{equation}
The common-flag rule
\eqref{rw:eq:general-flag-rule-Dt} gives the exact identities
\begin{equation}
 \begin{aligned}
 I_t^{\os}
 &=\frac12\sum_{i=0}^1D^t(\eta_i\Vert\tau_C)
   -D^t(\overline\eta\Vert\tau_C),\\
 I_t^{\ts}
 &=\frac12\sum_{i=0}^1D^t(\eta_i\Vert\overline\eta),\\
 I_t^{\rev}
 &=\frac12\sum_{i=0}^1D^t(\tau_C\Vert\eta_i)
   -D^t(\tau_C\Vert\overline\eta).
 \end{aligned}
 \label{eq:consolidated-cq-three-CMIs}
\end{equation}
We denote the Pauli matrices by $X$ and $Z$ and fix $0<a<1$.
We first choose
\begin{equation}
 \eta_0=\frac{\1+aZ}{2},
 \qquad
 \eta_1=\frac{\1+aX}{2},
 \qquad
 \tau_C=\pi_C.
\end{equation}
Since $\pi_C$ is a scalar multiple of the identity, every state on
$C$ commutes with it.  We shall use this to replace
$D^t(\,\cdot\Vert\pi_C)$ by $D(\,\cdot\Vert\pi_C)$ through the
classical-collapse identity \eqref{rw:eq:classical-collapse}.  Only
after this replacement do we invoke the elementary Umegaki
relative-entropy identity
\cite{donald1986relative}
\begin{equation}
 \frac12\sum_{i=0}^1D(\eta_i\Vert\xi)
 -D(\overline\eta\Vert\xi)
 =\frac12\sum_{i=0}^1D(\eta_i\Vert\overline\eta)
 \label{eq:consolidated-Donald-identity}
\end{equation}
for an arbitrary invertible $\xi$.  Thus \eqref{eq:consolidated-Donald-identity}
is deliberately $t$-independent: it is an identity for the Umegaki
relative entropy, not an asserted analogue for $D^t$.

For the first choice above, \eqref{eq:consolidated-cq-three-CMIs},
the commuting-pair collapse, and then
\eqref{eq:consolidated-Donald-identity} with $\xi=\pi_C$ give
\begin{equation}
 \begin{aligned}
 I_t^{\os}
 &=\frac12\sum_{i=0}^1D^t(\eta_i\Vert\pi_C)
   -D^t(\overline\eta\Vert\pi_C)=\frac12\sum_{i=0}^1D(\eta_i\Vert\overline\eta).
 \end{aligned}
\end{equation}
On the other hand,
$[\eta_i,\overline\eta]\neq0$ for both $i$.  Strict
$t$-monotonicity from \Cref{rw:thm:Dt-strict-monotonicity} therefore
gives
\begin{equation}
 I_t^{\os}
 <\frac12\sum_{i=0}^1D^t(\eta_i\Vert\overline\eta)
 =I_t^{\ts},
 \qquad t>0.
 \label{eq:consolidated-os-less-ts}
\end{equation}

For the opposite direction, we also fix $0<b<1$ and take
\begin{equation}
 \eta_+=\frac{\1+aZ}{2},
 \qquad
 \eta_-=\frac{\1-aZ}{2},
 \qquad
 \tau_C=\frac{\1+bX}{2}.
\end{equation}
Now $\overline\eta=\pi_C$,
$[\eta_\pm,\tau_C]\neq0$, and
$[\pi_C,\tau_C]=0$.  Applying
strict $t$-monotonicity to each positive summand in
\eqref{eq:consolidated-cq-three-CMIs}, the commuting-pair collapse
to the subtracted term, and
\eqref{eq:consolidated-Donald-identity} with $\xi=\tau_C$ yields
\begin{equation}
 \begin{aligned}
 I_t^{\os}
 &=\frac12\sum_{\epsilon\in\{+,-\}}
       D^t(\eta_\epsilon\Vert\tau_C)
      -D^t(\pi_C\Vert\tau_C)\\
 &>\frac12\sum_{\epsilon\in\{+,-\}}
        D(\eta_\epsilon\Vert\tau_C)
       -D(\pi_C\Vert\tau_C)\\
 &=\frac12\sum_{\epsilon\in\{+,-\}}
        D(\eta_\epsilon\Vert\pi_C)
 =I_t^{\ts},
 \qquad t>0.
 \end{aligned}
 \label{eq:consolidated-os-greater-ts}
\end{equation}
This proves the two directions in (i) analytically, including at
$t=1$.

The reverse orientation is already incomparable in the classical
case.  In \eqref{eq:consolidated-cq-ordering-state}, we take
\begin{equation}
 \eta_0=\operatorname{diag}\!\left(\frac9{10},\frac1{10}\right),
 \qquad
 \eta_1=\operatorname{diag}\!\left(\frac12,\frac12\right),
 \qquad
 \overline\eta=\operatorname{diag}\!\left(\frac7{10},\frac3{10}\right).
\end{equation}
All pairs commute, so the values are independent of $t$.  The two
forward quantities equal
\begin{equation}
 \chi:=\frac12D(\eta_0\Vert\overline\eta)
       +\frac12D(\eta_1\Vert\overline\eta)
 =0.101749\ldots,
\end{equation}
whereas, for $\tau_C=\operatorname{diag}(u,1-u)$,
\begin{equation}
 I_t^{\rev}
 =u\log\frac7{3\sqrt5}+(1-u)\log\frac3{\sqrt5}.
 \label{eq:consolidated-classical-reverse-values}
\end{equation}
The full-support choices $u=9/10$ and $u=1/10$ give, respectively,
\begin{equation}
 I_t^{\rev}=0.067710\ldots<\chi,
 \qquad
 I_t^{\rev}=0.268762\ldots>\chi,
 \qquad 0\leq t\leq1.
\end{equation}
This proves (ii).
\end{proof}

\begin{remark}[Boundary scope at the Umegaki endpoint]
The identity \eqref{eq:consolidated-forward-CMI-tzero} holds for
arbitrary tripartite states under the support extension, since both
forward quantities reduce to the entropy formula for the ordinary
conditional mutual information.  The negative conclusions of the
proposition already follow from invertible examples and therefore
also rule out universal comparisons on any larger state class.
\end{remark}

\paragraph{Numerical visualization at the BS endpoint.}
The analytic incomparability at $t=1$ established in
\Cref{prop:consolidated-CMI-incomparability} can also be seen directly
in a random-state experiment.  We take $A$, $B$, and $C$ to be qubits
and generate $1200$ independent pairs $(\rho_{ABC},\tau_C)$ as
follows.  For normalized independent complex Gaussian vectors
$\ket{\psi}\in\mathbb C^8$ and $\ket{\phi}\in\mathbb C^2$, we set
\begin{equation}
 \rho_{ABC}=\frac9{20}\ketbra{\psi}{\psi}
             +\frac{11}{20}\frac{\1_{ABC}}8,
 \qquad
 \tau_C=\frac12\ketbra{\phi}{\phi}
             +\frac12\frac{\1_C}2.
 \label{eq:consolidated-BS-scatter-ensemble}
\end{equation}
Every sampled pair is invertible.  Using the fixed seed $20260811$ and
evaluating the three quantities at $t=1$ gives the pairwise scatter
plots in \Cref{fig:consolidated-BS-CMI-scatter}.  From left to right,
the numbers of points above and below the diagonal are
$(539,661)$, $(591,609)$, and $(708,492)$, respectively.

\FloatBarrier
\begin{figure}[!htbp]
 \centering
 \includegraphics[width=0.99\linewidth]{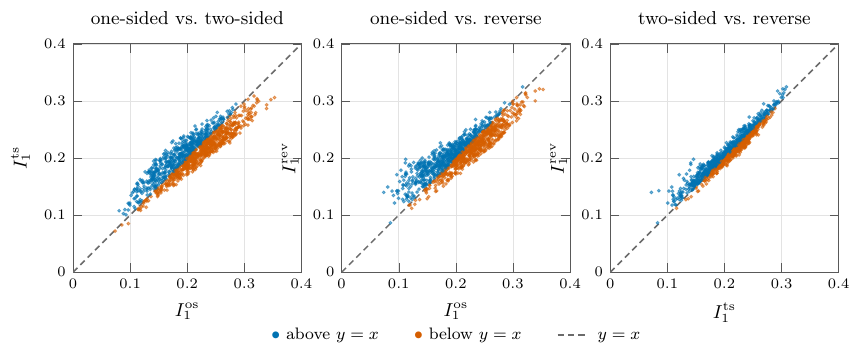}
 \caption{Pairwise comparison of the three BS conditional mutual
 informations for the random invertible pairs in
 \eqref{eq:consolidated-BS-scatter-ensemble}; the same sample is used
 in all three panels, and logarithms are natural.  The dashed line is
 $y=x$.  Blue points lie above the diagonal and orange points below
 it.  The occurrence of both colors in every panel visualizes the
 absence of a universal pairwise order at the BS endpoint.}
 \label{fig:consolidated-BS-CMI-scatter}
\end{figure}
\FloatBarrier

\paragraph{Comparisons up to reference-dependent constants.}
Despite this lack of a universal ordering, the two
forward quantities differ only through the replacement of the fixed
reference $\tau_C$ by the actual marginal $\rho_C$.  This observation
gives an elementary comparison controlled by the relative condition
number of these two states.

\begin{proposition}[Projective stability of the two forward
$t$-CMIs]
\label{prop:consolidated-forward-projective-comparison}
We define
\begin{equation}
 \kappa_C
 :=\frac{\lambda_{\max}
 \!\left(\rho_C^{-1/2}\tau_C\rho_C^{-1/2}\right)}
 {\lambda_{\min}
 \!\left(\rho_C^{-1/2}\tau_C\rho_C^{-1/2}\right)}.
 \label{eq:consolidated-reference-condition-number}
\end{equation}
For every $t\in[0,1]$,
\begin{equation}
 \boxed{
 \left|I_t^{\os}(A;C\mid B)-I_t^{\ts}(A;C\mid B)\right|
 \leq\log\kappa_C.}
 \label{eq:consolidated-forward-additive-comparison}
\end{equation}
The constant $\log\kappa_C$ is the Hilbert projective distance
between $\tau_C$ and $\rho_C$.  Equivalently, the exponentiated
deficits satisfy the genuine multiplicative comparison
\begin{equation}
 \kappa_C^{-1}e^{I_t^{\ts}(A;C\mid B)}
 \leq e^{I_t^{\os}(A;C\mid B)}
 \leq\kappa_Ce^{I_t^{\ts}(A;C\mid B)}.
 \label{eq:consolidated-forward-exponential-comparison}
\end{equation}
Consequently, on the region where
$\min\{I_t^{\os},I_t^{\ts}\}\geq\delta>0$,
\begin{equation}
 \frac{1}{1+\delta^{-1}\log\kappa_C}\,I_t^{\ts}
 \leq I_t^{\os}
 \leq
 \left(1+\frac{\log\kappa_C}{\delta}\right)I_t^{\ts}.
 \label{eq:consolidated-forward-away-zero-factor}
\end{equation}
In particular, $\tau_C=\rho_C$ gives equality, as it must.
\end{proposition}

\begin{proof}
The explicit expression in \eqref{rw:eq:Dt-explicit}, extended in
the second argument to arbitrary positive definite operators, has
the following two elementary properties when $R$ is a normalized
state:
\begin{align}
 S_1\leq S_2
 &\quad\Longrightarrow\quad
 D^t(R\Vert S_1)\geq D^t(R\Vert S_2),
 \label{eq:consolidated-Dt-second-order-reversal}\\
 D^t(R\Vert cS)&=D^t(R\Vert S)-\log c,
 \qquad c>0.
 \label{eq:consolidated-Dt-second-homogeneity}
\end{align}
Indeed, inversion reverses the L\"owner order, congruence by
$R^{t/2}$ preserves it, and the logarithm is operator monotone; the
homogeneity follows by pulling $c^{-1}$ out of the logarithm.

For the rest of this proof, we abbreviate the minimum and maximum
eigenvalues appearing in
\eqref{eq:consolidated-reference-condition-number} by
\begin{equation*}
 a_C:=\lambda_{\min}
 \!\left(\rho_C^{-1/2}\tau_C\rho_C^{-1/2}\right),
 \qquad
 b_C:=\lambda_{\max}
 \!\left(\rho_C^{-1/2}\tau_C\rho_C^{-1/2}\right).
\end{equation*}
Then $\kappa_C=b_C/a_C$, and conjugating the spectral bounds by
$\rho_C^{1/2}$ gives
$a_C\rho_C\leq\tau_C\leq b_C\rho_C$.  Since both states are
normalized, $a_C\leq1\leq b_C$.

For $(R,X)=(\rho_{ABC},\rho_{AB})$ and
$(R,X)=(\rho_{BC},\rho_B)$, we put
\begin{equation}
 \Delta_{R,X}:=
 D^t(R\Vert X\otimes\tau_C)
 -D^t(R\Vert X\otimes\rho_C).
\end{equation}
The preceding order bounds, together with
\eqref{eq:consolidated-Dt-second-order-reversal} and
\eqref{eq:consolidated-Dt-second-homogeneity}, give
\begin{equation}
 -\log b_C\leq\Delta_{R,X}\leq-\log a_C.
\end{equation}
By \eqref{eq:consolidated-three-CMIs}, the difference
$I_t^{\os}-I_t^{\ts}$ is the difference of the two corresponding
values of $\Delta_{R,X}$.  The width of the displayed interval is
$\log(b_C/a_C)$, proving
\eqref{eq:consolidated-forward-additive-comparison}.
If both CMIs are at least $\delta$, we write $h:=\log\kappa_C$.  Applying
the estimate in both directions gives
$I_t^{\os}\leq I_t^{\ts}+h\leq(1+h/\delta)I_t^{\ts}$ and
$I_t^{\ts}\leq I_t^{\os}+h\leq(1+h/\delta)I_t^{\os}$, which is
\eqref{eq:consolidated-forward-away-zero-factor}.
\end{proof}

\begin{remark}[Common marginal support]
The same estimate holds for the support-extended forward quantities
when
$\operatorname{supp}\rho_C=\operatorname{supp}\tau_C=:P_C$,
with the inverse and the extreme eigenvalues in
\eqref{eq:consolidated-reference-condition-number} taken on $P_C$.
If the two supports differ, the Hilbert projective distance is
infinite and the estimate has no finite content.
\end{remark}

The dependence on spectral conditioning cannot be removed if one
wants to compare all three orientations.  This already follows from
a full-support classical binary family, and hence the obstruction is
independent of $t$.  We fix $e\in(0,1)$, take $q\in(0,1/2)$, put
$d_q:=qe/(1-q)$, and consider the joint distribution, with rows
indexed by $A$ and columns by $C$,
\begin{equation}
 p_{AC}^{(q)}=\frac12
 \begin{pmatrix}
  q(1+e)&(1-q)(1-d_q)\\
  q(1-e)&(1-q)(1+d_q)
 \end{pmatrix}.
 \label{eq:consolidated-factor-counterexample}
\end{equation}
Its marginals are $p_A=(1/2,1/2)$ and $p_C=(q,1-q)$.  With
\begin{equation}
 g(x):=\frac{(1+x)\log(1+x)+(1-x)\log(1-x)}2,
\end{equation}
the two forward quantities coincide and obey
\begin{equation}
 I_t^{\os}=I_t^{\ts}
 =qg(e)+(1-q)g(d_q)=qg(e)+O(q^2).
 \label{eq:consolidated-factor-counterexample-forward}
\end{equation}
For $\tau_C=(u,1-u)$, direct cancellation of the classical
likelihoods gives
\begin{equation}
 I_t^{\rev}
 =-\frac u2\log(1-e^2)
  -\frac{1-u}{2}\log(1-d_q^2).
 \label{eq:consolidated-factor-counterexample-reverse}
\end{equation}
Taking $u=1/2$ makes
$I_t^{\rev}/I_t^{\os}\to+\infty$ as $q\downarrow0$, whereas
$u=q^2$ makes $I_t^{\os}/I_t^{\rev}\to+\infty$.  Thus no factor
depending only on $t$ and the dimensions can compare the reverse
quantity with either forward quantity.  Any prospective reverse
factor must therefore retain spectral-conditioning information.  The
projective condition number in
\eqref{eq:consolidated-reference-condition-number} provides precisely
such information for the two forward quantities.

For $\alpha>0$, $\alpha\neq1$, we define
$I_{\alpha,t}^{\bullet}$, with
$\bullet\in\{\os,\ts,\rev\}$, by replacing every occurrence of
$D^t$ in \eqref{eq:consolidated-three-CMIs} by $D_\alpha^t$.
For $0<\alpha<1$ and $1<\alpha\leq2$, these are nonnegative DPI
deficits and will be called the R\'enyi $t$-conditional mutual
informations.  For $\alpha>2$, the same notation denotes only the
corresponding algebraic differences: no general nonnegativity or CMI
interpretation is asserted.

At the exceptional quadratic order, the same comparison of reference
states instead yields a multiplicative estimate, with no restriction
away from the common zero set.

\begin{proposition}[Multiplicative comparison at R\'enyi order two]
\label{prop:consolidated-quadratic-CMI-comparison}
With $\kappa_C$ defined in
\eqref{eq:consolidated-reference-condition-number}, for every
$t\in[0,1]$,
\begin{equation}
 \boxed{
 \kappa_C^{-1}I_{2,t}^{\ts}(A;C\mid B)
 \leq I_{2,t}^{\os}(A;C\mid B)
 \leq\kappa_C I_{2,t}^{\ts}(A;C\mid B).}
 \label{eq:consolidated-quadratic-CMI-comparison}
\end{equation}
Both sides are in fact independent of $t$.
\end{proposition}

\begin{proof}
For an invertible state $\omega_C$, we define
\begin{align}
 Q_{ABC}(\omega_C)
 &:=\operatorname{Tr}\!\left[
  \rho_{ABC}^2(\rho_{AB}^{-1}\otimes\omega_C^{-1})
 \right],\notag\\
 Q_{BC}(\omega_C)
 &:=\operatorname{Tr}\!\left[
  \rho_{BC}^2(\rho_B^{-1}\otimes\omega_C^{-1})
 \right],\qquad
 J(\omega_C):=Q_{ABC}(\omega_C)-Q_{BC}(\omega_C).
 \label{eq:consolidated-quadratic-reference-functionals}
\end{align}
The quadratic identity
\eqref{eq:consolidated-quadratic-independence} and its DPI imply
$J(\omega_C)\geq0$ for every invertible $\omega_C$.  Both
$Q_{BC}(\omega_C)$ and $J(\omega_C)$ are linear functionals of
$\omega_C^{-1}$.  More explicitly, we define
\begin{align}
 M_C&:=\operatorname{Tr}_{AB}\!\left[
  (\rho_{AB}^{-1/2}\otimes I_C)\rho_{ABC}^2
  (\rho_{AB}^{-1/2}\otimes I_C)\right],\notag\\
 L_C&:=\operatorname{Tr}_B\!\left[
  (\rho_B^{-1/2}\otimes I_C)\rho_{BC}^2
  (\rho_B^{-1/2}\otimes I_C)\right],
 \qquad K_C:=M_C-L_C.
\end{align}
Then $M_C,L_C\geq0$, and cyclicity of the trace gives
\begin{equation}
 J(\omega_C)=\operatorname{Tr}(K_C\omega_C^{-1}),\qquad
 Q_{BC}(\omega_C)=\operatorname{Tr}(L_C\omega_C^{-1}).
\end{equation}
The operator $K_C$ is Hermitian and is positive as well.  Indeed,
the quadratic DPI gives $\operatorname{Tr}(K_C\omega_C^{-1})\geq0$ for every
invertible state $\omega_C$.  Every $Y_C>0$ is, up to a positive
scalar, the inverse of such a state, so
$\operatorname{Tr}(K_CY_C)\geq0$ for every $Y_C>0$.  Self-duality of the
positive-semidefinite cone therefore implies $K_C\geq0$.

For this proof, we let $a_C$ and $b_C$ again denote the minimum and
maximum eigenvalues, respectively, of
$\rho_C^{-1/2}\tau_C\rho_C^{-1/2}$.  Thus
$\kappa_C=b_C/a_C$, and inverting the corresponding order bounds gives
\begin{equation}
 b_C^{-1}\rho_C^{-1}\leq\tau_C^{-1}
 \leq a_C^{-1}\rho_C^{-1}.
\end{equation}
Testing this inequality against $K_C$ and $L_C$ yields
\begin{align}
 b_C^{-1}J(\rho_C)&\leq J(\tau_C)
                       \leq a_C^{-1}J(\rho_C),\notag\\
 b_C^{-1}Q_{BC}(\rho_C)&\leq Q_{BC}(\tau_C)
                       \leq a_C^{-1}Q_{BC}(\rho_C).
 \label{eq:consolidated-quadratic-functional-comparison}
\end{align}
Consequently, for
$x_\omega:=J(\omega_C)/Q_{BC}(\omega_C)\geq0$,
\begin{equation}
 \kappa_C^{-1}x_{\rho_C}
 \leq x_{\tau_C}\leq\kappa_Cx_{\rho_C}.
\end{equation}
Since
\begin{equation}
 I_{2,t}^{\os}=\log(1+x_{\tau_C}),
 \qquad
 I_{2,t}^{\ts}=\log(1+x_{\rho_C}),
\end{equation}
the result follows from
\begin{equation}
 c\log(1+x)\leq\log(1+cx)\qquad(0<c\leq1),
 \qquad
 \log(1+Cx)\leq C\log(1+x)\qquad(C\geq1).
\end{equation}
The $t$-independence is again
\eqref{eq:consolidated-quadratic-independence}.
\end{proof}

\begin{remark}[Common marginal support at order two]
The proposition also holds when $\rho_C$ and $\tau_C$ have a common
proper support.  One uses the support-extended forward quantities,
takes every inverse on the relevant marginal support, and defines
$\kappa_C$ from the nonzero spectrum in that corner.
\end{remark}

The corresponding linear statement fails at R\'enyi order one, even
at fixed projective distance between the two reference states.

\begin{proposition}[No projective linear comparison at the BS endpoint]
\label{prop:consolidated-no-BS-projective-factor}
For every $\kappa>1$ there is no finite constant $C_\kappa$ such that
\begin{equation}
 \widehat I^{\os}(A;C\mid B)
 \leq C_\kappa\widehat I^{\ts}(A;C\mid B)
 \label{eq:consolidated-false-BS-factor}
\end{equation}
for all invertible states and invertible references satisfying
$\kappa_C=\kappa$.  In fact, there is a full-support cq family with
$\kappa_C=\kappa$ for which
\begin{equation}
 \frac{\widehat I^{\os}(A;C\mid B)}
      {\widehat I^{\ts}(A;C\mid B)}\longrightarrow+\infty.
 \label{eq:consolidated-BS-factor-divergence}
\end{equation}
\end{proposition}

\begin{proof}
It is enough to take $B$ one-dimensional, $A$ a classical bit, and
$C=\mathbb C^2$.  We fix $\kappa>1$ and put
\begin{equation}
 r:=\frac{\kappa-1}{\kappa+1}\in(0,1).
\end{equation}
For a parameter $s\in(0,1)$ to be chosen below and
$0<\varepsilon<(1+s)^{-1}$, we define
\begin{align}
 m_\varepsilon
 &:=\operatorname{diag}(1-\varepsilon,\varepsilon),\notag\\
 \eta_{\varepsilon,\pm}
 &:=\operatorname{diag}
 \bigl(1-(1\pm s)\varepsilon,(1\pm s)\varepsilon\bigr),\notag\\
 \rho_{AC}^{(\varepsilon)}
 &:=\frac12\sum_{u\in\{+,-\}}
 \ketbra{u}{u}_A\otimes\eta_{\varepsilon,u},
 \label{eq:consolidated-BS-factor-cq-family}
\end{align}
and we choose the reference
\begin{equation}
 \tau_C^{(\varepsilon)}
 :=\begin{pmatrix}
  1-\varepsilon&r\sqrt{\varepsilon(1-\varepsilon)}\\
  r\sqrt{\varepsilon(1-\varepsilon)}&\varepsilon
 \end{pmatrix}.
 \label{eq:consolidated-BS-factor-reference}
\end{equation}
All these states are invertible and $\rho_C^{(\varepsilon)}=m_\varepsilon$.
Moreover,
\begin{equation}
 m_\varepsilon^{-1/2}\tau_C^{(\varepsilon)}m_\varepsilon^{-1/2}
 =\begin{pmatrix}1&r\\r&1\end{pmatrix},
\end{equation}
so its extreme eigenvalues are $1-r$ and $1+r$.  Consequently,
$\kappa_C=(1+r)/(1-r)=\kappa$ for every $\varepsilon$.

The common-flag rule gives
\begin{align}
 \widehat I^{\os}_\varepsilon
 &=\frac12\sum_{u\in\{+,-\}}
   \widehat D(\eta_{\varepsilon,u}\Vert
                    \tau_C^{(\varepsilon)})
   -\widehat D(m_\varepsilon\Vert
                    \tau_C^{(\varepsilon)}),\notag\\
 \widehat I^{\ts}_\varepsilon
 &=\frac12\sum_{u\in\{+,-\}}
   \widehat D(\eta_{\varepsilon,u}\Vert m_\varepsilon).
 \label{eq:consolidated-BS-factor-flag-identities}
\end{align}
The pairs in the second line commute.  Expanding their classical
relative entropies therefore yields
\begin{equation}
 \widehat I^{\ts}_\varepsilon
 =\varepsilon h(s)+O(\varepsilon^2),
 \qquad
 h(s):=\frac{(1+s)\log(1+s)+(1-s)\log(1-s)}2>0.
 \label{eq:consolidated-BS-ts-asymptotic}
\end{equation}

We next calculate the limit of the first line of
\eqref{eq:consolidated-BS-factor-flag-identities}.  For $a>0$, we set
\begin{equation}
 d_a:=\sqrt{(1-a)^2+4ar^2}
\end{equation}
and
\begin{equation}
 F_r(a):=\frac12\log\frac{a}{1-r^2}
 +\frac{1-a}{2d_a}
  \log\frac{1+a+d_a}{1+a-d_a}.
 \label{eq:consolidated-BS-limit-function}
\end{equation}
Then
\begin{equation}
 \lim_{\varepsilon\downarrow0}
 \widehat D\!\left(
  \operatorname{diag}(1-a\varepsilon,a\varepsilon)
  \middle\Vert\tau_C^{(\varepsilon)}\right)=F_r(a).
 \label{eq:consolidated-BS-divergence-limit}
\end{equation}
Indeed, the sum and product of the two eigenvalues of
$(\tau_C^{(\varepsilon)})^{-1/2}
 \operatorname{diag}(1-a\varepsilon,a\varepsilon)
 (\tau_C^{(\varepsilon)})^{-1/2}$ converge, respectively, to
\begin{equation}
 \frac{1+a}{1-r^2}
 \quad\text{and}\quad
 \frac{a}{1-r^2}.
\end{equation}
Thus the eigenvalues converge to
$(1+a\pm d_a)/(2(1-r^2))$.  For a two-dimensional positive
operator $Y$ with eigenvalues $u\neq v$, the functional calculus
writes $Y\log Y=\alpha Y+\beta I$.  Using
$\operatorname{Tr}[\tau_C^{(\varepsilon)}Y]=1$ and
$\operatorname{Tr}\tau_C^{(\varepsilon)}=1$, followed by direct simplification,
gives \eqref{eq:consolidated-BS-divergence-limit}.

It follows that
\begin{equation}
 \lim_{\varepsilon\downarrow0}\widehat I^{\os}_\varepsilon
 =J_r(s):=\frac{F_r(1+s)+F_r(1-s)}2-F_r(1).
 \label{eq:consolidated-BS-os-asymptotic}
\end{equation}
Direct differentiation of \eqref{eq:consolidated-BS-limit-function}
at $a=1$ gives
\begin{equation}
 F_r''(1)
 =\frac{\log\!\left(\frac{1+r}{1-r}\right)-2r}{4r}>0.
 \label{eq:consolidated-BS-limit-curvature}
\end{equation}
The strict inequality follows because
$\log\kappa>2(\kappa-1)/(\kappa+1)=2r$ for $\kappa>1$.
Hence $J_r(s)=\frac12F_r''(1)s^2+O(s^4)>0$ after fixing $s>0$
sufficiently small.  With this choice of $s$,
\eqref{eq:consolidated-BS-ts-asymptotic} tends to zero whereas
\eqref{eq:consolidated-BS-os-asymptotic} has a strictly positive
limit.  This proves \eqref{eq:consolidated-BS-factor-divergence}.
\end{proof}

Thus \Cref{prop:consolidated-quadratic-CMI-comparison} is genuinely
quadratic: its proof relies on the fact that both the quadratic DPI
defect and its denominator are positive linear functionals of the
inverse reference state.  At the BS endpoint, the logarithm destroys
this linearity, and
\Cref{prop:consolidated-no-BS-projective-factor} shows that projective
control of the references does not control the ratio of the two
deficits.  The uniform replacement based only on $\kappa_C$ is the
additive comparison
\eqref{eq:consolidated-forward-additive-comparison}, or equivalently
the multiplicative comparison of exponentiated deficits in
\eqref{eq:consolidated-forward-exponential-comparison}.  A linear
comparison is recovered away from their common zero set through
\eqref{eq:consolidated-forward-away-zero-factor}.

For general order-one $t$-CMIs, including the BS endpoint, the
elementary argument only gives
\eqref{eq:consolidated-forward-additive-comparison}; a linear
comparison through the common zero set would require a uniform
second-order nondegeneracy estimate.  One nevertheless obtains a
pairwise power comparison when the relevant smallest eigenvalues are
bounded below by fixed positive constants.  More
precisely, anticipating the common-zero-set theorem in
\Cref{cor:consolidated-tQMC-hierarchy}, we fix the dimensions and
$\varepsilon>0$ and restrict to
\begin{equation}
 \rho_{ABC}\geq\varepsilon I_{ABC},\qquad
 \tau_C\geq\varepsilon I_C.
 \label{eq:consolidated-CMI-invertible-slab}
\end{equation}
For every pair $x,y\in\{\os,\ts,\rev\}$ there are constants
$C_{x,y}>0$ and $\theta_{x,y}>0$, depending only on $\varepsilon$,
the dimensions, and the ordered pair $(x,y)$, such that, uniformly
for $t\in[0,1]$,
\begin{equation}
 I_t^x(A;C\mid B)
 \leq C_{x,y}\bigl(I_t^y(A;C\mid B)\bigr)^{\theta_{x,y}}.
 \label{eq:consolidated-CMI-Lojasiewicz-comparison}
\end{equation}
Indeed, the three quantities are nonnegative real-analytic functions
of $(t,\rho_{ABC},\tau_C)$ on a neighborhood of $[0,1]$ times the
compact set in \eqref{eq:consolidated-CMI-invertible-slab}, and
\Cref{cor:consolidated-tQMC-hierarchy} identifies their zero sets.
We apply the compact two-function \L{}ojasiewicz inequality recalled in
\eqref{eq:consolidated-two-function-Lojasiewicz} with
$u=I_t^y$ and $v=I_t^x$.  It gives $I_t^y\geq c(I_t^x)^r$;
hence \eqref{eq:consolidated-CMI-Lojasiewicz-comparison} holds with
$\theta_{x,y}=1/r$ and $C_{x,y}=c^{-1/r}$.  This last bound
is qualitative and its exponent need not be one; the explicit
linear-factor statement is the quadratic result
\eqref{eq:consolidated-quadratic-CMI-comparison}.

\begin{corollary}[Tripartite quantitative DPI]
\label{cor:consolidated-tripartite-dpi}
For a state $\rho_{ABC}>0$, a normalized state $\tau_C>0$, and
$t\in[0,1]$, we set, for the reverse orientation,
\begin{equation}
 P:=\rho_{AB}\otimes\tau_C,\qquad S:=\rho_{ABC},\qquad
 P_0:=\rho_B\otimes\tau_C,\qquad S_0:=\rho_{BC},
\end{equation}
and we define
\begin{align}
 G_t&:=L_{P^{1-t}}R_{S^{-1/2}P^tS^{-1/2}},&
 G_{0,t}&:=L_{P_0^{1-t}}R_{S_0^{-1/2}P_0^tS_0^{-1/2}},\notag\\
 \mathsf V_{\rev}(Z)
 &:=\left[I_A\otimes(ZS_0^{-1/2})\right]S^{1/2},&
 \mathsf U_{\rev,t}
 &:=G_t^{1/2}\mathsf V_{\rev}G_{0,t}^{-1/2},\notag\\
 q_{\rev,t}
 &:=\lVert S^{1/2}-\mathsf U_{\rev,t}S_0^{1/2}\rVert_2,&
 M_{\rev,t}
 &:=\max\{\lVert G_t^{-1}\rVert_\infty,
                 \lVert G_{0,t}^{-1}\rVert_\infty\},
 \label{eq:consolidated-reverse-data}
\end{align}
where $L_X$ and $R_X$ denote left and right multiplication by $X$,
respectively.  Then
\begin{equation}
 I_t^{\rev}(A;C\mid B)
 \geq\left(\frac\pi4\right)^4
 M_{\rev,t}^{-2}q_{\rev,t}^4.
 \label{eq:consolidated-reverse-q-bound}
\end{equation}
For $0<\alpha<1$ and $1<\alpha<2$, the reverse R\'enyi deficit
satisfies \eqref{eq:consolidated-Renyi-log-bounds} with
$(\Gamma_t,\Gamma_t^\Phi,q_t^\Phi,M_t)$ replaced by
$(G_t,G_{0,t},q_{\rev,t},M_{\rev,t})$.
\end{corollary}

Exact equality and recovery follow from
\Cref{sec:consolidated-recovery}.  The strict-interior quantitative
reconstruction consequences, including the asymmetric cubic estimate,
are collected in
\Cref{sec:consolidated-quantitative-reconstruction}.  No two-sided
reconstruction bound is claimed.

\subsection{Tripartite recovery and the
\texorpdfstring{$t$-QMC}{t-QMC} classes}

\begin{mainobject}{Main definition: Invertible $t$-quantum Markov chains}
For $t\in[0,1]$, an invertible normalized state $\tau_C$, and an
orientation $\bullet\in\{\os,\ts,\rev\}$, an invertible tripartite
state $\rho_{ABC}$ is a $\bullet$-oriented $t$-quantum Markov chain
($t$-QMC) when its corresponding $t$-conditional mutual information
vanishes:
\begin{equation}
 I_t^\bullet(A;C\mid B)_\rho=0.
 \label{eq:consolidated-tqmc-definition}
\end{equation}
Its invertible zero-set class is denoted by
\begin{equation}
 \mathfrak Q_t^\bullet
 :=\left\{\rho_{ABC}\in
 \mathcal S_{+}(\mathcal H_{ABC}):
 I_t^\bullet(A;C\mid B)_\rho=0\right\}.
 \label{eq:consolidated-tqmc-zero-set-definition}
\end{equation}
The dependence on the fixed ancillary state $\tau_C$ is suppressed
from the notation.  \Cref{cor:consolidated-tQMC-hierarchy}
below shows that the three orientation classes coincide; thereafter
``$t$-QMC'' may be used without an orientation qualifier.
\end{mainobject}

For the three ordered pairs in
\eqref{eq:consolidated-three-CMIs}, vanishing is equivalent,
respectively, to
\begin{align}
 \mathcal L_t(\rho_{ABC}\mid\rho_{AB}\otimes\tau_C)
 &=I_A\otimes
 \mathcal L_t(\rho_{BC}\mid\rho_B\otimes\tau_C),
 \label{eq:consolidated-os-recovery}\\
 \mathcal L_t(\rho_{ABC}\mid\rho_{AB}\otimes\rho_C)
 &=I_A\otimes
 \mathcal L_t(\rho_{BC}\mid\rho_B\otimes\rho_C),
 \label{eq:consolidated-ts-recovery}\\
 \mathcal L_t(\rho_{AB}\otimes\tau_C\mid\rho_{ABC})
 &=I_A\otimes
 \mathcal L_t(\rho_B\otimes\tau_C\mid\rho_{BC}).
 \label{eq:consolidated-rev-recovery}
\end{align}
The three equations are not algebraically identical for
$0<t<1$: the arbitrary ancillary state cannot be cancelled, and
the reverse equation exchanges the ordered arguments.  Their zero
sets nevertheless coincide.

\begin{corollary}[Complete invertible $t$-QMC hierarchy]
\label{cor:consolidated-tQMC-hierarchy}
For the classes just defined,
\begin{equation}
 \mathfrak Q_t^{\os}
 =\mathfrak Q_t^{\ts}
 =\mathfrak Q_t^{\rev}
 =\{\text{invertible QMCs}\},
 \qquad 0\leq t<1.
 \label{eq:consolidated-three-zero-sets}
\end{equation}
At $t=1$, all three zero sets are the invertible BS-QMCs.  Hence the
only strict step in the invertible hierarchy is the jump from the
common nonmaximal QMC class to the BS endpoint.
\end{corollary}

\begin{proof}
For $t<1$, we apply
\Cref{thm:consolidated-interior-rigidity} to the three ordered
pairs.  Thus \eqref{eq:consolidated-three-zero-sets} is precisely the
tripartite specialization of
\eqref{eq:consolidated-general-hierarchy}.  The two forward Umegaki
deficits are the ordinary CMI.  For the reverse orientation,
\Cref{thm:consolidated-interior-rigidity} gives on the imaginary axis
\begin{equation*}
 (\rho_{AB}\otimes\tau_C)^{iu}\rho_{ABC}^{-iu}
 =I_A\otimes
 \bigl[(\rho_B\otimes\tau_C)^{iu}\rho_{BC}^{-iu}\bigr],
 \qquad u\in\mathbb R.
\end{equation*}
Taking adjoints, rather than passing an inverse through the adjoint
channel, yields
\begin{equation*}
 \rho_{ABC}^{iu}(\rho_{AB}\otimes\tau_C)^{-iu}
 =I_A\otimes
 \bigl[\rho_{BC}^{iu}(\rho_B\otimes\tau_C)^{-iu}\bigr],
 \qquad u\in\mathbb R.
\end{equation*}
This is the forward Petz cocycle criterion for the pair
$(\rho_{ABC},\rho_{AB}\otimes\tau_C)$ and hence gives the QMC
condition recalled in \eqref{rw:eq:qmc-equivalent-conditions}.  At
$t=1$, all three
likelihood conditions reduce to the BS equality condition.
\end{proof}

The coincidence of the three zero sets is therefore a qualitative
recovery statement; it does not imply a pointwise comparison of the
three conditional mutual informations, as
\Cref{prop:consolidated-CMI-incomparability,fig:consolidated-BS-CMI-scatter}
shows.

\begin{corollary}[The $f$- and R\'enyi Markov hierarchy]
\label{cor:consolidated-f-Renyi-tQMC-hierarchy}
For any of the three orientations
$\bullet\in\{\os,\ts,\rev\}$, we define
$I_{f,t}^\bullet$ as the DPI deficit obtained from
\eqref{eq:consolidated-three-CMIs} by replacing $D^t$ with
$D_f^t$, i.e. the $(t,f)$-CMI with orientation $\bullet$.  If the
representing measure of $f$ is equality determining in the sense of
\eqref{eq:consolidated-equality-determining-support} for each of the two
finite-dimensional modular spectral problems involved, then
\begin{equation}
 I_{f,t}^\bullet=0
 \quad\Longleftrightarrow\quad
 \rho_{ABC}\text{ is a quantum Markov chain},
 \qquad 0\leq t<1.
\end{equation}
The recovery channel is in every case the ordinary Petz channel for
the corresponding partial trace and reference state.  At $t=1$,
every non-affine operator-convex generator instead has the BS-QMC
zero set.

In particular, for the R\'enyi quantities introduced in
\Cref{subsec:consolidated-tripartite-tcmi},
\begin{equation}
 \{I_{\alpha,t}^\bullet=0\}
 =\{\text{QMCs}\},
 \qquad
 0\leq t<1,\quad
 \alpha\in(0,1)\cup(1,2),
\end{equation}
with the same conclusion at $\alpha=1$ by continuity.  The
quadratic order is exceptional:
\begin{equation}
 \{I_{2,t}^\bullet=0\}
 =\{\text{BS-QMCs}\},
 \qquad 0\leq t\leq1.
\end{equation}
\end{corollary}

\begin{proof}
We apply \Cref{thm:consolidated-interior-rigidity} to the ordered
pair defining the chosen partial-trace deficit.  The full-support
power generators cover the two open R\'enyi ranges.  At order two the
power functional is $t$-independent and has the maximal/BS equality
class.
\end{proof}

\begin{remark}[Strictness witnesses]
By \Cref{cor:consolidated-tQMC-hierarchy}, the invertible $t$-QMCs
coincide with the invertible QMCs for every $0\leq t<1$.
Consequently, every invertible BS-QMC which is not a QMC is
automatically a witness to the strict inclusion of each nonmaximal
$t$-QMC class in the BS endpoint class.  In particular, the explicit
finite-dimensional example of \cite[Example~4.1]{Bluhm:2025kai} and
the Werner-state construction of
\cite[Proposition~4.3]{Bluhm:2025kai} apply directly: both are
BS-QMCs but fail to be $t$-QMCs for every $t<1$.
\end{remark}

On each fixed-dimensional compact set with a uniform positive lower
eigenvalue bound and for
fixed $t<1$, the same zero-set theorem gives qualitative stability:
if a $t$-CMI tends to zero, the distance to the QMC set tends to
zero.  Compactness alone supplies no computable rate.  For the Petz
state-recovery error, \Cref{cor:consolidated-nonmaximal-Petz-holder}
gives a nonconstructive H\"older rate, while the explicit estimates
available from strengthened DPI are collected in
\Cref{sec:consolidated-quantitative-reconstruction}.  None of these
comparisons is uniform as $t\uparrow1$.

\subsection{The common modular origin of the endpoint exception}
\label{subsec:consolidated-weighted-symmetry}

There is a useful operator-theoretic counterpart to the equality-set
collapse in \Cref{prop:consolidated-general-hierarchy}.  For a state
$\omega>0$ and $s\in[0,1]$, we equip
$\mathcal B(\mathcal H)$ with the power-weighted inner product
\begin{equation}
 \langle X,Y\rangle_{\omega,s}
 :=\operatorname{Tr}\!\left(\omega^sX^*\omega^{1-s}Y\right)
 =\langle X,J_{\omega,s}Y\rangle_{\rm HS},
 \qquad
 J_{\omega,s}:=L_{\omega^{1-s}}R_{\omega^s}.
 \label{eq:consolidated-power-weighted-product}
\end{equation}
The cases $s=1$ and $s=\frac12$ are, respectively, the GNS and KMS
inner products.  A linear map $\mathcal K$ is called
$s$-\emph{symmetric} if it is self-adjoint with respect to
\eqref{eq:consolidated-power-weighted-product}.  This family of weighted
inner products and the corresponding adjoints are discussed in
\cite[Section~2.2]{carlen2017gradient}.

The rigidity statement below is known in the detailed-balance
literature.  The GNS case goes back to Alicki
\cite{alicki1976detailed}; its extension to every
$s\neq\frac12$ in the quantum-Markov-semigroup setting was proved by
Fagnola and Umanit\`a
\cite[Proposition~8.1]{fagnola2007generators}.  The
formulation for arbitrary $*$-preserving linear maps, together with
the propagation from one non-KMS power weight to every power weight,
is given in
\cite[Lemmas~2.5 and~2.8 and Theorem~2.9]{carlen2017gradient}.  We
include the short finite-dimensional argument in order to fix the
conventions and make its connection with the $t$-divergences
explicit.

\begin{proposition}[Rigidity of power-weighted symmetry]
\label{prop:consolidated-power-weighted-rigidity}
For a complex-linear, $*$-preserving map
$\mathcal K:\mathcal B(\mathcal H)\to\mathcal B(\mathcal H)$, we set
\begin{equation}
 \Delta_\omega:=L_\omega R_{\omega^{-1}}.
\end{equation}
For every $s\neq\frac12$, the following conditions are equivalent:
\begin{enumerate}
 \item $\mathcal K$ is $s$-symmetric;
 \item $\mathcal K$ is KMS-symmetric and commutes with
       $\Delta_\omega$;
 \item $\mathcal K$ is $r$-symmetric for every $r\in[0,1]$.
\end{enumerate}
Consequently, if $\mathfrak S_{\omega,s}$ denotes the class of
$*$-preserving $s$-symmetric maps, then
\begin{equation}
 \mathfrak S_{\omega,s}
 =\bigcap_{0\leq r\leq1}\mathfrak S_{\omega,r}
 \subseteq \mathfrak S_{\omega,1/2},
 \qquad s\neq\frac12,
 \label{eq:consolidated-weighted-symmetry-hierarchy}
\end{equation}
and the inclusion can be strict.  Thus KMS symmetry is the unique
potentially larger symmetry class within the power-weighted family.
\end{proposition}

\begin{proof}
We denote the Hilbert--Schmidt adjoint of $\mathcal K$ by
$\mathcal K^\dagger$.  Symmetry at $s$ is equivalent to
\begin{equation}
 J_{\omega,s}\mathcal K
 =\mathcal K^\dagger J_{\omega,s}.
 \label{eq:consolidated-weighted-adjoint-relation}
\end{equation}
Since $\mathcal K$ is $*$-preserving, applying this identity to
$Y^*$ and $X^*$, taking complex conjugates, and using cyclicity of
the trace shows that $s$-symmetry also implies $(1-s)$-symmetry.
Hence \eqref{eq:consolidated-weighted-adjoint-relation} holds both
with $s$ and with $1-s$.  Eliminating
$\mathcal K^\dagger$ from the two identities gives
\begin{equation}
 \left[\mathcal K,
 J_{\omega,s}^{-1}J_{\omega,1-s}\right]=0,
 \qquad
 J_{\omega,s}^{-1}J_{\omega,1-s}
 =\Delta_\omega^{\,2s-1}.
 \label{eq:consolidated-weighted-modular-power}
\end{equation}
If $s\neq\frac12$, the function
$\lambda\mapsto\lambda^{2s-1}$ is injective on $(0,\infty)$.
Therefore $\Delta_\omega^{2s-1}$ and $\Delta_\omega$ have the same
spectral projections, and the first identity in
\eqref{eq:consolidated-weighted-modular-power} implies
$[\mathcal K,\Delta_\omega]=0$.  Finally,
\begin{equation}
 J_{\omega,r}=J_{\omega,s}\Delta_\omega^{\,s-r},
 \qquad 0\leq r\leq1,
 \label{eq:consolidated-weighted-propagation}
\end{equation}
so modular covariance propagates
\eqref{eq:consolidated-weighted-adjoint-relation} from $s$ to every
$r$.  This proves the equivalences.  When $s=\frac12$, the modular
power in \eqref{eq:consolidated-weighted-modular-power} is merely the
identity and yields no modular covariance.  Completely positive
KMS-symmetric maps which are not modular covariant show that the
inclusion in \eqref{eq:consolidated-weighted-symmetry-hierarchy} can
indeed be strict
\cite[Appendix~B]{carlen2017gradient}.
\end{proof}

\begin{remark}[Singular weights]
If $\omega\geq0$ is singular, the proposition remains valid on the
corner $P\mathcal B(\mathcal H)P$, where
$P:=\operatorname{supp}\omega$ and $\omega|_{P\mathcal H}$ is
invertible.  On the full algebra the displayed sesquilinear form is
only semidefinite and therefore is not an inner product.
\end{remark}

The fixed-point formulation in
\Cref{cor:consolidated-interior-recovery-dictionary} makes the
connection concrete.  We put $\tau:=\Phi(\sigma)$; the conditioned
channel $\Phi_\sigma$ is defined in
\eqref{eq:consolidated-conditioned-channel}.  A direct calculation
gives
\begin{equation}
 \begin{aligned}
 \langle X,\Phi^*(Y)\rangle_{\sigma,1/2}
 &=\operatorname{Tr}\!\left[
   \Phi(\sigma^{1/2}X^*\sigma^{1/2})Y\right]\\
 &=\langle\Phi_\sigma(X),Y\rangle_{\tau,1/2}.
 \end{aligned}
 \label{eq:consolidated-conditioned-KMS-adjoint}
\end{equation}
Thus $\Phi_\sigma$ is the KMS adjoint of $\Phi^*$, and
$\Phi^*\circ\Phi_\sigma$ is KMS-symmetric.  The last condition in
\Cref{cor:consolidated-interior-recovery-dictionary} says precisely
that the whitened likelihood
$\sigma^{-1/2}\rho\sigma^{-1/2}$ is fixed by this KMS
reversibilization.  This is the fixed-point form of Petz sufficiency
\cite{Petz-SufficiencyChannels-1978} and
\cite[Theorem~3.18(ix)]{hiai2017different}.

The preceding observation can be tied directly to BS-QMCs, at two
equivalent levels.  The canonical QMC shadow is controlled by the
usual KMS reversibilization.  After transporting its predual through
the congruence which undoes the shadow, one obtains a
$\frac12$-symmetric map which fixes the original BS-QMC itself.
The two ingredients were established separately in previous work:
the canonical shadow and the exact BS recovery fixed points in
\cite[Theorem~3.3 and Corollary~3.9]{Bluhm:2025kai}, and the
KMS-adjoint formulation of Petz reversibility in
\cite{Petz-SufficiencyChannels-1978,petz1986quasi} and
\cite[Section~2.2]{carlen2017gradient}.  The explicit congruence and
transported KMS weight below combine these results.

\begin{proposition}[BS-QMCs as transported KMS fixed points]
\label{prop:consolidated-BS-KMS-fixed-points}
For $\rho_{ABC}>0$, we set $d_B:=\dim\mathcal H_B$ and define
\begin{equation}
 \eta_{ABC}:=\frac1{d_B}\rho_B^{-1/2}\rho_{ABC}\rho_B^{-1/2},
 \qquad
 \Sigma_\eta:=\eta_{AB}\otimes\pi_C,
 \qquad
 \mathcal E:=\operatorname{Tr}_A.
 \label{eq:consolidated-BS-KMS-shadow}
\end{equation}
We define
\begin{equation}
 \mathcal K_\eta:=\mathcal E^*\circ\mathcal E_{\Sigma_\eta},
 \qquad
 \mathcal E_{\Sigma_\eta}(X)
 =d_B\operatorname{Tr}_A\!\left[
  (\eta_{AB}^{1/2}\otimes I_C)X
  (\eta_{AB}^{1/2}\otimes I_C)
 \right].
 \label{eq:consolidated-BS-KMS-reversibilization}
\end{equation}
Then $\mathcal K_\eta$ is unital, completely positive, and
$\frac12$-symmetric with respect to $\Sigma_\eta$.  Its predual
$\mathcal K_\eta^\dagger$ is $\frac12$-symmetric with respect to the
normalized inverse weight
\begin{equation}
 \widehat\Sigma_\eta
 :=\frac{\Sigma_\eta^{-1}}{\operatorname{Tr}(\Sigma_\eta^{-1})}.
 \label{eq:consolidated-inverse-KMS-weight}
\end{equation}
Moreover,
\begin{equation}
 \boxed{
 \begin{aligned}
 \rho_{ABC}\text{ is a BS-QMC}
 &\Longleftrightarrow \eta_{ABC}\text{ is a QMC}\\
 &\Longleftrightarrow
 \mathcal K_\eta^\dagger(\eta_{ABC})=\eta_{ABC}\\
 &\Longleftrightarrow
 \Sigma_\eta^{-1/2}\eta_{ABC}\Sigma_\eta^{-1/2}
 \in\operatorname{Fix}(\mathcal K_\eta).
 \end{aligned}}
 \label{eq:consolidated-BS-KMS-equivalence}
\end{equation}
In fact, we define the invertible congruence, with identity embeddings
on the complementary systems understood, by
\begin{equation}
 \mathcal C_{\rho_B}(X):=d_B\rho_B^{1/2}X\rho_B^{1/2}.
\end{equation}
For the linear BS recovery map in
\eqref{eq:consolidated-BS-CP-map}, we set
\begin{equation}
 \mathcal M_\rho^{\rm BS}
 :=\Phi_{B\to AB}^{\rm BS}\circ\operatorname{Tr}_A.
 \label{eq:consolidated-BS-fixed-point-map}
\end{equation}
We also put
\begin{equation}
 \mathcal M_\eta(X)
 :=d_B\eta_{AB}^{1/2}(I_A\otimes\operatorname{Tr}_A[X])\eta_{AB}^{1/2},
 \qquad X\in\mathcal B(\mathcal H_{AB}).
 \label{eq:consolidated-shadow-fixed-point-map}
\end{equation}
One has the similarity relation
\begin{equation}
 \mathcal M_\rho^{\rm BS}
 =\mathcal C_{\rho_B}\circ\mathcal M_\eta
  \circ\mathcal C_{\rho_B}^{-1},
 \qquad
 \mathcal K_\eta^\dagger=\mathcal M_\eta\otimes\id_C.
 \label{eq:consolidated-BS-KMS-similarity}
\end{equation}
We define the invertible state
\begin{equation}
 A_{AB}:=(I_A\otimes\rho_B^{-1/2})\eta_{AB}^{-1/2}
          (I_A\otimes\rho_B^{-1/2}),
 \qquad
 \omega_{AB}^{\rm BS}:=\frac{A_{AB}^2}{\operatorname{Tr}(A_{AB}^2)}.
 \label{eq:consolidated-BS-KMS-weight}
\end{equation}
Then $\mathcal M_\rho^{\rm BS}$ is completely positive and
$\frac12$-symmetric with respect to $\omega_{AB}^{\rm BS}$, and
\begin{equation}
 \boxed{
 \rho_{ABC}\text{ is a BS-QMC}
 \quad\Longleftrightarrow\quad
 (\mathcal M_\rho^{\rm BS}\otimes\id_C)(\rho_{ABC})
 =\rho_{ABC}.}
 \label{eq:consolidated-direct-BS-KMS-fixed-point}
\end{equation}
Thus the BS recovery fixed-point map is itself a KMS-symmetric map
for the explicit, state-dependent weight
\eqref{eq:consolidated-BS-KMS-weight}.
\end{proposition}

\begin{proof}
The definition in \eqref{eq:consolidated-BS-KMS-shadow} gives
$\eta_B=I_B/d_B$.  Consequently
$\mathcal E(\Sigma_\eta)=\pi_B\otimes\pi_C$, and substituting this
into \eqref{eq:consolidated-conditioned-channel} gives the second
identity in \eqref{eq:consolidated-BS-KMS-reversibilization}.
Equation \eqref{eq:consolidated-conditioned-KMS-adjoint} then shows
that $\mathcal K_\eta$ is KMS-symmetric; complete positivity and
unitality follow directly from its definition.
The inverse-weight symmetry of its predual may also be seen directly:
up to the irrelevant normalization of the weight,
\begin{equation}
 \begin{aligned}
 \langle X,\mathcal K_\eta^\dagger(Y)
 \rangle_{\Sigma_\eta^{-1},1/2}
 &=d_Bd_C\operatorname{Tr}\!\left[
       (\operatorname{Tr}_A X)^*\operatorname{Tr}_A Y\right]\\
 &=\langle\mathcal K_\eta^\dagger(X),Y
 \rangle_{\Sigma_\eta^{-1},1/2}.
 \end{aligned}
 \label{eq:consolidated-predual-inverse-KMS-symmetry}
\end{equation}
Equivalently, taking $C$ one-dimensional shows that
$\mathcal M_\eta$ is $\frac12$-symmetric with respect to the
normalized inverse weight
$\eta_{AB}^{-1}/\operatorname{Tr}(\eta_{AB}^{-1})$.

The first equivalence in
\eqref{eq:consolidated-BS-KMS-equivalence} is the canonical BS-QMC/QMC
correspondence
\cite[Theorem~3.3]{Bluhm:2025kai}.  The Petz recovery map for the
pair $(\eta_{ABC},\Sigma_\eta)$ under $\mathcal E$ is
\begin{equation}
 \mathcal E_{\Sigma_\eta}^\dagger(Y)
 =d_B(\eta_{AB}^{1/2}\otimes I_C)Y
       (\eta_{AB}^{1/2}\otimes I_C).
 \label{eq:consolidated-shadow-Petz-map}
\end{equation}
In particular,
$\mathcal K_\eta^\dagger=\mathcal M_\eta\otimes\id_C$.
Hence the QMC condition is equivalent to
$\mathcal K_\eta^\dagger(\eta_{ABC})=\eta_{ABC}$.
The last equivalence is the Petz fixed-point condition in
\Cref{cor:consolidated-interior-recovery-dictionary}.

It remains to verify the direct relation between the two fixed-point
maps.  On $AB$ one has
\begin{equation}
 \mathcal C_{\rho_B}^{-1}(X)
 =\frac1{d_B}\rho_B^{-1/2}X\rho_B^{-1/2},
\end{equation}
so tensoring with $\id_C$ and using
\eqref{eq:consolidated-shadow-Petz-map} yields, for $X$ on $ABC$,
\begin{align}
 (\mathcal C_{\rho_B}\otimes\id_C)\mathcal K_\eta^\dagger
 (\mathcal C_{\rho_B}^{-1}\otimes\id_C)(X)
 &={d_B}\rho_B^{1/2}\eta_{AB}^{1/2}\rho_B^{-1/2}
       \operatorname{Tr}_A[X]\rho_B^{-1/2}\eta_{AB}^{1/2}\rho_B^{1/2}\notag\\
 &= (\Phi_{B\to AB}^{\rm BS}\otimes\id_C)(\operatorname{Tr}_A[X]),
 \end{align}
because
$d_B\eta_{AB}=\rho_B^{-1/2}\rho_{AB}\rho_B^{-1/2}$.
This proves \eqref{eq:consolidated-BS-KMS-similarity}; its
BS-QMC fixed-point interpretation is precisely
\cite[Corollary~3.9]{Bluhm:2025kai}, and gives
\eqref{eq:consolidated-direct-BS-KMS-fixed-point}.

It remains only to identify the transported KMS weight.  For
$S:=I_A\otimes\rho_B^{1/2}$ and operators $X,Y$ on $AB$, cyclicity
of the trace gives
\begin{equation}
 \begin{aligned}
 &\left\langle
 \mathcal C_{\rho_B}^{-1}(X),
 \mathcal C_{\rho_B}^{-1}(Y)
 \right\rangle_{\eta_{AB}^{-1},1/2}\\
 &\qquad
 =\frac1{d_B^2}\operatorname{Tr}\!\left[
  S^{-1}\eta_{AB}^{-1/2}S^{-1}X^*
  S^{-1}\eta_{AB}^{-1/2}S^{-1}Y
 \right]\\
 &\qquad
 =\frac{\operatorname{Tr}(A_{AB}^2)}{d_B^2}
   \langle X,Y\rangle_{\omega_{AB}^{\rm BS},1/2}.
 \end{aligned}
 \label{eq:consolidated-congruence-KMS-isometry}
\end{equation}
Thus $\mathcal C_{\rho_B}$ is, up to a common scalar, an isometry
between the two KMS geometries.  Conjugating the
$\eta_{AB}^{-1}$-symmetric map $\mathcal M_\eta$ by this congruence
and using
\eqref{eq:consolidated-BS-KMS-similarity} proves that
$\mathcal M_\rho^{\rm BS}$ is $\frac12$-symmetric with respect to
$\omega_{AB}^{\rm BS}$.
\end{proof}

Two qualifications are essential.  First, both
$\mathcal M_\rho^{\rm BS}$ and its KMS weight
$\omega_{AB}^{\rm BS}$ depend on the marginal $\rho_{AB}$; arbitrary
KMS-symmetric maps do not characterize BS-QMCs.  Second,
$\mathcal M_\rho^{\rm BS}$ is generally neither unital nor trace
preserving.  Thus this is a fixed-point characterization by a
completely positive KMS-symmetric endomorphism, not by a quantum
Markov channel.  Finally, the word ``symmetric'' in the sandwich-form
BS recovery maps of \Cref{sec:consolidated-recovery} originally
refers to the operator ordering that guarantees positivity, not to
KMS self-adjointness.  The latter is the additional consequence of
the canonical-shadow transport established above.

The analogy with the $t$-divergences becomes exact at the level of
the modular exponent.  We set
\begin{equation}
 s_t:=1-\frac t2,
 \qquad 2s_t-1=1-t.
 \label{eq:consolidated-t-weighted-parameter}
\end{equation}
For $t<1$, Steps~1 and~3 in the proof of
\Cref{thm:consolidated-interior-rigidity} sample the modular orbit at
the nonzero increments $n(1-t)$.  Just as in
\eqref{eq:consolidated-weighted-modular-power}, this nonzero power
has the full spectral resolution of the modular operator; it
therefore determines the complete cocycle and forces the common Petz
equality class.  At $t=1$, equivalently $s_t=\frac12$, the sampled
power is $\Delta_\rho^0=I$: the modular orbit collapses, KMS symmetry alone
does not recover the missing modular information, and the equality
class can enlarge to the BS class.  Hence the coincidence
$\mathfrak F_{f,t}=\mathfrak F_{\rm P}$ for fixed
equality-determining $f$ and every $t<1$, and the equivalence of all
non-KMS power-weighted symmetry notions, are not literally the same
statement---the former concerns state pairs saturating a channel
DPI, while the latter concerns adjoints of linear maps---but they are
two manifestations of the same modular-spectral rigidity principle:
a nonzero power of an invertible modular operator retains all of its
spectral projections, whereas its zeroth power cannot distinguish
between them.

\begin{remark}[Why the qualifier ``power-weighted'' is essential]
The more general modularly weighted inner products are discussed in
\cite[Section~2.2]{carlen2017gradient}.  For a positive modular weight
$g:(0,\infty)\to(0,\infty)$, one may instead set
$J_{\omega,g}:=R_\omega g(\Delta_\omega)$.  The preceding adjoint
argument then produces commutation only with
\begin{equation}
 h_g(\Delta_\omega),
 \qquad
 h_g(x):=\frac{xg(x^{-1})}{g(x)}.
\end{equation}
It forces full modular covariance when $h_g$ separates the actual
modular spectrum, but not in general.  The KMS product is therefore
the unique exception among the power weights, not among all possible
modularly weighted inner products.  Likewise, the divergence-side
statement requires the equality-determining hypothesis on $f$; the
quadratic and sparse exceptions described after
\Cref{prop:consolidated-general-hierarchy} remain.
\end{remark}

\subsection{Canonical BS shadows and nonmaximal compatibility}

The equality hierarchy explains exactly which parts of the BS-QMC
correspondence of \cite{Bluhm:2025kai} survive along the
interpolation.  We write
$d_B=\dim\mathcal H_B$ and define the canonical BS shadow
\begin{equation}
 \eta^{\rm BS}_{ABC}(\rho)
 :=\frac1{d_B}\rho_B^{-1/2}\rho_{ABC}\rho_B^{-1/2}.
 \label{eq:consolidated-BS-shadow}
\end{equation}

\begin{proposition}[Canonical QMC shadow]
\label{prop:consolidated-canonical-shadow}
If $\rho_{ABC}$ is an invertible $t$-QMC for any
$t\in[0,1]$, then:
\begin{enumerate}
\item $\eta_B^{\rm BS}=I_B/d_B$ and
$\eta_{ABC}^{\rm BS}$ is a QMC;
\item
\begin{equation}
 [\eta_{AB}^{\rm BS},\eta_{BC}^{\rm BS}]=0,\qquad
 \rho_{ABC}
 =d_B^2\rho_B^{1/2}
  \eta_{AB}^{\rm BS}\eta_{BC}^{\rm BS}\rho_B^{1/2};
 \label{eq:consolidated-shadow-product}
\end{equation}
\item there are a unitary
$U_B:\mathcal H_B\to\bigoplus_n
 (\mathcal H_{B_L^n}\otimes\mathcal H_{B_R^n})$, a probability
distribution $(p_n)_n$, and states
$\widetilde\eta_{AB_L^n}$ and
$\widetilde\eta_{B_R^nC}$ such that
\begin{equation}
 \rho_{ABC}
 =\rho_B^{1/2}U_B^*\left[
 \bigoplus_n d_Bp_n\,
 \widetilde\eta_{AB_L^n}\otimes
 \widetilde\eta_{B_R^nC}\right]U_B\rho_B^{1/2};
 \label{eq:consolidated-shadow-block}
\end{equation}
\item equivalently,
\begin{equation}
 \rho_{ABC}=Z_\rho^{-1}\rho_B^{1/2}
 e^{-H_{AB}-H_{BC}}\rho_B^{1/2},
 \qquad [H_{AB},H_{BC}]=0;
 \label{eq:consolidated-shadow-Gibbs}
\end{equation}
\item $\rho\mapsto\eta^{\rm BS}(\rho)$ is an idempotent nonlinear
retraction from the invertible BS-QMC class onto the invertible QMCs with
maximally mixed $B$-marginal.  Its restriction to every
$t$-QMC class is surjective onto that same normalized QMC set.
\end{enumerate}
\end{proposition}

This is an immediate consequence of
\Cref{cor:consolidated-tQMC-hierarchy} and the canonical
BS-QMC/QMC correspondence
\cite[Theorem~3.3 and Remarks~3.6--3.7]{Bluhm:2025kai}.  Indeed, for
$t<1$ an invertible $t$-QMC is a QMC, while at $t=1$ it is a
BS-QMC; the cited correspondence gives all the displayed shadow
properties, including the retraction and its fibres.

Inside the BS class, the additional nonmaximal condition admits several
equivalent forms.

\begin{corollary}[Nonmaximal compatibility inside the BS class]
\label{cor:consolidated-BS-interior-tests}
For an invertible BS-QMC $\rho_{ABC}$, we denote its canonical shadow
by $\eta^{\rm BS}$ and define $W_{AB}$ through the polar decomposition
\begin{equation}
 \rho_{AB}^{1/2}\rho_B^{-1/2}
 =d_B^{1/2}W_{AB}(\eta_{AB}^{\rm BS})^{1/2}.
 \label{eq:consolidated-shadow-polar-unitary}
\end{equation}
For every fixed $0\leq t<1$, the
following are equivalent:
\begin{enumerate}
\item $\rho_{ABC}$ is a $t$-QMC;
\item $\rho_{ABC}$ is a QMC;
\item for every $r\in\mathbb R$,
\begin{equation}
 [\rho_B^{ir}\eta_{AB}^{\rm BS}\rho_B^{-ir},
   \eta_{BC}^{\rm BS}]=0;
 \label{eq:consolidated-BS-interior-tests}
\end{equation}
\item $W_{AB}\eta_{BC}^{\rm BS}W_{AB}^*
       =\eta_{BC}^{\rm BS}$;
\item
\begin{equation}
 d_B^{-1}\rho_{AB}^{-1/2}\rho_{ABC}\rho_{AB}^{-1/2}
 =\eta_{BC}^{\rm BS}.
 \label{eq:consolidated-polar-test}
\end{equation}
\end{enumerate}
The corresponding structural and recovery descriptions appear in
\eqref{rw:eq:qmc-block-structure} and
\eqref{rw:eq:qmc-equivalent-conditions}.
\end{corollary}

This is a direct consequence of
\Cref{cor:consolidated-tQMC-hierarchy}: for $t<1$, being a $t$-QMC
is the same as being a QMC.  Within the BS class, the
remaining equivalences are exactly
\cite[Proposition~4.2]{Bluhm:2025kai}.

\subsection{The restricted inverse BS correspondence}

The inverse correspondence requires precisely the compatibility just
identified.
For the explicit formulas below we choose $\tau_C=\pi_C$.  The
invertible zero class is independent of this choice by
\Cref{cor:consolidated-tQMC-hierarchy}; moreover, replacing
$\pi_C$ by the unnormalized $I_C$ multiplies both sides of the
likelihood equation by the same scalar and therefore does not change
it.
We take a QMC $\eta_{ABC}>0$ with $\eta_B=I_B/d_B$ and a state
$X_B>0$, and we define
\begin{equation}
 \omega_{ABC}(\eta,X)
 :=d_BX_B^{1/2}\eta_{ABC}X_B^{1/2}
 =d_B^2X_B^{1/2}\eta_{AB}\eta_{BC}X_B^{1/2}.
 \label{eq:consolidated-inverse-BS}
\end{equation}
For positive $R,S$, we put, using the principal square root,
\begin{equation}
 F_t(R,S):=(R^tS^{-1})^{1/2}R^{(1-t)/2}.
\end{equation}

\begin{proposition}[Restricted inverse correspondence]
\label{prop:consolidated-restricted-inverse}
For $t\in[0,1]$, the assignments
\begin{equation}
 \rho\longmapsto(\eta^{\rm BS}(\rho),\rho_B),
 \qquad
 (\eta,X)\longmapsto\omega(\eta,X)
\end{equation}
are inverse bijections between invertible $t$-QMCs and the pairs
$(\eta,X)$ satisfying
\begin{multline}
 F_t\!\left(X_B^{1/2}\eta_{ABC}X_B^{1/2},
            X_B^{1/2}\eta_{AB}X_B^{1/2}\otimes I_C\right)
 =I_A\otimes
 F_t\!\left(X_B^{1/2}\eta_{BC}X_B^{1/2},
            X_B^{1/2}\eta_BX_B^{1/2}\otimes I_C\right).
 \label{eq:consolidated-restricted-constraint}
\end{multline}
Equivalently, $\omega(\eta,X)$ obeys the one-sided likelihood
condition \eqref{eq:consolidated-os-recovery} with
$\tau_C=\pi_C$.  At $t=1$, every
pair is allowed.  For $t<1$, the allowed pairs are exactly those
for which $\omega(\eta,X)$ is a QMC; hence the inverse
fibres are constant on $[0,1)$ and jump at $t=1$.
\end{proposition}

\begin{proof}
The unrestricted assignments are the mutually inverse BS
correspondence of
\cite[Theorem~3.3 and Remark~3.6]{Bluhm:2025kai}; in particular,
$\omega_B=X_B$ and the canonical shadow of $\omega$ is $\eta$.
If $\mathcal C_X(Z):=X_B^{1/2}ZX_B^{1/2}$, then
\begin{equation}
 \begin{aligned}
 \omega_{ABC}&=d_B\mathcal C_X(\eta_{ABC}),&
 \omega_{AB}\otimes I_C
   &=d_B\bigl(\mathcal C_X(\eta_{AB})\otimes I_C\bigr),\\
 \omega_{BC}&=d_B\mathcal C_X(\eta_{BC}),&
 \omega_B\otimes I_C
   &=d_B\bigl(\mathcal C_X(\eta_B)\otimes I_C\bigr),
 \end{aligned}
\end{equation}
where $d_B\mathcal C_X(\eta_B)=X_B$.
The homogeneity $F_t(cR,cS)=F_t(R,S)$ therefore turns the
one-sided $t$-QMC equation exactly into
\eqref{eq:consolidated-restricted-constraint}.  The likelihood form
is equivalent by \eqref{eq:consolidated-os-recovery}.  At $t=1$,
the unrestricted BS correspondence applies.  For every $t<1$,
\Cref{cor:consolidated-tQMC-hierarchy} identifies the
additional equation with QMC compatibility, independently
of $t$.
\end{proof}

Every invertible $t$-QMC also obeys the BS completely positive
reconstruction defined in \eqref{eq:consolidated-BS-CP-map}:
\begin{equation}
 (\Phi_{B\to AB}^{\rm BS}\otimes\id_C)(\rho_{BC})
 =\rho_{ABC}.
 \label{eq:consolidated-BS-CP-recovery}
\end{equation}
This is \cite[Corollary~3.9]{Bluhm:2025kai}.
This map need not be trace preserving, although
\begin{equation}
 \operatorname{Tr}\Phi_{B\to AB}^{\rm BS}(Z)
 \leq d_A\lVert\rho_B^{-1}\rVert_\infty\operatorname{Tr} Z,\qquad Z\geq0.
\end{equation}
The estimate appears in the discussion immediately preceding
\cite[Corollary~3.9]{Bluhm:2025kai}.
For $t<1$, the CP identity is necessary but not sufficient: it
characterizes the larger BS-QMC class, and
\eqref{eq:consolidated-restricted-constraint} is the missing
condition.

\subsection{Fractional shadows and explicit reconstruction}

The transformation below conjugates the state by an inverse
fractional power of its $B$-marginal and then renormalizes it.  In
this sense it progressively \emph{whitens}, or flattens, the
$B$-marginal: at $r=0$ it leaves the state unchanged, while at $r=1$
the transformed marginal is $I_B/d_B$.  For every $r<1$, however,
the transformed marginal is proportional to $\rho_B^{1-r}$, so
$\rho_B$ and hence the original state can still be recovered.  This
invertibility is lost only at the fully whitened endpoint $r=1$.

\begin{proposition}[Fractional shadows preserve and parametrize QMCs]
\label{prop:consolidated-fractional-shadow}
For $0\leq r<1$, we define
\begin{equation}
 \eta_{ABC}(r)
 :=\lambda(r)^{-1}
 \rho_B^{-r/2}\rho_{ABC}\rho_B^{-r/2},
 \qquad
 \lambda(r):=\operatorname{Tr}\rho_B^{1-r}.
 \label{eq:consolidated-fractional-shadow}
\end{equation}
Then
\begin{equation}
 \eta(r)\text{ is a QMC}
 \quad\Longleftrightarrow\quad
 \rho\text{ is a QMC}.
 \label{eq:consolidated-fractional-shadow-equivalence}
\end{equation}
Moreover, $\rho\mapsto\eta(r)$ is a bijection of the invertible QMC
set onto itself.  Given $\eta$, its inverse is
\begin{equation}
 q=\frac1{1-r},\qquad
 X_B=\frac{\eta_B^q}{\operatorname{Tr}\eta_B^q},\qquad
 \lambda=(\operatorname{Tr}\eta_B^q)^{-1/q},\qquad
 \rho_{ABC}=\lambda X_B^{r/2}\eta_{ABC}X_B^{r/2}.
 \label{eq:consolidated-fractional-shadow-inverse}
\end{equation}
At $r=1$, in contrast, the shadow forgets $\rho_B$, fixes its
middle marginal to $I_B/d_B$, and parametrizes the nontrivial BS
fibre.
\end{proposition}

\begin{proof}
We first suppose that $\eta(r)$ is a QMC and write it in the block form
recalled in \eqref{rw:eq:qmc-block-structure}.  Its middle marginal
obeys
\begin{equation}
 \eta_B(r)=\lambda(r)^{-1}\rho_B^{1-r},\qquad
 \rho_B=(\lambda(r)\eta_B(r))^{1/(1-r)}.
\end{equation}
Functional calculus preserves every direct-sum tensor block, so
$\rho_B$ has the same left--right block factorization.  Congruencing
the Markov blocks of $\eta(r)$ by $\rho_B^{r/2}$, and absorbing
scalar normalizations into the block weights, gives a Markov
decomposition of $\rho$.  The converse follows by the same argument
with the congruence by $\rho_B^{-r/2}$.

For bijectivity, solving the marginal relation gives exactly
\eqref{eq:consolidated-fractional-shadow-inverse}; substitution shows
that the two assignments are inverse.  The same block argument proves
that both directions preserve the invertible QMC set.  At $r=1$,
$\eta_B(1)=I_B/d_B$, so the marginal can no longer be recovered.
\end{proof}

\begin{corollary}[Explicit reconstruction from a fractional shadow]
\label{cor:consolidated-fractional-shadow-reconstruction}
For $0\leq r<1$, we put
\begin{equation}
 X_{AB}^{(r)}:=
 \rho_B^{-r/2}\rho_{AB}\rho_B^{-r/2}.
\end{equation}
Then $\eta_{ABC}(r)$ is a QMC if and only if
\begin{equation}
 \rho_{ABC}
 =\bigl(\Psi_{B\to AB}^{(r)}\otimes\id_C\bigr)(\rho_{BC}),
 \label{eq:consolidated-fractional-shadow-reconstruction}
\end{equation}
where, with identities on $A$ suppressed,
\begin{align}
 \Psi_{B\to AB}^{(r)}(Z)
 :={}&\rho_B^{r/2}\bigl(X_{AB}^{(r)}\bigr)^{1/2}
 \rho_B^{-1/2}Z\rho_B^{-1/2}
 \bigl(X_{AB}^{(r)}\bigr)^{1/2}\rho_B^{r/2}.
 \label{eq:consolidated-fractional-shadow-map}
\end{align}
In particular, $\Psi_{B\to AB}^{(r)}$ is a completely positive
reconstruction map determined by $\rho_{AB}$ and $\rho_B$.
\end{corollary}

\begin{proof}
The Petz characterization of a QMC
\cite[Eqs.~(10)--(11)]{hayden-2004} gives
\begin{equation}
 \eta_{ABC}(r)=\eta_{AB}(r)^{1/2}\eta_B(r)^{-1/2}
 \eta_{BC}(r)\eta_B(r)^{-1/2}\eta_{AB}(r)^{1/2}.
\end{equation}
We substitute
\begin{equation}
 \eta_{AB}=\lambda^{-1}X_{AB}^{(r)},\quad
 \eta_{BC}=\lambda^{-1}\rho_B^{-r/2}\rho_{BC}\rho_B^{-r/2},
 \quad
 \eta_B=\lambda^{-1}\rho_B^{1-r},
\end{equation}
and conjugate by $\lambda\rho_B^{r/2}$.  The adjacent powers of
$\rho_B$ reduce to $\rho_B^{-1/2}$, yielding
\eqref{eq:consolidated-fractional-shadow-reconstruction}.  Reversing
the calculation proves the converse.
\end{proof}

\begin{remark}[Non-invertible QMCs]
For $0\leq r<1$, the proposition and corollary above extend to
arbitrary QMCs after restricting $B$ to
$P_B:=\operatorname{supp}\rho_B$ and interpreting negative powers as
support inverses.  Since
$\operatorname{supp}\eta_B(r)=P_B$, the inverse formula remains in
the same support stratum, and the block-decomposition argument is
unchanged.  This observation does not extend the endpoint BS-fibre
statement at $r=1$.
\end{remark}

The exact tripartite structure is now complete.  We finish the
recovery analysis by isolating the quantitative statements that are
specific to the strict interior of the interpolation.

\section{Strict-interior defect and reconstruction estimates}
\label{sec:consolidated-quantitative-reconstruction}

By \Cref{cor:consolidated-interior-recovery-dictionary}, for an
invertible pair and every $0\leq t<1$, equality is exactly Petz
sufficiency.  In
the partial-trace specialization, \Cref{cor:consolidated-tQMC-hierarchy}
therefore says that every invertible $t$-QMC is a QMC.  Consequently,
below the BS endpoint the interpolation produces no new fixed-point
algebra, state-space block decomposition, complement, or Markov
structure.  These are the standard Petz and QMC structures recalled
in \Cref{rw:subsec:qmc-bsqmc} and described for general channels in
\cite[Theorems~3.18--3.19]{hiai2017different}.  Applying the same
result to a complementary channel likewise gives its Petz equality
class, so there is no additional nonmaximal complementary-channel
structure.  The genuinely larger BS endpoint is treated in
\Cref{sec:consolidated-tqmc}, and its shadow descriptions come from
\cite{Bluhm:2025kai}.

We therefore do not repeat the standard fixed-point decomposition,
the QMC block form, their set-theoretic complements, or the known
$t=0$ Petz-recovery remainders.  Instead, throughout this section we
take $0<t<1$ and retain only estimates whose nonzero values depend on
$t$.  Although the exact zero set of each strict-interior DPI defect
is the Petz/QMC set, that fact alone does not determine the
quantitative size of the defect.

\subsection{General-channel comparison estimates}

We take $\rho,\sigma>0$ and a channel
$\Phi:\mathcal B(\mathcal H)\to\mathcal B(\mathcal K)$
with invertible output states, as in \Cref{sec:consolidated-dpi}.  The
first estimate controls the asymmetric BS residual from
\eqref{eq:consolidated-asymmetric-map}.  Its right-hand side is useful
away from the BS equality class, but it is not a distance to the
smaller Petz equality class.

\begin{proposition}[Cubic estimate for the asymmetric BS residual]
\label{prop:consolidated-asymmetric-cubic}
For $0<t<1$, we put
$H_t:=\lVert\Gamma_t\rVert_\infty^2+
\lVert\Gamma_t^\Phi\rVert_\infty^2$.
Then the $D^t$ data-processing defect satisfies
\begin{equation}
 \begin{aligned}
 \delta_t^\Phi
 &:=D^t(\rho\Vert\sigma)
   -D^t\!\left(\Phi(\rho)\Vert\Phi(\sigma)\right)
\geq
 \frac{\left\lVert
 [\rho-\mathcal R_{\sigma,\Phi}^{\rm asym}(\Phi(\rho))]
 \sigma^{-1/2}\right\rVert_2^3}
 {27H_t},
 \end{aligned}
 \label{eq:consolidated-asymmetric-cubic}
\end{equation}
where $\mathcal R_{\sigma,\Phi}^{\rm asym}$ is the asymmetric recovery
map introduced in \Cref{prop:consolidated-asymmetric-recovery} and
defined explicitly in \eqref{eq:consolidated-asymmetric-map}.
\end{proposition}

\begin{proof}
For the resolvent error in \eqref{eq:consolidated-wst}, we set
\begin{equation}
 r_t:=\Gamma_t\sigma^{1/2}
 -\mathsf V_{\sigma,\Phi}\Gamma_t^\Phi[\Phi(\sigma)]^{1/2}
 =\bigl[\rho-
 \mathcal R_{\sigma,\Phi}^{\rm asym}(\Phi(\rho))\bigr]
 \sigma^{-1/2}.
\end{equation}
The resolvent identity gives
$\lVert r_t+s^2w_{s,t}\rVert_2\leq H_t/s$.
Moreover, \Cref{thm:consolidated-resolvent-dpi} and
$\mathrm d\mu_{x\log x}(s)=\mathrm ds$ give
\begin{equation}
 \delta_t^\Phi\geq
 \int_0^\infty s^2\lVert w_{s,t}\rVert_2^2\,\mathrm ds.
\end{equation}
Integrating over $s\in[S,2S]$, choosing $S$ appropriately, and
weakening the optimized constant yields
$\lVert r_t\rVert_2^3\leq27H_t\delta_t^\Phi$.
\end{proof}

\begin{remark}[Non-invertible first state]
The cubic estimate remains valid when
$\operatorname{supp}\rho\leq\operatorname{supp}\sigma$.  We
compress to the reference supports, apply the proof to
$\rho_\varepsilon=(1-\varepsilon)\rho+\varepsilon\sigma$, and let
$\varepsilon\downarrow0$.  Because only shifted resolvents occur,
the limit gives \eqref{eq:consolidated-asymmetric-cubic} with support
inverses.
\end{remark}

The next construction gives a positive comparison operator.  It is
important that the construction depends on the ordered input pair
$(\rho,\sigma)$; hence it is not a recovery channel determined by
$\Phi$ and the reference state.  For $X,S>0$, we define the geodesic
amplitude
\begin{equation}
 \mathfrak a_t(X\mid S)
 :=X^{(1-t)/2}S^{1/2}
 \left(S^{-1/2}X^tS^{-1/2}\right)^{1/2}.
 \label{eq:consolidated-amplitude}
\end{equation}
Functional calculus gives
\begin{equation}
 \mathfrak a_t(X\mid S)=\Gamma_t(X,S)^{1/2}S^{1/2},
\end{equation}
where $\Gamma_t(X,S)$ denotes \eqref{eq:consolidated-modular-operators}
for the ordered pair $(X,S)$.  In particular, it satisfies
$\mathfrak a_t(X\mid S)\mathfrak a_t(X\mid S)^*=X$.
The pair-dependent positive comparison map is
\begin{equation}
 \mathcal R_{\rho,\sigma,\Phi}^{(t)}(X)
 :=[\mathsf U_t\mathfrak a_t(X\mid\Phi(\sigma))]
   [\mathsf U_t\mathfrak a_t(X\mid\Phi(\sigma))]^*.
 \label{eq:consolidated-amplitude-map}
\end{equation}
It is positively homogeneous, trace non-increasing, and generally
nonlinear on the positive-definite cone; it extends continuously to
the positive cone.  The amplitude identity and
\eqref{eq:consolidated-transpose-relations} give the first relation
below, while \eqref{eq:consolidated-qt} and the Schatten estimate
$\lVert AA^*-BB^*\rVert_1\leq
(\lVert A\rVert_2+\lVert B\rVert_2)\lVert A-B\rVert_2$ give the
second:
\begin{equation}
 \mathcal R_{\rho,\sigma,\Phi}^{(t)}(\Phi(\rho))=\rho,
 \qquad
 \left\lVert\sigma-
 \mathcal R_{\rho,\sigma,\Phi}^{(t)}(\Phi(\sigma))\right\rVert_1
 \leq2q_t^\Phi.
 \label{eq:consolidated-amplitude-properties}
\end{equation}
If
$F_{\rm r}(X,Y):=\lVert X^{1/2}Y^{1/2}\rVert_1$ denotes root
fidelity, the same comparison gives
\begin{equation}
 F_{\rm r}\!\left(\sigma,
 \mathcal R_{\rho,\sigma,\Phi}^{(t)}(\Phi(\sigma))\right)
 \geq(1-q_t^\Phi)_+.
 \label{eq:consolidated-amplitude-fidelity}
\end{equation}
Combining \eqref{eq:consolidated-amplitude-properties} with
\Cref{thm:consolidated-Dt-quartic} yields
\begin{equation}
 \delta_t^\Phi
 \geq\left(\frac\pi8\right)^4M_t^{-2}
 \left\lVert\sigma-
 \mathcal R_{\rho,\sigma,\Phi}^{(t)}(\Phi(\sigma))
 \right\rVert_1^4.
 \label{eq:consolidated-amplitude-trace-bound}
\end{equation}
Thus \eqref{eq:consolidated-amplitude-map} is only a device for
turning the intrinsic defect into a trace-norm estimate.  The exact
recovery channel throughout the strict interior remains the Petz map
$\mathcal P_{\sigma,\Phi}$, independently of $t$, by
\Cref{cor:consolidated-interior-recovery-dictionary}.

\subsection{Tripartite strict-interior estimates}

The estimates in this subsection are particular cases of the
strengthened data-processing inequalities proved in
\Cref{thm:consolidated-resolvent-dpi,thm:consolidated-Dt-quartic,thm:consolidated-Renyi-dpi},
obtained by specializing the channel to a partial trace and choosing
the reference operators that define the corresponding tripartite
quantity.  Their interpretation is inspired by the recoverability
programme initiated by Fawzi and Renner: a small conditional mutual
information controls the distance to a reconstructed state
\cite{fawzi2015approximate}.  Subsequent strengthened-DPI results
developed this principle through rotated, pinched, and universal
recovery maps
\cite{berta2015monotonicity,sutter2016strengthened,junge2018universal},
with quantitative Petz and BS developments in
\cite{carlen2018recovery,carlen2020recovery-map,bluhm2020strengthened}.
Here we apply the same guiding principle to the strict-interior
$t$-conditional mutual informations.  The bounds below are thus
tripartite corollaries of the general strengthened DPIs and do not
constitute an independent recovery theory.

We now specialize only the preceding non-endpoint estimates.  We fix
$0<t<1$ and retain from \eqref{eq:consolidated-reverse-data} the
states $P,S,P_0,S_0$, the modular operators $G_t,G_{0,t}$, and the
quantities $\mathsf U_{\rev,t}$, $q_{\rev,t}$, and $M_{\rev,t}$.

\paragraph{Reverse orientation.}
The adjective ``reverse'' refers to the order of the two arguments of
the divergence.  Indeed, the
one-sided deficit uses the ordered pairs $(S,P)$ and $(S_0,P_0)$,
whereas $I_t^{\rev}$ uses the reversed pairs $(P,S)$ and $(P_0,S_0)$.
Since $\operatorname{Tr}_A P=P_0$ and
$\operatorname{Tr}_A S=S_0$, the general positive comparison
construction applied in this reverse order maps $P_0$ exactly to $P$
and compares $S=\rho_{ABC}$ with the image of
 $S_0=\rho_{BC}$.  To display every ingredient explicitly, for
 $X>0$ on $BC$, we define
\begin{equation}
 \begin{aligned}
 \mathcal C_t^{\rev}(X)
 &:=(\rho_B\otimes\tau_C)^{-\frac{1-t}{2}}X^{\frac{1-t}{2}}
   \rho_{BC}^{\frac12}
   \left(\rho_{BC}^{-\frac12}X^t\rho_{BC}^{-\frac12}\right)^{\frac12}
   \left(
    \rho_{BC}^{-\frac12}(\rho_B\otimes\tau_C)^t
    \rho_{BC}^{-\frac12}
   \right)^{-\frac12}\rho_{BC}^{-\frac12},\\
 \mathcal R_t^{\rev}(X)
 &:=(\rho_{AB}\otimes\tau_C)^{\frac{1-t}{2}}
   \left[I_A\otimes\mathcal C_t^{\rev}(X)\right]
   (\rho_{AB}\otimes\tau_C)^t
   \left[I_A\otimes\mathcal C_t^{\rev}(X)\right]^*
   (\rho_{AB}\otimes\tau_C)^{\frac{1-t}{2}}.
 \end{aligned}
 \label{eq:consolidated-reverse-map}
\end{equation}
All inverse powers in \eqref{eq:consolidated-reverse-map} exist by
invertibility, and the formula extends continuously to $X\geq0$.
It is a marginal-dependent, positively homogeneous reconstruction on
the positive cone; it is trace non-increasing and generally
nonlinear, so it is not asserted to be a quantum channel.  The term
``reconstruction'' is justified by
$\mathcal R_t^{\rev}(\rho_B\otimes\tau_C)
=\rho_{AB}\otimes\tau_C$, while
$\mathcal R_t^{\rev}(\rho_{BC})$ is the positive candidate for
$\rho_{ABC}$ controlled below.  At $t=0$,
\eqref{eq:consolidated-reverse-map} reduces to the usual Petz map for
the reference $\rho_{AB}\otimes\tau_C$.  Moreover,
\begin{equation}
 \boxed{
 I_t^{\rev}(A;C\mid B)
 \geq\left(\frac\pi8\right)^4M_{\rev,t}^{-2}
 \left\lVert\rho_{ABC}
 -\mathcal R_t^{\rev}(\rho_{BC})\right\rVert_1^4.}
 \label{eq:consolidated-reverse-trace-bound}
\end{equation}
Moreover,
\begin{equation}
 q_{\rev,t}\geq\frac12
 \lVert\rho_{ABC}-\mathcal R_t^{\rev}(\rho_{BC})\rVert_1.
\end{equation}
Together with
\Cref{cor:consolidated-tripartite-dpi} and
\eqref{eq:consolidated-Renyi-log-bounds}, this places the
reconstruction error inside the logarithmic lower bounds to the power
$6-2\alpha$ for $0<\alpha<1$ and to the power $2\alpha+2$ for
$1<\alpha<2$.  These are the strict-interior $(t,\alpha)$-R\'enyi
consequences; the quadratic order $\alpha=2$ is $t$-independent and is
therefore excluded here.

\paragraph{One-sided orientation.}
The one-sided orientation has the trace-preserving linear map
\begin{equation}
 \mathcal R_{A\mid BC}^{\rm asym}(X)
 :=[I_A\otimes(XP_0^{-1})]P.
 \label{eq:consolidated-tripartite-asymmetric-map}
\end{equation}
It is generally neither positive nor Hermiticity preserving.  With
\begin{equation}
 \Gamma_{ABC,t}^{\os}:=\Gamma_t(S,P),\qquad
 \Gamma_{BC,t}^{\os}:=\Gamma_t(S_0,P_0),
\end{equation}
\Cref{prop:consolidated-asymmetric-cubic} gives
\begin{equation}
 \boxed{
 I_t^{\os}(A;C\mid B)
 \geq
 \frac{\left\lVert
 [\rho_{ABC}
 -\mathcal R_{A\mid BC}^{\rm asym}(\rho_{BC})]
 P^{-1/2}\right\rVert_2^3}
 {27\left(
 \lVert\Gamma_{ABC,t}^{\os}\rVert_\infty^2+
 \lVert\Gamma_{BC,t}^{\os}\rVert_\infty^2\right)}.}
 \label{eq:consolidated-one-sided-cubic}
\end{equation}
With the inverses taken on the supports of $P$ and $P_0$, the same
one-sided estimate holds for a non-invertible tripartite state by the
support extension in
\Cref{prop:consolidated-asymmetric-cubic}.
The reverse bound supplies a positive pair-dependent comparison,
whereas the one-sided bound supplies a simple linear but nonpositive
one.  Its asymmetric residual already vanishes on the larger BS
class, so \eqref{eq:consolidated-one-sided-cubic} gives no numerical
estimate of distance to a Petz-recovered state.  Neither construction
is an alternative recovery theory.  For every $0<t<1$, the exact
conditions \eqref{eq:consolidated-os-recovery}--%
\eqref{eq:consolidated-rev-recovery} are equivalent to Petz recovery
and hence, by \Cref{cor:consolidated-tQMC-hierarchy}, to the QMC
condition.

The only estimate in this paper whose right-hand side is directly a
distance to the Petz output throughout a nonmaximal interval
is the H\"older bound under positive lower eigenvalue assumptions in
\Cref{cor:consolidated-nonmaximal-Petz-holder}.  We do not repeat it
here: unlike the two explicit estimates above, it has the correct
Petz/QMC zero set, but its exponent and constant are nonconstructive.
At $t=0$ and $t=1$, the intrinsic quartic bound underlying these
estimates reduces to the known relative-entropy and BS remainders,
respectively, as shown in \cite[Corollary~5.1]{carlen2018recovery} and
\cite[Theorems~5.3 and~7.1]{bluhm2020strengthened}.

%% file: continuity_conditional_sections.tex
\section{Continuity of conditional quantities and geodesic divergences}
\label{sec:continuity-interpolated-conditionals}

Continuity behaves differently at the two ends of the geodesic.
The Umegaki conditional entropy
\cite[Theorem~1.1 and Proposition~2.1]{berta2026sharpconditional},
and the optimized Petz--R\'enyi
conditional entropies of orders $\alpha\in[1/2,1)$, admit sharp
continuity moduli without assuming a positive lower bound on the
smallest eigenvalues of the relevant states
\cite[Theorem~2.1]{cheng2026sharprenyi}.  By contrast, for every fixed
$t>0$, the logarithmic and R\'enyi down-arrow conditional quantities
studied below are discontinuous at some rank-deficient state, where the
smallest eigenvalue is zero.  The natural positive result in the latter
regime is therefore continuity
under fixed positive lower bounds on the smallest eigenvalues of the
relevant states and marginals.  We organize the
section in the order in which the results are used: we first prove
joint two-input estimates for general $(t,f)$-divergences and their
logarithmic and power specializations, then derive bounds for
conditional quantities, and finally treat the three $t$-conditional
mutual informations.

The organization and comparison of the estimates are inspired by the
almost locally affine (ALAFF) method of
\cite[Section~4, especially Theorem~4.6]{bluhm2023continuity}, which
develops the Alicki--Fannes--Winter continuity strategy, and by the
three complementary approaches to sandwiched R\'enyi quantities in
\cite[Sections~4--5]{bluhm2026unified}.  The functional-calculus
estimates for the geodesic families are proved directly below; whenever
an endpoint theorem is imported from the literature, we identify the
precise result at the point where it is used.

Throughout, logarithms are natural and
\begin{equation}
 T(\rho,\sigma):=\frac12\lVert\rho-\sigma\rVert_1.
 \label{eq:continuity-trace-distance}
\end{equation}
For $K>1$, we use the modulus
\begin{equation}
 g_K(u):=
 \begin{cases}
  h_2(u)+u\log(K-1),&0\leq u\leq1-K^{-1},\\
  \log K,&1-K^{-1}\leq u\leq1,
 \end{cases}
 \qquad g_1\equiv0,
 \label{eq:def-gK-continuity}
\end{equation}
where $h_2(u):=-u\log u-(1-u)\log(1-u)$ is the binary entropy.
For integral $K$, the function $g_K$ is the plateau extension of the
sharp Fannes--Audenaert entropy modulus
\cite[Theorem~1]{audenaert2007sharp}; the same real-parameter
parametrization is used in
\cite[Eq.~(4)]{berta2026sharpconditional}.

The definitions in \Cref{rw:sec:interpolated-divergences} extend to
positive, possibly nonnormalized arguments by lower-semicontinuous
regularization.  Inverses and negative powers are generalized
inverses on the relevant support.  For a state $\rho$ and $p,c>0$
this convention gives
\begin{equation}
 D^t(p\rho\Vert\sigma)
 =pD^t(\rho\Vert\sigma)+p\log p,
 \qquad
 Q_{\alpha,t}(p\rho\Vert\sigma)
 =p^\alpha Q_{\alpha,t}(\rho\Vert\sigma),
 \label{eq:conditional-first-input-scaling}
\end{equation}
and
\begin{equation}
 D^t(\rho\Vert c\sigma)=D^t(\rho\Vert\sigma)-\log c,
 \qquad
 D_\alpha^t(\rho\Vert c\sigma)
 =D_\alpha^t(\rho\Vert\sigma)-\log c.
 \label{eq:conditional-scaling-convention}
\end{equation}

\input{continuity_joint_divergences}

\input{continuity_conditional_quantities}
\input{continuity_cmi_bounds}

%% file: continuity_joint_divergences.tex
\subsection{Joint continuity of the
\texorpdfstring{$(t,f)$}{(t,f)}-divergences}
\label{subsec:joint-continuity-tf}

We begin by allowing both inputs of the divergence to vary.  Relative
entropies need not be continuous when the inputs are allowed to approach
singular positive operators, whose smallest eigenvalue is zero, as
illustrated by the Umegaki and BS analyses in
\cite[Sections~5.3 and~6.2]{bluhm2023continuity}.  The
natural general statement therefore concerns sets on which the
smallest eigenvalues of both inputs are bounded below by fixed
positive constants.  We first introduce the constants shared by the
joint and conditional estimates.  For $0<m,n\leq1$, we set
\begin{align}
 \mathfrak K_t(m,n)
 &:=n^{-1}\bigl((1-t)m^{-t}+tm^{t-1}\bigr),
 \label{eq:continuity-frak-Kt}\\
 C_t(m,n)
 &:=(1-t)\log m^{-1}
 +\log(n^{-1}m^{-t})+tn^{-1}m^{-1-t/2}.
 \label{eq:continuity-Ct}
\end{align}
For $\alpha>0$, $\alpha\neq1$, we define
\begin{align}
 q_\alpha(m,n)
 &:={}
 \begin{cases}
  \max\{m^\alpha,n^{1-\alpha}\},&0<\alpha<1,\\
  1,&\alpha>1,
 \end{cases}
 \label{eq:continuity-qalpha}\\
 \Lambda_\alpha(m,n)
 &:={}
 \begin{cases}
  \alpha m^{\alpha-1},&0<\alpha<1,\\
  \alpha n^{-\alpha}m^{-1},&\alpha>1,
 \end{cases}
 \label{eq:continuity-Lambda-alpha}\\
 A_{\alpha,t}(m,n)
 &:=\frac{\Lambda_\alpha(m,n)\mathfrak K_t(m,n)}
          {|\alpha-1|q_\alpha(m,n)}.
 \label{eq:continuity-A-alpha-t}
\end{align}

We now prove the joint result for a general generator.  We then
exploit the additional structure of the logarithmic and power
generators to obtain the sharper estimates used for the conditional
quantities later in this section.

The separation of the two input perturbations and the use of uniform
spectral control follow the viewpoint of
\cite[Sections~4 and~5.3]{bluhm2023continuity}.  For the R\'enyi case,
we also follow the organizational idea of controlling the unnormalized
power functional before applying its logarithmic normalization, in
parallel with \cite[Example~4.3 and the proof of Theorem~4.4]{bluhm2026unified}.
The arguments below are direct because the almost-additive and
operator-space methods in those works do not apply verbatim to the
present Petz--geometric functional.

\input{continuity_general_tf}
\input{continuity_log_power}

%% file: continuity_general_tf.tex
\subsubsection{A general operator-Lipschitz bound}
\label{subsec:joint-continuity-general-tf}

We fix $0<m,n\leq1$ and write
\begin{equation}
 J_{m,n}:=[m,n^{-1}].
 \label{eq:general-tf-spectral-interval}
\end{equation}
For a real function $f$ on $J_{m,n}$, we define on the
Hilbert--Schmidt space at hand
\begin{align}
 \mathsf L_f(m,n)
 &:={}\sup_{\substack{m\1_{\rm HS}\leq G,H\leq
                                n^{-1}\1_{\rm HS}\\
             G\neq H}}
 \frac{\lVert f(G)-f(H)\rVert_\infty}
      {\lVert G-H\rVert_\infty},
 \label{eq:operator-Lipschitz-constant-tf}\\
 \Omega_f(m,n)
 &:=\max_{x\in J_{m,n}}f(x)-\min_{x\in J_{m,n}}f(x).
 \label{eq:oscillation-constant-tf}
\end{align}
Thus $\mathsf L_f(m,n)$ is the operator-Lipschitz constant of $f$
on the relevant interval.  We assume it is finite.  This holds, in
particular, for the operator-convex generators used in this paper
when they are finite on a neighbourhood of $J_{m,n}$, as follows
from their integral representation
\cite[Theorem~V.4.6 and Eq.~(V.31), pp.~133--134]
{bhatia1997matrix}.

We consider invertible states $\rho_0,\rho_1,\sigma_0,\sigma_1$ satisfying
\begin{equation}
 \rho_0,\rho_1\geq m\1,
 \qquad
 \sigma_0,\sigma_1\geq n\1,
 \label{eq:general-tf-slab}
\end{equation}
and introduce
\begin{equation}
 \varepsilon:=T(\rho_0,\rho_1),
 \qquad
 \delta:=T(\sigma_0,\sigma_1).
 \label{eq:general-tf-distances}
\end{equation}

\begin{lemma}[Perturbing the cyclic vector and modular operator]
\label{lemma:general-tf-perturbations}
Under \eqref{eq:general-tf-slab}, the following estimates hold:
\begin{align}
 \lVert\sigma_0^{1/2}-\sigma_1^{1/2}\rVert_2
 &\leq\frac{\delta}{\sqrt{2n}},
 \label{eq:general-tf-vector-bound}\\
 m\1_{\rm HS}
 &\leq\Gamma_t(\rho_i,\sigma_i)
 \leq n^{-1}\1_{\rm HS},
 \qquad i\in\{0,1\},
 \label{eq:general-tf-Gamma-spectrum}\\
 \lVert\Gamma_t(\rho_0,\sigma_0)
       -\Gamma_t(\rho_1,\sigma_1)\rVert_{2\to2}
 &\leq
 \mathfrak K_t(m,n)\varepsilon+n^{-2}\delta.
 \label{eq:general-tf-Gamma-bound}
\end{align}
\end{lemma}

\begin{proof}
We set $X=\sigma_0^{1/2}-\sigma_1^{1/2}$.  The Sylvester equation
\begin{equation}
 \bigl(L_{\sigma_0^{1/2}}+R_{\sigma_1^{1/2}}\bigr)X
 =\sigma_0-\sigma_1
 \label{eq:general-tf-Sylvester-equation}
\end{equation}
and the bound
$L_{\sigma_0^{1/2}}+R_{\sigma_1^{1/2}}
 \geq2\sqrt n\,\1_{\rm HS}$ give
\begin{equation}
 \lVert X\rVert_2
 \leq\frac{1}{2\sqrt n}
       \lVert\sigma_0-\sigma_1\rVert_2.
 \label{eq:general-tf-Sylvester-bound}
\end{equation}
If $\sigma_0-\sigma_1=\Delta_+-\Delta_-$ is its Jordan
decomposition, then
$\operatorname{Tr}\Delta_+=\operatorname{Tr}\Delta_-=\delta$, so
\begin{equation}
 \lVert\sigma_0-\sigma_1\rVert_2^2
 \leq(\operatorname{Tr}\Delta_+)^2
     +(\operatorname{Tr}\Delta_-)^2
 =2\delta^2.
\end{equation}
This proves \eqref{eq:general-tf-vector-bound}.

For $P,Q\geq a\1$ and $0<u<1$, the standard fractional-power
representation \cite[Eq.~(1.39), pp.~21--22]{bhatia2009positive}
and the resolvent identity yield
\begin{equation}
 P^u-Q^u
 =\frac{\sin(\pi u)}{\pi}\int_0^\infty
 \lambda^u(P+\lambda)^{-1}(P-Q)(Q+\lambda)^{-1}
 \,\mathrm d\lambda
 \label{eq:general-tf-power-resolvent}
\end{equation}
and hence
\begin{equation}
 \lVert P^u-Q^u\rVert_\infty
 \leq ua^{u-1}\lVert P-Q\rVert_\infty.
 \label{eq:general-tf-power-Lipschitz}
\end{equation}
Likewise, the inverse-power representation
\cite[Exercise~1.5.10, p.~23]{bhatia2009positive} reads
\begin{equation}
 P^{-1/2}=\frac1\pi\int_0^\infty
 \lambda^{-1/2}(P+\lambda)^{-1}\,\mathrm d\lambda
\end{equation}
and gives, for $P,Q\geq n\1$,
\begin{equation}
 \lVert P^{-1/2}-Q^{-1/2}\rVert_\infty
 \leq\frac{1}{2n^{3/2}}\lVert P-Q\rVert_\infty.
 \label{eq:general-tf-inverse-square-root}
\end{equation}
Thus both Lipschitz estimates follow directly from the cited integral
representations.

We recall from \eqref{rw:eq:gamma-explicit} that
\begin{equation}
 \Gamma_t(\rho,\sigma)
 =L_{\rho^{1-t}}R_{B_t(\rho,\sigma)},
 \qquad
 B_t(\rho,\sigma):=\sigma^{-1/2}\rho^t\sigma^{-1/2}.
 \label{eq:general-tf-Gamma-explicit}
\end{equation}
Since every state is bounded above by the identity,
\begin{equation}
 m^t\1\leq B_t(\rho_i,\sigma_i)\leq n^{-1}\1,
 \qquad
 m^{1-t}\1\leq\rho_i^{1-t}\leq\1.
\end{equation}
The commuting left and right multiplication operators have spectral
products in $J_{m,n}$, proving
\eqref{eq:general-tf-Gamma-spectrum}.

We insert a mixed product in \eqref{eq:general-tf-Gamma-explicit}.
Using $\lVert L_AR_B\rVert_{2\to2}
\leq\lVert A\rVert_\infty\lVert B\rVert_\infty$ and
\eqref{eq:general-tf-power-Lipschitz} gives
\begin{equation}
 \begin{aligned}
 &\lVert\Gamma_t(\rho_0,\sigma_0)
       -\Gamma_t(\rho_1,\sigma_0)\rVert_{2\to2}\leq n^{-1}\bigl((1-t)m^{-t}+tm^{t-1}\bigr)
       \lVert\rho_0-\rho_1\rVert_\infty.
 \end{aligned}
 \label{eq:general-tf-first-input-Gamma}
\end{equation}
For the second input, we write $S_i=\sigma_i^{-1/2}$ and expand
\begin{equation}
 S_0\rho_1^tS_0-S_1\rho_1^tS_1
 =(S_0-S_1)\rho_1^tS_0
 +S_1\rho_1^t(S_0-S_1).
\end{equation}
Together with \eqref{eq:general-tf-inverse-square-root}, this yields
\begin{equation}
 \lVert\Gamma_t(\rho_1,\sigma_0)
       -\Gamma_t(\rho_1,\sigma_1)\rVert_{2\to2}
 \leq n^{-2}\lVert\sigma_0-\sigma_1\rVert_\infty.
 \label{eq:general-tf-second-input-Gamma}
\end{equation}
The Jordan decomposition of the difference of two states gives
$\lVert\rho_0-\rho_1\rVert_\infty\leq\varepsilon$ and
$\lVert\sigma_0-\sigma_1\rVert_\infty\leq\delta$.
Combining the last two estimates proves
\eqref{eq:general-tf-Gamma-bound}.
\end{proof}

This technical result immediately allows us to prove the following continuity bound.

\begin{theorem}[Joint continuity of $D_f^t$ under positive lower
eigenvalue bounds]
\label{theo:joint-slab-general-tf}
We take $t\in[0,1]$ and assume that $f$ is operator Lipschitz on
$J_{m,n}$.  Under \eqref{eq:general-tf-slab},
\begin{equation}
 \boxed{\begin{aligned}
 \left|D_f^t(\rho_0\Vert\sigma_0)
       -D_f^t(\rho_1\Vert\sigma_1)\right|\leq
 \mathsf L_f(m,n)\mathfrak K_t(m,n)\varepsilon
 +\left(\frac{\mathsf L_f(m,n)}{n^2}
       +\frac{\Omega_f(m,n)}{\sqrt{2n}}\right)\delta.
 \end{aligned}}
 \label{eq:joint-slab-general-tf-bound}
\end{equation}
\end{theorem}

\begin{proof}
We set
\begin{equation}
 c_f:=\frac12\left(
 \max_{x\in J_{m,n}}f(x)+\min_{x\in J_{m,n}}f(x)\right),
 \qquad g:=f-c_f.
\end{equation}
Since $\lVert\sigma_i^{1/2}\rVert_2=1$, subtracting $c_f$ does not
change the difference of the two divergences.  Moreover,
\begin{align}
 \lVert g(\Gamma_t(\rho_i,\sigma_i))\rVert_{2\to2}
 &\leq\frac12\Omega_f(m,n),\qquad i=0,1,\notag\\
 \lVert g(\Gamma_t(\rho_0,\sigma_0))
          -g(\Gamma_t(\rho_1,\sigma_1))\rVert_{2\to2}
 &\leq\mathsf L_f(m,n)
 \lVert\Gamma_t(\rho_0,\sigma_0)
       -\Gamma_t(\rho_1,\sigma_1)\rVert_{2\to2}.
\end{align}
Expanding the quadratic forms gives
\begin{align}
 &\left\langle\sigma_0^{1/2},
   g(\Gamma_t(\rho_0,\sigma_0))\sigma_0^{1/2}
  \right\rangle_{\rm HS}
 -\left\langle\sigma_1^{1/2},
   g(\Gamma_t(\rho_1,\sigma_1))\sigma_1^{1/2}
  \right\rangle_{\rm HS}\notag\\
 &\quad=
 \left\langle\sigma_0^{1/2}-\sigma_1^{1/2},
   g(\Gamma_t(\rho_0,\sigma_0))\sigma_0^{1/2}
 \right\rangle_{\rm HS}\notag\\
 &\qquad
 +\left\langle\sigma_1^{1/2},
   g(\Gamma_t(\rho_0,\sigma_0))
   (\sigma_0^{1/2}-\sigma_1^{1/2})
 \right\rangle_{\rm HS}\notag\\
 &\qquad
 +\left\langle\sigma_1^{1/2},
   \bigl(g(\Gamma_t(\rho_0,\sigma_0))
        -g(\Gamma_t(\rho_1,\sigma_1))\bigr)\sigma_1^{1/2}
 \right\rangle_{\rm HS}.
 \label{eq:general-tf-difference-expansion}
\end{align}
Cauchy--Schwarz and
\Cref{lemma:general-tf-perturbations} give
\eqref{eq:joint-slab-general-tf-bound}.
\end{proof}

\begin{remark}[Fixed-support compression]
The theorem also applies when the four inputs are non-invertible only
because of unused ambient directions: one compresses to a fixed
common support on which the four operators satisfy
\eqref{eq:general-tf-slab}, and computes $m,n$ and the functional
calculus in those corners.  It gives no uniform boundary estimate
when the supports vary or a positive eigenvalue tends to zero.
\end{remark}

\begin{remark}[Affine normalization]
\label{rem:general-tf-affine-normalization}
For every affine function $\ell(x)=ax+b$,
\begin{equation}
 D_{f-\ell}^t(\rho\Vert\sigma)
 =D_f^t(\rho\Vert\sigma)-(a+b),
\end{equation}
because \eqref{rw:eq:gamma-normalization} gives
$\langle\sigma^{1/2},\Gamma_t\sigma^{1/2}\rangle_{\rm HS}=1$.
Thus the difference on the left of
\eqref{eq:joint-slab-general-tf-bound} is unchanged under affine
normalization, and the right-hand side may be minimized over affine
representatives of $f$.
\end{remark}

%% file: continuity_log_power.tex
\subsubsection{Sharper logarithmic and power bounds}
\label{subsec:joint-continuity-log-power}

The general theorem does not use the special form of $x\log x$ or
$x^\alpha$.  We now exploit that structure.  Besides improving the
constants, this gives the R\'enyi estimates used in the conditional
and CMI results below.

We take invertible states $\rho,\rho'$ and invertible positive
contractions $\sigma,\sigma'$ of equal trace on the same space.  We assume
\begin{equation}
 \rho,\rho'\geq m\1,
 \qquad
 \sigma,\sigma'\geq n\1,
 \label{eq:joint-slab-assumptions}
\end{equation}
and set
\begin{equation}
 \varepsilon:=T(\rho,\rho'),
 \qquad
 \delta:=T(\sigma,\sigma'),
 \qquad
 \delta_\infty:=\lVert\sigma-\sigma'\rVert_\infty.
 \label{eq:joint-slab-perturbations}
\end{equation}
Since $\operatorname{Tr}\sigma=\operatorname{Tr}\sigma'$, the
Jordan decomposition of $\sigma-\sigma'$ gives
$\delta_\infty\leq\delta$.  We retain both parameters because the
second-input arguments below use $\delta_\infty$, whereas
$\delta$ gives a convenient trace-distance-only relaxation.  We set
\begin{equation}
 X_t(\rho,\sigma):=\rho^{t/2}\sigma^{-1}\rho^{t/2}.
 \label{eq:joint-slab-X}
\end{equation}

\begin{lemma}[Functional calculus under common lower eigenvalue bounds]
\label{lemma:joint-slab-functional-calculus}
Under \eqref{eq:joint-slab-assumptions}, for every $t\in[0,1]$,
\begin{align}
 m^t\1&\leq X_t(\rho,\sigma)\leq n^{-1}\1,
 \label{eq:joint-slab-X-spectrum}\\
 m\1_{\rm HS}&\leq\Gamma_t(\rho,\sigma)
                    \leq n^{-1}\1_{\rm HS},
 \label{eq:joint-slab-Gamma-spectrum}\\
 \lVert X_t(\rho,\sigma)-X_t(\rho',\sigma)\rVert_\infty
 &\leq tn^{-1}m^{t/2-1}
          \lVert\rho-\rho'\rVert_\infty,
 \label{eq:joint-slab-X-first-input}\\
 \lVert\Gamma_t(\rho,\sigma)
       -\Gamma_t(\rho',\sigma)\rVert_{2\to2}
 &\leq\mathfrak K_t(m,n)
          \lVert\rho-\rho'\rVert_\infty.
 \label{eq:joint-slab-Gamma-first-input}
\end{align}
Moreover,
\begin{equation}
 \left\langle\sigma^{1/2},
 \Gamma_t(\rho,\sigma)\sigma^{1/2}\right\rangle_{\rm HS}=1.
 \label{eq:joint-slab-first-moment}
\end{equation}
If $c_\infty:=1+\delta_\infty/n$, then
\begin{equation}
 c_\infty^{-1}\sigma\leq\sigma'\leq c_\infty\sigma,
 \qquad
 c_\infty^{-1}X_t(\rho,\sigma)
 \leq X_t(\rho,\sigma')
 \leq c_\infty X_t(\rho,\sigma).
 \label{eq:joint-slab-order-comparison}
\end{equation}
Finally, if $a\1\leq G,H\leq b\1$, then
\begin{align}
 \lVert\log G-\log H\rVert_\infty
 &\leq a^{-1}\lVert G-H\rVert_\infty,
 \label{eq:joint-slab-log-Lipschitz}\\
 \lVert G^\alpha-H^\alpha\rVert_\infty
 &\leq
 \begin{cases}
  \alpha a^{\alpha-1}\lVert G-H\rVert_\infty,
      &0<\alpha<1,\\
  \alpha b^\alpha a^{-1}\lVert G-H\rVert_\infty,
      &\alpha>1.
 \end{cases}
 \label{eq:joint-slab-power-Lipschitz}
\end{align}
\end{lemma}

\begin{proof}
The cited fractional-power resolvent estimates
\eqref{eq:general-tf-power-resolvent}--%
\eqref{eq:general-tf-power-Lipschitz} give the fractional-power
bounds.  The resolvent formula for the logarithm
\cite[Example~2.11, Eq.~(2.6)]{vershynina2019upper} gives
\eqref{eq:joint-slab-log-Lipschitz}; the Duhamel identity for the
matrix exponential
\cite[the identity preceding Eq.~(6.42), p.~225]
{bhatia2009positive}, applied to
$G^\alpha=\exp(\alpha\log G)$, gives the second branch of
\eqref{eq:joint-slab-power-Lipschitz}.

Since $\rho,\sigma\leq\1$,
\begin{equation}
 \rho^{t/2}\sigma^{-1}\rho^{t/2}
 \geq\rho^t\geq m^t\1,
 \qquad
 \rho^{t/2}\sigma^{-1}\rho^{t/2}\leq n^{-1}\1.
\end{equation}
The analogous bounds for
$B_t=\sigma^{-1/2}\rho^t\sigma^{-1/2}$, together with the commuting
left and right multiplication operators in
\eqref{rw:eq:gamma-explicit}, prove
\eqref{eq:joint-slab-Gamma-spectrum}.  Inserting a mixed product
proves \eqref{eq:joint-slab-X-first-input} and
\eqref{eq:joint-slab-Gamma-first-input}.
\eqref{eq:joint-slab-first-moment} is
\eqref{rw:eq:gamma-normalization}.

Finally,
$\sigma'\leq\sigma+\delta_\infty\1
\leq(1+\delta_\infty/n)\sigma$, and the same
argument with $\sigma,\sigma'$ interchanged gives the other
comparison.  Inversion and congruence prove
\eqref{eq:joint-slab-order-comparison}.
\end{proof}

For the $t$-relative entropies, we can prove the following result.

\begin{theorem}[Joint continuity of $D^t$ under positive lower
eigenvalue bounds]
\label{theo:joint-slab-Dt-continuity}
Under \eqref{eq:joint-slab-assumptions}, with
$\varepsilon,\delta,\delta_\infty$ defined in
\eqref{eq:joint-slab-perturbations}, for every $t\in[0,1]$,
\begin{equation}
 \left|D^t(\rho\Vert\sigma)-D^t(\rho'\Vert\sigma')\right|
 \leq C_t(m,n)\varepsilon
 +\log\left(1+\frac{\delta_\infty}{n}\right).
 \label{eq:joint-slab-Dt-bound}
\end{equation}
Since $\delta_\infty\leq\delta$, the right-hand side remains valid
with $\delta_\infty$ replaced by $\delta$.
\end{theorem}

\begin{proof}
We first fix $\sigma$ and vary the state.  If $\Delta$ is traceless
Hermitian and $K$ is Hermitian, then
\begin{equation}
 |\operatorname{Tr}(\Delta K)|
 \leq\frac12\lVert\Delta\rVert_1
 \bigl(\lambda_{\max}(K)-\lambda_{\min}(K)\bigr).
 \label{eq:joint-slab-traceless-expectation}
\end{equation}
Indeed, we subtract from $K$ the midpoint of its extreme eigenvalues and
apply trace-norm/operator-norm duality
\cite[Eqs.~(1.173)--(1.174)]{watrous2018theory}.  Using
\eqref{rw:eq:Dt-explicit}, we integrate along the segment $A_s$ from
$\rho$ to $\rho'$.  The Daleckii--Krein formula
\eqref{rw:eq:daleckii-krein} gives
\begin{equation}
 \frac{\mathrm d}{\mathrm ds}\operatorname{Tr}(A_s\log A_s)
 =\operatorname{Tr}\!\left[\dot A_s(\log A_s+\1)\right].
\end{equation}
The identity term drops against the traceless increment,
and the spectrum in $[m,1]$ gives
$(1-t)\varepsilon\log m^{-1}$.

We next split the variation of
$\operatorname{Tr}[\rho\log X_t(\rho,\sigma)]$ into the change of
the expectation and the change of the logarithm.  By
\eqref{eq:joint-slab-X-spectrum} and
\eqref{eq:joint-slab-traceless-expectation}, the first is at most
$\varepsilon\log(n^{-1}m^{-t})$.  By
\eqref{eq:joint-slab-X-first-input} and
\eqref{eq:joint-slab-log-Lipschitz}, the second is at most
$tn^{-1}m^{-1-t/2}\varepsilon$.  We used
$\lVert\rho-\rho'\rVert_\infty\leq\varepsilon$, which follows from
the Jordan decomposition of the difference of two states.  These
three contributions are $C_t(m,n)\varepsilon$.

To vary the second input, we use
\eqref{eq:joint-slab-order-comparison} and operator monotonicity of
the logarithm \cite[Section~5.3.7, pp.~158--159]
{bhatia2009positive}:
\begin{equation}
 -\log c_\infty\,\1
 \leq\log X_t(\rho',\sigma')
       -\log X_t(\rho',\sigma)
 \leq\log c_\infty\,\1.
\end{equation}
Taking the expectation in $\rho'$ gives the remaining term.
\end{proof}

\begin{theorem}[Joint continuity of $D_\alpha^t$ under positive lower
eigenvalue bounds]
\label{theo:joint-slab-Dalpha-continuity}
Under \eqref{eq:joint-slab-assumptions}, with
$\varepsilon,\delta,\delta_\infty$ defined in
\eqref{eq:joint-slab-perturbations}, assume in addition that $\sigma$
and $\sigma'$ are states.  For every $t\in[0,1]$ and $\alpha>0$,
$\alpha\neq1$,
\begin{equation}
 \begin{aligned}
 &\left|D_\alpha^t(\rho\Vert\sigma)
       -D_\alpha^t(\rho'\Vert\sigma')\right|\leq A_{\alpha,t}(m,n)\varepsilon
 +\begin{cases}
 \displaystyle\log\left(1+\frac{\delta_\infty}{n}\right),
 &0<\alpha<1\ \text{or}\ 1<\alpha\leq2,\\[2mm]
 \displaystyle n^{-\alpha-1}m^{-t}\delta_\infty,
 &\alpha>2.
 \end{cases}
 \end{aligned}
 \label{eq:joint-slab-Dalpha-bound}
\end{equation}
Since $\delta_\infty\leq\delta$, the same estimate holds with
$\delta_\infty$ replaced throughout by $\delta$.
\end{theorem}

\begin{proof}
We first keep $\sigma$ fixed.  By
\eqref{eq:joint-slab-Gamma-spectrum} and
\eqref{eq:joint-slab-first-moment}, the spectral measure of
$\Gamma_t(\rho,\sigma)$ in the unit vector $\sigma^{1/2}$ has support
in $[m,n^{-1}]$ and first moment one.  For $0<\alpha<1$,
\begin{equation}
 \lambda^\alpha\geq m^\alpha,
 \qquad
 \lambda^\alpha
 =\lambda\lambda^{\alpha-1}\geq n^{1-\alpha}\lambda,
\end{equation}
whereas for $\alpha>1$ Jensen's inequality gives a moment at least
one.  Thus
\begin{equation}
 Q_{\alpha,t}(\rho\Vert\sigma)\geq q_\alpha(m,n).
 \label{eq:joint-slab-Q-lower-bound}
\end{equation}
The power estimate \eqref{eq:joint-slab-power-Lipschitz} and
\eqref{eq:joint-slab-Gamma-first-input} imply
\begin{equation}
 \left|Q_{\alpha,t}(\rho\Vert\sigma)
       -Q_{\alpha,t}(\rho'\Vert\sigma)\right|
 \leq\Lambda_\alpha(m,n)\mathfrak K_t(m,n)\varepsilon.
 \label{eq:joint-slab-Q-first-input}
\end{equation}
We now use
$|\log x-\log y|\leq|x-y|/\min\{x,y\}$.  Separating the power
functional from its logarithmic normalization parallels the
sandwiched R\'enyi strategy of
\cite[Example~4.3 and the proof of Theorem~4.4]
{bluhm2026unified}, although the functional here is different.

For the second input, \eqref{rw:eq:power-functional-push-through}
gives
\begin{equation}
 Q_{\alpha,t}(\rho\Vert\sigma)
 =\operatorname{Tr}\!\left[
 \rho^{t+\alpha(1-t)}
 X_t(\rho,\sigma)^{\alpha-1}\right].
 \label{eq:joint-slab-Q-X-representation}
\end{equation}
By the L\"owner--Heinz inequality and order reversal under inversion
\cite[Theorem~1.5.9, p.~22]{bhatia2009positive},
$x^{\alpha-1}$ is operator monotone decreasing for $0<\alpha<1$
and operator monotone increasing for $1<\alpha\leq2$.  Therefore
\eqref{eq:joint-slab-order-comparison} gives
\begin{equation}
 c_\infty^{-|\alpha-1|}Q_{\alpha,t}(\rho\Vert\sigma)
 \leq Q_{\alpha,t}(\rho\Vert\sigma')
 \leq c_\infty^{|\alpha-1|}Q_{\alpha,t}(\rho\Vert\sigma),
 \label{eq:joint-slab-Q-order-comparison}
\end{equation}
which proves the logarithmic second-input term.

If $\alpha>2$, then $t+\alpha(1-t)\geq1$, so
$\lVert\rho^{t+\alpha(1-t)}\rVert_1\leq1$.  Applying
\eqref{eq:joint-slab-power-Lipschitz} to $x^{\alpha-1}$ on
$[m^t,n^{-1}]$, and using
\begin{equation}
 \lVert X_t(\rho,\sigma)-X_t(\rho,\sigma')\rVert_\infty
 \leq n^{-2}\delta_\infty,
\end{equation}
gives
\begin{equation}
 \left|Q_{\alpha,t}(\rho\Vert\sigma)
       -Q_{\alpha,t}(\rho\Vert\sigma')\right|
 \leq(\alpha-1)n^{-\alpha-1}m^{-t}\delta_\infty.
\end{equation}
We use \eqref{eq:joint-slab-Q-lower-bound}, which equals one in this
range, and divide by $\alpha-1$.
\end{proof}

\begin{remark}[Fixed-support compression]
The two preceding continuity theorems remain valid for states that
are singular only through a common unused ambient subspace.  One
compresses all four inputs to their fixed common active support and
uses the positive lower eigenvalues there.  This is not a uniform
extension to changing supports or to eigenvalues approaching zero.
\end{remark}

\begin{remark}[The order-one limit]
The coefficient $A_{\alpha,t}$ diverges as $\alpha\to1$ because an
absolute error in $Q_{\alpha,t}$ is divided by $|\alpha-1|$.
\Cref{theo:joint-slab-Dt-continuity} is therefore the genuine
order-one result, not the limit of
\Cref{theo:joint-slab-Dalpha-continuity}.  The latter remains a
continuity theorem for the algebraic R\'enyi expression when
$\alpha>2$, but no DPI claim is made in that range.
\end{remark}

\subsubsection{Umegaki and BS specializations}
\label{subsec:joint-continuity-endpoints}

For $f(x)=x\log x$, an explicit operator-Lipschitz constant in
\Cref{theo:joint-slab-general-tf} is
\begin{equation}
 \mathsf L_{x\log x}(m,n)
 \leq\max\{\log m^{-1},\log n^{-1}\}+\frac1{mn}.
 \label{eq:xlogx-operator-Lipschitz}
\end{equation}
Indeed,
\begin{equation}
 A\log A-B\log B=(A-B)\log A+B(\log A-\log B),
\end{equation}
and \eqref{eq:joint-slab-log-Lipschitz} applies on $J_{m,n}$.  The
other constant in the general theorem is
\begin{equation}
 \Omega_{x\log x}(m,n)
 =n^{-1}\log n^{-1}
 -\begin{cases}
  -e^{-1},&m\leq e^{-1},\\
  m\log m,&m>e^{-1}.
 \end{cases}
 \label{eq:xlogx-oscillation}
\end{equation}
The direct logarithmic theorem is sharper.  To compare below with
trace-distance bounds from the literature, we use
$\delta_\infty\leq\delta$.  The resulting endpoint estimates read
\begin{gather}
 \boxed{\begin{aligned}
 |D(\rho\Vert\sigma)-D(\rho'\Vert\sigma')|
 &\leq\log((mn)^{-1})\varepsilon
 +\log\left(1+\frac{\delta}{n}\right)
 \end{aligned}}
 \label{eq:our-joint-Umegaki-slab}\\
 \boxed{\begin{aligned}
 |\widehat D(\rho\Vert\sigma)
       -\widehat D(\rho'\Vert\sigma')|
 &\leq
 \left[\log((mn)^{-1})+n^{-1}m^{-3/2}\right]\varepsilon
 +\log\left(1+\frac{\delta}{n}\right)
 \end{aligned}}
 \label{eq:our-joint-BS-slab}
\end{gather}
Thus the same proof supplies two-input estimates for the Umegaki and
BS entropies and every $t$-relative entropy.

\paragraph{Comparison at the Umegaki endpoint.}
We assume that $\sigma,\sigma'$ are states,
$\delta=T(\sigma,\sigma')$, and
\begin{equation}
 \rho-\rho'=\varepsilon(\rho_+-\rho_-)
 \label{eq:Umegaki-normalized-Jordan}
\end{equation}
is the normalized Jordan decomposition.  We set
\begin{equation}
 M:=\exp\!\left[
 \max\{D_{\max}(\rho_+\Vert\sigma),
       D_{\max}(\rho_-\Vert\sigma')\}\right].
 \label{eq:Umegaki-Jordan-M}
\end{equation}
For the exact value $\varepsilon=T(\rho,\rho')$,
\cite[Corollary~1, Eq.~(69)]{audenaert2025fundamental} proves
\begin{equation}
 \begin{aligned}
 |D(\rho\Vert\sigma)-D(\rho'\Vert\sigma')|
 &\leq\varepsilon\log(M-1)+h_2(\varepsilon)+\log\left(1+\frac{\delta}{n}\right).
 \end{aligned}
 \label{eq:recent-two-input-Umegaki}
\end{equation}
If only $T(\rho,\rho')\leq\varepsilon_0$ is known, we take the monotone
envelope.  The state-dependent value of $M$ in
\eqref{eq:Umegaki-Jordan-M} depends on the Jordan states and is
therefore retained in that version of the result:
\begin{equation}
 \begin{aligned}
 |D(\rho\Vert\sigma)-D(\rho'\Vert\sigma')|
 &\leq
 \begin{cases}
  \varepsilon_0\log(M-1)+h_2(\varepsilon_0),
  &0\leq\varepsilon_0<1-M^{-1},\\
  \log M,&1-M^{-1}\leq\varepsilon_0\leq1
 \end{cases}\\
 &\quad+\log\left(1+\frac{\delta}{n}\right).
 \end{aligned}
 \label{eq:recent-two-input-Umegaki-envelope}
\end{equation}
Since any state $\omega$ satisfies
$D_{\max}(\omega\Vert\sigma_i)\leq\log n^{-1}$ when
$\sigma_i\geq n\1$, we can further replace $M$ by $n^{-1}$.
With $g_K$ from \eqref{eq:def-gK-continuity}, this gives a form whose
constants depend only on the prescribed positive lower eigenvalue
bound:
\begin{equation}
 |D(\rho\Vert\sigma)-D(\rho'\Vert\sigma')|
 \leq g_{n^{-1}}(\varepsilon_0)
 +\log\left(1+\frac{\delta}{n}\right).
 \label{eq:recent-two-input-Umegaki-slab}
\end{equation}
Unlike \eqref{eq:our-joint-Umegaki-slab}, this requires no positive
lower eigenvalue bound for $\rho$ or $\rho'$.  It improves the earlier
joint ALAFF estimate \cite[Theorem~5.13, Eq.~(22)]
{bluhm2023continuity}.  With the correction in
\cite{bluhm2024corrections}, that estimate reads
\begin{equation}
 \begin{aligned}
 |D(\rho\Vert\sigma)-D(\rho'\Vert\sigma')|
 &\leq
 \left(\varepsilon+\frac{3\delta}{1-n/2}\right)\log(2n^{-1})
 +(1+\varepsilon)h_2\!\left(\frac{\varepsilon}{1+\varepsilon}\right)\\
 &\quad
 +2\log\left(1+\frac{2n^{-1}\delta}{1-n/2+\delta}\right).
 \end{aligned}
 \label{eq:ALAFF-two-input-Umegaki-comparison}
\end{equation}
Among the Umegaki estimates derived in this paper, the direct
logarithmic estimate \eqref{eq:our-joint-Umegaki-slab} is the
strongest: it improves the specialization of
\Cref{theo:joint-slab-general-tf}, which does not exploit the
logarithmic structure.  Once the estimates from the literature are
included, however, there is no single strongest formula.  We denote the
right-hand sides of \eqref{eq:our-joint-Umegaki-slab},
\eqref{eq:recent-two-input-Umegaki-slab}, and
\eqref{eq:ALAFF-two-input-Umegaki-comparison} by
$\mathcal B_{\rm dir}$, $\mathcal B_{\rm Aud}$, and
$\mathcal B_{\rm Blu}$, respectively.  In the nonconstant branch of
$g_{n^{-1}}$, the first two satisfy the exact comparison
\begin{equation}
 \mathcal B_{\rm dir}(\varepsilon,\delta)
 -\mathcal B_{\rm Aud}(\varepsilon,\delta)
 =\varepsilon\log\frac{1}{m(1-n)}-h_2(\varepsilon).
 \label{eq:Umegaki-direct-Audenaert-difference}
\end{equation}
Thus the direct estimate is stronger sufficiently close to
$\varepsilon=0$, whereas the Audenaert estimate can be stronger for
larger perturbations.  Under the common lower bounds used here, the
strongest uniform consequence depending only on the prescribed lower
eigenvalue bounds that is recorded in this subsection is
therefore
\begin{equation}
 \min\left\{
   \mathcal B_{\rm dir}(\varepsilon,\delta),
   \mathcal B_{\rm Aud}(\varepsilon,\delta),
   \log n^{-1}
 \right\}.
 \label{eq:best-Umegaki-floor-bound}
\end{equation}
The last term is the common diameter bound: from
$\sigma,\sigma'\geq n\1$ one has
$0\leq D(\rho\Vert\sigma),D(\rho'\Vert\sigma')\leq\log n^{-1}$.
The state-dependent estimate
\eqref{eq:recent-two-input-Umegaki-envelope} may be sharper still;
it is not a uniform function of $m,n,\varepsilon,$ and $\delta$ and
is consequently not part of the comparison below.

\Cref{fig:umegaki-continuity-bounds-comparison} displays the three
uniform formulas on the admissible qubit slice
$m=n=0.1$ and $\varepsilon=\delta=u\in[0,0.8]$.  The endpoint $0.8$
is the largest possible trace distance between qubit states bounded
below by $0.1\1$.  The nonzero crossing $u_\star$ of the direct and
Audenaert curves is characterized by
\begin{equation}
 h_2(u_\star)=u_\star\log\frac{100}{9}.
 \label{eq:Umegaki-CB-crossover}
\end{equation}

\begin{figure}[t]
 \centering
 \includegraphics[width=0.99\linewidth]
 {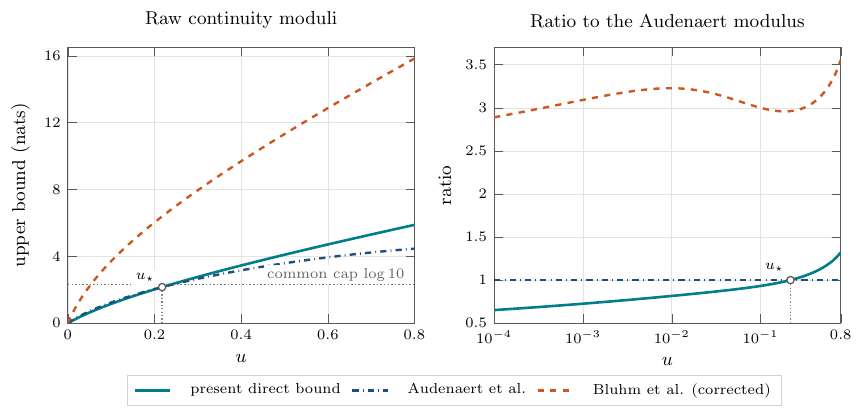}
 \caption{Comparison of joint Umegaki continuity bounds for invertible
 qubit states with
 $\rho,\rho',\sigma,\sigma'\geq0.1\1$ and
 $T(\rho,\rho')=T(\sigma,\sigma')=u$.
 The teal curve is the right-hand side of
 \eqref{eq:our-joint-Umegaki-slab}, the direct logarithmic estimate
 obtained here.  It is locally strongest;
 the estimate depending only on the lower eigenvalue bound from
 \cite[Corollary~1]{audenaert2025fundamental} becomes stronger after
 $u_\star$, while the corrected ALAFF estimate of
 \cite[Theorem~5.13]{bluhm2023continuity} and
 \cite{bluhm2024corrections} is larger throughout this slice.  The
 left panel shows the raw analytic moduli before applying the common
 cap $\log 10$, and the right panel shows ratios relative to the
 Audenaert modulus; a smaller ratio means a stronger bound.}
 \label{fig:umegaki-continuity-bounds-comparison}
\end{figure}

Related integral-representation bounds for Umegaki and Petz
quasi-relative entropies appear in
\cite[Theorems~3.1, 3.3, and~4.1]{vershynina2019upper}.

\paragraph{Comparison at the BS endpoint.}
We take states $\sigma,\sigma'$ with
$\sigma,\sigma'\geq n\1$.  Combining the fixed-reference ALAFF
bound of \cite[Corollary~6.13]{bluhm2023continuity} with the
second-input order comparison
\eqref{eq:joint-slab-order-comparison} gives the genuine two-input
estimate
\begin{equation}
 \begin{aligned}
 |\widehat D(\rho\Vert\sigma)
       -\widehat D(\rho'\Vert\sigma')|
 &\leq\varepsilon\log n^{-1}
 +(1+\varepsilon)n^{-1}
 h_2\!\left(\frac{\varepsilon}{1+\varepsilon}\right)
 +\log\left(1+\frac{\delta}{n}\right).
 \end{aligned}
 \label{eq:ALAFF-BS-two-input}
\end{equation}
Indeed, we insert $\widehat D(\rho'\Vert\sigma)$ between the two terms,
apply the cited result to the change $\rho\mapsto\rho'$, and use
operator monotonicity of the logarithm exactly as in the last part of
the proof of \Cref{theo:joint-slab-Dt-continuity} for the change
$\sigma\mapsto\sigma'$, followed by
$\delta_\infty\leq\delta$.  The cited paper itself proves only the
fixed-reference BS bound and explicitly omits the technically more
involved two-input analogue
\cite[Section~6.2.4]{bluhm2023continuity}; hence
\eqref{eq:ALAFF-BS-two-input} is the corresponding two-input extension
obtained here by this decomposition.  It requires no lower bound on
$\rho$ or $\rho'$, whereas \eqref{eq:our-joint-BS-slab} is Lipschitz
in $\varepsilon$ under the common lower bound $m\1$.

On this common domain, the direct and ALAFF-based estimates admit an
exact comparison.  We denote the right-hand sides of
\eqref{eq:our-joint-BS-slab} and \eqref{eq:ALAFF-BS-two-input} by
$\mathcal B_{\rm dir}^{\rm BS}$ and
$\mathcal B_{\rm ALAFF}^{\rm BS}$, respectively, and we set
\begin{equation}
 \begin{aligned}
 \Psi(\varepsilon)
 &:= (1+\varepsilon)
 h_2\!\left(\frac{\varepsilon}{1+\varepsilon}\right)
  =(1+\varepsilon)\log(1+\varepsilon)
    -\varepsilon\log\varepsilon,\\
 \kappa_{m,n}&:=m^{-3/2}+n\log m^{-1},
 \qquad
 r(\varepsilon):=\frac{\Psi(\varepsilon)}{\varepsilon}.
 \end{aligned}
 \label{eq:BS-comparison-functions}
\end{equation}
Here and below $\Psi(0):=0$.  The common second-input term cancels,
and direct calculation gives
\begin{equation}
 \mathcal B_{\rm dir}^{\rm BS}(\varepsilon,\delta)
 -\mathcal B_{\rm ALAFF}^{\rm BS}(\varepsilon,\delta)
 =\frac{\varepsilon}{n}\bigl[\kappa_{m,n}-r(\varepsilon)\bigr].
 \label{eq:BS-direct-ALAFF-difference}
\end{equation}
In particular, the ordering is independent of $\delta$.  Moreover,
\begin{equation}
 r'(\varepsilon)
 =-\frac{\log(1+\varepsilon)}{\varepsilon^2}<0,
 \qquad
 \lim_{\varepsilon\downarrow0}r(\varepsilon)=+\infty,
 \qquad
 r(1)=2\log2.
 \label{eq:BS-comparison-r-monotonicity}
\end{equation}
In every nontrivial $d$-dimensional state space, $m\leq d^{-1}$ and
hence $\kappa_{m,n}>2\log2$.  There is therefore a unique positive
value $\varepsilon_\star\in(0,1)$ at which the two analytic bounds
coincide, characterized by
\begin{equation}
 (1+\varepsilon_\star)\log(1+\varepsilon_\star)
 -\varepsilon_\star\log\varepsilon_\star
 =\kappa_{m,n}\varepsilon_\star.
 \label{eq:BS-CB-crossover}
\end{equation}
The direct estimate is smaller for
$0<\varepsilon<\varepsilon_\star$, whereas the ALAFF-based estimate
is smaller for $\varepsilon>\varepsilon_\star$.  This local advantage
of the direct bound reflects
$\Psi(\varepsilon)=\varepsilon(1-\log\varepsilon)+O(\varepsilon^2)$:
the direct modulus is linear, whereas the ALAFF-based first-input term
has order $\varepsilon\log(1/\varepsilon)$.
The positive lower
eigenvalue bound also restricts the possible range of $\varepsilon$:
writing
$\rho=m\1+(1-dm)\tau$ and
$\rho'=m\1+(1-dm)\tau'$ gives
$\varepsilon\leq1-dm$.  Thus $\varepsilon_\star<1-dm$ precisely when
$\kappa_{m,n}>r(1-dm)$.  Equality places
the crossover at the endpoint, whereas
$\kappa_{m,n}<r(1-dm)$ makes the direct estimate smaller throughout
the possible range.

Both estimates have the common diameter cap $\log n^{-1}$.  Indeed,
$\sigma\geq n\1$ implies
\begin{equation}
 \rho^{1/2}\sigma^{-1}\rho^{1/2}
 \leq n^{-1}\rho\leq n^{-1}\1,
\end{equation}
and consequently
$0\leq\widehat D(\rho\Vert\sigma)\leq\log n^{-1}$, with the same
statement for the primed pair.  The strongest uniform consequence of
the two formulas on their common domain is therefore
\begin{equation}
 \min\left\{
  \mathcal B_{\rm dir}^{\rm BS}(\varepsilon,\delta),
  \mathcal B_{\rm ALAFF}^{\rm BS}(\varepsilon,\delta),
  \log n^{-1}
 \right\}.
 \label{eq:best-BS-floor-bound}
\end{equation}

For the qubit choice $m=n=0.1$ used in
\Cref{fig:umegaki-continuity-bounds-comparison}, the crossover occurs
at approximately $4\times10^{-14}$.  Thus a linear-scale plot on that
slice would hide the interval where the direct estimate is stronger.
To display both regimes, \Cref{fig:BS-continuity-bounds-comparison}
uses the admissible qubit slice
$m=n=0.4$ and $\varepsilon=\delta=u\in[0,0.2]$.

\begin{figure}[t]
 \centering
 \includegraphics[width=0.99\linewidth]
 {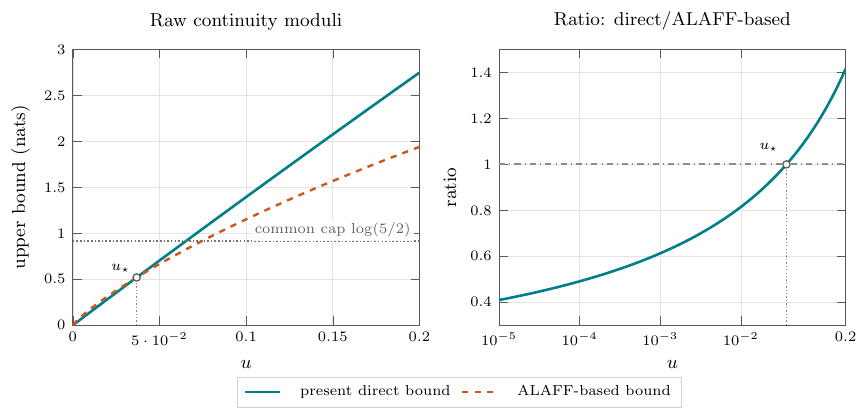}
 \caption{Comparison of joint BS continuity bounds for invertible
 qubit states with
 $\rho,\rho',\sigma,\sigma'\geq0.4\1$ and
 $T(\rho,\rho')=T(\sigma,\sigma')=u$.
 The teal curve is the right-hand side of
 \eqref{eq:our-joint-BS-slab}, the direct logarithmic estimate proved
 here.  The orange curve is the right-hand side of
 \eqref{eq:ALAFF-BS-two-input}, obtained here by combining the
 fixed-reference estimate of
 \cite[Corollary~6.13]{bluhm2023continuity} with the second-input
 order comparison.  The direct analytic modulus is smaller below
 $u_\star$, and the ALAFF-based analytic modulus is smaller above it.
 The left panel shows
 the uncapped analytic moduli together with their common diameter cap
 $\log(5/2)$; the right panel shows their ratio, with values below one
 favouring the direct bound.}
 \label{fig:BS-continuity-bounds-comparison}
\end{figure}

Finally, \Cref{theo:joint-slab-general-tf} covers arbitrary $t$ and
every admissible operator-Lipschitz generator.  The sandwiched R\'enyi
bounds of \cite[Sections~4--5]{bluhm2026unified} concern a different
family.

\FloatBarrier

%% file: continuity_conditional_quantities.tex
\subsection{Conditional geodesic entropies}
\label{subsec:conditional-geodesic-entropies}

\subsubsection{Definitions}

For a bipartite state $\omega_{AB}$, we define the own-marginal, or
down-arrow, conditional quantities
\begin{align}
 H_t^\downarrow(A|B)_\omega
 &:=-D^t\!\left(\omega_{AB}\middle\Vert
                 \1_A\otimes\omega_B\right),
 \label{eq:def-Ht-down-continuity}\\
 H_{\alpha,t}^\downarrow(A|B)_\omega
 &:=-D_\alpha^t\!\left(\omega_{AB}\middle\Vert
                 \1_A\otimes\omega_B\right),
 \qquad \alpha>0,\quad\alpha\neq1.
 \label{eq:def-Halpha-t-down-continuity}
\end{align}
Equivalently, with $\pi_A=\1_A/d_A$,
\begin{align}
 H_t^\downarrow(A|B)_\omega
 &=\log d_A-D^t(\omega_{AB}\Vert\pi_A\otimes\omega_B),
 \label{eq:Ht-down-normalized-reference}\\
 H_{\alpha,t}^\downarrow(A|B)_\omega
 &=\log d_A-D_\alpha^t(\omega_{AB}\Vert
                               \pi_A\otimes\omega_B).
 \label{eq:Halpha-t-down-normalized-reference}
\end{align}
At the Petz endpoint,
\begin{equation}
 H_0^\downarrow(A|B)=H(A|B),
 \qquad
 H_{\alpha,0}^\downarrow(A|B)_\omega
 =-\overline D_\alpha
   (\omega_{AB}\Vert\1_A\otimes\omega_B).
 \label{eq:conditional-Petz-endpoints}
\end{equation}
The second expression is the \emph{nonoptimized}, or down-arrow,
Petz--R\'enyi conditional entropy.  This down-arrow/up-arrow
terminology and the optimized construction follow the standard
conditional R\'enyi definitions
\cite[Section~II, Eqs.~(13)--(16)]{tomamichel2014relating}.
For comparison, we define the
optimized Petz quantity
\begin{equation}
 H_{\alpha,0}^\uparrow(A|B)_\omega
 :=-\inf_{\tau_B\in\mathcal S(B)}
 \overline D_\alpha
     (\omega_{AB}\Vert\1_A\otimes\tau_B).
 \label{eq:def-Halpha-zero-up}
\end{equation}
For $0<\alpha<1$, the quantum Sibson identity
\cite[Lemma~3 in the Supplemental Material]{sharma2013fundamental},
also stated in \cite[Lemma~1]{tomamichel2014relating}, gives
\begin{equation}
 H_{\alpha,0}^\uparrow(A|B)_\omega
 =\frac{\alpha}{1-\alpha}
 \log\operatorname{Tr}\!\left[
 \bigl(\operatorname{Tr}_A\omega_{AB}^{\alpha}\bigr)^{1/\alpha}
 \right].
 \label{eq:Sibson-optimized-Petz}
\end{equation}

\subsubsection{Continuity under positive lower eigenvalue bounds}
\label{subsec:conditional-spectral-slab-bounds}

We now apply the joint divergence bounds from
\Cref{subsec:joint-continuity-log-power}, retaining the separate
operator-norm parameter $\delta_\infty$ for the moving reference.
This distinction is what yields the sharper conditional constants
below.  The constants
$\mathfrak K_t$, $C_t$, $q_\alpha$, $\Lambda_\alpha$, and
$A_{\alpha,t}$ are defined in
\eqref{eq:continuity-frak-Kt}--\eqref{eq:continuity-A-alpha-t}.

\begin{theorem}[Conditional continuity under positive lower
eigenvalue bounds]
\label{theo:conditional-geodesic-slab-continuity}
We consider two invertible bipartite states, denoted by
$\omega_{AB}$ and $\zeta_{AB}$, satisfying
\begin{equation}
 \omega_{AB},\zeta_{AB}\geq m\1_{AB},
 \qquad
 \omega_B,\zeta_B\geq n\1_B,
 \label{eq:conditional-common-spectral-slab}
\end{equation}
and set
\begin{equation}
 \varepsilon:=T(\omega_{AB},\zeta_{AB}),
 \qquad
 \delta:=T(\omega_B,\zeta_B).
 \label{eq:conditional-two-distances}
\end{equation}
Then, for every $t\in[0,1]$,
\begin{equation}
 \left|H_t^\downarrow(A|B)_\omega
       -H_t^\downarrow(A|B)_\zeta\right|
 \leq C_t(m,n)\varepsilon
       +\log\left(1+\frac{\delta}{n}\right).
 \label{eq:conditional-Dt-slab-bound}
\end{equation}
For every $t\in[0,1]$, $\alpha>0$, $\alpha\neq1$, we set
$\widetilde n_A:=n/d_A$.  Then
\begin{equation}
 \begin{aligned}
 &\left|H_{\alpha,t}^\downarrow(A|B)_\omega
       -H_{\alpha,t}^\downarrow(A|B)_\zeta\right|\leq A_{\alpha,t}(m,\widetilde n_A)\varepsilon
 +\begin{cases}
 \displaystyle\log\left(1+\frac{\delta}{n}\right),
 &0<\alpha<1\ \text{or}\ 1<\alpha\leq2,\\[2mm]
 \displaystyle d_A^\alpha n^{-\alpha-1}m^{-t}\delta,
 &\alpha>2.
 \end{cases}
 \end{aligned}
 \label{eq:conditional-Dalpha-t-slab-bound}
\end{equation}
The estimates are explicit, but no sharpness is asserted.
\end{theorem}

\begin{proof}
Put
\begin{equation*}
 \delta_{\infty,B}:=\lVert\omega_B-\zeta_B\rVert_\infty\leq\delta.
\end{equation*}
The inequality follows from the Jordan decomposition of the
traceless Hermitian operator $\omega_B-\zeta_B$.  For the
unnormalized reference operators one has
\begin{equation*}
 \begin{aligned}
 T(\1_A\otimes\omega_B,\1_A\otimes\zeta_B)
 &=d_A\delta,\\
 \lVert\1_A\otimes(\omega_B-\zeta_B)\rVert_\infty
 &=\delta_{\infty,B}.
 \end{aligned}
\end{equation*}
Thus their trace distance carries an unnecessary factor $d_A$, but
the operator-norm perturbation parameter in
\Cref{theo:joint-slab-Dt-continuity} does not.  The two references
have common lower eigenvalue bound $n$, so that theorem gives
\begin{equation*}
 \begin{aligned}
 &\left|D^t(\omega_{AB}\Vert\1_A\otimes\omega_B)
       -D^t(\zeta_{AB}\Vert\1_A\otimes\zeta_B)\right|\\
 &\qquad\leq C_t(m,n)\varepsilon
       +\log\left(1+\frac{\delta_{\infty,B}}{n}\right)
 \leq C_t(m,n)\varepsilon
       +\log\left(1+\frac{\delta}{n}\right),
 \end{aligned}
\end{equation*}
which is \eqref{eq:conditional-Dt-slab-bound}.

For the R\'enyi quantity, we use
\eqref{eq:Halpha-t-down-normalized-reference}.  The normalized
reference $\pi_A\otimes\omega_B$ has floor $n/d_A$, while
\begin{equation}
 T(\pi_A\otimes\omega_B,\pi_A\otimes\zeta_B)=\delta,
 \qquad
 \lVert\pi_A\otimes(\omega_B-\zeta_B)\rVert_\infty
 =\frac{\delta_{\infty,B}}{d_A}.
\end{equation}
Since $\widetilde n_A=n/d_A$, the second-input terms in
\Cref{theo:joint-slab-Dalpha-continuity} become
\begin{equation*}
 \log\left(1+\frac{\delta_{\infty,B}/d_A}{n/d_A}\right)
 =\log\left(1+\frac{\delta_{\infty,B}}{n}\right)
 \leq\log\left(1+\frac{\delta}{n}\right)
\end{equation*}
for $0<\alpha<1$ and $1<\alpha\leq2$, and
\begin{equation*}
 \left(\frac{n}{d_A}\right)^{-\alpha-1}
 m^{-t}\frac{\delta_{\infty,B}}{d_A}
 =d_A^\alpha n^{-\alpha-1}m^{-t}\delta_{\infty,B}
 \leq d_A^\alpha n^{-\alpha-1}m^{-t}\delta
\end{equation*}
for $\alpha>2$.  Together with the first-input term
$A_{\alpha,t}(m,\widetilde n_A)\varepsilon$, these estimates prove
\eqref{eq:conditional-Dalpha-t-slab-bound}.  In particular, the
factor is $d_A^\alpha$, rather than the weaker $d_A^{\alpha+1}$ that
would result from using only trace distance in the second input.
\end{proof}

\begin{remark}[Fixed local supports]
The theorem also covers states embedded singularly in a larger
ambient space when there are fixed local projections $P_A,P_B$ such
that both $\omega_{AB}$ and $\zeta_{AB}$ are invertible on
$P_A\mathcal H_A\otimes P_B\mathcal H_B$ and vanish on its
orthogonal complement.  One applies the result in these supported
matrix algebras and computes the dimensions and positive eigenvalue
bounds there.  Varying supports are not covered.
\end{remark}

The coefficient in \eqref{eq:conditional-Dalpha-t-slab-bound}
diverges as $\alpha\to1$, because an absolute perturbation of the
power functional is divided by $|\alpha-1|$.  The order-one bound
\eqref{eq:conditional-Dt-slab-bound} must therefore be proved
separately.  For $\alpha>2$,
\eqref{eq:conditional-Dalpha-t-slab-bound} is a continuity estimate
for the algebraic R\'enyi expression; it does not assert DPI.

\subsubsection{Sharp endpoint benchmarks}
\label{subsec:conditional-sharp-endpoints}

The result proved in this paper is
\Cref{theo:conditional-geodesic-slab-continuity}.  In this subsection
we make its Umegaki, down-arrow Petz--R\'enyi, and BS endpoint
specializations explicit.  We state external bounds only when they
apply to the same conditional quantity and hence permit a direct
comparison.  The optimized Petz--R\'enyi result mentioned below
concerns a different quantity and is therefore not reproduced.

We recall $g_K$ from \eqref{eq:def-gK-continuity} and set
\begin{equation}
 r_{A:B}:=\min\{d_A,d_B\},
 \qquad D_{A:B}:=d_A\cdot r_{A:B}.
 \label{eq:effective-conditional-dimension}
\end{equation}
Here $\operatorname{SN}(\omega_{AB})$ denotes the Schmidt number: the
least integer $s$ for which $\omega_{AB}$ has a convex decomposition
into pure states of Schmidt rank at most $s$
\cite{terhal2000schmidt}.
\begin{theorem}[Sharp Umegaki conditional modulus]
\label{theo:sharp-order-one-conditional-modulus}
For arbitrary states $\rho_{AB},\sigma_{AB}$ satisfying
$T(\rho,\sigma)\leq u$,
\begin{equation}
 \left|H(A|B)_\rho-H(A|B)_\sigma\right|
 \leq
 \begin{cases}
  h_2(u)+u\log(D_{A:B}-1),
    &0\leq u\leq1-D_{A:B}^{-1},\\
  \log D_{A:B},
    &1-D_{A:B}^{-1}\leq u\leq1,
 \end{cases}
 \label{eq:sharp-order-one-effective-dimension}
\end{equation}
when $D_{A:B}>1$; if $D_{A:B}=1$, both sides vanish.  Here
$h_2(u)=-u\log u-(1-u)\log(1-u)$ is the binary entropy.  The modulus
is optimal for every distance constraint.
\end{theorem}

\Cref{theo:sharp-order-one-conditional-modulus} is
\cite[Theorem~1.1 and Proposition~2.1]
{berta2026sharpconditional}.  It strengthens the
Alicki--Fannes--Winter estimates
\cite[Eq.~(1)]{alicki2004continuity} and
\cite[Lemma~2]{winter2016tight}.  The sharp
equal-conditioning-marginal case had previously been obtained,
independently, in
\cite[Theorem~5]{berta2025integral} and
\cite[Theorem~5 and Remark~4]{audenaert2025fundamental}.

\paragraph{Reformulation in the notation of this paper.}
We next extract the forms needed for our endpoint comparison.  The
following packaging of the state-dependent and Schmidt-number
statements is carried out here and is not part of the quoted statement
of \Cref{theo:sharp-order-one-conditional-modulus}.  We define
\begin{align}
 H_{\min}(A|B)_{\sigma|\sigma}
 &:=-\log\inf\{c>0:\sigma_{AB}
       \leq c\1_A\otimes\sigma_B\},\notag\\
 \kappa_\sigma
 &:=d_A\exp[-H_{\min}(A|B)_{\sigma|\sigma}].
 \label{eq:kappa-anchor-continuity}
\end{align}
Then
\begin{equation}
 H(A|B)_\rho-H(A|B)_\sigma
 \leq\min\{g_{\kappa_\sigma}(u),
             \log d_A-H(A|B)_\sigma\},
 \label{eq:anchor-refined-order-one-bound}
\end{equation}
and
\begin{equation}
 1\leq\kappa_\sigma
 \leq d_A\cdot\operatorname{SN}(\sigma_{AB})
 \leq D_{A:B}\leq d_A^2.
 \label{eq:kappa-Schmidt-hierarchy}
\end{equation}
For every $1\leq s\leq r_{A:B}$,
\begin{equation}
 \sup_{\substack{T(\rho,\sigma)\leq u\\
       \operatorname{SN}(\rho),\operatorname{SN}(\sigma)\leq s}}
 \left|H(A|B)_\rho-H(A|B)_\sigma\right|
 =g_{d_A\cdot s}(u).
 \label{eq:exact-Schmidt-number-modulus}
\end{equation}
Here $d_A\cdot s$ denotes the ordinary product of the dimension
$d_A=\dim A$ and the Schmidt-number cutoff $s$; no tensor power of
$A$ is intended.

Since $H_0^\downarrow=H$,
\eqref{eq:sharp-order-one-effective-dimension} is exactly the sharp
$t=0$ result for the conditional family.  For
each $1\leq s\leq r_{A:B}$, we choose subspaces
$A_0\subseteq A$ and $B_0\subseteq B$ with
$\dim A_0=\dim B_0=s$, and denote the orthogonal projection
onto $B_0$ by $P_{B_0}$.  An extremizing pair is
\begin{equation}
 \sigma=\Phi_s,
 \qquad
 \rho_v=(1-v)\Phi_s+\frac{v}{d_A\cdot s-1}
 (\1_A\otimes P_{B_0}-\Phi_s),
 \quad v=\min\{u,1-(d_A\cdot s)^{-1}\},
 \label{eq:order-one-isotropic-extremizer}
\end{equation}
where $\Phi_s$ is a maximally entangled state on $A_0\otimes B_0$.
This equality family is the one used in the proof of
\cite[Proposition~2.1]{berta2026sharpconditional}.

\begin{remark}[What is proved here at the Umegaki endpoint]
\label{rem:our-Umegaki-conditional-lower-eigenvalue-bound}
Under the hypotheses and notation of
\Cref{theo:conditional-geodesic-slab-continuity}, one has
$H_0^\downarrow=H$ and
$C_0(m,n)=\log((mn)^{-1})$.  Consequently, the literal $t=0$
specialization of the result proved here is
\begin{equation}
 \left|H(A|B)_\omega-H(A|B)_\zeta\right|
 \leq
 \varepsilon\log((mn)^{-1})
 +\log\left(1+\frac{\delta}{n}\right).
 \label{eq:our-Umegaki-conditional-general-specialization}
\end{equation}

At the Umegaki endpoint the identity
$H(A|B)=H(AB)-H(B)$ gives a uniformly stronger consequence of the
same lower eigenvalue assumptions.  For $0<a\leq d^{-1}$, we set
\begin{equation}
 \chi_d(a):=\log\frac{1-(d-1)a}{a}.
 \label{eq:Umegaki-entropy-local-constant}
\end{equation}
Then, with $d_{AB}=d_Ad_B$,
\begin{equation}
 \boxed{
 \left|H(A|B)_\omega-H(A|B)_\zeta\right|
 \leq
 \varepsilon\chi_{d_{AB}}(m)
 +\delta\chi_{d_B}(n).}
 \label{eq:our-Umegaki-conditional-local-bound}
\end{equation}
Indeed, we take states $\tau_0,\tau_1$ on a $d$-dimensional space with
$\tau_0,\tau_1\geq a\1$, set
$\tau_s=(1-s)\tau_0+s\tau_1$, and write
$\Delta=\tau_1-\tau_0$.  Since $\operatorname{Tr}\tau_s=1$,
\begin{equation}
 a\1\leq\tau_s\leq[1-(d-1)a]\1.
\end{equation}
The entropy derivative used in the proof of
\Cref{theo:joint-slab-Dt-continuity} reads
\begin{equation}
 \frac{\mathrm d}{\mathrm ds}H(\tau_s)
 =-\operatorname{Tr}[\Delta\log\tau_s].
\end{equation}
Because $\operatorname{Tr}\Delta=0$,
\eqref{eq:joint-slab-traceless-expectation} therefore gives
\begin{equation}
 \left|\frac{\mathrm d}{\mathrm ds}H(\tau_s)\right|
 \leq T(\tau_0,\tau_1)\chi_d(a).
\end{equation}
Integrating this estimate first for the two states on $AB$ and then
for their $B$ marginals, and using the triangle inequality in
$H(A|B)=H(AB)-H(B)$, proves
\eqref{eq:our-Umegaki-conditional-local-bound}.

Moreover, trace-distance contractivity under partial trace
\cite[Eq.~(1.183)]{watrous2018theory} gives
$\delta\leq\varepsilon$, while
$\chi_d(a)\leq\log a^{-1}$.  Hence
\begin{equation}
 \varepsilon\chi_{d_{AB}}(m)+\delta\chi_{d_B}(n)
 \leq\varepsilon\log m^{-1}+\delta\log n^{-1}
 \leq\varepsilon\log((mn)^{-1}),
 \label{eq:Umegaki-local-improves-general-specialization}
\end{equation}
so \eqref{eq:our-Umegaki-conditional-local-bound} improves the raw
specialization \eqref{eq:our-Umegaki-conditional-general-specialization}.
Combining it with the sharp unrestricted benchmark
\eqref{eq:sharp-order-one-effective-dimension} yields the following
uniform estimate on this restricted class:
\begin{equation}
 \left|H(A|B)_\omega-H(A|B)_\zeta\right|
 \leq\min\left\{
 g_{D_{A:B}}(\varepsilon),
 \varepsilon\chi_{d_{AB}}(m)+\delta\chi_{d_B}(n)
 \right\}.
 \label{eq:combined-Umegaki-conditional-bound}
\end{equation}
The first branch is sharp without positive lower eigenvalue
assumptions.  The second branch, derived here, retains the actual
marginal perturbation and is linear for fixed $m,n$; no novelty or
sharpness is claimed for this elementary refinement.
\end{remark}

\paragraph{Optimized Petz--R\'enyi conditional entropy.}
A sharp unrestricted continuity bound for the optimized quantity
$H_{\alpha,0}^\uparrow$ in the range $\alpha\in[1/2,1)$ was obtained
in \cite[Theorem~2.1]{cheng2026sharprenyi}.  The geodesic family
studied here instead contains the nonoptimized quantity
$H_{\alpha,0}^\downarrow$.  Although
$H_{\alpha,0}^\uparrow\geq H_{\alpha,0}^\downarrow$, pointwise
domination does not compare differences of their values.  The cited
theorem therefore gives no direct benchmark for our bound.  Sharp unrestricted continuity for the down-arrow
Petz conditional entropy remains open
\cite[Section~7.5]{cheng2026sharprenyi}.  The result proved here at
the Petz endpoint is the $t=0$ specialization of
\eqref{eq:conditional-Dalpha-t-slab-bound}, under positive lower
bounds on the relevant smallest eigenvalues.

\paragraph{BS conditional entropy.}
At the other endpoint, $H_1^\downarrow$ is the BS-conditional
entropy of \cite[Eq.~(3)]{bluhm2023continuity}.  For comparison with
\cite[Corollary~6.8]
{bluhm2023continuity}, we define
\begin{align}
 f_{a,b}(p)
 &:=p\log\bigl(p+(1-p)a\bigr)
 +(1-p)\log\bigl((1-p)+pb\bigr),\notag\\
 \ell_\mu&:=1-d_{AB}\mu,\notag\\
 g_\mu^{\rm BS}(u)
 &:=\frac{\ell_\mu+u}{\ell_\mu}
 \bigl(f_{\mu^{-1},\mu^{-1}}+\mu^{-1}h_2\bigr)
 \!\left(\frac{u}{\ell_\mu+u}\right).
 \label{eq:BS-conditional-literature-modulus}
\end{align}
If $0<\mu<d_{AB}^{-1}$, $\rho_{AB},\sigma_{AB}\geq\mu\1_{AB}$, and
$T(\rho,\sigma)\leq u$, that result reads
\begin{equation}
 |H_1^\downarrow(A|B)_\rho-H_1^\downarrow(A|B)_\sigma|
 \leq\frac{2u}{\ell_\mu}\log d_A+g_\mu^{\rm BS}(u).
 \label{eq:BS-conditional-literature-bound}
\end{equation}
Our bound \eqref{eq:conditional-Dt-slab-bound} at $t=1$ reads
\begin{equation}
 |H_1^\downarrow(A|B)_\rho-H_1^\downarrow(A|B)_\sigma|
 \leq
 \left[\log(n^{-1}m^{-1})+n^{-1}m^{-3/2}\right]u
 +\log\left(1+\frac{\delta}{n}\right),
 \label{eq:our-BS-conditional-bound}
\end{equation}
when the smallest eigenvalues of the joint states and their marginals
are bounded below by $m$ and $n$, respectively, and
$\delta=T(\rho_B,\sigma_B)$.  The ALAFF bound of
\cite[Corollary~6.8]{bluhm2023continuity} retains only a positive
lower eigenvalue bound for the joint states, whereas
\eqref{eq:our-BS-conditional-bound} can use the actual marginal
perturbation and marginal lower eigenvalue bound.  Neither estimate is claimed
to dominate the other or to be sharp.  Thus
\eqref{eq:our-BS-conditional-bound} is the $t=1$ endpoint estimate
proved here, while \eqref{eq:BS-conditional-literature-bound} is the
quoted ALAFF benchmark.

\subsubsection{Why a positive lower eigenvalue bound is necessary for every
\texorpdfstring{$t>0$}{t>0}}
\label{subsec:conditional-boundary-discontinuity}

The need for a positive lower eigenvalue bound is structural, not an
artifact of the proof.
The argument starts from the two-qubit coalescing-state construction
in \cite[proof of Proposition~6.7]{bluhm2023continuity}.  The
rank-one calculation below extends it to every fixed $t>0$ and to all
R\'enyi powers.

\begin{lemma}[Rank-one block]
\label{lemma:rank-one-interpolated-block}
For $P=|e\rangle\langle e|$, $p>0$, and $\beta\geq0$ satisfying
$P\leq\operatorname{supp}\beta$, we set
$c=\langle e|\beta^{-1}|e\rangle$, where $\beta^{-1}$ is the
support inverse.  For every $t>0$,
\begin{align}
 D^t(pP\Vert\beta)&=p\log(pc),
 \label{eq:rank-one-Dt-block}\\
 Q_{\alpha,t}(pP\Vert\beta)&=p^\alpha c^{\alpha-1},
 \qquad \alpha>0.
 \label{eq:rank-one-Qalpha-block}
\end{align}
\end{lemma}

\begin{proof}
On the support of $P$,
$(pP)^{t/2}\beta^{-1}(pP)^{t/2}=p^tcP$.
Substitution in \eqref{rw:eq:Dt-explicit} proves
\eqref{eq:rank-one-Dt-block}; substitution in
\begin{equation}
 Q_{\alpha,t}(\rho\Vert\sigma)
 =\operatorname{Tr}\!\left[
 \rho^{t+\alpha(1-t)}
 (\rho^{t/2}\sigma^{-1}\rho^{t/2})^{\alpha-1}
 \right]
 \label{eq:boundary-Q-push-through}
\end{equation}
which is \eqref{rw:eq:power-functional-push-through}, proves
\eqref{eq:rank-one-Qalpha-block}.
\end{proof}

\begin{remark}[Support compression]
The singular case in the lemma is obtained simply by applying its
calculation in the corner $\operatorname{supp}\beta$.
\end{remark}

We take $A=B=\mathbb C^2$.  For $0<u<1$, we set
\begin{align}
 |e_0\rangle&:=|0\rangle,
 &|e_1^{(u)}\rangle&:=\sqrt{1-u}\,|0\rangle+\sqrt u\,|1\rangle,
 \notag\\
 P_0&:=|e_0\rangle\langle e_0|,
 &P_1^{(u)}&:=|e_1^{(u)}\rangle\langle e_1^{(u)}|,
 \label{eq:coalescing-pure-signals}\\
 \omega_u&:=\frac12\sum_{i=0}^1
 |i\rangle\langle i|_A\otimes P_i^{(u)},
 &\omega_0&:=\pi_A\otimes P_0.
 \label{eq:conditional-discontinuity-states}
\end{align}
Their $B$-marginals are
$\beta_u=(P_0+P_1^{(u)})/2$ and $\beta_0=P_0$, and
\begin{equation}
 T(\omega_u,\omega_0)=\frac{\sqrt u}{2},
 \qquad
 \langle e_i^{(u)}|\beta_u^{-1}|e_i^{(u)}\rangle=2
 \quad(i=0,1).
 \label{eq:coalescing-state-identities}
\end{equation}

\begin{theorem}[Discontinuity at rank-deficient states of
$H_t^\downarrow$]
\label{theo:Ht-down-discontinuous}
For every fixed $t\in(0,1]$,
\begin{equation}
 H_t^\downarrow(A|B)_{\omega_u}=0\quad(u>0),
 \qquad
 H_t^\downarrow(A|B)_{\omega_0}=\log2.
 \label{eq:Ht-down-discontinuity-values}
\end{equation}
Consequently, $H_t^\downarrow$ is discontinuous at the
rank-deficient state $\omega_0$, whose smallest eigenvalue is zero,
and there is no dimension-only modulus tending to zero even if the
bound is required only for invertible two-qubit states.  Equivalently,
no such modulus can hold uniformly when their smallest eigenvalues are
allowed to approach zero.
\end{theorem}

\begin{proof}
The state and reference are block diagonal in $A$.  In each block of
$\omega_u$, \Cref{lemma:rank-one-interpolated-block} has $p=1/2$
and $c=2$, so
\begin{equation}
 D^t(\omega_u\Vert\1_A\otimes\beta_u)
 =\sum_{i=0}^1\frac12\log\left(\frac12\,2\right)=0.
\end{equation}
At $u=0$ the two operators commute and the divergence is $-\log2$.

To retain invertibility, we introduce
\begin{equation}
 \omega_{u,\eta}:=(1-\eta)\omega_u+\eta\pi_{AB},
 \qquad 0<\eta<1.
 \label{eq:invertible-regularization-omega}
\end{equation}
Then
\begin{align}
 H_t^\downarrow(A|B)_{\omega_{0,\eta}}&=\log2,
 \label{eq:invertible-anchor-product}\\
 \lim_{\eta\downarrow0}
 H_t^\downarrow(A|B)_{\omega_{u,\eta}}&=0
 \quad(u>0),
 \label{eq:invertible-regularized-Ht-limit}\\
 T(\omega_{u,\eta},\omega_{0,\eta})
 &=(1-\eta)\frac{\sqrt u}{2}.
 \label{eq:invertible-regularized-distance}
\end{align}
We choose $u_k\downarrow0$ and then $\eta_k\downarrow0$ so that the
absolute error in \eqref{eq:invertible-regularized-Ht-limit} is at most
$1/k$.  The pairs of invertible states have distance tending to zero and
entropy gap tending to $\log2$.
\end{proof}

\begin{theorem}[Discontinuity at rank-deficient states of
$H_{\alpha,t}^\downarrow$]
\label{theo:Halpha-t-down-discontinuous}
For fixed $t\in(0,1]$ and $\alpha>0$, $\alpha\neq1$, we have
\begin{equation}
 H_{\alpha,t}^\downarrow(A|B)_{\omega_u}=0\quad(u>0),
 \qquad
 H_{\alpha,t}^\downarrow(A|B)_{\omega_0}=\log2.
 \label{eq:Halpha-t-down-discontinuity-values}
\end{equation}
The same regularizations by invertible states show that this quantity is
discontinuous at the rank-deficient state $\omega_0$, whose smallest
eigenvalue is zero.  In particular, it has no dimension-only
continuity modulus that remains uniform when the smallest eigenvalues
of the states are allowed to approach zero.
\end{theorem}

\begin{proof}
For $u>0$, \eqref{eq:rank-one-Qalpha-block} gives
$Q_{\alpha,t}=2(1/2)^\alpha2^{\alpha-1}=1$; at $u=0$,
commutativity gives $Q_{\alpha,t}=2^{1-\alpha}$.  These are precisely
the values in \eqref{eq:Halpha-t-down-discontinuity-values}.  For
fixed $u>0$, continuity of the power-functional expression
\eqref{rw:eq:power-functional} in the form
\begin{equation}
 Q_{\alpha,t}(\rho\Vert\sigma)
 =\operatorname{Tr}\!\left[
 \sigma^{1/2}\rho^{\alpha(1-t)}\sigma^{1/2}
 (\sigma^{-1/2}\rho^t\sigma^{-1/2})^\alpha
 \right]
 \label{eq:boundary-Q-nonnegative-powers}
\end{equation}
along \eqref{eq:invertible-regularization-omega} proves the claim for
invertible states.  At $t=1$, the zero power is interpreted by the
defining limit through invertible states.
\end{proof}

At $t=0$ the order-one obstruction disappears and
\Cref{theo:sharp-order-one-conditional-modulus} applies.  At $t=1$
the construction is exactly the BS discontinuity mechanism of
\cite[Proposition~6.7]{bluhm2023continuity}.  The argument does not
settle the different open problem of sharp unrestricted continuity
for $H_{\alpha,0}^\downarrow$.

%% file: continuity_cmi_bounds.tex
\subsection{Continuity bounds for
\texorpdfstring{$t$}{t}-conditional mutual informations}
\label{subsec:t-CMI-continuity-boundary}

We now transfer the conditional-entropy estimates to the tripartite
quantities introduced in
\Cref{subsec:consolidated-tripartite-tcmi}.  The one-sided orientation
comes first because, for the maximally mixed auxiliary state, it is
exactly a difference of two conditional quantities.  We then treat
all three orientations by using the joint two-input bounds, compare
the Umegaki and BS endpoints with the literature, and finish with the
obstruction to uniform continuity as the smallest eigenvalues tend to
zero.
The transfer by differences of conditional entropies and the triangle
inequality follows the Umegaki and BS treatments in
\cite[Sections~5.2.3 and~6.2.3]{bluhm2023continuity} and the
sandwiched R\'enyi treatment in
\cite[Section~5.3]{bluhm2026unified}.

\subsubsection{The one-sided orientation}

We fix $\tau_C=\pi_C$.  The definitions
\eqref{eq:consolidated-three-CMIs} and the scaling rule
\eqref{eq:conditional-scaling-convention} give
\begin{align}
 I_t^{\os}(A;C|B)_\rho
 &=H_t^\downarrow(C|B)_\rho
   -H_t^\downarrow(C|AB)_\rho,
 \label{eq:one-sided-t-CMI-entropy-difference}\\
 I_{\alpha,t}^{\os}(A;C|B)_\rho
 &=H_{\alpha,t}^\downarrow(C|B)_\rho
   -H_{\alpha,t}^\downarrow(C|AB)_\rho.
 \label{eq:one-sided-alpha-t-CMI-entropy-difference}
\end{align}
For a conditioned system $X$, we define
\begin{equation}
 \mathfrak M_{\alpha,t}^{X|Y}(m,n;x,y)
 :=A_{\alpha,t}\!\left(m,\frac{n}{d_X}\right)x
 +\begin{cases}
 \displaystyle\log\left(1+\frac{y}{n}\right),
 &0<\alpha<1\ \text{or}\ 1<\alpha\leq2,\\[2mm]
 \displaystyle d_X^\alpha n^{-\alpha-1}m^{-t}y,
 &\alpha>2.
 \end{cases}
 \label{eq:tCMI-Renyi-slab-modulus}
\end{equation}

\begin{proposition}[One-sided $t$-CMI continuity]
\label{prop:tCMI-spectral-slab-continuity}
Let $\mathcal H_A,\mathcal H_B,\mathcal H_C$ be finite-dimensional,
and let $\rho_{ABC},\zeta_{ABC}\in
\mathcal S_+(\mathcal H_A\otimes\mathcal H_B\otimes\mathcal H_C)$.
For both inputs, define the one-sided quantity using the same
auxiliary state $\tau_C=\pi_C:=\1_C/d_C$.  Suppose that there are
constants $m_{BC},n_B,m_{ABC},n_{AB}>0$ such that
\begin{align}
 \rho_{BC},\zeta_{BC}&\geq m_{BC}\1_{BC},
 &\rho_B,\zeta_B&\geq n_B\1_B,\notag\\
 \rho_{ABC},\zeta_{ABC}&\geq m_{ABC}\1_{ABC},
 &\rho_{AB},\zeta_{AB}&\geq n_{AB}\1_{AB}.
 \label{eq:tCMI-common-spectral-floors}
\end{align}
For every marginal $X$, we write
$\delta_X:=T(\rho_X,\zeta_X)$.
Then, for every $t\in[0,1]$,
\begin{align}
 &|I_t^{\os}(A;C|B)_\rho-I_t^{\os}(A;C|B)_\zeta|
 \notag\\
 &\quad\leq
 C_t(m_{BC},n_B)\delta_{BC}
 +\log\left(1+\frac{\delta_B}{n_B}\right)
 +C_t(m_{ABC},n_{AB})\delta_{ABC}
 +\log\left(1+\frac{\delta_{AB}}{n_{AB}}\right).
 \label{eq:tCMI-order-one-slab-bound}
\end{align}
For every $t\in[0,1]$, $\alpha>0$, $\alpha\neq1$,
\begin{align}
 &|I_{\alpha,t}^{\os}(A;C|B)_\rho
          -I_{\alpha,t}^{\os}(A;C|B)_\zeta|
 \notag\\
 &\quad\leq
 \mathfrak M_{\alpha,t}^{C|B}
       (m_{BC},n_B;\delta_{BC},\delta_B)
 +\mathfrak M_{\alpha,t}^{C|AB}
       (m_{ABC},n_{AB};\delta_{ABC},\delta_{AB}).
 \label{eq:tCMI-Renyi-slab-bound}
\end{align}
For $\alpha\in(0,1)\cup(1,2]$, this is a continuity bound for the
nonnegative DPI deficit; for $\alpha>2$, it concerns only the
algebraic difference.
\end{proposition}

\begin{proof}
We apply \Cref{theo:conditional-geodesic-slab-continuity} to $C|B$ and
$C|AB$ and use the triangle inequality in
\eqref{eq:one-sided-t-CMI-entropy-difference} and
\eqref{eq:one-sided-alpha-t-CMI-entropy-difference}.
\end{proof}

\subsubsection{All three orientations}

For nonnegative $x,y$, we introduce the single-divergence moduli
\begin{align}
 \mathfrak C_t(m,n;x,y)
 &:=C_t(m,n)x+\log\left(1+\frac{y}{n}\right),
 \label{eq:three-orientation-Ct-modulus}\\
 \mathfrak C_{\alpha,t}(m,n;x,y)
 &:=A_{\alpha,t}(m,n)x+
 \begin{cases}
  \displaystyle\log\left(1+\frac{y}{n}\right),
  &0<\alpha<1\ \text{or}\ 1<\alpha\leq2,\\[2mm]
  \displaystyle n^{-\alpha-1}m^{-t}y,&\alpha>2.
 \end{cases}
 \label{eq:three-orientation-Calpha-modulus}
\end{align}
We use the same normalized auxiliary state $\tau_C>0$ for both
tripartite inputs and set
\begin{equation}
 \delta_X:=T(\rho_X,\zeta_X),
 \qquad
 r_X:=\min\{\lambda_{\min}(\rho_X),
             \lambda_{\min}(\zeta_X)\},
 \qquad
 \vartheta:=\lambda_{\min}(\tau_C).
 \label{eq:three-orientation-local-data}
\end{equation}
For the moving product references, we define
\begin{align}
 \delta_{AB:C}
 &:=T(\rho_{AB}\otimes\rho_C,
       \zeta_{AB}\otimes\zeta_C),
 &s_{AB:C}
 &:=\min\{\lambda_{\min}(\rho_{AB})\lambda_{\min}(\rho_C),
           \lambda_{\min}(\zeta_{AB})\lambda_{\min}(\zeta_C)\},
 \notag\\
 \delta_{B:C}
 &:=T(\rho_B\otimes\rho_C,
       \zeta_B\otimes\zeta_C),
 &s_{B:C}
 &:=\min\{\lambda_{\min}(\rho_B)\lambda_{\min}(\rho_C),
           \lambda_{\min}(\zeta_B)\lambda_{\min}(\zeta_C)\}.
 \label{eq:three-orientation-product-data}
\end{align}
Trace-norm contractivity under partial trace
\cite[Eq.~(1.183)]{watrous2018theory} and telescoping give
\begin{equation}
 \delta_X\leq\delta_{ABC},
 \qquad
 \delta_{AB:C}\leq\delta_{AB}+\delta_C,
 \qquad
 \delta_{B:C}\leq\delta_B+\delta_C.
 \label{eq:three-orientation-distance-comparison}
\end{equation}

\begin{theorem}[Local continuity in all three orientations]
\label{theo:three-orientation-local-continuity}
Let $\mathcal H_A,\mathcal H_B,\mathcal H_C$ be finite-dimensional,
let $\rho_{ABC},\zeta_{ABC}\in
\mathcal S_+(\mathcal H_A\otimes\mathcal H_B\otimes\mathcal H_C)$,
and fix the same auxiliary state
$\tau_C\in\mathcal S_+(\mathcal H_C)$ in the definitions of the
one-sided and reverse quantities for both inputs.  With the local
trace-distance and smallest-eigenvalue data defined in
\eqref{eq:three-orientation-local-data}--%
\eqref{eq:three-orientation-product-data}, all of
$r_X$, $s_{AB:C}$, $s_{B:C}$, and $\vartheta$ are strictly positive.
For every $t\in[0,1]$,
\begin{align}
 |I_t^{\os}(\rho)-I_t^{\os}(\zeta)|
 &\leq
 \mathfrak C_t(r_{ABC},\vartheta r_{AB};
                  \delta_{ABC},\delta_{AB})
 +\mathfrak C_t(r_{BC},\vartheta r_B;
                  \delta_{BC},\delta_B),
 \label{eq:three-orientation-os-local}\\
 |I_t^{\ts}(\rho)-I_t^{\ts}(\zeta)|
 &\leq
 \mathfrak C_t(r_{ABC},s_{AB:C};
                  \delta_{ABC},\delta_{AB:C})
 +\mathfrak C_t(r_{BC},s_{B:C};
                  \delta_{BC},\delta_{B:C}),
 \label{eq:three-orientation-ts-local}\\
 |I_t^{\rev}(\rho)-I_t^{\rev}(\zeta)|
 &\leq
 \mathfrak C_t(\vartheta r_{AB},r_{ABC};
                  \delta_{AB},\delta_{ABC})
 +\mathfrak C_t(\vartheta r_B,r_{BC};
                  \delta_B,\delta_{BC}).
 \label{eq:three-orientation-rev-local}
\end{align}
For every $t\in[0,1]$ and
$\alpha\in(0,\infty)\setminus\{1\}$, the same statements hold for
$I_{\alpha,t}^\bullet$ after replacing $\mathfrak C_t$ by
$\mathfrak C_{\alpha,t}$.  For
$\alpha\in(0,1)\cup(1,2]$, these are continuity bounds for
nonnegative R\'enyi DPI deficits; for $\alpha>2$, they bound only the
corresponding algebraic differences.
\end{theorem}

\begin{proof}
Tensoring with $\tau_C$ preserves trace distance and multiplies the
spectral floor by $\vartheta$.  For moving products,
\begin{equation}
 \eta_1\otimes\xi_1-\eta_2\otimes\xi_2
 = (\eta_1-\eta_2)\otimes\xi_1
   +\eta_2\otimes(\xi_1-\xi_2),
\end{equation}
and multiplicativity of the trace norm on tensor products
\cite[Section~1.1]{watrous2018theory} gives
\begin{equation}
 T(\eta_1\otimes\xi_1,\eta_2\otimes\xi_2)
 \leq T(\eta_1,\eta_2)+T(\xi_1,\xi_2).
\end{equation}
We write each $t$-CMI as the difference of the two divergences in
\eqref{eq:consolidated-three-CMIs}, use the triangle inequality, and
apply \Cref{theo:joint-slab-Dt-continuity}.  Every moving reference
here is a normalized state, so its operator-norm perturbation is at
most the corresponding trace distance used in
\eqref{eq:three-orientation-local-data}--%
\eqref{eq:three-orientation-product-data}.  The R\'enyi statement
then follows from \Cref{theo:joint-slab-Dalpha-continuity}.
\end{proof}

\begin{remark}[Fixed local supports]
The theorem remains valid for ambiently singular inputs when fixed
local projections $P_A,P_B,P_C$ support both tripartite states and
$\tau_C$, and all of these operators are invertible in the resulting
tensor-product corners.  The proof is then applied after local
compression, with all dimensions and positive eigenvalue bounds
computed there.  It does not provide a uniform estimate for changing
supports.
\end{remark}

If only a common global floor
$\rho_{ABC},\zeta_{ABC}\geq\mu\1_{ABC}$ is known, then, for example,
$\operatorname{Tr}_C(\mu\1_{ABC})=d_C\mu\1_{AB}$, and partial trace
gives
\begin{equation}
 \begin{aligned}
 \rho_{AB},\zeta_{AB}&\geq d_C\mu\1_{AB},&
 \rho_{BC},\zeta_{BC}&\geq d_A\mu\1_{BC},\\
 \rho_B,\zeta_B&\geq d_Ad_C\mu\1_B,&
 \rho_C,\zeta_C&\geq d_Ad_B\mu\1_C,
 \end{aligned}
 \label{eq:three-orientation-marginal-floors}
\end{equation}
and every marginal distance is at most
$T(\rho_{ABC},\zeta_{ABC})$.  Thus
\Cref{theo:three-orientation-local-continuity} immediately yields a
dimension-and-$\mu$ bound for all three orientations.

\subsubsection{Comparison with CMI and BS-CMI bounds}
\label{subsec:tCMI-endpoint-comparison}

By \eqref{eq:consolidated-forward-CMI-tzero}, for arbitrary
tripartite states the two forward orientations at $t=0$ equal the
ordinary CMI.  For
$\varepsilon=T(\rho_{ABC},\zeta_{ABC})$, the best established direct
dimension-only estimate is
\begin{equation}
 |I(A:C|B)_\rho-I(A:C|B)_\zeta|
 \leq2\varepsilon\log\min\{d_A,d_C\}
 +2(1+\varepsilon)h_2\!\left(
        \frac{\varepsilon}{1+\varepsilon}\right).
 \label{eq:ordinary-CMI-literature-bound}
\end{equation}
This estimate appears in \cite[Corollary~1]{shirokov2017tight} and
\cite[Corollary~5.7]{bluhm2023continuity}.  The sharp CMI modulus is
still unknown to the best of our knowledge.

\Cref{theo:sharp-order-one-conditional-modulus} gives a complementary estimate
that can be better for asymmetric dimensions or smaller marginal
distances.  With $\delta_X=T(\rho_X,\zeta_X)$, we set
\begin{align}
 D_{C:B}&:=d_C\min\{d_C,d_B\},
 &D_{C:AB}&:=d_C\min\{d_C,d_Ad_B\},\notag\\
 D_{A:B}&:=d_A\min\{d_A,d_B\},
 &D_{A:BC}&:=d_A\min\{d_A,d_Bd_C\}.
 \label{eq:CMI-effective-dimensions}
\end{align}
Using
\begin{equation}
 I(A:C|B)=H(C|B)-H(C|AB)=H(A|B)-H(A|BC),
\end{equation}
the two chain-rule representations imply
\begin{equation}
 \begin{aligned}
 |I(A:C|B)_\rho-I(A:C|B)_\zeta|
 \leq\min\bigl\{&
 g_{D_{C:B}}(\delta_{BC})
 +g_{D_{C:AB}}(\delta_{ABC}),\\
 &g_{D_{A:B}}(\delta_{AB})
 +g_{D_{A:BC}}(\delta_{ABC})\bigr\}.
 \end{aligned}
 \label{eq:ordinary-CMI-two-term-bound}
\end{equation}
We can therefore take the minimum of the right-hand sides of
\eqref{eq:ordinary-CMI-literature-bound} and
\eqref{eq:ordinary-CMI-two-term-bound}.  Although each summand in
the latter is sharp as a conditional-entropy modulus, the resulting
CMI estimate need not be sharp because the triangle inequality
discards cancellations.
Neither of these two $t=0$ estimates requires invertibility.

At $t=1$ and $\tau_C=\pi_C$, the one-sided quantity is, after
relabelling the systems, the one-sided BS-CMI of
\cite[Eq.~(4)]{bluhm2023continuity}.  We assume that
$0<\mu<d_{ABC}^{-1}$ and
$\rho_{ABC},\zeta_{ABC}\geq\mu\1_{ABC}$, and we set
\begin{align}
 \ell_\mu^{ABC}&:=1-d_{ABC}\mu,\notag\\
 G_\mu^{ABC}(u)
 &:=\frac{\ell_\mu^{ABC}+u}{\ell_\mu^{ABC}}
 \bigl(f_{\mu^{-1},\mu^{-1}}+\mu^{-1}h_2\bigr)
 \!\left(\frac{u}{\ell_\mu^{ABC}+u}\right).
 \label{eq:BS-CMI-literature-modulus}
\end{align}
Then \cite[Corollary~6.11]{bluhm2023continuity} gives
\begin{equation}
 \begin{aligned}
 &|I_1^{\os}(A;C|B)_\rho-I_1^{\os}(A;C|B)_\zeta|\leq
 \frac{2\varepsilon}{\ell_\mu^{ABC}}
 \log\min\{d_C,\sqrt{d_{ABC}}\}
 +2G_\mu^{ABC}(\varepsilon).
 \end{aligned}
 \label{eq:BS-CMI-literature-bound}
\end{equation}
Section~3.2 of \cite{bluhm2023continuity} deliberately focuses on
the one-sided BS-CMI, and the quoted ALAFF theorem covers that
orientation only.  No analogous
bound is proved there for the two-sided or reverse orientation.
Thus \Cref{theo:three-orientation-local-continuity} supplies new
estimates for those two orientations under uniform positive lower
bounds on the relevant smallest eigenvalues.  The results of
\cite[Sections~5.1 and~5.3]{bluhm2026unified} concern sandwiched
R\'enyi conditionals and
CMIs, not the geodesic quantities considered here.

\subsubsection{Failure of uniform continuity as the smallest
eigenvalues tend to zero}

We take $A,B,C\cong\mathbb C^2$.  With $P_i^{(u)}$ as in
\eqref{eq:coalescing-pure-signals}, we define the copied-label lift of
the construction in
\cite[proof of Proposition~6.7]{bluhm2023continuity} by
\begin{equation}
 \Omega_{ABC}^{(u)}
 :=\frac12\sum_{i=0}^1
 |i\rangle\langle i|_A\otimes P_i^{(u)}{}_B
 \otimes|i\rangle\langle i|_C.
 \label{eq:copied-label-tripartite-state}
\end{equation}

\begin{corollary}[Forward $t$-CMI discontinuity]
\label{cor:one-sided-t-CMI-discontinuity}
For every $t\in(0,1]$,
\begin{equation}
 \begin{aligned}
 I_t^{\os}(A;C|B)_{\Omega^{(u)}}
 &=I_t^{\ts}(A;C|B)_{\Omega^{(u)}}=0 &&(u>0),\\
 I_t^{\os}(A;C|B)_{\Omega^{(0)}}
 &=I_t^{\ts}(A;C|B)_{\Omega^{(0)}}=\log2.
 \end{aligned}
 \label{eq:one-sided-t-CMI-jump}
\end{equation}
The same identities hold for
$(I_{\alpha,t}^{\os},I_{\alpha,t}^{\ts})$ for every
$\alpha>0$, $\alpha\neq1$.  In the R\'enyi DPI range these are
discontinuities of the deficits; for $\alpha>2$ they concern the
algebraic differences.  The jumps persist along pairs of invertible tripartite states whose trace distance tends to zero.
\end{corollary}

\begin{proof}
The $BC$ marginal of $\Omega^{(u)}$ is $\omega_u$ with its classical
label renamed $C$, while $A$ is a perfect copy of $C$.  Thus the
$C|AB$ conditional term vanishes, and
\eqref{eq:one-sided-t-CMI-entropy-difference}--%
\eqref{eq:one-sided-alpha-t-CMI-entropy-difference}, together with
\Cref{theo:Ht-down-discontinuous,theo:Halpha-t-down-discontinuous},
give the values above.  Moreover, $\Omega_C^{(u)}=\pi_C$, so the
one-sided and two-sided references coincide.

For the statement restricted to invertible states, we set
\begin{equation}
 \Omega_{u,\eta}:=(1-\eta)\Omega^{(u)}+\eta\pi_{ABC}.
 \label{eq:invertible-tripartite-regularization}
\end{equation}
Its $BC$ marginal is \eqref{eq:invertible-regularization-omega}.  The
two $C|AB$ terms being compared are equal: a unitary controlled by
$A$ rotates the two signal vectors while leaving the maximally mixed
regularizing term invariant, and simultaneous unitary invariance is
\eqref{rw:eq:unitary-invariance-Dft}.  Finally,
\begin{equation}
 T(\Omega_{u,\eta},\Omega_{0,\eta})
 =(1-\eta)\frac{\sqrt u}{2}.
\end{equation}
The diagonal sequences $(u_k,\eta_k)$ constructed in the proofs of
\Cref{theo:Ht-down-discontinuous,theo:Halpha-t-down-discontinuous}
therefore prove the claim.
\end{proof}

At $t=0$ this mechanism disappears and
\eqref{eq:ordinary-CMI-literature-bound} gives a continuity modulus
without assuming a positive lower bound on the smallest eigenvalues.
The construction does not yield a finite reverse quantity at its
singular endpoint, so it does not establish an analogous no-modulus
theorem for $I_t^{\rev}$.

%% file: operational_interpretations_section.tex
\section{Operational interpretations}
\label{sec:operational-interpretations}

For every ordered pair $(\rho,\sigma)$ and every $t\in[0,1]$, \cite{capel2026informationgeometry} constructs a pair-dependent finite classical
modular-record experiment whose likelihood cumulant is generated by
the geodesic modular operator
\cite[Theorem~3.1]{capel2026informationgeometry}.  It proves an exact
finite-blocklength reduction from tests in the associated commutative
algebra to tests of the resulting classical laws
\cite[Proposition~3.2]{capel2026informationgeometry}.  For these
restricted i.i.d. tests, the endpoint slope and curvature give
$D^t$ and its information variance
\cite[Proposition~3.3(i)]{capel2026informationgeometry}.  The same
likelihood surface determines the Stein and second-order exponents
\cite[Proposition~3.3(ii)]{capel2026informationgeometry}, the Chernoff
information \cite[Proposition~3.3(iii)]{capel2026informationgeometry},
and the corresponding Cram\'er rate functions
\cite[Proposition~3.3(iv)]{capel2026informationgeometry}.  At $t=1$, the
construction recovers the optimal reverse-test record of
\cite{matsumoto2010reverse}, and $\widehat D$ is the Stein exponent of
that canonical experiment
\cite[Section~3.1, Eqs.~(57)--(58), and Remark~3.4]{capel2026informationgeometry}.

These statements concern the prescribed modular record, which is part
of the operational task.  They are not statements about unrestricted
quantum discrimination between $\rho^{\otimes n}$ and
$\sigma^{\otimes n}$: the exponent in the ordinary quantum Stein lemma
is the Umegaki relative entropy $D(\rho\Vert\sigma)$
\cite{ogawa2000strong}.  The construction of the record, its exact
classical representation, and this restricted statistical theory all
belong to the companion paper.  The following logarithmic
specialization of its geodesic likelihood-surface theorem is the only
result needed for the capacity theorem below
\cite[Theorem~3.1]{capel2026informationgeometry}.

\begin{proposition}[Canonical modular-record representation]
\label{prop:op-canonical-record}
We take states $\rho,\sigma>0$ on a finite-dimensional Hilbert space
and $t\in[0,1]$, and we write the spectral resolution of the geodesic
modular operator on the Hilbert--Schmidt space as
\begin{equation}
 \Gamma_t(\rho,\sigma)
 =\sum_{z\in\mathcal Z_t}\lambda_{t,z}E_{t,z},
 \qquad \lambda_{t,z}>0,
 \label{eq:op-Gamma-spectral-resolution}
\end{equation}
where the eigenvalues are distinct.  We define
\begin{equation}
 q_t(z):=\big\langle\sigma^{1/2},
 E_{t,z}[\sigma^{1/2}]\big\rangle_{\rm HS},
 \qquad
 p_t(z):=\lambda_{t,z}q_t(z),
 \label{eq:op-pq-def}
\end{equation}
and omit outcomes for which $q_t(z)=0$.  Then $p_t$ and $q_t$ are
probability distributions, $p_t\ll q_t$, and
\begin{equation}
 \frac{p_t(z)}{q_t(z)}=\lambda_{t,z}.
 \label{eq:op-record-likelihood-ratio}
\end{equation}
Moreover, the geodesic relative entropy is exactly the classical
relative entropy of this experiment:
\begin{equation}
 D^t(\rho\Vert\sigma)
 =\sum_{z\in\mathcal Z_t}p_t(z)
   \log\frac{p_t(z)}{q_t(z)}
 =D_{\rm KL}(p_t\Vert q_t).
 \label{eq:op-D-KL}
\end{equation}
\end{proposition}

\begin{remark}[Support extension]
The algebraic statement of the proposition, and hence the capacity
theorem below, extends to
$\operatorname{supp}\rho\leq\operatorname{supp}\sigma$.  We
compress to $\operatorname{supp}\sigma$, use the regularized endpoint
definitions at $t=0,1$, allow zero eigenvalues of $\Gamma_t$, and
adopt $0\log0=0$.  Spectral calculus gives
\begin{equation}
 \sum_zq_t(z)=1,
 \qquad
 \sum_zp_t(z)
 =\left\langle\sigma^{1/2},
   \Gamma_t\sigma^{1/2}\right\rangle_{\rm HS}
 =\operatorname{Tr}\rho=1,
\end{equation}
and $p_t(z)=\lambda_{t,z}q_t(z)$ still proves
\eqref{eq:op-D-KL}.  This observation concerns only the canonical
record and its capacity calculation; it does not extend the companion
paper's finite-blocklength or reverse-test assertions.
\end{remark}

\subsection{Capacity per unit cost}
\label{subsec:op-present}

Starting from \Cref{prop:op-canonical-record}, we package these two
laws as a
binary-input memoryless channel
$\mathsf W_t$ with output alphabet $\mathcal Z_t$ by setting
\begin{equation}
 \mathsf W_t(z\mid0)=q_t(z),\qquad
 \mathsf W_t(z\mid1)=p_t(z),\qquad c(0)=0,\quad c(1)=1.
 \label{eq:op-on-off-channel}
\end{equation}
The free letter is a background use, the costly letter is a
$t$-aligned signal pulse, and the output is the prescribed modular
record.  Access to the undephased quantum preparations used to
construct that record is not part of this classical channel.

\begin{theorem}[Capacity per unit cost]
\label{thm:op-capacity}
The capacity per unit cost of $\mathsf W_t$ is
\begin{equation}
 \boxed{
 C_{\rm puc}(\mathsf W_t,c)=D^t(\rho\Vert\sigma)
 \quad\text{nats per costly pulse}.}
 \label{eq:op-capacity}
\end{equation}
Equivalently, $\nu$ costly preparations communicate
$\nu D^t(\rho\Vert\sigma)+o(\nu)$ nats reliably in the sparse-pulse
limit, while arbitrarily many zero-cost background uses are allowed.
\end{theorem}

\begin{proof}
The zero-cost-symbol theorem \cite[Theorem~3]{verdu1990capacity} gives
\begin{equation}
 C_{\rm puc}(\mathsf W_t,c)
 =\sup_{x:c(x)>0}
 \frac{D_{\rm KL}(\mathsf W_t(\cdot\mid x)
 \Vert\mathsf W_t(\cdot\mid0))}{c(x)}.
\end{equation}
There is only one nonzero-cost letter, so \eqref{eq:op-D-KL} proves the
claim.  Equivalently, if the pulse is used with probability $\theta$ and
$m_\theta=(1-\theta)q_t+\theta p_t$, then
\begin{equation}
 I_\theta(X;Z)
 =\theta D_{\rm KL}(p_t\Vert q_t)-D_{\rm KL}(m_\theta\Vert q_t).
\end{equation}
Since the alphabet is finite and $p_t\ll q_t$, the last term is
$O(\theta^2)$.  Therefore
$\lim_{\theta\downarrow0}I_\theta(X;Z)/\theta
=D^t(\rho\Vert\sigma)$.
\end{proof}

\begin{remark}[BS capacity per unit cost]
At $t=1$, the companion construction reduces to the canonical optimal
reverse-test record for the BS entropy
\cite[Section~3.1, Eqs.~(57)--(58)]{capel2026informationgeometry},
originally obtained in \cite{matsumoto2010reverse}.  Writing
$\mathsf W_{\rm BS}:=\mathsf W_1$, \Cref{thm:op-capacity} gives
\begin{equation}
 C_{\rm puc}(\mathsf W_{\rm BS},c)=\widehat D(\rho\Vert\sigma)
 \quad\text{nats per costly BS pulse}.
 \label{eq:op-BS-capacity}
\end{equation}
This is the capacity of the classical reverse-test precursor, not the
capacity obtained by unrestricted use of the original quantum states.
\end{remark}

\paragraph{Scope of this result.}
The channel in \eqref{eq:op-on-off-channel} depends on the ordered pair
$(\rho,\sigma)$ and on $t$; it is therefore a prescribed resource rather
than a pair-independent communication model.  At $t=1$ it reduces to
the BS reverse-test channel, while for commuting $\rho$ and $\sigma$ all
values of $t$ give the same classical experiment.  For $0<t<1$, no
reverse-test variational characterization is asserted.  Thus
\Cref{thm:op-capacity}, together with its BS specialization, is the
operational result presented and proved here.  The construction and
exact classical representation of the modular record belong to
\cite[Theorem~3.1]{capel2026informationgeometry}; its finite-blocklength
reduction and restricted statistical consequences are developed in
\cite[Propositions~3.2--3.3]{capel2026informationgeometry}.

%% file: discussion_outlook_section.tex
\section{Discussion and outlook}
\label{sec:discussion-outlook}

The affine-invariant geodesic between the standard and maximal
relative modular operators produces a unified family of quantum
divergences.  Its logarithmic member joins the Umegaki and
Belavkin--Staszewski relative entropies, while its power members join
the Petz and geometric R\'enyi divergences.  This construction retains
the expected algebraic properties, is monotone in the R\'enyi order,
and has a sharply parameter-dependent behavior along the geodesic.
In particular, the relative entropies increase strictly from Umegaki
to BS for every pair of noncommuting invertible states, whereas the direction for
the R\'enyi family changes with the order.  The comparison with the
$(\alpha,z)$ family, the complete matrix characterization of the
$C^1$ universal monotonicity cone, and the representing-measure
characterization of its operator-convex slice show that the
geodesic family is structurally distinct from previously studied
two-parameter families, as detailed in
\Cref{rw:sec:interpolated-divergences,rw:sec:basic-properties,rw:sec:t-monotonicity}.

The central analytic result is data processing for every finite
operator-convex generator and every $t\in[0,1]$.  This includes the
$t$-relative entropies and the $(t,\alpha)$-R\'enyi divergences for
$0<\alpha<1$ and $1<\alpha\leq2$.  Retaining the positive terms in
the proof gives resolvent remainders, an intrinsic Hilbert--Schmidt
defect, a quartic lower bound for $D^t$, and order-dependent bounds for
the R\'enyi family.  Thus the interpolation preserves not only DPI
but also a quantitative record of the information lost under a
channel, as shown in
\Cref{thm:consolidated-dpi,thm:consolidated-resolvent-dpi,thm:consolidated-Dt-quartic,thm:consolidated-Renyi-dpi}.

The equality theory displays an abrupt endpoint transition rather
than an interpolation of recovery classes.  For invertible input and
output states and an equality-determining operator-convex generator,
every nonmaximal parameter $0\leq t<1$ has exactly the Petz
sufficiency class.  At $t=1$ this changes to the generally larger BS
class.  The quadratic generator is the distinguished exception: its
divergence is independent of $t$ and has the BS equality class along
the whole path.  The equivalent likelihood-ratio, functional-calculus,
and state-recovery formulations are collected in
\Cref{sec:consolidated-recovery}.  For partial traces, the same
rigidity implies that all three $t$-QMC classes restricted to
invertible states coincide with the invertible QMC class for $t<1$,
whereas at $t=1$ they coincide with the invertible BS-QMC class,
which can be strictly larger.  The
power-weighted inner-product analysis in
\Cref{subsec:consolidated-weighted-symmetry} explains why all
nonmaximal parameters share one recovery class while the KMS/BS
endpoint is exceptional, as further discussed in
\Cref{cor:consolidated-tQMC-hierarchy,cor:consolidated-f-Renyi-tQMC-hierarchy}.

Although their zero sets agree for $t<1$, the three ordered
$t$-conditional mutual informations are quantitatively different.
They admit no universal ordering, and the available comparisons range
from an additive projective estimate for the two forward orientations
to an exact multiplicative relation at R\'enyi order two and a no-go
result for an analogous universal factor at the BS endpoint, as established in
\Cref{prop:consolidated-CMI-incomparability,prop:consolidated-forward-projective-comparison,prop:consolidated-quadratic-CMI-comparison,prop:consolidated-no-BS-projective-factor}.
The strict-interior estimates in
\Cref{sec:consolidated-quantitative-reconstruction} further show how
the common Petz zero set coexists with genuinely $t$-dependent
nonzero defects and reconstruction formulas.

The continuity theory is organized from general to conditional
quantities.  We first obtain two-input bounds for $(t,f)$-divergences
and sharper logarithmic and power estimates, then pass to conditional
entropies and to all three $t$-CMI orientations.  These bounds are
explicit when the relevant states and marginals have prescribed
positive lower bounds on their smallest eigenvalues.  At the Umegaki,
Petz, and BS endpoints, we compare the resulting estimates with the
appropriate sharp or ALAFF-based results in the literature.  Moreover,
for every fixed $t>0$, the logarithmic and R\'enyi down-arrow
conditional quantities, as well as the corresponding two forward
$t$-CMIs, can be discontinuous when one or more smallest eigenvalues
tend to zero.
The positive lower-eigenvalue assumptions are therefore essential for
this class of uniform estimates, as discussed in
\Cref{sec:continuity-interpolated-conditionals}.

The present paper contains the construction and structural theory of
the geodesic divergences, their DPI and quantitative remainders, the
classification of equality and QMCs, the comparison and reconstruction
results for conditional quantities, the continuity theory, and the
capacity-per-unit-cost theorem in \Cref{thm:op-capacity}.  The last
result imports only the canonical classical representation
$D^t(\rho\Vert\sigma)=D_{\rm KL}(p_t\Vert q_t)$ from the companion
paper
\cite[Theorem~3.1 and Proposition~3.3(i)]{capel2026informationgeometry}.
The construction and exact classical representation of the
pair-dependent experiment for operator-convex generators are proved
there in \cite[Theorem~3.1]{capel2026informationgeometry}.  That work
also contains the Hessian metrics, their universal RLD endpoint, and
their BKM--RLD interpolation
\cite[Theorem~2.3, Corollary~2.4, and Theorem~2.5]{capel2026informationgeometry},
as well as the boundary-index characterization of $t$ and comparison
with the Hasegawa--Petz metric path
\cite[Theorem~2.6 and Proposition~2.7]{capel2026informationgeometry},
the finite-blocklength reduction and restricted Stein, second-order,
Chernoff, and large-deviation results
\cite[Propositions~3.2--3.3]{capel2026informationgeometry}, and the
Busemann and horospherical interpretation
\cite[Propositions~4.1--4.2 and Section~4.3]{capel2026informationgeometry}.

\subsection{Open problems}

The paper leaves a collection of questions worth mentioning here and exploring in future works.

\begin{enumerate}[label=\textbf{\arabic*.},leftmargin=*]

\item \textbf{Structure of the universal monotonicity cone.}
Its $C^1$ part is already characterized exactly by the ordinary-matrix
test in \Cref{rw:thm:universal-cone-matrix-characterization} and its
finite operator-convex slice by the representing-measure polar condition in
\Cref{rw:thm:universal-cone-measure-characterization}.  What remains
open is a simpler structural description: for example, the
identification of its extreme rays, a reduction to a minimal matrix
dimension or test family, or a tractable scalar criterion equivalent
to the polar inequalities.  The counterexamples in
\Cref{rw:prop:universal-cone-differential-test,rw:rem:universal-cone-false-simplifications}
show that the most immediate pointwise derivative conditions do not
provide such a description.

\item \textbf{Maximal data-processing range.}
Operator convexity gives a uniform sufficient condition for DPI, but
it need not describe every generator that is monotone for a fixed
interior value of $t$.  It would be useful to characterize the maximal
fixed-$t$ class of generators satisfying DPI.  For the power family,
this includes deciding whether any R\'enyi orders above two admit DPI
for a nontrivial interior parameter or dimension, or whether a sharp
obstruction rules them out.  No DPI is asserted in that range by the
present continuity results, which concern the algebraic R\'enyi
expressions only.

\item \textbf{Explicit quantitative Petz recovery.}
For $0<t<1$, the intrinsic defect has the correct Petz zero set, and
\Cref{cor:consolidated-nonmaximal-Petz-holder} converts it into a
H\"older estimate for the distance to the Petz-recovered state when
the relevant smallest eigenvalues have fixed positive lower bounds.
The exponent and constant obtained from the \L{}ojasiewicz argument
are nonconstructive.  A major open problem is to obtain computable
dimension- and spectrum-dependent constants, or a state-independent
estimate that does not require positive lower eigenvalue bounds.  Such
constants cannot remain uniform as $t\uparrow1$, because the equality
class then enlarges to the BS class.  It is also open whether the
strict-interior remainders admit a direct formulation through a
pair-independent quantum recovery channel: the positive comparison
map used here is pair dependent and nonlinear, while the simpler
asymmetric map is generally not positive.

\item \textbf{Relations among the three conditional orientations.}
The common zero set of the three $t$-CMIs does not produce either a
universal ordering or a universal multiplicative comparison.  Natural
remaining questions are to find sharp comparison constants under
prescribed lower bounds on the smallest eigenvalues, to determine
whether reference- or spectrum-dependent constants can replace the
failed universal BS factor, and to derive a two-sided reconstruction
estimate with the correct QMC zero set.  The distinction between a shared equality
class and inequivalent nonzero deficits is likely essential in any
such result.

\item \textbf{Sharp continuity with positive lower eigenvalue bounds.}
The moduli in \Cref{sec:continuity-interpolated-conditionals} are
explicit but are generally not claimed to be optimal.  Determining
the sharp moduli for $D_f^t$, $D^t$, $D_\alpha^t$, the down-arrow
conditional entropies, and all three $t$-CMIs remains open.  Even for
the endpoint $t=0$, the sharp continuity modulus of quantum
conditional mutual information is not known.  It would also be useful
to know whether anchor parameters or Schmidt-number information can
sharpen the conditional bounds once positive lower bounds on the
smallest eigenvalues are prescribed.  The endpoint crossovers in
\Cref{fig:umegaki-continuity-bounds-comparison,fig:BS-continuity-bounds-comparison}
compare available analytic estimates but do not establish optimality.

\item \textbf{Continuity as smallest eigenvalues tend to zero.}
The discontinuity constructions in
\Cref{theo:Ht-down-discontinuous,theo:Halpha-t-down-discontinuous,cor:one-sided-t-CMI-discontinuity}
move the conditioning marginal and treat the two forward CMI
orientations.  They therefore leave open whether fixing the
conditioning marginal restores continuity for some $t>0$, and whether
the reverse $t$-CMI admits a dimension-only modulus.  Energy-constrained
and infinite-dimensional analogues would require a replacement for
the finite-dimensional lower-eigenvalue hypotheses.

\end{enumerate}